\documentclass[11pt]{amsart}

\usepackage{a4wide}
\usepackage{color}
\usepackage{esint,amssymb}
\usepackage{graphicx}
\usepackage{mathtools}
\usepackage[colorlinks=true, pdfstartview=FitV, linkcolor=blue, citecolor=blue, urlcolor=blue,pagebackref=false]{hyperref}
\usepackage{microtype}

\usepackage{bm}
\usepackage{bbm}
\usepackage{scalerel} 
\usepackage{dsfont}
\usepackage[font={footnotesize}]{caption}
\usepackage{empheq}
\usepackage[T1]{fontenc}
\usepackage{subfiles}

\newtheorem{theo}{Theorem}
\newtheorem{prop}{Proposition}[section]
\newtheorem{lem}[prop]{Lemma}
\newtheorem{coro}[prop]{Corollary}
\newtheorem{remark}[prop]{Remark}

\theoremstyle{plain}

\theoremstyle{definition}
\newtheorem{defi}[prop]{Definition}

\numberwithin{equation}{section}

\newcommand{\R}{\mathbb{R}}
\newcommand{\E}{\mathbb{E}}
\renewcommand{\P}{\mathbb{P}}

\newcommand{\I}{\mathcal{I}}

\newcommand{\g}{\mathbf{g}}
\newcommand{\f}{\mathbf{f}}

\def \be{\begin{equation}}
\def \ee{\end{equation}}

\def \t0{\rightarrow 0} 

\def \supp{\mathrm{supp }} 
\def \div{\mathrm{div} \,} 
\def \1{\mathbf{1}} 

\def \dist{\mathrm{dist}}

\def\({\left(}
\def\){\right)}

\def \PNbeta{\mathbb{P}_{N, \beta}} 

\renewcommand{\subset}{\subseteq}

\renewcommand{\subset}{\subseteq}

\renewcommand{\tilde}{\widetilde}

\renewcommand{\div}{\divg}
\newcommand{\indic}{\mathds{1}}

\renewcommand{\hat}{\widehat}

\def\Xint#1{\mathchoice
   {\XXint\displaystyle\textstyle{#1}}%
   {\XXint\textstyle\scriptstyle{#1}}%
   {\XXint\scriptstyle\scriptscriptstyle{#1}}%
   {\XXint\scriptscriptstyle\scriptscriptstyle{#1}}%
   \!\int}
\def\XXint#1#2#3{{\setbox0=\hbox{$#1{#2#3}{\int}$}
     \vcenter{\hbox{$#2#3$}}\kern-.5\wd0}}
\def\dashint{\Xint-}

\def \XN{{X}_N}

\def\N{{n_\Old}}

\def\n{n}

\def\K{\mathsf{K}}

\def\G{\mathsf{F}}
\def\H{\mathsf{H}}
\def\P{\mathsf{P}}
\def\Q{\mathsf{Q}}

\def\HN{\mathcal{H}_N}
\def\HNV{\mathcal{H}_N^{V}}
\def\HNVt{\mathcal{H}_N^{V_t}}

\def \FN{F_N}

\def \Fluct{\mathrm{Fluct}}
\def \fluct{\mathrm{fluct}}

\def \muv{\mu_V} 

\def\Esp{\mathbb{E}} 
\def \E{\Esp}

\def \ZNbeta{Z_{N,\beta}}

\def \Vt{V_{t}}

\def\I{\mathcal{I}}

\def \muv{\mu_V}
\def \muvt{\mu_{t}}

\def\g{\mathsf{g}}

\def\mn{\mathrm{n}}

\def\muv{\mu_V}

\def \zetat{\zeta_t}

\def\indic{\mathbf{1}}
\def \E{\mathcal{E}}

\def\rr{\mathsf{r}}
\def\rrc{\tilde{\mathsf{r}}}
\def\rrh{\hat{\mathsf{r}}}
\def\div{\mathrm{div}\,}

\def \Escr{E^0}

\def\Old{\mathcal{O}}
\def\New{\mathcal{N}}

\def\GG{\mathsf{h}}

\def\omc{\overset{\circ}{\Omega}}

\def\B{\mathsf{B}}

\begin{document}
\title{Local Laws and a Mesoscopic CLT for $\beta$-ensembles}
\author{Luke Peilen}
\address{Courant Institute of Mathematical Sciences, New York University, 251 Mercer St., New York, NY 10012}
\email{lkp279@cims.nyu.edu}
\date{\today}

\maketitle

\begin{abstract}
We study the statistical mechanics of the log-gas, or $\beta$-ensemble, for general potential and inverse temperature. By means of a bootstrap procedure, we prove local laws on the next order energy that are valid down to microscopic length scales. To our knowledge, this is the first time that this kind of a local quantity has been controlled for the log-gas. Simultaneously, we exhibit a control on fluctuations of linear statistics that is valid at all mesoscales using the expansion introduced in \cite{J98} and the transport approach of \cite{BLS18}. Using these local laws, we are able to exhibit for the first time a CLT at arbitrary mesoscales, improving upon previous results (\cite{BL18}) that were true only for power mesoscales.
\end{abstract}

\tableofcontents

\section{Introduction}
\subsection{General Set-Up}
We are interested in studying the statistical mechanics of log-gases, or $\beta$-ensembles, in $1$-dimension. These are collections of point masses $\XN=(x_1,x_2,\dots,x_n) \in (\R)^N$ that interact via the Hamiltonian 
\begin{equation}\label{Hamiltonian}
\HN^V(\XN):=\frac{1}{2}\sum_{i \ne j}\g( x_i-x_j)+N\sum_{i=1}^N V(x_i),
\end{equation}
where $V$ is a confinement potential with sufficiently rapid growth at infinity and $\g$ is the \textit{log kernel}
\begin{equation}\label{log kernel}
\g(x)=-\log |x|.
\end{equation}
The distribution of points is then governed by the Gibbs measure
\begin{equation}\label{Gibbs measure}
d\PNbeta:=\frac{1}{\ZNbeta}\int_{\R^N}\exp \left(-\beta \HN(\XN)\right)~d\XN,
\end{equation}
where $\ZNbeta$ is a normalizing constant
\begin{equation}\label{Partition function}
\ZNbeta:=\int_{\R^N}\exp (-\beta \HN(\XN))~d\XN
\end{equation}
referred to as the \textit{partition function}. This is a model that has attracted particular interest in recent years due to its connections with random matrix theory (see \cite{F10} for a detailed reference). At particular values of $\beta$, one recovers the eigenvalue distributions of classical random matrix ensembles: $\beta=1$ corresponds to the Gaussian Orthogonal Ensemble (GOE), $\beta=2$ to the Gaussian Unitary Ensemble (GUE) and $\beta=4$ to the Gaussian Sympectic Ensemble (GSE). These connections allow for simplified analysis due to connections with the theory of orthogonal polynomials. For general $\beta$ and quadratic $V$, these measures also describe the eigenvalue distributions of certain tridiagonal matrices, as described in \cite{DE02}. However, general $\beta$ and $V$ remain interesting from a statistical mechanical perspective. In this setting there are no specific formulae for the correlation functions and much of the study in recent years has been on developing tools for approaching these general situations. 

Our goal here is to continue this study by providing precise controls on a local energy quantity, defined below, down to a minimal length scale for $\beta$-ensembles with general inverse temperature and potential. As a consequence, we will obtain discrepancy estimates and good controls on fluctuations of linear statistics at varying scales; in particular, we will obtain a CLT for fluctuations at all mesoscales. The method we propose for this study we expect to generalize well to the study of higher dimensional Riesz gases, unlike many tools from random matrix theory which are specific to the $1$-d log gas.

\subsection{First-Order Asymptotics}Much is known about the first order asymptotic behavior of these systems (see \cite{ST97} for a more detailed discussion). From Frostman \cite{F35}, if the potential $V$ is lower semicontinuous, bounded below and satisfies the growth condition
\begin{equation}
\lim_{|x|\rightarrow +\infty}\left(V(x)+\log|x|\right)=+\infty
\end{equation}
then the continuous approximation of the Hamiltonian 
\begin{equation}\label{Continuous energy}
\I_V(\mu)=\frac{1}{2}\int \g(x-y)~d\mu(x)d\mu(y)+\int V(x)~d\mu(x)
\end{equation}
has a unique, compactly supported minimizer $\muv$ among the set of probability measures on $\R$, characterized by the Euler-Lagrange equation
\begin{equation}\label{Euler-Langrange equation}
\begin{cases}
h^{\muv}+V=c_V & \text{on }\Sigma_V \\
h^{\muv}+V \geq c_V & \text{otherwise,}
\end{cases}
\end{equation}
where $h^{\muv}=\g\ast \muv$ and $\Sigma_V$ denotes the support of $\muv$. $c_V$ is a fixed constant depending on the potential $V$. In the following, we denote the corresponding \textit{effective potential} by 
\begin{equation}\label{effective potential}
\zeta_V:=h^{\muv}+V-c_V.
\end{equation}
Furthermore, if we assume that $V$ is continuous, then the empirical measures $\mu_N$ given by 
\begin{equation}\label{Empirical measures}
\mu_N:=\frac{1}{N}\sum_{i=1}^N \delta_{x_i}
\end{equation}
converge almost surely to $\muv$ under $\PNbeta$ (see \cite{BG97} for the case of quadratic potential).

\subsection{Next-Order Asymptotics}
Given this leading term understanding of the asymptotic behavior of $\beta$-ensembles, it is a natural question to ask about the next-order behavior. Splitting the Hamiltonian, one can write (see $\S \ref{Section Basic Results}$ below)
\begin{align*}
\HNV(\XN)&=N^2\I_V(\muv)+N\sum_{i=1}^N\zeta_V(x_i)+\FN(\XN,\muv)
\end{align*}
where $\FN(\XN,\muv)$ is a next-order energy defined by 
\begin{equation}\label{next order energy}
\FN(\XN,\muv)=\frac{1}{2}\int_{\Delta^c}\g(x-y) \left(\sum_{i=1}^N \delta_{x_i}-N\muv\right)(x) \left(\sum_{i=1}^N \delta_{x_i}-N\muv\right)(y)
\end{equation}
for a configuration of points $\XN \subset \R^N$ and $\Delta \subset \R^2$ denotes the \textit{diagonal} 
\begin{equation}\label{diagonal}
\Delta:=\{(x,y) \in \R^2:x=y\}.
\end{equation}
Using the expansion of the partition function $\ZNbeta$ proven in \cite{LS15}, one can prove (as in \cite{BLS18}) that this next order energy is well controlled in the sense of exponential moments. Namely,
\begin{equation}\label{Global Energy Control}
\left|\log\Esp_{\PNbeta}\exp \left(-\beta \left(F_N(\XN,\muv)+\frac{1}{2}N\log N\right)\right)\right|\lesssim \beta N.
\end{equation}
This is innately connected to understanding a different next-order observable: \textit{fluctuations of linear statistics}. For sufficiently smooth $\xi$, one can consider the quantity
\begin{equation}\label{Fluctuations}
\Fluct_N(\xi):=\int \xi \left(\sum_{i=1}^N \delta_{x_i}-N\muv\right)=\int \xi ~d\fluct_N
\end{equation}
and ask about its asymptotic behavior. Morally, the control (\ref{Global Energy Control}) can be used to control these fluctuations using the so-called "electrostatic approach" introduced in \cite{SS15-1}.

In two dimensions, $\g$ is the Coulomb kernel and as such the fundamental solution of the Laplacian (up to a constant). As introduced by Caffarelli and Silvestre in \cite{CS07}, this provides a useful way of understanding the log kernel in $1$ dimension. Embedding $\R$ into the plane (using capital letters for coordinates in $\R^2$) we can rewrite $\FN(\XN,\muv)$ as 
\begin{equation*}
\FN(\XN,\muv)=\frac{1}{2}\int_{\R^2\setminus \Delta}\g(X-Y) \left(\sum_{i=1}^N \delta_{(x_i,0)}-N\muv\delta_\R\right)(X) \left(\sum_{i=1}^N \delta_{(x_i,0)}-N\muv\delta_\R\right)(Y),
\end{equation*}
where $\muv\delta_\R$ is a singular measure supported on $\R \times \{0\} \subset \R^2$ satisfying 
\begin{equation}\label{delta measure}
\int_{\R^2}\varphi(x,y)~d\muv\delta_\R(x,y)=\int_{\R}\varphi(x,0)~d\muv
\end{equation}
for sufficiently smooth test functions $\varphi$. We define the \textit{electrostatic potential} $u:\R^2 \rightarrow \R$ as 
\begin{equation}\label{electrostatic potential}
u(\cdot)=\g\ast \left(\sum_{i=1}^N \delta_{(x_i,0)}-N\muv\delta_\R\right)
\end{equation}
and notice that the \textit{electric field} $\nabla u$ is equivalent in norm to the Stieljes transform of the fluctuation measure $\sum_{i=1}^N \delta_{x_i}-N\muv$. Observing that $u$ solves
\begin{equation*}
-\Delta u=2\pi \left(\sum_{i=1}^N \delta_{(x_i,0)}-N\muv\delta_\R\right)
\end{equation*}
in $\R^2$, one can morally write (and formalize with renormalization)
\begin{equation*}
\FN(\XN,\muv)=\frac{1}{4\pi}\int_{\R^2 \setminus \Delta}u(-\Delta u)``="\frac{1}{4\pi}\int_{\R^2}|\nabla u|^2.
\end{equation*}
Letting $\tilde{\xi}$ denote the harmonic extension of $\xi$ away from the axis and using the control of (\ref{Global Energy Control}), one would then like to find the a priori estimate
\begin{equation*}
|\Fluct_N(\xi)|=\left|\frac{1}{2\pi}\int_{\R^2}\nabla \tilde{\xi}\cdot \nabla u\right| \lesssim\|\nabla \tilde{\xi}\|_{L^2}\sqrt{\FN(\XN,\muv)} \lesssim \|\xi\|_{H^{1/2}}\sqrt{N\log N},
\end{equation*}
where we have identified the $L^2$ norm of the gradient of the harmonic extension of a test function $\xi$ with its $H^{1/2}$ norm. This is made rigorous in \cite{BLS18}, and shows how $\FN$ controls fluctuations of linear statistics. Moreover, they show that $\Fluct_N(\xi)$ converges in distribution to a Gaussian random variable with mean $\int \sqrt{-\Delta}\xi \log \muv$ and variance $\frac{2}{\beta}\|\xi\|_{H^{1/2}}^2$, with a quantitative convergence of Laplace transforms.

Similar asymptotics for fluctuations had already been shown (see \cite{J98},\cite{BG13},\cite{BG16}, \cite{Sh13} and \cite{Sh14}), but the above result has the benefit of requiring significant less regularity on the potential $V$ and test function $\xi$; previous results had required analytic potential. Furthermore, the above approach is not dimension specific, and a similar embedding approach has been used to understand Riesz gases in higher dimensions (see \cite{LS15} and \cite{PS17}). While the higher dimensional analysis adds additional complications, we would like the reader to see the following approach to the log-gas as a blueprint for studying similar questions for Riesz gases in such regimes.

\subsection{Multiscale Analysis}
Our discussion up to this point provides a good description of the \textit{macroscopic} (order $1$) behavior of the system.  One would like to also give a more precise description of the local behavior of the system, and provide similarly precise asymptotics at various scales down to the typical distance between points, $\sim \frac{1}{N}$ (the \textit{microscopic scale}). 
%
To do this, we introduce a scaling parameter $L=L(N)$ and consider both a new local energy quantity $F^\Omega(\XN)$ which behaves like $\int_{\Omega \times [-L,L]}|\nabla u|^2$ (defined below, see $\S \ref{Section Basic Results}$) restricted to sets $|\Omega|\sim L$ in the bulk, and fluctuations of test functions of the form $\theta \left(\frac{\cdot-z}{L}\right)$.
\begin{defi}\label{scales}
Let $L:=L(N)$ denote a scaling parameter that depends on $N$. We say that $L$ is
\begin{enumerate}
\item \textbf{macroscopic} if there exists constants $a,b >0$ such that $a \leq L \leq b$;
\item \textbf{mesoscopic} if $L \rightarrow 0$ and $LN \rightarrow \infty$;
\item \textbf{microscopic} if there exists constants $a,b>0$ such that $a \leq LN \leq b$.
\end{enumerate}
\end{defi}
Can one provide similar control on $F^\Omega$ for $L$ down to some minimal length scale, what we call a \textit{local law}? Such a result was obtained for an analogous quantity for the Coulomb gas in any dimension $d \geq 2$ in \cite{AS21},  and \cite{BMP22} recently obtained a control on the Stieljes transform of the log gas using loop equations methods that yields control on the microscopic behavior of the gas.

Can one then also control the fluctuations of $\xi$, and obtain some sort of Central Limit Theorem at scales other than macroscopic? Such a result was obtained in \cite{BL18} using higher order loop equations, but they were only able to demonstrate Gaussian fluctuations for power mesoscales $L=N^{-\alpha}$, $\alpha \in (0,1)$.

These questions are answered affirmatively in the following sections.
\subsection{Statement of Assumptions}
We need to make some additional assumptions on $V$ and the form of the equilibrium measure $\muv$. We list those assumptions in detail here for ease of reference.

\begin{enumerate}
\item[A1](growth at $\infty$): We assume that the potential $V \in C^3(\R)$ grows fast enough at infinity, namely
\begin{equation}\label{Growth condition}
\lim_{|x|\rightarrow +\infty}\left(V(x)+\log|x|\right)=+\infty
\end{equation}

\item[A2](nondegeneracy of the equilibrium measure): We assume that the equilibrium measure is nondegenerate, namely that it is a finite union of nondegenerate closed intervals
\begin{equation}\label{equilibrium measure}
\Sigma_V=\bigcup_{i=1}^n [a_i,b_i].
\end{equation}
We will refer to the \textit{bulk}, $\B$, of $\Sigma_V$ as the set of points that do not approach the boundary in a controlled sense. For ease, we let $d=\frac{1}{4}\min_{i}|b_i-a_i|$, and let 
\begin{equation}\label{bulk}
\B:=\{x \in \Sigma_V:\dist(x, \partial \Sigma_V)\geq d\}.
\end{equation}

\item[A3](boundary decay):  We also assume that $\muv$ has typical square root decay at the boundary. Namely, letting 
\begin{equation}\label{Boundary decay}
\sigma(x):=\prod_{i=1}^n \sqrt{|x-a_i\|x-b_i|}
\end{equation}
for $x \in \Sigma_V$, we assume that $\muv$ takes the form
\begin{equation}\label{Density}
\muv(x)=S(x)\sigma(x)
\end{equation}
for some sufficiently smooth $S(x)>0$ on $\Sigma_V$.

\item[A4](nondegeneracy of effective potential):We also make the nondegeneracy assumption that 
\begin{equation}\label{Nondegeneracy}
\zeta_V >0
\end{equation}
on $\Sigma_V^c$.
\end{enumerate}
These assumptions taken together are often referred to as the \textit{multi-cut, noncritical} regime. As discussed in \cite{BEY12}, these assumptions are not very strong and in fact \cite{KM00} proves that a more restrictive class of assumptions still yields a set of potentials that are in a sense generic. In particular, these assumptions are satisfied at least for all real-analytic potentials, and uniformly convex smooth potentials (\cite{BLS18}). An important benefit of our approach is that we can significantly relax the assumptions on the potential $V$, which is often required to be analytic for other approaches to log-gases. We also expect similar results to hold in the critical setting with appropriate restrictions on the test function in Theorems \ref{Uniform Fluctuations} and \ref{Gaussian asymptotics}, but the approach here is a rather more technically difficult in that situation.  

\subsection{Statement of Results}\label{Statement of Results}
The first result, a control on local energies, is best stated at a ``blown up" scale, where the configuration has been blown up so that the typical distance between particles is of order $1$. $F^\Omega$ denotes a corresponding next order energy with respect to a set $\Omega$ in the bulk of the support of the equilibrium measure, defined in $\S \ref{Section Basic Results}$. It behaves like $\int_{\Omega \times [-L,L]} |\nabla u|^2$, and one should think of it as analogous to $\FN(\XN,\muv)$ but restricted to a set $|\Omega|\sim L$. $\mu$ is the corresponding blown up version of the equilibrium measure, also defined in $\S \ref{Section Basic Results}$. We will denote the blown-up configuration here as $\XN'$, and use $L'$ to denote the blown-up scale $L'=LN$.
%
%
\begin{theo}[Local Law]\label{Local Law}
Let $\Omega \subset \B'$, the blowup of the bulk $\B$, satisfy $|\Omega|\in [L',2L']$ for some $L'>\omega$ where $\omega$ is an order one minimal scale. Let $C_0:=\frac{C}{8\pi}$, where $C$ is the constant of Lemma \ref{local energy control}. Then there exists a constant $\mathcal{C}$, independent of $L'$, and a good event $\mathcal{G}_\Omega$ such that 
\begin{equation}
\G^{\Omega}(\XN', \mu)+C_0\#(\XN' \cap \Omega) \leq \mathcal{C}L'
\end{equation}
on $\mathcal{G}_\Omega$, with 
\begin{equation}
\PNbeta(\mathcal{G}_\Omega^c)\leq C_1e^{-C_2\beta L'}
\end{equation}
for some constants $C_1$ and $C_2$ dependent only on $\mu$.
\end{theo}

$\omega$ is an order one constant which we define in $\S \ref{Section Bootstrap on Scales}$ and can be optimized, but we make no attempt in this article to identify it. Observe that here we do not consider inverse temperatures $\beta$ that depend on $N$, but in principle one can adapt our arguments and estimates to examine these cases. Such regimes are of interest, and the minimal scale would at that point possibly depend on $\beta$ as in \cite{AS21} and \cite{S22}.

$\G^\Omega$ is a new quantity, and provides good control on fluctuations of linear statistics. Coupling Lemma \ref{Local Laws Energy Estimate} and Lemma \ref{local energy control}, one obtains for an arbitrary test function $\zeta$ supported inside of $\Omega$ that 
\begin{align}\label{energy controls fluctuations}
\nonumber \left|\int \zeta~d\left(\sum_{i=1}^N \delta_{x_i}-\mu \right)\right|& \leq \|\zeta'\|_{L^\infty(\Omega)}(|\Omega|+|\Omega|^{1/4}\sqrt{8\pi(\G^{\Omega}(\XN, U)+C_0\#(\XN \cap \Omega))})+ \\
&C\left(\sqrt{\mathfrak{h}}\|\zeta'\|_{L^2(\Omega)}+\frac{1}{\sqrt{\mathfrak{h}}}\|\zeta\|_{L^2(\Omega)}\right)\sqrt{8\pi(\G^{\Omega}(\XN, U)+C_0\#(\XN \cap \Omega))},
\end{align}
for arbitrary $\mathfrak{h}>0$. In addition to controlling local fluctuations, $\G^\Omega(\XN',  U)+C_0\#(\XN' \cap \Omega)$ is a positive quantity (see Lemma \ref{local energy control}) that controls local discrepancies (i.e. the actual number of points observed in $\Omega$; see Lemma \ref{discrepancy estimate}).

Our next result concerns fluctuations of \textit{rescaled test functions}, which we define as any function of the form
\begin{equation}\label{rescaled test function}
\xi_{z,L}(\cdot)=\theta \left(\frac{\cdot-z}{L}\right)
\end{equation}
where $\theta$ is a fixed compactly supported function, and $z \in \R$. It is stated in traditional coordinates to match the existing literature. We recall that the $H^{1/2}$ norm of a suitably regular test function $\theta$ can be defined by the $L^2$ norm of the gradient of its harmonic extension $\tilde{\theta}$ to the upper half plane $\R^2$, and use the following normalization:
\begin{equation}\label{H1/2 definition}
\|\theta\|_{H^{1/2}}^2:=\frac{1}{2\pi}\int_{y>0} |\nabla \tilde{\theta}|^2.
\end{equation}
Observe that it is scale and translation invariant: namely, with $\xi_{z,L}$ as in (\ref{rescaled test function}), $||\xi_{z,L}||_{H^{1/2}}=||\theta||_{H^{1/2}}$. We will often elide the $z,L$ subscript on $\xi_{z,L}$ when it is clear, and note that $L$ may depend on $N$.

\begin{theo}[Bounded Fluctuations]\label{Uniform Fluctuations}

Suppose that $\theta\in C^3$ is a compactly supported test function, and let $\xi_{z,L}$ denote the associated rescaled test function at scale $L > \frac{\omega}{N}$ with $z \in \B$.   Let $N_c$ be as in Definition \ref{critical N}. Then, for every $N\geq N_c$ there is an event $\mathcal{G}_N$ with 
\begin{equation}\label{good event}
\PNbeta( \mathcal{G}_N^c)\leq C_1e^{-C_2\beta L N}
\end{equation}
for some fixed constants $C_1$ and $C_2$ dependent only $\muv$ such that 
\begin{align}\label{uniform bound}
&\left|\log \Esp_{\PNbeta}\left[\exp \left(s\Fluct_N(\xi)\right)\indic_{\mathcal{G}_N}\right] \right|= \frac{s^2}{\beta}\|\xi\|_{H^{1/2}}^2 \\
&\nonumber+O\left(|s|\|\theta\|_{C^3}+|s|\left|1-\frac{1}{\beta}\right|\|\theta\|_{C^3}+\frac{s^2}{\beta}\max \left(\|\theta\|_{C^1},\|\theta\|_{C^3}^2\right)+\frac{|s|^3}{\beta^2 LN}\max\left(\|\theta\|_{C^2}^2,\|\theta\|_{C^3}^3\right)\right).
\end{align}
\end{theo}

Once we have this control and the Local Law at all mesoscales, we can upgrade the above control to a Central Limit Theorem for mesoscopic fluctuations.

\begin{theo}[Central Limit Theorem]\label{Gaussian asymptotics}
Suppose that $\theta\in C^{13}$ is a compactly supported test function, and let $\xi_{z,L}$ denote the associated rescaled test function at scale $L>\frac{\omega}{N}$ with $z \in \B$. Let $\mathcal{G}_N$ be the good event of Theorem \ref{Uniform Bound on Fluctuations}. Then, for every $N\geq N_c$ there is a good event $\mathcal{G}_N'\subset \mathcal{G}_N$ with 
\begin{equation}
\PNbeta(\mathcal{G}_N\setminus \mathcal{G}_N') \lesssim \frac{1}{(LN)^{1/4}}
\end{equation} 
such that
\begin{align}
\log \mathbb{E}_{\mathbb{P}_{N,\beta}}[\exp(s\Fluct_N(\xi))\indic_{\mathcal{G}_N'}]&=\frac{s^2}{\beta}\|\xi\|_{H^{1/2}}^2 + s\left(1-\frac{1}{\beta}\right)\int\psi'(x)~d\muv(x) \\ \nonumber&+O_\beta\left(\frac{|s|}{(LN)^{1/4}}\|\theta\|_{C^{13}}+\frac{s^2}{(LN)^{1/4}}\|\theta\|_{C^5}^2+\frac{|s|^3}{LN}\|\theta\|_{C^3}^3  \right),
\end{align}
where $\psi$ is a transport map associated to $\xi$, defined in (\ref{definition of transport}). As a consequence (see Section \ref{Section CLT Estimate}), when $\omega <L\ll1$ $\Fluct_N(\xi_{z,L})$ converges in distribution to a centered Gaussian random variable with variance $\frac{2}{\beta}\|\theta\|_{H^{1/2}}^2$. When $L$ is a fixed macroscopic constant, we recover \cite[Theorem 1]{BLS18} in the nonsingular case.
\end{theo}

\subsection{Relation to Other Work}
Our method adapts a bootstrap on scales that had been used to prove local laws (i.e. exponential moment controls of $\G^\Omega$) in \cite{AS21} for the Coulomb cases. Similar bootstrap approaches had been used before to study Coulomb gases by Serfaty and other authors (see \cite{NS15}, \cite{L17} and \cite{PR18}), and with the loop equation machinery of random matrices (\cite{BBNY17}). This article is the first step towards generalizing the approach of \cite{AS21} to Riesz gases in all dimensions. Unlike other studies of log-gases that use methods which rely deeply on the specific case of the $1$-d log gas via ordering of particles and rigidity estimates, we expect similar arguments to apply in the study of higher dimensional Riesz gases. Due to the added difficulties in defining the transport map in higher dimensions and the increased singularity of anisotropy type error terms, we restrict ourselves here to the $1$-d log case.

A macroscopic CLT for $\beta$-ensembles using the Laplace transform approach of \cite{J98} that we employ here was proven using a transport approach in \cite{BLS18}. Similar CLT results had also been established before using an expansion of the partition function at orders $\frac{1}{N}$ obtained via the Johansson method; see for instance \cite{BG13},\cite{BG16}, \cite{Sh13} and \cite{Sh14}. All of these results required real analytic potentials, or in the case of \cite{BFG15}, very smooth potentials ($C^{31}$).  Such methods were also used to obtain universality of eigenvalue gaps in \cite{BEY14} and \cite{BEY12} under more general assumptions on the potential $V$, and a transport approach was also used to obtain universality in \cite{BFG15}. A quantitative CLT was also obtained using Stein's method in \cite{LLW19}, with relaxed assumptions on the potential. 

Other novel methods have also recently been successful in studying related one-dimensional models. \cite{Lam21} obtained a mesoscopic CLT using the loop equations machinery for the circular $\beta$-ensembles, and \cite{HL21} used techniques from \cite{LLW19} to establish a CLT in the so-called ``high temperature" regime with inverse temperature that scales as $\beta/N$. Making use of a comparison principle for the Helffer-Sj\"ostrand equation and Stein's method, \cite{B21} provided optimal local laws and a quantitative CLT for $1$-d Riesz gases on the circle.

There has also been recent interest in CLTs for $\beta$-ensembles at various scales. \cite{BL18} used higher order loop equations to prove a mesoscopic CLT at all power mesoscales $L=N^{-\alpha}$ for $\alpha \in (0,1)$. We replace the loop equations approach of \cite{BL18} with the transport method from \cite{BLS18}, which has the benefit of requiring less regularity on the potential. Our result also has the benefit of extending their CLT to arbitrary mesoscales, and also yielding a bounded fluctuations result for microscopic fluctuations. The mesoscopic analysis introduces the additional difficulty that the transport map of \cite{BLS18} is no longer localized to the support of the test function, but instead has global tails. This requires significant additional computational analysis, which we discuss in Appendix C.

One can also ask about the microscopic behavior of $\beta$-ensembles in the thermodynamic limit, which was shown independently in \cite{VV09} and \cite{KS09} to be the $\text{Sine}_\beta$ process. CLTs for this process in the case $\beta=2$ go back to \cite{S87} and \cite{S00} for sufficiently smooth test functions. \cite{L21} recently obtained a CLT using Johansson's method (\cite{J98}) and the transport approach employed by \cite{BLS18} which we make use of here, coupled with the representation of $\text{Sine}_\beta$ via the DLR equations of \cite{DHLM21}. A CLT had been obtained previously in \cite{KVV12} using the SDE description of $\text{Sine}_\beta$, and was also obtained in \cite{Lam21} by using a coupling between the circular $\beta$-ensembles and $\text{Sine}_\beta$ introduced in \cite{VV20}.

CLTs for the Coulomb gas are also of interest, and have been obtain using similar approaches to this paper at macroscales in \cite{LS18} in dimension $2$ and at all scales in \cite{S22} in dimension $2$ and for special cases in dimension $3$. The argument of \cite{S22} uses the local laws of \cite{AS21} as input, but differs from the study here in that it does not require a CLT type estimate to run the bootstrap to obtain the local laws of \cite{AS21}. The difficulty in obtaining such results in higher dimensions is connected to the increased difficulty in obtaining precise expansions of the free energy. 

\cite{BMP22} recently demonstrated an optimal local law for fluctuations of $\beta$-ensembles, phrased in terms of the Stieljes transform of the fluctuations, using loop equations. Since their local law is a control on the $L^q$ norms of the Stieljes transform evaluated at specific points in the bulk of the equilibrium measure and ours is a consequence of exponential moment control of an $L^2$ average of the Stieljes transform, it seems that neither their local law nor ours directly implies the other. Using their local law, they were able to prove optimal rigidity for particle locations and Central Limit Theorems for singular fluctuations. In particular, they consider fluctuations of 
\begin{equation*}
L_N:=-\log |\cdot| \ast \left(\sum_{i=1}^N \delta_{x_i}-N\muv\right),
\end{equation*}
which they show behaves like a log correlated field. As a result, they are able to deduce a Central Limit Theorem for the location of the eigenvalues. It is important to point out that there is some difference between what we call a CLT (namely, one for fluctuations of smooth linear statistics) and what these and other authors sometimes refer to as a CLT (namely, one for fluctuations of the number of points). As with our local laws, there is no clear implication in either direction. Our method is unable to directly handle such singular test functions since a certain smoothness is required for our transport approach. However, we are able to consider the case of much more general potential (we only need mild smoothness conditions on $V$, where they require analyticity) and are also able to consider the multicut case. We expect our method to work in the critical case as well. We do not pursue that here, as it introduces additional difficulties in estimating the transport map.

\subsection{Outline of the Proof}\label{Outline of the Proof}
We would like to comment here on the method of proof and the outline of the paper. The approach is inspired by the bootstrap on scales used in \cite{L17} and \cite{AS21} to prove local laws for the Coulomb Gas, although several additional difficulties are present in our study due to the logarithm here being the solution kernel of a nonlocal operator. In particular, we need to couple the proofs of Theorem \ref{Local Law} with Theorem \ref{Uniform Fluctuations} in the same bootstrap procedure in order to control the electric field away from the axis in $\R^2$. More exactly, we need to show that Theorem \ref{Local Law} at scale $2L'$ implies Theorem \ref{Uniform Fluctuations} at scale $2L'$, so we can use the estimate in Theorem \ref{Uniform Fluctuations} to obtain the control we need to prove Theorem \ref{Local Law} down to scale $L'$. The approach we present has the benefit of translating well to the analysis of higher dimensional Riesz gases.

In $\S \ref{Section Basic Results}$ we introduce the main objects of the paper, including the next order subadditive and superadditive local energies and partition functions. There we describe the extension procedure (\cite{CS07}) that allows us to view the energies in $\Omega$ as connected to solutions of local problems in $\Omega \times [-L,L]$. We prove some basic a priori estimates on these quantities, and cite some known results at macroscopic scales. 

$\S \ref{Section Bootstrap on Scales}$ contains the main bootstrap argument, which assumes Theorem \ref{Local Law} at scales $2^kL'$ for $k \geq 1$ and seeks to prove the same result at scale $|\Omega|\sim L'$ by coupling a Chernoff bound with an estimate of the form 
\begin{equation*}
\Esp \exp \left(\frac{\beta}{2}\left(\G^{\Omega}(\XN', \mu)+C_0 \#(\XN' \cap \Omega)\right)\right)\indic_{\mathcal{G}} \leq \exp(\mathcal{C}L')
\end{equation*}
for some good event $\mathcal{G}$ and constants $\mathcal{C}$ and $C_0$ independent of $L'$. Using the sub and superadditivity of the energies established in $\S \ref{Section Basic Results}$, one can compute on a good event $\mathcal{G}_n$ with $n$ points in $\Omega$
\begin{align*}
&\Esp \left(\exp \left(\frac{\beta}{2}\G^{\Omega}(\XN', \mu)\right)\indic_{\mathcal{G}_{n}}\right) \\
&\leq \frac{1}{N^N\K^{\mathrm{int}}(\R)}{N \choose n}\int_{\Omega^n \cap \mathcal{G}_n}\exp \left(-\frac{\beta}{2}\H^{\mathrm{int}}_{\epsilon,\frac{L}{2}}(X_n,\Omega)+\frac{\beta}{2}C_0n\right)~d\rho^{\otimes n}(X_n) \\
&\times \int_{(\Omega^c)^{N-n}\cap \mathcal{G}_n}\exp \left(-\beta \H^{\mathrm{ext}}_{\epsilon, \frac{3L}{2}}(X_{N-n},\Omega^c)\right)~d\rho^{\otimes N-n}(X_{N-n})
\end{align*}
where $\K^{\mathrm{int}}$ denotes the next-order partition function, $\rho$ is a measure that incorporates the confinement effects of the effective potential (defined in $\S \ref{Section Basic Results}$) and $\H_{\epsilon, h}$ are interior and exterior ``screenable" minimal energies defined in $\S \ref{Section Basic Results}$ with adequate decay from the axis phrased in terms of parameters $\epsilon$ and $h$ that will be specified in $\S \ref{Section Basic Results}$ and $\S \ref{Section Bootstrap on Scales}$. Using a screening procedure modified from \cite{PS17} and optimized as in \cite{AS21}, we can modify these generic good configurations and replace their associated electric fields with ``screened" fields having nice Neumann boundary conditions, effectively splitting the field into two independent components. Replacing the field external to $\Omega$ requires a new approach to the ``inner" screening introduced in \cite{AS21}, which we describe in Appendix B. Then, we can use subadditivity of the local next order partition functions to write
\begin{align*}
\Esp \left(\exp \left(\frac{\beta}{2}\G^{\Omega}(\XN', \mu)\right)\indic_{\mathcal{G}_{n}}\right)\lesssim \frac{\K^{\mathrm{int}}\left(\Omega, \frac{\beta}{2}\right)}{\K^{\mathrm{int}}(\Omega, \beta)}\exp \left(\frac{\beta}{2}C_0n+ \text{ screening errors }\right).
\end{align*}
Many of the screening errors can be easily controlled using the control on next-order energies at scale $2^kL'$, although we run into difficulties controlling the energy away from the axis
\begin{equation*}
L'\int_{\Omega \times \{\pm h\}}|\nabla u|^2,
\end{equation*}
which a priori is of order $L'$ (the same order that we are looking to obtain energy bounds at). To handle this term, we observe that the components of $\nabla u$ are fluctuations of smooth test functions of the forms 
\begin{equation}\label{definition of rho}
\kappa_{a,h}(x)=2\pi \frac{-(a-x)}{(a-x)^2+h^2}
\end{equation}
and 
\begin{equation}\label{definition of zeta}
\zeta_{a,h}(x)=-2\pi h\frac{1}{(a-x)^2+h^2},
\end{equation}
where $a$ and $h$ are scale dependent quantities that are defined in $\S \ref{Section Bootstrap on Scales}$, and can thus be well controlled as soon as we have good estimates for these fluctuations. The control from (\ref{energy controls fluctuations}) (which controls fluctuations at an order which is a square root of the energy) yields insufficiently strong bounds, which forces us to seek a stronger bound like the one in Theorem \ref{Uniform Fluctuations}. We obtain the requisite estimate assuming that control, and postpone its proof to $\S \ref{Section Uniform Bound on Fluctuations}$.

$\S \ref{Section Uniform Bound on Fluctuations}$ proves the boundedness of fluctuations that we require for the bootstrap on scales, in traditional coordinates to match the existing literature. Using the Laplace transform method of \cite{J98} and the transport approach of \cite{BLS18} and \cite{L21}, we can write for sufficiently smooth test functions $\xi$
\begin{align*}
&\nonumber \Esp_{\PNbeta}[\exp(s\Fluct_N(\xi))\indic_{\mathcal{G}}]\\
&=\exp \left(\frac{s^2}{\beta}\|\xi\|_{H^{1/2}}^2+N\textsf{Error}_1\right)\Esp_{\PNbeta}\left[\exp \left(-\frac{\beta}{2}\left(\textsf{Error}_2+\textsf{Error}_3\right)\right)\indic_{\Phi_t^{-1}(\mathcal{G})}\right],
\end{align*}
where $\Phi_t$ is the map on coordinates $x \mapsto x+t\psi(x)$. $\psi$ is chosen in such a way that the push-forward $\tilde{\muvt}=\Phi_t\#\muv$ more or less approximates the equilibrium measure $\muvt$ for the potential $V_t=V+t\xi$. This is accomplished by searching for a transport $\psi$ that satisfies the Euler-Lagrange equation
\begin{equation*}
(h^{\tilde{\muvt}}+V_t)\circ \Phi_t=\text{constant}
\end{equation*}
on $\Sigma_V$ to leading order in $t$. Since $\Phi_t\#\muv \sim \muv-t(\psi \muv)'$, we seek $\psi$ satisfying 
\begin{equation*}
h^{(\psi \muv)'}=\xi+\text{constant}
\end{equation*}
in $\Sigma_V$, which leads us (after an integration by parts) to the definition (\ref{transport}) that is common in the literature (see \cite{BLS18}). 

The terms $\textsf{Error}_i$ are fluctuations involving $\psi$. $\textsf{Error}_1$ and $\textsf{Error}_3$ are easily dealt with using fine estimates on $\psi$ obtained in Lemmas \ref{Rescaled Transport Estimates} and \ref{Local Laws Transport Estimates} coupled with (\ref{energy controls fluctuations}) and Theorem \ref{Local Law}, up to scales $2^kL$. There is some difficulty in obtaining the necessary estimates on $\psi$, since it depends nonlocally on $\xi$; these estimates are postponed to Appendix C. $\textsf{Error}_2$ is more subtle, and involves an anisotropy term
\begin{equation*}
\int \frac{\psi(x)-\psi(y)}{x-y}d\fluct_N (x)d\fluct_N (y)
\end{equation*}
which is just barely too singular for (\ref{energy controls fluctuations}) to provide sufficiently good bounds. This however can be localized and controlled using the functional inequality \cite[Proposition 1.1]{S20} which is rephrased as a commutator estimate in \cite[Proposition 3.1]{NRS21} and localized to the support of the transport in \cite{RS22}. Our estimate is specific to one dimension, and similar more singular anisotropy terms will appear in the analogous study of higher dimensional Riesz gases. We don't expect this to be a large obstacle in that study since commutator estimates are also available in that regime, but it will be necessarily more difficult.

$\S \ref{Section CLT Estimate}$ upgrades the boundedness of $\S \ref{Section Uniform Bound on Fluctuations}$ to prove a CLT for all mesoscopic fluctuations. The proof relies on an upgraded analysis of the error terms using the local law of Theorem \ref{Local Law} at all scales. $\textsf{Error}_1$ and $\textsf{Error}_3$ are easily upgraded, but $\textsf{Error}_2$ requires a finer analysis which involves upgrading the control on anisotropy. We expand upon a Fourier transform trick introduced in \cite{BLS18}, writing
\begin{align*}
\frac{\psi(x)-\psi(y)}{x-y}=\int_0^1 \psi'(sx+(1-s)y)~ds&
=c\int _0^1 \int_{\R}\mathcal{F}(\psi')(\lambda)e^{i\lambda sx}e^{i\lambda(1-s)y}~d\lambda ds
\end{align*}
where $\mathcal{F}$ denotes the Fourier transform. This allows us to estimate anisotropy using fine estimates on $\psi$ from Lemma \ref{Rescaled Transport Estimates} to control $\mathcal{F}(\psi')$, and by using Theorem \ref{Uniform Fluctuations} to provide improved control on the fluctuations of $e^{i\lambda sx}$ and $e^{i\lambda(1-s)y}$ as compared to that given by (\ref{energy controls fluctuations}) and Theorem \ref{Local Law}.

The rest of the paper is devoted to various technical estimates, which we push to the appendices. Appendix A discusses some auxiliary computations on the energies, and proves controls on fluctuations and discrepancies. Appendix B contains the proof of the screening statement Proposition \ref{screening result}, which has been adapted and optimized from \cite{PS17}, and a discussion of our new approach to ``inner" screening. Appendix C contain the technical estimates on the transport $\psi$ and associated maps of importance to the fluctuation estimates, and Appendix D proves a technical moment bound that may exist but we were unable to locate in the literature.
%

\subsection{Acknowledgements} The author would like to thank their Ph.D. advisor Sylvia Serfaty for suggesting the problem, and for helpful discussions and guidance throughout the process. They would also like to thank Thomas Leblé for useful conversations in the development of the argument, and for close readings that improved the paper. This material is based upon work supported by the National Science Foundation under Grant No. DGE1839302.


\section{Basic Results}\label{Section Basic Results}
Our interest is in next order asymptotics at all scales. To proceed then, we split off the first order deterministic asymptotics and ``blow up" so that microscopic behavior is visible at order $1$. This splitting is accomplished by the following lemma introduced in \cite{SS15-2} and for the $\beta$-ensembles in \cite{SS15-1}.
\begin{lem}[Splitting Formula]\label{Splitting Formula}
Given any configuration $\XN \in \R^N$,
\begin{align*}
\HNV(\XN)&=N^2\I_V(\muv)+N\sum_{i=1}^N\zeta_V(x_i)+\FN(\XN,\muv)
\end{align*}
where $\FN(\XN,\muv)$ is defined by 
\begin{equation*}
\FN(\XN,\muv)=\frac{1}{2}\int_{\Delta^c}\g(x-y) \left(\sum_{i=1}^N \delta_{x_i}-N\muv\right)(x) \left(\sum_{i=1}^N \delta_{x_i}-N\muv\right)(y)
\end{equation*}
for a configuration of points $\XN \subset \R^N$.
\end{lem}
\begin{proof}
The proof is standard, see for instance \cite{BLS18} for this specific formulation. One writes 
\begin{equation*}
\HNV(\XN)=\frac{1}{2}\int \g(x-y)\sum_{i \ne j}\delta_i(x)\delta_j(y)+N \int V\sum_{i=1}^N \delta_{x_i}(x)
\end{equation*}
and expands $\sum \delta_{x_i}=\left(d\fluct_N\right)+N\muv$.
\end{proof}
We would now like to change coordinates so that the typical distance between particles is of order $1$. Setting $\XN'=N\XN$ and $\muv'(x)=\muv \left(\frac{x}{N}\right)$, we compute directly 
\begin{equation*}
\FN(\XN,\muv)=\frac{1}{2}\int \g(x-y)\left(\sum_{i=1}^N \delta_{x_i'}-\muv'\right)(x)\left(\sum_{i=1}^N \delta_{x_i'}-\muv'\right)(y)-\frac{1}{2}N\log N.
\end{equation*}
Motivated by this computation, for any measure $\mu$ of total mass $N$ we define 
\begin{equation}\label{definition of next order energy}
\G(\XN,\mu)=\frac{1}{2}\int \g(x-y)\left(\sum_{i=1}^N \delta_{x_i}-\mu\right)(x)\left(\sum_{i=1}^N \delta_{x_i}-\mu\right)(y)
\end{equation}
and observe that the above computation yields the identification 
\begin{equation*}
\FN(\XN,\muv)+\frac{1}{2}N\log N=\G(\XN', \muv').
\end{equation*}
A similar computation yields the relation 
\begin{equation*}
N\zeta_V(x_i)=\zeta'_V(x_i'),
\end{equation*}
where 
\begin{equation*}
\zeta'_V(x)=\int \g(x-z)~d\muv'(z)+NV \left(\frac{x}{N}\right)-Nc_V+N\log N.
\end{equation*}
Observe that with this identification and the splitting formula (Lemma \ref{Splitting Formula}), we have
\begin{align*}
\ZNbeta
&=\exp \left(-\beta N^2 \I_V(\muv)+\frac{\beta}{2}N\log N\right)\int_{\R^N}\exp \left(-\beta \G(\XN',\muv')\right)\prod_{i=1}^N\frac{e^{-\beta \zeta_V'(x_i')}}{N}dx_i' \\
&=\exp \left(-\beta N^2 \I_V(\muv)+\frac{\beta}{2}N\log N\right)\K(\muv',\zeta_V')
\end{align*}
where in the last line we have defined the next-order partition function 
\begin{equation}\label{definition of next order partition function}
\K(\mu, \zeta)=\frac{1}{N^N}\int_{\R^N}\exp \left(-\beta \G(\XN,\mu)\right)~d\rho^{\otimes N}
\end{equation}
where $\rho$ is the blow up of the confinement density $e^{-\beta \zeta'(x')}~dx'$ under $\XN'=N\XN$ for a suitably integrable $\zeta \geq 0$. Observe that this measure is identically one off of the support of an effective potential $\zeta$, and decays exponentially fast elsewhere. We then also define the next-order Gibbs measure
\begin{equation*}
d\Q(\mu, \zeta)=\frac{1}{\K(\mu, \zeta)}\exp  \left(-\beta \G(\XN,\mu)\right)~d\rho^{\otimes N}.
\end{equation*}
These are well-defined as soon as $\zeta$ grows fast enough at infinity; in particular, they are well-defined for $\zeta_V'$. From our above expansion of the main partition function, we have 
\begin{equation*}
d\PNbeta=d\Q(\mu,\zeta)
\end{equation*}
for $\mu=\muv'$, $\zeta=\zeta_V'$.

\subsection{The Electrostatic Approach}
A useful tool in the analysis of Coulomb and Riesz gases is the so-called ``electrostatic" approach, introduced in \cite{SS15-2} and \cite{RS16} and used extensively in \cite{AS21}. This method relies on recognizing the log kernel as the fundamental solution of the Laplacian in two dimensions, and makes use of a potential theoretic approach to recognize the above quantity $\G(\XN, \mu)$ as a renormalized electrostatic energy. As detailed in \cite{PS17}, a similar approach can be made to work using the Caffarelli-Silvestre extension procedure \cite{CS07}. For a general measure $\mu$ we define 
\begin{equation}\label{mu potential}
h^\mu(x)=\g \ast \mu
\end{equation}
and observe that we can extend $h^\mu$ to $\R^2$ by 
\begin{equation}\label{mu potential extension}
h^\mu(X)=\int \g( X-(y,0))~d\mu(y)=\g \ast \mu\delta_\R
\end{equation}
with $X=(x_1,x_2)$. It follows that
\begin{equation}\label{laplacian mu potential}
-\Delta h^\mu=2\pi \left(\sum_{i=1}^N \delta_{(x_i,0)}-\mu\delta_\R\right)
\end{equation}
with $\mu\delta_{\R}$ as in (\ref{delta measure}). As observed in \cite{CS07}, in $\R$ this allows us to recover 
\begin{equation}\label{half laplacian mu potential}
\sqrt{-\Delta} h^\mu=\pi \left(\sum_{i=1}^N \delta_{x_i}-\mu\right).
\end{equation}
Ideally, we would now be able to integrate by parts in the plane to identify 
\begin{equation*}
\G(\XN,\mu)=\frac{1}{2}\int\g(x-y)|\left(\sum_{i=1}^N \delta_{(x_i,0)}-\mu\delta_{\R}\right)(x)\left(\sum_{i=1}^N \delta_{(x_i,0)}-\mu\delta_{\R}\right)(y)
\end{equation*}
with $\int |\nabla h^\mu|^2$, but this is impossible due to the blowup of $h^\mu$ at the points $(x_i,0)$. As such we use a renormalization procedure introduced for the log-gas in \cite{SS15-1} with truncation parameter equal to the minimal distance, as in \cite{AS21}.

To do this, for any $\eta>0$ define the truncation
\begin{equation}\label{planar truncation}
\overline{\f}_\eta(X)=(\g(X)-\g(\eta) )_+
\end{equation}
for $X=(x,y)\in \R^2$. Observe that 
\begin{equation}\label{laplacian truncation}
-\Delta \overline{\f}_{\eta}=2\pi \left(\delta_0-\delta_0^{(\eta)}\right),
\end{equation}
where $\delta_0^{\eta}$ is the uniform measure on a sphere of radius $\eta$, which we often refer to as a smeared Dirac. For any truncation vector $\vec{\eta}=(\eta_1,\dots,\eta_N)$ then, we define 
\begin{equation}\label{truncated potential}
h^\mu_{\vec{\eta}}(X)=h^\mu(X)-\sum_{i=1}^N \overline{\f}_{\eta_i}(X-X_i)
\end{equation}
and observe that 
\begin{equation}\label{laplacian truncated potential}
-\Delta h^\mu_{\vec{\eta}}=2\pi \left(\sum_{i=1}^N \delta_{(x_i,0)}^{(\eta_i)}-\mu\delta_\R\right).
\end{equation}
For our computations, we will often make use of the identification
\begin{equation}\label{planar truncation is 1d truncation}
\int_{U \times [-h,h]}\overline{\f}_\eta(x,y)~d\mu\delta_{\R}=\int_U \overline{\f}_\eta(x,0)~d\mu
\end{equation}
for $h>\eta$, where we have defined 
\begin{equation}\label{1d truncation}
\f_{\eta}(x)=(\g(x)-\g(\eta))_+.
\end{equation}
The truncation parameter we will choose is the following minimal distance, as in \cite{AS21}:
\begin{equation}\label{minimal distance}
\rr_i:=\frac{1}{4}\min \left(\min_{j \ne i} |x_i-x_j|, 1\right).
\end{equation}
With this truncation defined, we can now integrate by parts.
\begin{lem}\label{electrical view of next order energy}
Let $\XN$ be a configuration of points in $\R$, and let $\vec{\eta}$ denote a truncation vector with $\eta_i \leq \rr_i$ for all $1 \leq i \leq N$. Then,
\begin{equation*}
\G(\XN,\mu)=\frac{1}{4\pi}\left(\int_{\R^2}|\nabla h_{\vec{\eta}}^\mu|^2-2\pi \sum_{i=1}^N\g(\eta_i)\right)-\sum_{i=1}^N \int_{\R}\f_\eta(x-x_i)~d\mu(x).
\end{equation*}
\end{lem}
\begin{proof}
The proof is a standard integration by parts using the decay of $h_{\vec{\eta}}^\mu$ at infinity, cf. \cite[Proposition 1.6]{PS17}.
\end{proof}
This electrostatic understanding provides a natural way of approaching local configuration energy, which we define in the following subsection.

\subsection{Local Energies and Minimal Distance}
In what follows, we consider $\Omega \subset \R$ with $\Omega=\R$ or a closed and bounded interval, or the complement of such a set. Many of the quantities here are analogous to those defined for the Coulomb gas in \cite{AS21}, and have natural analogs for Riesz gases in higher dimensions. For the rest of this section and $\S 3$, we remove the $'$ from the blown-up scales $L'$ when it is clear and simply denote them as $L$ for ease of notation, since it is evident that we are working at blown up scale.

We will need a minimal distance that respects distance from the boundary of $\Omega$. Let
\begin{equation}\label{local minimal distance}
\rrc_i:=\frac{1}{4}\begin{cases}\min \left(\min_{x_j \in \Omega, j \ne i} |x_i-x_j|,1\right) & \text{if }\dist(x_i,\partial \Omega)\geq \frac{1}{2} \\
1 & \text{otherwise.}
\end{cases}
\end{equation}
Notice that if $x_i \in \Omega ' \subset \Omega$, then
\begin{equation}\label{decreasing minimal distance}
\rrc_i(\Omega')\geq \rrc_i(\Omega).
\end{equation}
Suppose $\mu$ is a measure on $\R$ with $\mu(\R)=N \in \mathbb{N}$ and $X_N$ is a configuration in $\R$. Let $|\Omega|\in[L,2L]$, and denote $\tilde{\Omega}$ by 
\begin{equation}\label{extension set}
\tilde{\Omega}:=\Omega \times [-L,L].
\end{equation} 
We define the \textit{true local energy} $\G^\Omega$ by
\begin{equation}\label{true local energy}
\G^\Omega(\XN,\mu)=\frac{1}{4\pi}\left(\int_{\tilde{\Omega}}|\nabla u_{\rrc}|^2-2\pi \sum_{x_i \in \Omega} \g(\rrc_i)\right)-\sum_{x_i \in \Omega} \int_{\R}\f_{\rrc_i}(x-x_i)~d\mu(x),
\end{equation}
with $u$ the electrostatic potential (\ref{electrostatic potential}). At times we will need to consider local energies of fields defined by other potentials as well. If $w$ is any electrostatic potential solving 
\begin{equation}\label{viable potential}
-\Delta w=2\pi \left(\sum_{i=1}^N \delta_{(x_i,0)}-\mu\delta_{\R}\right),
\end{equation}
in $\tilde{\Omega}$ with $\tilde{\Omega}$ as in (\ref{extension set}), then we will correspondingly define the \textit{interior local energy} $\E^\mathrm{int}(w,\Omega)$ of that potential by 
\begin{equation}\label{general potential local energy}
\E^\mathrm{int}(w, \Omega)=\frac{1}{4\pi}\left(\int_{\tilde{\Omega}}|\nabla w_{\rrc}|^2-2\pi \sum_{x_i \in \Omega} \g(\rrc_i)\right)-\sum_{x_i \in \Omega} \int_{\R}\f_{\rrc_i}(x-x_i)~d\mu(x)
\end{equation}
and relative to $\Omega' \subset \Omega$ the \textit{exterior local energy} $\E^\mathrm{ext}(w,\Omega^c)$ of that potential by 
\begin{equation}\label{general potential local energy external}
\E^\mathrm{ext}(w, \Omega^c)=\frac{1}{4\pi}\left(\int_{\tilde{\Omega}\setminus \tilde{\Omega'})}|\nabla w_{\rrc}|^2-2\pi \sum_{x_i \in \Omega \setminus \Omega'} \g(\rrc_i)\right)-\sum_{x_i \in \Omega \setminus \Omega'} \int_{\R}\f_{\rrc_i}(x-x_i)~d\mu(x),
\end{equation}
with $\tilde{\Omega},$ $\tilde{\Omega'}$ as in (\ref{extension set}).

In our analysis, we will need sub and superadditive approximations of our local energy that are purely local quantities, unlike the above energies which depend on the global configuration. We proceed to define those now.

\subsubsection{Superadditive Approximation}
Our superadditive quantities are analogous to Dirichlet energies. We cannot literally take a Dirichlet energy as our definition, since we will want to ``screen" the electric field corresponding to our superadditive quantity. In a rough sense, screening modifies an electric field to create a Neumann boundary condition that allows us to effectively isolate a region of the field. This process requires an appropriate quantitative decay of the field away from the axis (\ref{screenability}), and the electric field corresponding to a Dirichlet potential may or may not possess adequate decay. To get around this issue, we instead minimize over all screenable configurations that exhibit appropriate decay.

Suppose that $X_{\mn}$ is a configuration in $\Omega$, and $\Omega\subset \R$ is as described at the beginning of this subsection. Let $\tilde{\Omega}$ be as in (\ref{extension set}). Define
\begin{equation}\label{Dirichlet energy}
\H_{a,h}^{\mathrm{int}}(X_{\mn}, \mu,\Omega)=\frac{1}{4\pi}\left(\int_{\tilde{\Omega}}|\nabla w_{\rrc}|^2-2\pi\sum_{i=1}^{\mn} \g(\rrc_i)\right)-\sum_{i=1}^{\mn} \int_{\R}\f_{\rrc_i}(x-x_i)~d\mu(x)
\end{equation}
where $w$ is a potential minimizing
\begin{equation*}
\frac{1}{4\pi}\left(\int_{\tilde{\Omega}}|\nabla w_{\rrc}|^2-2\pi\sum_{i=1}^{\mn} \g(\rrc_i)\right)-\sum_{i=1}^{\mn} \int_{\R}\f_{\rrc_i}(x-x_i)~d\mu(x)
\end{equation*}
such that $w$ solves (\ref{viable potential}) in $\tilde{\Omega}$ while satisfying the screenability condition (\ref{screenability}) with additional decay $\int_{2\Omega \times \{\pm h\}}|\nabla w|^2 \leq a$. If the Dirichlet potential for $\Omega$ is screenable and has the appropriate decay away from the axis, then we can write
\begin{equation*}
\H_{a,h}^{\mathrm{int}}(X_{\mn}, \mu, \Omega):=\frac{1}{4\pi}\left(\int_{\tilde{\Omega}}|\nabla w_{\rrc}|^2-2\pi\sum_{i=1}^{\mn} \g(\rrc_i)\right)-\sum_{i=1}^{\mn} \int_{\R}\f_{\rrc_i}(x-x_i)~d\mu(x),
\end{equation*}
where $w$ solves the Dirichlet problem
\begin{equation*}
\begin{cases}
-\Delta w=2\pi \left(\sum_{i=1}^{\mn} \delta_{x_i}-\mu \delta_{\R}\right)& \text{in }\tilde{\Omega}, \\
w_{\widetilde{r}}=0 & \text{on }\partial \tilde{\Omega}.
\end{cases}
\end{equation*}
Analogously, for $\Omega'\subset \Omega$ with $\tilde{\Omega}$, $\tilde{\Omega'}$ as in (\ref{extension set}) we define an external energy by 
\begin{equation}\label{Dirichlet energy - outside}
\H_{a,h}^{\mathrm{ext}}(X_{\mn}, \mu,\Omega' \setminus \Omega)=\frac{1}{4\pi}\left(\int_{\tilde{\Omega}\setminus \tilde{\Omega'}}|\nabla w_{\rrc}|^2-2\pi\sum_{i=1}^{\mn} \g(\rrc_i)\right)-\sum_{i=1}^{\mn} \int_{\R}\f_{\rrc_i}(x-x_i)~d\mu(x) 
\end{equation}
where $w$ is a potential minimizing 
\begin{equation*}
\frac{1}{4\pi}\left(\int_{\tilde{\Omega}\setminus \tilde{\Omega'}}|\nabla w_{\rrc}|^2-2\pi\sum_{i=1}^{\mn} \g(\rrc_i)\right)-\sum_{i=1}^{\mn} \int_{\R}\f_{\rrc_i}(x-x_i)~d\mu(x)
\end{equation*}
such that $w$ solves (\ref{viable potential}) in $\tilde{\Omega}\setminus \tilde{\Omega'}$ while satisfying the screenability condition (\ref{screenability}) with additional decay $\int_{2\Omega \times \{\pm h\}}|\nabla w|^2 \leq a$.

\subsubsection{Subadditive Approximation}
Our subadditive quantity is analogous to a Neumann energy. For any $\Omega \subset \R$ we first define an increasing minimal distance
\begin{equation}\label{subadditive local minimal distance}
\rrh_i:=\frac{1}{4}\min \left(\min_{x_j \in \Omega, j \ne i} |x_i-x_j|, \dist(x_i, \partial \Omega),1\right)
\end{equation}
and observe that if $x_i \in \Omega' \subset \Omega$,
\begin{equation}\label{increasing minimal distance}
\rrh_i(\Omega')\leq \rrh_i(\Omega).
\end{equation}
At times (see Lemma \ref{global bound on partition functions}), we will also need to add additional energy to point near the boundary in order to maintain a coercive energy. With that in mind, we define
\begin{equation*}
\GG(x_i):=\left(\g\left(\frac{1}{4}\dist(x_i, \partial \Omega)\right)-\g\left(\frac{1}{4}\right)\right)_+.
\end{equation*} 
Now, let $X_{\mn}\subset \Omega$ with $\Omega=\mn$. We define an internal subadditive energy by 
\begin{equation}\label{Neumann energy}
\G^{\mathrm{int}}(X_{\mn},\mu, \Omega):=\frac{1}{4\pi}\left(\int_{\tilde{\Omega}}|\nabla u_{\rrh}|^2-2\pi\sum_{i=1}^{\mn} \g(\rrh_i)\right)-\sum_{i=1}^{\mn} \int_{\R}\f_{\rrh_i}(x-x_i)d\mu(x)+\sum_{i=1}^{\mn} \GG(x_i)
\end{equation}
where $\tilde{\Omega}$ is as in (\ref{extension set}). $u_{\rrh}$ is the truncated potential with $\rrh$ defined relative to $\Omega$, where $u$ is the unique solution to the Neumann problem
\begin{equation}\label{Neumann potential}
\begin{cases}
-\Delta u=2\pi \left(\sum_{i=1}^{\mn} \delta_{x_i}-\mu\delta_{\R}\right) & \text{in }\tilde{\Omega} \\
\nabla u \cdot \widehat{n}=0 &\text{on }\partial \tilde{\Omega}.
\end{cases}
\end{equation}
We note that this equation has a solution due to the assumption that $\mu(\Omega)=\mn$. Notice that when we take $\Omega=\R$, our Neumann condition requires decay at infinity which forces $u=\g\ast \left(\sum_{i=1}^N \delta_{x_i}-\mu\delta_{\R}\right)$, and so $\G^\Omega(\XN,\mu)$ and $\G^{\mathrm{int}}(\XN,\mu,\Omega)$ agree with $\G(\XN,\mu)$. Similarly, relative to $\Omega' \subset \Omega$ with $\mu(\Omega \setminus \Omega')=\mn$ we define an external subadditive energy by 
\begin{equation}\label{Neumann energy - external}
\G^{\mathrm{ext}}(X_{\mn},\mu, \Omega\setminus \Omega'):=\frac{1}{4\pi}\left(\int_{\tilde{\Omega}\setminus \tilde{\Omega'}}|\nabla u_{\rrh}|^2-2\pi\sum_{i=1}^{\mn} \g(\rrh_i)\right)-\sum_{i=1}^{\mn} \int_{\R}\f_{\rrh_i}(x-x_i)d\mu(x)+\sum_{i=1}^{\mn} \GG(x_i),
\end{equation}
where $\Omega$ and $\Omega'$ are as in (\ref{extension set}).

Finally, we define a local version of the next-order partition function and Gibbs measure. Assuming $\mu(\Omega)=\mn$ again, we define the internal partition function 
\begin{equation}\label{local next order partition function}
\K^{\mathrm{int}}(\mu,\zeta, \Omega):=\frac{1}{\mn^{\mn}}\int_{\Omega^{\mn}}\exp\left(-\beta \G^{\mathrm{int}}(X_{\mn},\mu, \Omega)\right)~d\rho^{\otimes \mn}
\end{equation}
where $\rho=e^{-\beta \zeta(x)}~dx$ is a measure that is of order $\mn$ defined in (\ref{definition of next order partition function}) and relative to $\Omega'\subset \Omega$ with $\mu(\Omega \setminus \Omega')=\mn$ the external partition function
\begin{equation}\label{local next order partition function - external}
\K^{\mathrm{ext}}(\mu,\zeta, \Omega\setminus \Omega'):=\frac{1}{\mn^{\mn}}\int_{(\Omega\setminus \Omega')^{\mn}}\exp\left(-\beta \G^{\mathrm{ext}}(X_{\mn},\mu, \Omega\setminus \Omega')\right)~d\rho^{\otimes \mn}.
\end{equation}

When it is clear what reference measure or potential we are working with respect to, we will often elide $\mu$ or $\zeta$ in the notation.

\subsection{Basic Results on Local Energies and Partition Functions}
In this section, we introduce a variety of lemmata that will be of use to us throughout the paper. Some results generalize directly from the analogous propositions for Coulomb gases discussed in \cite{AS21} with little to no proof; we remark and comment upon those here. We postpone any additionally necessary arguments to Appendix A. Various technical bounds on the energies and truncation parameters are also postponed to Appendix A.

The first tells us that we can restrict our analysis to intervals of integer mass. This will be useful to us when we need to screen configurations in ways which are neutralized by the background noise $\mu$. We state it without argument, as it translates exactly from \cite{AS21}. 
\begin{lem}\label{quantization}
Let $U$ be $\R$, a finite collection of closed intervals, or the complement of such a set. Let $\mathcal{Q}_L=\{\text{intervals }Q,\text{ length }\in [L,2L], \: \mu(Q) \in \mathbb{N}\}$. Assume $\mu \geq m>0$ in $U$. Then, $\exists~L_0>0$ such that for all $L>L_0$, there is a collection $\mathcal{K}_L \subset \mathcal{Q}_L$ such that 
\begin{equation*}
U=\bigcup_{K \in \mathcal{K}_L}K.
\end{equation*}
\end{lem}

We now turn our attention to the local energies described above, first demonstrating the superadditivity of $\H_{a,h}$ and $\E$ in the $\Omega$ argument. The proof is exactly as in \cite{AS21}.
\begin{lem}\label{superadditivity}
Let $\XN$ be a configuration of points in $\R$, and suppose that $\Omega' \subset \Omega$. Let $w$ be associated to $\Omega$ as in (\ref{viable potential}). Then,
\begin{align*}
\H_{a,h}^{\mathrm{int}}(\XN \vert_{\Omega},\Omega) &\geq \H_{a,h}^{\mathrm{int}}(\XN\vert_{\Omega'}, \Omega')+\H_{a,h}^{\mathrm{ext}}(\XN\vert_{\Omega \setminus \Omega'},\Omega \setminus \Omega') \\
\E^{\mathrm{int}}(w, \Omega)&\geq \E^{\mathrm{int}}(w,\Omega)+\E^{\mathrm{ext}}(w, \Omega \setminus \Omega')
\end{align*}
\end{lem}
\begin{proof}
Let $w$ be minimal in the definition of $\H^{\mathrm{int}}(\XN \vert_\Omega,\Omega)$. Split the defining integral into the contributing pieces over $\Omega'$ and $\Omega \setminus \Omega'$; changing the respective minimal distances to those defined relative to $\Omega'$ and $\Omega \setminus \Omega'$ can only increase the truncation parameter (\ref{decreasing minimal distance}), which in turn decreases the quantities by Lemma \ref{Monotonicity}. Furthermore, if the requisite decay at height $h$ from the axis is satisfied on $\Omega$ then it is also satisfied on $\Omega'$ and $\Omega \setminus \Omega'$, so we are minimizing over a larger class of potentials on the right hand side. We conclude using the minimality in the definition of $\H$.

The second line follows immediately from the fact that $\rrc$ only increases when defined over smaller sets (\ref{decreasing minimal distance}), and the monotonicity of Lemma \ref{Monotonicity}.
\end{proof}

The minimality of $\H$ also tells us that, in particular, it is smaller than $\G^\Omega$ if the true field has the appropriate decay.

\begin{lem}\label{relaxation}
Suppose $\XN$ is a configuration of points in $\R$ with $N=\mu(\R)$. Let $\Omega' \subset \Omega \subset \R$. Let $w$ satisfy (\ref{viable potential}) in $\tilde{\Omega}$ with a screenable electric field (\ref{screenability}) such that $\int_{2\Omega \times \{\pm h\}}|\nabla w|^2\leq a$ for $h>\max_i \rrc_i$. Then, 
\begin{equation*}
\E^{\mathrm{int}}(w, \Omega) \geq \H_{a,h}^{\mathrm{int}}(\XN\vert_\Omega, \mu, \Omega)
\end{equation*}
and
\begin{equation*}
\E^{\mathrm{ext}}(w,\Omega \setminus \Omega') \geq \H_{a,h}^{\mathrm{ext}}(\XN\vert_{\Omega \setminus \Omega'}, \mu, \Omega \setminus \Omega').
\end{equation*}
\end{lem}
\begin{proof}
This is just a result of the minimality of the definition of $\H$.
\end{proof}

We next would like to show subadditivity of the Neumann energy $\G(\XN,\mu,\Omega)$ in the $\Omega$ argument. This requires the following lemma, which is stated for our case for instance in \cite{SS15-1}, \cite{PS17} and \cite{LS15}. It states that gradients minimize electrostatic energy. It is common in the literature, so we omit its proof.
\begin{lem}\label{projection}
Suppose $\Omega \subset \R$ is a union of finitely many closed intervals, and $E$ and $\nabla u$ are vector fields that both satisfy
\begin{equation*}
\begin{cases}
-\div(A)=c_d\left(\sum_{i=1}^N \delta_{x_i}-\mu\delta_\R\right) & \text{in }\tilde{\Omega} \\
A \cdot \nu=0 & \text{on }\partial \tilde{\Omega},
\end{cases}
\end{equation*}
where $\tilde{\Omega}$ is as in (\ref{extension set}). Let $\vec{\eta}$ be any positive truncation parameter. Then,
\begin{equation*}
\int_{\tilde{\Omega}}|\nabla u_{\vec{\eta}}|^2\leq \int_{\tilde{\Omega}}|E_{\vec{\eta}}|^2,
\end{equation*}
where $E_{\vec{\eta}}$ for a general electric field is defined as $E-\sum_{i=1}^N \nabla \overline{\f}_{\eta_i}(X-(x_i,0))$. The same result holds for $\Omega'\subset \Omega$, with $\tilde{\Omega}$ replaced by $\tilde{\Omega}\setminus \tilde{\Omega'}$.
\end{lem}
This immediately gives us our desired subadditivity.

\begin{lem}\label{subadditivity}
Suppose $\Omega' \subset \Omega$, where $\Omega '$ and $\Omega \setminus \Omega'$ are either $\R$ or closed and bounded intervals. Suppose $X_{N'}$ is a configuration in $\Omega'$ and $Y_{N-N'}$ is a configuration in $\Omega \setminus \Omega'$. Suppose $\mu(\Omega')=N'$, $\mu(\Omega \setminus \Omega')=N-N'$. Then
\begin{equation*}
\G^{\mathrm{int}}(X_{N'} \cup Y_{N-N'}, \mu, \Omega) \leq \G^{\mathrm{int}}(X_{N'}, \mu, \Omega')+\G^{\mathrm{ext}}(Y_{N-N'}, \mu, \Omega \setminus \Omega').
\end{equation*}
\end{lem}
\begin{proof}
This is a direct application of Lemma \ref{projection}. Let $u_1$ solve the Neumann problem 
\begin{equation*}
\begin{cases}
-\Delta u=2\pi \left(\sum_{i=1}^{N'} \delta_{(x_i,0)}-\mu\delta_{\R}\right) & \text{in }\tilde{\Omega'}  \\
\nabla u \cdot \widehat{n}=0 &\text{on }\partial \tilde{\Omega'}
\end{cases}
\end{equation*}
where $\tilde{\Omega'}$ is as in (\ref{extension set}), and let $u_2$ solve the analogous in $\tilde{\Omega}\setminus \tilde{\Omega'}$. Set $E_i=\nabla u_i$, and let $E=E_1\indic_{\tilde{\Omega'}}+E_2\indic_{\tilde{\Omega}\setminus \tilde{\Omega'} }$. Then, $E$ satisfies 
\begin{equation*}
\begin{cases}
-\div(A)=c_d\left(\sum_{i=1}^{N} \delta_{(x_i,0)}-\mu\delta_\R\right) & \text{in }\tilde{\Omega} \\
A \cdot \nu=0 & \text{on }\partial \tilde{\Omega}.
\end{cases}
\end{equation*}
since the Neumann condition prevents additional divergence from being created at the interface $\partial(\tilde{\Omega'})$. Now, redefining $\rrh_i$ with respect to $\Omega$ can only increase $\rrh_i$ (\ref{increasing minimal distance}), which by the monotonicity in the truncation parameter of Lemma \ref{Monotonicity} only decreases the left hand side. Finally, $\GG(x_i)$ also only decreases when defined relative to the larger set, so we conclude the desired inequality.
\end{proof}

This gives us subadditivity in the partition functions as well. We state the following simplified version for ease.

\begin{lem}\label{subadditivity for partition functions}
Suppose $\Omega_1 \subset \Omega_2 \subset \cdots \subset \Omega_p=\Omega$, with $\mu(\Omega_1)=N_1$ and $\mu(\Omega_i \setminus \Omega_{i-1})=N_i$ for $i \geq 2$. Then
\begin{equation*}
\K^{\mathrm{int}}(\mu,\zeta,\Omega) \geq \frac{N! N^{-N}}{N_1! \cdots N_p!N_1^{-N_1} \cdots N_p^{-N_p}}\K^{\mathrm{int}}(\mu, \zeta, \Omega_1)\prod_{i=2}^p \K^{\mathrm{ext}}(\mu,\zeta, \Omega_i\setminus \Omega_{i-1}).
\end{equation*}
\end{lem}
\begin{proof}
The proof is combinatorial in nature. Using the fact that the way to place $N$ points in $p$ boxes with $N_i$ points in the $i$th box is given by $\frac{N!}{N_1! \cdots N_p!}$, we compute using Lemma \ref{subadditivity} and symmetry that 
\begin{align*}
\K^{\mathrm{int}}(\mu,\zeta,\Omega)&=\frac{1}{N^N}\int_{\Omega^N}\exp\left(-\beta \left(\G^{\mathrm{int}}(\XN,\mu, \Omega)+\sum_{i=1}^N \zeta(x_i)\right)\right)~d\XN \\
&\geq \frac{1}{N^N}\frac{N!}{N_1! \cdots N_p!} \int_{(\Omega_1)^{N_1}}\exp \left(-\beta(\G^{\mathrm{int}}(X_{N_1},\mu,\Omega_1)\right)\prod_{j=1}^{N_i}e^{-\zeta(x_j)}~dz_j \\
&\times \prod_{i=2}^p \int_{(\Omega_i \setminus \Omega_{i-1})^{N_i}}\exp \left(-\beta(\G^{\mathrm{ext}}(X_{N_i},\mu,\Omega_i \setminus \Omega_{i-1})\right)\prod_{j=1}^{N_i}e^{-\zeta(x_j)}~dz_j \\
&\geq\frac{N! N^{-N}}{N_1! \cdots N_p!N_1^{-N_1} \cdots N_p^{-N_p}}\K^{\mathrm{int}}(\mu, \zeta, \Omega_1)\prod_{i=2}^p \K^{\mathrm{ext}}(\mu,\zeta, \Omega_i\setminus \Omega_{i-1}).
\end{align*}
\end{proof}

Finally, we also have some a priori control on partition functions that will be useful to us. The first is a direct consequence of the expansion of the global partition function given by Corollary 1.1 in \cite{LS15}. As discussed in \cite{BLS18} (Lemma 2.3) it gives us a bound on the exponential moments of the global next order energy.
\begin{lem}\label{global bound on partition functions}
Let $\mu$, $\zeta$ be as above. Then,
\begin{equation*}
|\log \K(\mu, \zeta)|\lesssim_\beta N
\end{equation*}
where $\K(\mu,\zeta)$ is as in (\ref{definition of next order partition function}). As a result,
\begin{equation*}
\Esp_{\PNbeta}\left[\exp \left(-\beta \G(\XN,\mu)\right)\right]\lesssim_\beta N,
\end{equation*}
where $\G(\XN,\mu)$ is as in (\ref{definition of next order energy}).
\end{lem}
We can also prove a similar a priori bound on local partition functions analogous to Proposition 3.8 of \cite{AS21}, inspired by an argument from \cite{GZ19}.
\begin{lem}\label{energy integrals and local partition functions}
Suppose that $\Omega \subset N\Sigma_V$ is such that $0<m \leq \mu \leq \Lambda$ on $\Omega$, and that $\mu(\Omega)=\mn$ with $\mn$ a natural number at least $2$. Then,
\begin{equation}\label{a priori estimate}
\frac{1}{\mn^{\mn}}\int_{\Omega^\mn}\G^{\mathrm{int}}(X_{\mn},\mu,\Omega)~d\mu^{\otimes \mn}(X_\mn) \lesssim \mn.
\end{equation}
Furthermore, we also have 
\begin{equation}\label{a priori lower bound}
\G^{\mathrm{int}}(X_\mn,\mu,\Omega) \gtrsim -\mn.
\end{equation}
As a result, we can conclude that 
\begin{equation}\label{a priori partition function control}
\left|\log \K^{\mathrm{int}}(\Omega)\right|\lesssim_{\beta}\mn.
\end{equation}
\end{lem}
\begin{proof}
(\ref{a priori lower bound}) is an immediately consequence of the definition of $\G^{\mathrm{int}}(X_\mn,\mu,\Omega)$ and (\ref{a priori partition function control}) is a direct result of (\ref{a priori lower bound}) and (\ref{a priori estimate}), so we need only demonstrate (\ref{a priori estimate}). The approach is exactly as in Proposition 3.8 of \cite{AS21}; instead of then running through the technical details (which are all the same), we simply remark on the general structure and changes that are necessary.

First, we can integrate by parts as in \cite{AS21} to obtain 
\begin{align}\label{symmetrization}
\G^{\mathrm{int}}(X_\mn, \Omega)&=\frac{1}{2}\int\int_{\Delta^c}G_\Omega(X,Y)\left(\sum_{i=1}^\mn \delta_{\tilde{x_i}}-\mu\delta_{\R}\right)(X)\left(\sum_{i=1}^\mn \delta_{\tilde{x_i}}-\mu\delta_{\R}\right)(Y) \\
\nonumber &+\frac{1}{2}\sum_{i=1}^\mn H_\Omega(\tilde{x_i})+\sum_{i=1}^n\GG(x_i)
\end{align}
where $\tilde{x_i}=(x_i,0)$, $G_\Omega$ is the Neumann Green kernel for $\tilde{\Omega}$ (\ref{extension set}) solving
\begin{equation*}
\begin{cases}
-\Delta G_\Omega=2\pi\left(\delta_y(x)-\frac{\mu\delta_{\R}}{\mu(\Omega)}\right) & \text{in }\tilde{\Omega} \\
-\nabla G_\Omega \cdot \widehat{n}=0 & \text{on }\partial \tilde{\Omega}
\end{cases}
\end{equation*}
and 
\begin{equation*}
\lim_{Y\rightarrow X}G_\Omega(X,Y)+\log |X-Y|=H_\Omega(X).
\end{equation*}
Using (\ref{symmetrization}) we compute 
\begin{align*}
&\int_{\Omega^\mn}\G^{\mathrm{int}}(X_\mn,\mu,\Omega)~d\mu^{\otimes \mn}(X_\mn)=\frac{1}{2}\int_{\Omega^\mn}\left(\sum_{i=1}^\mn H_\Omega(\tilde{x_i})+2\sum_{i=1}^\mn\GG(x_i)\right)~d\mu^{\otimes \mn}(X_\mn)+ \\
&\frac{1}{2}\int_{\Omega^\mn}\left(\sum_{i\ne j}G_\Omega(\tilde{x_i}, \tilde{x_j})-2\sum_{i=1}^\mn \int G_\Omega(\tilde{x_i},Y)~d\mu\delta_{\R}(Y)+\int G_\Omega(X,Y)~d\mu\delta_\R^{\otimes 2}(X,Y)\right)~d\mu^{\otimes \mn}(X_\mn)  \\
&=\frac{1}{2}\mn(\mn-1)\mn^{\mn-2}\int_{\Omega^2}G_\Omega(X,Y)~d\mu\delta_\R^{\otimes 2}(X,Y)-\mn\mn^{\mn-1}\int_{\Omega^2}G_\Omega(X,Y)~d\mu\delta_\R^{\otimes 2}(X,Y)\\
&+\frac{1}{2}\mn^\mn\int_{\Omega^2}G_\Omega(X,Y)~d\mu\delta_\R^{\otimes 2}(X,Y)+\frac{1}{2}\mn\mn^{\mn-1}\int_{\Omega}H_\Omega(\tilde{x})+2\GG(x)~d\mu(x) \\
&=-\frac{1}{2}\mn^{\mn-1}\int_{\Omega^2}G_\Omega(X,Y)~d\mu\delta_\R^{\otimes 2}(X,Y)+\frac{1}{2}\mn^{\mn}\int_{\Omega}H_\Omega(\tilde{x})+2\GG(x)~d\mu(x).
\end{align*}
Now, up to an additive constant, we can choose $G_\Omega$ such that $\int G_\Omega(X,Y)~dx=0$ and
\begin{equation}\label{H bound}
H_\Omega(\tilde{x})\leq -\frac{1}{\mn}\int_\Omega\g(\tilde{x}-Z)~d\mu\delta_\R(Z)+C\max(\g(\dist(\tilde{x},\partial \Omega \times \R)),1)
\end{equation}
by the Green's function estimate of \cite[Proposition A.1]{AS21}. Since $\int G_\Omega(X,Y)~d\mu\delta_\R(Y)$ is harmonic with vanishing Neumann boundary condition, it is a constant. By Fubini,
\begin{equation*}
\int G_\Omega(X,Y)~d\mu\delta_\R(Y)~dx=\int G_\Omega(X,Y)~dXd\mu\delta_\R(Y)=0
\end{equation*}
so that constant is zero. Hence $\int G_\Omega(X,Y)~d\mu\delta_\R(Y)d\mu\delta_\R(X)=0$ too and so we only need to control 
\begin{equation*}
\frac{1}{2}\mn^{\mn}\int_{\Omega}H_\Omega(\tilde{x})+2\GG(x)~d\mu(x).
\end{equation*}
Using (\ref{H bound}), the definition of $\GG$ and that log is an integrable singularity at the origin, we conclude immediately that 
\begin{equation*}
\frac{1}{2}\mn^{\mn}\int_{\Omega}H_\Omega(\tilde{x})+2\GG(x)~d\mu(x) \lesssim \mn^\mn(\mn \log \mn)
\end{equation*}
and so 
\begin{equation}\label{intermediate estimate}
\frac{1}{\mn^\mn}\int_{\Omega^\mn}\G^{\mathrm{int}}(X_\mn,\mu,\Omega)~d\mu^{\otimes \mn}(X_\mn) \lesssim \mn \log \mn,
\end{equation}
which is almost the result we want aside from a pesky logarithm. To remove the logarithm for $\mn>3$, we partition $\Omega$ into a union $\Omega_1 \subset \Omega_2 \cdots \subset \Omega_k=\Omega$ with $\mu(\Omega_1)=n_1 \in \{2,3\}$ and $\mu(\Omega_i\setminus \Omega_{i=1})=n_i \in \{2,3\}$, and use subadditivity as in the proof of Lemma \ref{subadditivity for partition functions}. Observe the same argument for (\ref{intermediate estimate}) also applies to obtain
\begin{equation*}
\frac{1}{n_i^{n_i}}\int_{(\Omega_i \setminus \Omega_{i-1})^n_i}\G^{\mathrm{ext}}(X_{n_i},\mu,\Omega)~d\mu^{\otimes n_i}(X_{n_i}) \lesssim n_i \log n_i.
\end{equation*}
$k \sim \mn$, so we find
\begin{align*}
&\frac{1}{\mn^\mn}\int_{\Omega^\mn}\G^{\mathrm{int}}(X_\mn,\mu,\Omega)~d\mu^{\otimes \mn}(X_\mn)\\
&=\frac{1}{\mn^\mn}\frac{\mn!}{n_1! \cdots n_k!}\int_{\substack{\Omega_1 \times (\Omega_2\setminus \Omega_1)\times \\ \cdots \times (\Omega_k \setminus \Omega_{k-1})}}\G^{\mathrm{int}}(X_\mn \vert_{\Omega_1}, \mu, \Omega_1)+\sum_{i=2}^k \G^{\mathrm{ext}}(X_\mn\vert_{\Omega_i\setminus \Omega_{i-1}},\mu,\Omega_i\setminus \Omega_{i-1})~d\mu^{\otimes \mn} \\
&=\frac{\mn!}{\mn^\mn}\prod_{j=1}^k \frac{n_j^{n_j}}{n_j!}\left(\int_{\Omega_1^{n_1}}\G^{\mathrm{int}}(X_{n_1},\mu,\Omega_1)+\sum_{i=2}^k \frac{1}{n_i}\int_{(\Omega_i\setminus \Omega_{i-1})^{n_i}}\G^{\mathrm{ext}}(X_{n_i},\mu,\Omega_i \setminus \Omega_{i-1})~d\mu^{\otimes n_i}(X_{n_i})\right) \\
&\lesssim \mn\frac{\mn!}{\mn^\mn}\prod_{i=j}^k \frac{n_j^{n_j}}{n_j!}.
\end{align*}
Using Stirling's formula and $a+b \leq ab$ for integers $a,b>1$,
\begin{equation*}
\frac{\mn!}{\mn^\mn}\prod_{j=1}^k \frac{n_j^{n_j}}{n_j!} \lesssim \frac{\mn^\mn e^{-\mn}\sqrt{\mn}}{\mn^\mn}\prod_{j=1}^k \frac{n_j^{n_j}}{n_j^{n_j}e^{-n_j}\sqrt{n_j}}\leq e^{-\mn+\sum_{j=1}^kn_j}\sqrt{\frac{\mn}{\prod_{j=1}^kn_j}}\leq 1.
\end{equation*}
Thus, we actually have 
\begin{equation*}
\frac{1}{\mn^\mn}\int_{\Omega^\mn}\G^{\mathrm{int}}(X_\mn,\mu,\Omega)~d\mu^{\otimes \mn}(X_\mn) \lesssim \mn,
\end{equation*}
as desired.
(\ref{a priori estimate}) and (\ref{a priori lower bound}) immediately imply the partition function control. Using the lower bound on $\G$ (\ref{a priori lower bound}) and the definition (\ref{local next order partition function}), we have
\begin{equation*}
\K^{\mathrm{int}}(\Omega) \lesssim \frac{1}{n^n}\int_{\Omega^n} e^{\beta n}~d\rho ^{\otimes n} \lesssim \frac{|\rho(\Omega)|^n}{n^n}e^{\beta n}\lesssim  \frac{|\mu(\Omega)|^n}{n^n}e^{\beta n}=  e^{\beta n}.
\end{equation*}
using that $\mu$ is bounded above and below, and that $\rho$ agrees with Lebesgue measure on $\Omega\subset N\Sigma_V$.

To obtain a lower bound on $\K^{\mathrm{int}}(\Omega)$, we use the estimate (\ref{a priori estimate}) and Jensen's inequality. Tilting first by $\mu$ and using that $\rho$ agrees with Lebesgue measure on $\Omega$, we find 
\begin{align*}
\K^{\mathrm{int}}(\Omega)&=\frac{1}{\mn^\mn}\int_{\Omega^\mn}\exp \left(-\beta \G^{\mathrm{int}}(X_{\mn},\mu,\Omega)+\sum_{i=1}^\mn \log \frac{\rho(x_i)}{\mu(x_i)}\right)~d\mu^{\otimes \mn} \\
& \geq\exp \left(\frac{1}{\mn^\mn}\int_{\Omega^\mn}\left(-\beta \G^{\mathrm{int}}(X_{\mn},\mu,\Omega)+\sum_{i=1}^\mn \log \frac{\rho(x_i)}{\mu(x_i)}\right)~d\mu^{\otimes \mn}\right) \\
&\geq \exp \left(-C\beta \mn- \frac{1}{\mn^\mn}\int_{\Omega^\mn}\sum_{i=1}^\mn \log \mu(x_i)~d\mu^{\otimes \mn}\right) \\
&=\exp \left(-C\beta \mn-\int_{\Omega}\mu \log \mu\right)\geq \exp \left(-C_\beta \mn \right).
\end{align*}

\end{proof}

\section{Main Boostrap}\label{Section Bootstrap on Scales}
This section aims to prove the main probabilistic control on local energies, following a bootstrap on scales in the spirit of \cite{AS21}. We prove the following theorem.

\setcounter{theo}{0}
\begin{theo}[Local Law]\label{Local Law}
Let $\Omega \subset \B'$, the blowup of the bulk $\B$, satisfy $|\Omega|\in [L,2L]$ for some $L>\omega$ defined in the proof of Proposition \ref{decay control}. Let $C_0:=\frac{C}{8\pi}$, where $C$ is the constant of Lemma \ref{local energy control}. Then there exists a constant $\mathcal{C}$, independent of $L$, and a good event $\mathcal{G}_\Omega$ such that 
\begin{equation}\label{energy bound - local law}
\G^{\Omega}(\XN, \mu)+C_0\#(\XN \cap \Omega) \leq \mathcal{C}L
\end{equation}
on $\mathcal{G}_\Omega$, with 
\begin{equation}
\PNbeta(\mathcal{G}_\Omega^c)\leq C_1e^{-C_2\beta L}
\end{equation}
for some constants $C_1$ and $C_2$ dependent only on $\mu$.
\end{theo}

It suffices to demonstrate the following proposition.

\begin{prop}\label{Main Bootstrap}
Let $\square_{2^kL}\subset \B'$ be intervals of integer mass with radius $[2^kL,2^{k+1}L]$ and fixed center and suppose that 
\begin{equation}\label{energy control assumption}
\G^{\square_{2^kL}}(\XN, \mu)+C_0\#(\XN \cap \square_{2^kL}) \leq \mathcal{C}2^kL
\end{equation}
for all $k \geq 1$ on an event $\mathcal{G}'_L$. Then, there is an event $\mathcal{G}_L\subset \mathcal{G}_L'$ such that 
\begin{equation*}
\G^{\square_L}(\XN, \mu)+C_0\#(\XN \cap \square_L) \leq \mathcal{C}L
\end{equation*}
on $\mathcal{G}_L$, and 
\begin{equation*}
\PNbeta(\mathcal{G}'_L \setminus \mathcal{G}_L) \leq C_1e^{-C_2\beta L}+e^{-\frac{\beta}{4}\mathcal{C}L}
\end{equation*} 
with $C_1$ and $C_2$ only dependent on $\mu$.
\end{prop}
\begin{proof}[Proof of Theorem \ref{Local Law}]
The idea is simply to start at the macroscale in $\B'$, where we know the above local law holds outside of an exponentially small event (cf. \cite[Lemma 2.3]{BLS18}), and apply Proposition \ref{Main Bootstrap} iteratively down to $\square_L$. 
At each application, proving the result on $\square_{2^kL}$, we lose an event of probability no more than $\leq C_1e^{-C_2\beta 2^kL}+e^{-\frac{\beta}{4}\mathcal{C}2^kL}$. Thus, the local law at scale $L$ holds off of an event of size at most 
\begin{equation*}
\sum_{k=0}^\infty C_1e^{-C_2\beta 2^kL}+e^{-\frac{\beta}{4}\mathcal{C}2^kL} \lesssim C_1\sum_{k=1}^\infty \left(e^{-C_2\beta L}\right)^k+\sum_{k=1}^\infty \left(e^{-\frac{\beta}{4}\mathcal{C}L}\right)^k \leq C_1'e^{-C_2\beta L}+C_3e^{-\frac{\beta}{4}\mathcal{C}L},
\end{equation*}
which up to redefining constants independently of $L$ demonstrates the result. 
Finally, observing that neither the proof of Proposition \ref{Main Bootstrap} or the previous argument requires $\square_L$ to be centered at the origin but rather can be centered at any point, we arrive at the desired result for all $\Omega \in \mathcal{Q}_L$. Covering $\B'$ with $\mathcal{Q}_L$ allows us to conclude the result.
\end{proof}
Thus, we set about proving Proposition \ref{Main Bootstrap}. We will assume throughout the result of Section \ref{Section Uniform Bound on Fluctuations}; namely, if Theorem \ref{Local Law} holds down to scale $L$, then Theorem \ref{Uniform Bound on Fluctuations} holds down to scale $L$ too.

\subsection{Screening}
Our main tool is a screening result, based on ideas from \cite{ACO09} and \cite{SS12}. It was introduced for the log gas in \cite{SS15-1} and used extensively in \cite{PS17} and \cite{LS15}. The version presented below is refined and optimized for the log gas, as was done for the Coulomb gas in \cite{AS21}. We also introduce a new inner screening technique for the one-dimensional log gas, which allows us to screen fields external to set $\square_L$. This new tool is an essential part of our bootstrap procedure.
\begin{prop}\label{screening result}
Assume $\mu$ is a density satisfying $0<m \leq \mu \leq \Lambda$ in $\square_L $ (outer screening) or $\square_{2L} \setminus \square_L$ (inner screening), where $\square_L\in \mathcal{Q}_L$ satisfies $\mu(\square_L)=\mn \in \mathbb{N}$ (resp. $\mu(\square_L^c)=\mn \in \mathbb{N}$ for inner screening). Then, there exists $C$ depending only on $m, \Lambda$ such that the following holds. 

Suppose $L\geq \tilde{l}\geq l \geq C$, and assume in the inner case that $\square_L  \subset \{x \in \R: \dist(x, \partial \Sigma) \geq \tilde{l}\}$. In the outer screening, let $\frac{3L}{4} \geq h_1,h_2 \geq \frac{L}{2}$ and in the inner screening, let $\frac{3L}{2} \leq h_1,h_2 \leq 2L$. In the outer screening, let $X_n$ be a configuration of points in $\square_L$ and let $w$ solve 
\begin{equation}\label{outer screening w}
\begin{cases}
-\Delta w=2\pi \left(\sum_{i=1}^n \delta_{(x_i,0)}-\mu \delta_\R\right) \hspace{3mm} &\text{in } \square_L \times [-L,L].
\end{cases}
\end{equation}
In the inner screening, let $X_n$ be a configuration of points in $\square_L^c$ and let $w$ solve 
\begin{equation}\label{inner screening w}
\begin{cases}
-\Delta w=2\pi \left(\sum_{i=1}^n \delta_{(x_i,0)}-\mu \delta_\R\right) \hspace{3mm} &\text{in } (\square_L \times [-L,L])^c.
\end{cases}
\end{equation}

In the outer screening, denote
\begin{align}
&S(X_n)=\int_{(\square_{L-\tilde{l}}\setminus \square_{L-2\tilde{l}})\times [-h_2,h_1]}|\nabla w_{\rrc}|^2  \\
& e(X_n)=\int_{\square_{L} \times \{-h_2,h_1\}}|\nabla w|^2.
\end{align}
In the inner screening, denote
\begin{align}
&S(X_n)=\int_{(\square_{L+2\tilde{l}}\setminus \square_{L+\tilde{l}}) \times [-L,L]}|\nabla w_{\rrc}|^2  \\
& e(X_n)=\int_{\square_{L+2\tilde{l}}\times \{-h_2,h_1\}}|\nabla w|^2.
\end{align}
Let $h=\max\{h_1,h_2\}$. Assume the screenability condition 
\begin{equation}\label{screenability}
\max\left(M_0^2, \frac{hS(X_n)}{\tilde{l}l^2}\right)\leq \mathsf{c},
\end{equation}
where $M_0$ is a constant dependent on the configuration defined in the proof (\ref{definition of M} - \ref{definition of M - neg}), and $\mathsf{c}$ is a constant dependent only on $m$ and $n$ defined in the proof as well (\ref{definition of little c}). Then, there exists a $T \in [\tilde{l}, 2\tilde{l}]$, a set $\Old$ such that $\square_{L-T-1} \subset \Old \subset \square_{L-T+1}$ (outer screening) or $\R \setminus \square_{L-T+1} \subset \Old \subset \R \setminus \square_{L-T-1}$ (inner screening), a subset $I_{\partial} \subset \{1,2,\dots,n\}$, and a positive measure $\tilde{\mu}$ in $\New:=\square_L \setminus \Old$ in the outer screening and in $\square_L^c \setminus \Old$ in the inner screening such that the following holds:
\begin{enumerate}
    \item With $n_\Old$ defined below, we have
\begin{itemize}
\item $\tilde{\mu}(\New)=\mn-n_\Old$
\item $|\mu(\New)-\tilde{\mu}(\New)|\lesssim h+\frac{S(X_n)}{\tilde{l}}$
\item $\|\mu-\tilde{\mu}\|_{L^\infty(\New)}<\frac{m}{2}$
\item $\left|\frac{\mu-\tilde{\mu}}{\tilde{\mu}}\right|_{L^\infty(\New)}\leq\frac{1}{2}$
\item $\int_{\New}(\mu-\tilde{\mu})^2 \lesssim l$.
\end{itemize}
    \item $\#I_\partial \lesssim \frac{S(X_n)}{\tilde{l}}$.
    \item $\G^{\mathrm{int}}(Y_{\mn}, \mu, \Omega)$ and $\E^{\mathrm{int}}(w, \Omega)$ are comparable, i.e.
\begin{align}\label{screening error}
\nonumber \G^{\mathrm{int}}(Y_{\mn}, \mu, \Omega)-&\E^{\mathrm{int}}(w, \Omega)\lesssim_{m,\Lambda} \\
&\frac{lS(X_n)}{\tilde{l}}+\sum_{i,j}\g(x_i-z_j)+\tilde{l}+\G^{\mathrm{int}}(Z_{\mn-n_\Old}, \tilde{\mu}, \New) +|n-\mn|+\tilde{l}LM_0^2 +Le(X_n),
\end{align}
where
\begin{equation*}
J=\{(i,j) \in I_\partial \times \{1,\dots,\mn-n_\Old\}:|x_i-z_j|\leq \rrc_i\}.
\end{equation*}
\end{enumerate}
For inner screening, the same result holds with $\E^{\mathrm{int}}$ and $\G^{\mathrm{int}}$ in (\ref{screening error}) replaced by $\E^{\mathrm{ext}}$ and $\G^{\mathrm{ext}}$.
\end{prop}
The proof of this statement is given in Appendix B. Examining the error terms above, there is no reason a priori that $Le(X_n)$ should live at scale less than $L$. Below we show that with high probability this is arbitrarily smaller than $L$, and hence that $\E^{\mathrm{int}}(w, \square_L)$ and $\E^{\mathrm{ext}}(w,\square_L^c)$ for a screenable field and the Neumann energy of a corresponding screened configuration are comparable. 

\subsection{Control of Screening Errors}
We first show that on a good event we have requisite decay away from the axis at local scales.
\begin{prop}\label{decay control}
Let $\XN \subset \R$ be a configuration, and $u$ associated as in (\ref{electrostatic potential}). Let $\Omega=\square_L \subset \B'$ where $L>\omega$ for some $\omega$ defined below. Let $h_L =\frac{L}{2}$. Then, for $\epsilon>0$ arbitrarily small there are constants $C_1$ and $C_2$, dependent only on $m, \Lambda, K_1$ and $\epsilon$, and an event $\mathcal{G}$ with $\mathcal{G}\subset \mathcal{G}_L'$ such that for large enough $N$
\begin{equation}\label{local control of top energy}
\int_{2\Omega \times \{\pm h_L\}}|\nabla u|^2 \leq \epsilon 
\end{equation}
on $\mathcal{G}$, with 
\begin{equation}\label{local event}
\PNbeta(\mathcal{G}_L' \setminus \mathcal{G})\leq C_1e^{-C_2 \beta L}.
\end{equation}
The same result holds for $\int_{2\Omega\times \{\pm h_L\}}|\nabla u|^2$, with $h_L=\frac{3L}{2}$.
\end{prop}
\begin{proof}
Let $h=h_L$ for ease of notation. Letting $E$ be shorthand for $\nabla u$, we recall that $E_x$ and $E_y$ are fluctuations and make use of Theorem \ref{Uniform Bound on Fluctuations} at scale $2^kL$. Namely, 
\begin{equation*}
E_x(a,h)= \int_{\R}\kappa_{a,h}(x) d\left(\sum_{i=1}^N \delta_{x'_i}-\mu\right)(x),
\end{equation*}
with $\kappa_{a,h}(x)=2\pi \frac{-(a-x)}{(a-x)^2+h^2}$ and similarly, for the $y$-component, we have 
\begin{equation*}
E_y(a,h)= \int_{\R} \zeta_{a.h}(x) d\left(\sum_{i=1}^N \delta_{x_i'}-\mu\right)(x),
\end{equation*}
with $\zeta_{a,h}(x)=-2\pi h\frac{1}{(a-x)^2+h^2}$. We split $2\Omega$ into $K$ equally sized subintervals $I_i$, and let $z_i=(a_i,\pm h)$ denote the center of the subinterval $I_i$. On each subinterval $I_i$, we estimate
\begin{align*}
\int_{I_i}|E|^2&\leq 2\int_{I_i}|E_x(z)-E_x(z_i)|^2+2\int_{I_i}|E_x(z_i)|^2+2\int_{I_i}|E_y(z)-E_y(z_i)|^2+2\int_{I_i}E_y(z_i)^2 \\
&\leq 4 |\partial_x E|_{L^\infty(I_i)}^2\int_{I_i}|z-z_i|^2+\frac{4L}{K}|E_x(z_i)|^2+\frac{4L}{K}|E_y(z_i)|^2 \\
&\leq \frac{32L^3}{K^3} |\partial_x E|_{L^\infty(I_i)}^2+\frac{4L}{K}|E(z_i)|^2.
\end{align*}
\subsection*{Control of $|\partial_x E|_{L^\infty(I_i)}^2$}
We first estimate $|\partial_x E|_{L^\infty(I_i)}^2$ using the energy estimate (Lemma \ref{Local Laws Energy Estimate}) and the control (\ref{energy control assumption}). Without loss of generality we consider $|\partial_xE_y((0,h))|$, since the argument for other $a_i \in \B'$ is analogous.

There is some $\alpha>0$, independent of $N$, such that $[-\alpha N,\alpha N]\subset \B'$. We apply the control of \ref{energy control assumption} inside of $[-\alpha N,\alpha N]$, and exploit a uniform control outside. In order to make use of the decay of $\zeta_{0,h}$, we introduce a partition of unity $\{\varphi_k\}_{k=1}^{k_*+1}$ satisfying 
\begin{align*}
\supp(\varphi_k)& \subset \square_{2^kL}\setminus \square_{2^{k-1}L} \\
|\partial_x \zeta_{0,h} \varphi_k|_{C^m( \square_{2^kL}\setminus \square_{2^{k-1}L})}&\lesssim |\partial_x \zeta_{0,h}|_{C^m( \square_{2^kL}\setminus \square_{2^{k-1}L})}
\end{align*}
for $2\leq k \leq k_*$ with $2^{k_*}L \leq \alpha N<2^{k_*+1}L$. We suppose the same holds for $\varphi_1$ with $\square_{2L}\setminus \square_{L}$ replaced by $\square_{2L}$, and for $\varphi_{k_*+1}$ on the rest of the line. By Lemmas \ref{Local Laws Energy Estimate}, \ref{local energy control} and the control of (\ref{energy control assumption}), we find
\begin{align*}
|\partial_xE_y((0,h))|\leq A+B+C
\end{align*}
with 
\begin{align*}
A&= \|(\partial_x \zeta_{0,h} \varphi_1)'\|_{L^\infty(\square_{2L})}(|\square_{2L}|+|\square_{2L}|^{1/4}\|\nabla u_{\rrc}\|_{L^2(\tilde{\square}_{2L})}) \\
&+C\left(\sqrt{\mathfrak{h}_1}\|(\partial_x \zeta_{0,h} \varphi_1)'\|_{L^2(\square_{2L})}+\frac{1}{\sqrt{\mathfrak{h}_1}}\|\partial_x \zeta_{0,h} \varphi_1)\|_{L^2(\square_{2L})}\right)\|\nabla u_{\rrc}\|_{L^2(\tilde{\square}_{2L})},
\end{align*}
\begin{align*}
&B=\sum_{k=2}^{k_*} \|(\partial_x \zeta_{0,h} \varphi_k)'\|_{L^\infty(\square_{2^kL}\setminus \square_{2^{k-1}L})}(|\square_{2^kL}\setminus \square_{2^{k-1}L}|+|\square_{2^kL}\setminus \square_{2^{k-1}L}|^{1/4}\|\nabla u_{\rrc}\|_{L^2(\tilde{\square}_{2^kL}\setminus \tilde{\square}_{2^{k-1}L})}) \\
&+C\sum_{k=2}^{k_*}\left(\sqrt{\mathfrak{h}_k}\|(\partial_x \zeta_{0,h} \varphi_k)'\|_{L^2(\square_{2^kL}\setminus \square_{2^{k-1}L})}+\frac{1}{\sqrt{\mathfrak{h}_k}}\|\partial_x \zeta_{0,h} \varphi_k)\|_{L^2(\square_{2^kL}\setminus \square_{2^{k-1}L})}\right)\|\nabla u_{\rrc}\|_{L^2(\tilde{\square}_{2^kL}\setminus \tilde{\square}_{2^{k-1}L})}
\end{align*}
and 
\begin{align*}
C&=\|(\partial_x \zeta_{0,h} \varphi_{k_*+1})'\|_{L^\infty(\square_{2^{k_*}L}^c)}(|\square_{2^{k_*}L}^c|+|\square_{2^{k_*}L}^c|^{1/4}\|\nabla u_{\rrc}\|_{L^2(\tilde{\square}_{2^{k_*}L}^c )}) \\
&+C\left(\sqrt{\mathfrak{h}_{k_*+1}}\|(\partial_x \zeta_{0,h} \varphi_{k_*+1})'\|_{L^2(\square_{2^{k_*}L}^c)}+\frac{1}{\sqrt{\mathfrak{h}_{k_*+1}}}\|\partial_x \zeta_{0,h} \varphi_{k_*+1})\|_{L^2(\square_{2^{k_*}L}^c)}\right)\|\nabla u_{\rrc}\|_{L^2(\tilde{\square}_{2^{k_*}L}^c)},
\end{align*}
where $\tilde{\square}_{2^kL}$ is as in (\ref{extension set}). Observe that $\partial_x \zeta_{0,h}=\frac{4\pi hx}{(x^2+h^2)^2}$ and $(\partial_x\zeta_{0,h})'=\frac{-16\pi hx^2}{(x^2+h^2)^3}$. For $A$, we compute directly using Lemma \ref{local energy control} and (\ref{energy control assumption}) that 
\begin{align*}
A&\lesssim \frac{1}{h^3}\left((4L)+(4L)^{1/4}\sqrt{2\mathcal{C}L}\right)+\frac{1}{h^2}\sqrt{2\mathcal{C}L} \lesssim \sqrt{\mathcal{C}}\left(\frac{L}{h^3}+\frac{\sqrt{L}}{h^2}\right)
\end{align*}
with $\mathfrak{h}_1=h$. Similarly,
\begin{align*}
B&\lesssim \sum_{k=2}^{k_*}\left[\frac{1}{\left(2^{k-1}L\right)^3}\left(2^kL+(2^kL)^{1/4}\sqrt{\mathcal{C}2^kL}\right) + \frac{1}{\left(2^{k-1}L\right)^2}\sqrt{\mathcal{C}2^kL}\right]\lesssim \frac{\sqrt{\mathcal{C}}}{L^{3/2}}
\end{align*}
with $\mathfrak{h}_k=2^{k-1}L$. Finally, on the rest of $\R$ we use a global energy estimate to find 
\begin{equation*}
C \lesssim \frac{1}{N^3}\left(N+N^{1/4}\sqrt{N}\right)+\frac{1}{N^2}\sqrt{N}\lesssim \frac{1}{N^{3/2}}.
\end{equation*}
Since $L<N$, we have 
\begin{equation*}
|\partial_xE_y((0,h))|^2\lesssim \mathcal{C}\left(\frac{L^2}{h^6}+\frac{L}{h^4}+\frac{1}{L^3}\right).
\end{equation*}
Observe that we could have just as easily ran the same estimates for any $a \in \Omega$, and for $\kappa_{a,h}$ instead of $\zeta_{a,h}$. Thus, by adjusting constants, we find
\begin{equation}\label{gradient electric field bound}
|\partial_x E|_{L^\infty(I_i)}^2\lesssim \mathcal{C}\left(\frac{L^2}{h^6}+\frac{L}{h^4}+\frac{1}{L^3}\right)
\end{equation}
too. 

\subsection*{Control of $|E(z_i)|^2$}
Next, we control $|E(z_i)|$ using a using a Chernoff bound and the uniform bound on the electric field given by Theorem \ref{Uniform Bound on Fluctuations}. Fix a positive number $M\geq 2C$ to be determined, where $C$ here denotes the constant in the uniform bound Theorem \ref{Uniform Bound on Fluctuations}. The idea is that for most configurations, the electric field at height $h$ is controlled by $\frac{M}{h}$. 

Letting $s=ch$ and applying Theorem \ref{Uniform Bound on Fluctuations}, we find with a simple Chernoff bound that
\begin{align*}
\PNbeta(chE_x(a,h)\geq cM \cap \mathcal{G}'_L) &\leq \exp \left(-cM+cC\left(1+\left|1-\frac{1}{\beta}\right|\right)+\frac{Cc^2}{\beta}+\frac{Cc^3}{\beta^2 h}\right) \\
&\leq \exp \left(-\frac{cM}{2}+\frac{c^2C}{\beta}\right)
\end{align*}
where $C$ is the constant in Theorem \ref{Uniform Bound on Fluctuations}, so long as we chose $M$ and $c$ such that 
\begin{equation}\label{restriction}
M>2C\left(1+\left|1-\frac{1}{\beta}\right|+\frac{c^2}{\beta^2 h}\right).
\end{equation}
Optimizing over $c$ yields $c=\frac{\beta M}{4C}$ (which we will check later is a legitimate choice) and 
\begin{equation*}
\PNbeta(chE_x(a,h)\geq cM \cap \mathcal{G}'_L) \leq \exp \left(-\frac{\beta ^2 M^2}{16 C}\right).
\end{equation*}
It follows that
\begin{align*}
&\PNbeta\left(\left\{|E_x(a,h)|> \frac{M}{h}\right\} \cap \mathcal{G}_L'\right) \\
&\leq\PNbeta\left(\left\{e^{chE_x(a,h)} > e^{cM} \right\}\cap \mathcal{G}'_L\right)+\PNbeta\left(\left\{e^{-chE_x(a,h)} \geq e^{cM}\right\} \cap \mathcal{G}'_L\right) \\
&\leq2\exp \left(-\frac{\beta ^2 M^2}{16 C}\right).
\end{align*}
The same estimate works for $E_y$. Thus, we have 
\begin{equation}\label{pointwise electric field bound}
\PNbeta\left(\text{there exists an }i \text{ such that }|E_x(z_i)| \text{ or }|E_y(z_i)| > \frac{M}{h} \cap \mathcal{G}'_L\right) \leq 4K \exp \left(-\frac{\beta ^2 M^2}{16 C}\right).
\end{equation}

\subsection*{Conclusion}
Define $\mathcal{G}$ to be the event on which the above estimate holds. Using (\ref{gradient electric field bound}) and (\ref{pointwise electric field bound}), we find 
\begin{align*}
\int_{2\Omega \times \{\pm h\}}|\nabla u|^2=\sum_i \int_{I_i}|E|^2&\lesssim K \left( \frac{L^3}{K^3}(|\partial_xE_x(z_i)|^2+|\partial_xE_y(z_i)|^2)+\frac{L}{K}|E(z_i)|^2\right) \\
&\lesssim L\frac{M^2}{h^2}+\frac{\mathcal{C}L^3}{K^2}\left[\frac{L^2}{h^6}+\frac{L}{h^4}\right]
\end{align*}
on $\mathcal{G}$, with $\PNbeta(\mathcal{G}_L'\setminus \mathcal{G})\leq 4K \exp \left(-\frac{\beta ^2 M^2}{16 C}\right)$. Now, with $h=\frac{L}{2}$, $M=\frac{\sqrt{L}}{K_2}$ and $K$ a constant, we have on $\mathcal{G}$ that 
\begin{equation*}
\int_{\Omega \times \{\pm h\}}|\nabla u|^2\lesssim \frac{1}{K_2^2L}+\frac{\mathcal{C}}{K^2L}+\frac{\mathcal{C}}{K^2},
\end{equation*}
which can be made small by appropriate choice of constants. The astute reader will observe that a similarly sharp estimate for $\int_{2\Omega \times \{\pm h\}}|\nabla u|^2$ could be made via Theorem \ref{Uniform Bound on Fluctuations} and a second moment bound, however this approach has the benefit of yielding an estimate that holds on an event of much larger probability. 

We need to check that this choice is valid given our restriction (\ref{restriction}) on $M$ and $c$. Taking $L$ large enough, we easily have $M>2C\left(1+\left|1-\frac{1}{\beta}\right|\right)$ in the constant $\beta$ regime, although we note that in the regime where $\beta$ depends on $N$ this might introduce an additional constraint on our minimal scale. This choice of minimal scale $\omega$ of course depends on the $\epsilon$ that we will choose to apply this estimate; however, since the choice of $\epsilon$ is independent of the scale, this leads us to a well-defined minimal scale $\omega$. To guarantee that $M>2C \left(\frac{c^2}{\beta^2 h}\right)$ too, we substitute in $h=\frac{L}{2}$, $M=\frac{\sqrt{L}}{K_2}$ and our $c=\frac{\beta M}{4C}$ to find
\begin{equation*}
\frac{\sqrt{L}}{K_2}>2C\left(\frac{2L}{4CK_2^2L}\right)=\frac{1}{2K_2^2},
\end{equation*}
which is valid so long as we take $L$ large enough. Hence, if we take $K$ and $K_2$ large enough,   
\begin{equation*}
\int_{\Omega \times \{\pm h\}}|\nabla u|^2\leq \epsilon
\end{equation*}
with \begin{equation*}
\PNbeta(\mathcal{G}'_L\setminus \mathcal{G})\leq C_1e^{-C_2\beta L},
\end{equation*}
as desired.
\end{proof}
This allows us to conclude that the local energy of any field with similar decay is comparable to the Neumann energy.  
\begin{prop}\label{tailored screening}
Let $C_0:=\frac{C}{8\pi}$ for the constant $C$ defined in Lemma \ref{local energy control}. Let $\Omega=\square_L\subset \B'$ and $X_n=\XN\vert_\Omega$. Suppose $\mu(\Omega)=\mn$, and observe that $\mu$ satisfies the assumptions of Proposition \ref{screening result}. Then, for $\epsilon>0$ arbitrarily small, there are some constants $C_1:=C_1(\epsilon)$, $C_2:=C_2(\epsilon)$ dependent only on $m,\Lambda$ and set of configurations $\mathcal{G}\subset \mathcal{G}_L'$ such that for all configurations in $\mathcal{G}$, 
  \begin{align}\label{tailored outer screening error}
\G^{\mathrm{int}}(Y_{\mn}, \Omega)-\H^{\mathrm{int}}_{\epsilon,\frac{L}{2}}(X_n, \Omega)&\lesssim_{m,\Lambda} \sum_{i,j}\g(x_i-z_j)+\G^{\mathrm{int}}(Z_{\mn-n_\Old}, \tilde{\mu}, \New) +|n-\mn|+\epsilon L,
\end{align}
\begin{align}\label{tailored inner screening error}
\G^{\mathrm{ext}}(Y_{\mn}, \Omega^c)-\H^{\mathrm{ext}}_{\epsilon, \frac{3L}{2}}(\XN\setminus X_n, \Omega^c)&\lesssim_{m,\Lambda} \sum_{i,j}\g(x_i-z_j)+\G^{\mathrm{int}}(Z_{\mn-n_\Old}, \tilde{\mu}, \New) +|n-\mn|+\epsilon L
\end{align}
and 
\begin{equation*}
\PNbeta(\mathcal{G}'_L\setminus \mathcal{G})\leq C_1e^{-C_2\beta L}.
\end{equation*}
\end{prop}
\begin{proof}
Let us focus on the proof for the outer screening, since the inner case is analogous. Let $\epsilon>0$ be arbitrarily small, and let $\mathcal{G}$ be the event of Proposition \ref{decay control}. Let $w$ be any potential associated to $\Omega$ as in (\ref{outer screening w}) with $\E^{\mathrm{int}}(w, \Omega)\leq \G^\Omega(\XN, \mu)$ and $\int_{\Omega \times \pm h_L}|\nabla w|^2 \leq \epsilon$. (Recall that there is at least one potential that satisfies these assumptions, namely the true potential (\ref{electrostatic potential}).) We first claim that for any configuration in $\mathcal{G}$, all such $w$ are screenable in $\Omega$. We denote $X_n=\XN \vert_\Omega$, and set $l=\frac{L}{K_1}$, $\tilde{l}=K_3L<\frac{1}{4}L$ for $K_3\geq\frac{1}{K_1}$ and $K_1 \geq C_0$ to be defined later. Let us define $h_1=h_2=\frac{L}{2}.$ 

%
Observe using Jensen that we have 
\begin{equation*}
M_0^2\leq\left[\frac{1}{2\tilde{l}}\int_{\Old\times\{\pm h\}}\nabla w \cdot \widehat{n}\right]^2\lesssim\frac{|\Old|e(X_n)}{\tilde{l}^2}\lesssim \frac{L\epsilon}{K_3^2L^2}=\frac{\epsilon}{K_3^2L},
\end{equation*}
where we have used $e(X_n)\leq \epsilon$ on $\mathcal{G}$ by Proposition \ref{decay control}. Then also 
\begin{align*}
\frac{hS(X_n)}{\tilde{l}l^2}\leq \frac{L}{2}\frac{8\pi\mathcal{C}L}{K_3L \frac{L^2}{K_1^2}}=\frac{4\pi\mathcal{C}K_1^2}{K_3L}
\end{align*}
using Lemma \ref{local energy control} and the assumption $\E^{\mathrm{int}}(w, \Omega)\leq \G^\Omega(\XN, \mu)$ and we see that all configurations in $\mathcal{G}$ are screenable for some uniform $L>\omega$. Thus, we can apply the screening result Proposition \ref{screening result}, and the requisite control will be obtained by appropriately controlling the errors introduced by the screening. First, we can control the $M_0^2$ error term by $e(X_n)$:
\begin{equation*}
\tilde{l}LM_0^2 = \tilde{l}L \left(\frac{1}{2\tilde{l}}\int_{\Old\times \{\pm h\}}\nabla w \cdot \widehat{n}\right)^2\leq \tilde{l}L \frac{|\Old|}{2\tilde{l}^2} e(X_n)\leq \frac{L^2}{\tilde{l}}e(X_n)\leq\frac{\epsilon L }{K_3}.
\end{equation*}
We can now rewrite the error bound of Proposition \ref{screening result} as  
 \begin{align*}
&\G^{\mathrm{int}}(Y_{\mn}, \Omega)-\E^{\mathrm{int}}(w,\Omega) \\
&\lesssim_{m,\Lambda} 2l\frac{S(X_n)}{\tilde{l}}+K_2L+\sum_{i,j}\g(x_i-z_j)+\G^{\mathrm{int}}(Z_{\mn-n_\Old}, \tilde{\mu}, \New) +|n-\mn|+\frac{\epsilon L}{K_3}.
\end{align*}
Now, notice that $ l\frac{S(X_n)}{\tilde{l}} \leq \frac{2C\mathcal{C}L^2}{K_1K_3L}=\frac{2C\mathcal{C}L}{K_1K_3} $. Hence, if we take $K_1$ large enough relative to $K_3$,   
  \begin{align*}
\G^{\mathrm{int}}(Y_{\mn}, \Omega)-\E^{\mathrm{int}}(w, \Omega)&\lesssim_{m,\Lambda} \sum_{i,j}\g(x_i-z_j)+\G^{\mathrm{int}}(Z_{\mn-n_\Old}, \tilde{\mu}, \New) +|n-\mn|+\epsilon L.
\end{align*}
$w$ was an arbitrary potential with the desired decay, so we might as well take a potential minimizing the argument in the definition of $\H^{\mathrm{int}}_{\epsilon, \frac{L}{2}}(\XN\vert_\Omega, \Omega)$ and obtain
  \begin{align*}
\G^{\mathrm{int}}(Y_{\mn}, \Omega)-\H^{\mathrm{int}}_{\epsilon, \frac{L}{2}}(\XN\vert_\Omega, \Omega)&\lesssim_{m,\Lambda} \sum_{i,j}\g(x_i-z_j)+\G^{\mathrm{int}}(Z_{\mn-n_\Old}, \tilde{\mu}, \New) +|n-\mn|+\epsilon L.
\end{align*}
with 
\begin{equation*}
\PNbeta(\mathcal{G}'_L\setminus \mathcal{G})\leq C_1e^{-C_2\beta L},
\end{equation*}
as desired.
\end{proof}

\subsection{Analysis of Exponential Moments}
It is not sufficient to simply have the above energy bound on one configuration, but rather on a whole volume to maintain good control on Gibbs measures. We make use of the volume of configurations generated by screening in the following proposition. 

\begin{prop}\label{volume of configurations}
Keep the same assumptions as in Proposition \ref{tailored screening}. Let $\epsilon>0$ be arbitrarily small, and let $\mathcal{G}$ denote the good event of Proposition \ref{tailored screening}. Let $n \leq N$, and define 
\begin{align}\label{integration restriction of good event}
\mathcal{G}^\Omega&=\{X_n \in \Omega^n: X_n=\XN \vert_\Omega \text{ for some }\XN \in \mathcal{G}\} \\
\mathcal{G}^{\Omega^c}&=\{X_{N-n} \in (\Omega^c)^{N-n}: X_{N-n}=\XN \vert_{\Omega^c} \text{ for some }\XN \in \mathcal{G}\}
\end{align}
Then,we have 
\begin{equation*}
n^{-n}\int_{\mathcal{G}^\Omega}\exp\left\{-\beta H^{\mathrm{int}}_{\epsilon,\frac{L}{2}}(X_n, \mu,\Omega)\right\}~d\rho^{\otimes n}(X_n)\leq K^{\mathrm{int}}(\Omega)\exp(\beta \varepsilon_e+\varepsilon_v)
\end{equation*}
and 
\begin{equation*}
(N-n)^{N-n}\int_{\mathcal{G}^{\Omega^c}}\exp\left\{-\beta H^{\mathrm{ext}}_{\epsilon,\frac{3L}{2}}(X_{N-n}, \mu,\Omega^c)\right\}~d\rho^{\otimes N-n}(X_{N-n})\leq K^{\mathrm{ext}}(\Omega^c)\exp(\beta \varepsilon_e+\varepsilon_v)
\end{equation*}
where $\varepsilon_e$ denotes the energy error
\begin{equation*}
\varepsilon_e=C\left(|\mn-n|+\epsilon L+\tilde{l}+\frac{\mathcal{C}L}{\tilde{l}}\right)
\end{equation*}
and $\varepsilon_v$ denotes the volume error, with estimate
\begin{equation*}
\varepsilon_v\leq C\tilde{l}+(\mn-n-\alpha)\log \frac{\alpha}{\alpha'}+\left(\mn-n-\alpha-\frac{1}{2}\right)\log \left(1+\frac{n-\mn}{\alpha}\right)+\frac{1}{2}\log \frac{n}{\mn},
\end{equation*}
where $\alpha,\alpha'$ are integers satisfying $\alpha,\alpha' \lesssim \tilde{l}$ and $|\alpha-\alpha'|\lesssim h+\frac{L}{\tilde{l}}$.
\end{prop}
\begin{proof}
We focus on the integral over $\mathcal{G}^\Omega$, since the argument for the integral over $\mathcal{G}^{\Omega^c}$ is an analogous application of Proposition \ref{tailored screening}. We first need to do some combinatorial accounting, as in \cite[Proposition 4.2]{AS21}.

Each screenable configuration $X_n$ yields a number of points $n-\N$ of points that are removed. There are ${n \choose \N}$ ways of doing this, and each deletion corresponds to a volume of configurations of no more than $\rho(\New)^{n-\N}$. Then we insert $\mn-\N$ new points, but the resulting configurations are equivalent up to a permutation of indices; this leads to an overcounting of factor ${\mn \choose \N}$. Coupling this with the error estimate of Proposition \ref{tailored screening}, we find 
\begin{align*}
&\mn^{\mn}\K^{\mathrm{int}}(\Omega) \geq\int_{\mathcal{G}^\Omega}\exp \left(-\beta \left(\H^{\mathrm{int}}_{\epsilon,\frac{L}{2}}(X_n,\Omega) +|\mn-n|+\epsilon L\right)\right) \\
&\int_{\New^{\mn-\N}}\exp \left(-\beta C\left(\sum_{(i,j)\in J}\g(x_i-z_j)+\G^{\mathrm{int}}(Z_{\mn-\N},\tilde{\mu},\New)\right)\right) \frac{{\mn \choose \N}}{{n \choose \N}|\New|^{n-\N}}~dZ_{\mn-\N}d\rho^{\otimes n}(X_n),
\end{align*}
where we have replaced $\rho$ with the Lebesgue measure in the bulk. Notice that this computation is valid for inner screening as well, since we only modify the configuration in the bulk. We will make use of a Jensen inequality trick inspired by \cite{GZ19} coupled with a tilt as in the proof of Proposition \ref{global bound on partition functions}. Integrating against $\tilde{\mu}$ instead of Lebesgue measure, we can rewrite the interior integral as
\begin{equation*}
\int_{\New^{\mn-\N}}\exp \left(-\beta C\left(\sum_{(i,j)\in J}\g(x_i-z_j)+\G^{\mathrm{int}}(Z_{\mn-\N},\tilde{\mu},\New)\right)-\sum_{i=1}^{\mn-\N}\log \tilde{\mu}(z_i)\right)~d\tilde{\mu}\vert_{\New}^{\otimes {\mn-\N}}(Z_{\mn-\N}).
\end{equation*}
Applying Jensen's inequality, we then have 
\begin{align*}
&\int_{\New^{\mn-\N}}\exp \left(-\beta C\left(\sum_{(i,j)\in J}\g(x_i-z_j)+\G^{\mathrm{int}}(Z_{\mn-\N},\tilde{\mu},\New)\right)-\sum_{i=1}^{\mn-\N}\log \tilde{\mu}(z_i)\right)~d\tilde{\mu}\vert_{\New}^{\otimes {\mn-\N}}(Z_{\mn-\N}) \\
&\geq \tilde{\mu}(\New)^{\mn-\N}\exp (A+B+C)
\end{align*}
with 
\begin{align*}
A&=\tilde{\mu}(\New)^{\N-\mn}\int_{\New^{\mn-\N}}-\beta C\sum_{(i,j)\in J}\g(x_i-z_j)~d\tilde{\mu}\vert_{\New}^{\otimes {\mn-\N}}(Z_{\mn-\N}) \\
B&= \tilde{\mu}(\New)^{\N-\mn}\int_{\New^{\mn-\N}}-\beta C\G^{\mathrm{int}}(Z_{\mn-\N},\tilde{\mu},\New)~d\tilde{\mu}\vert_{\New}^{\otimes {\mn-\N}}(Z_{\mn-\N}) \\
C&= \tilde{\mu}(\New)^{\N-\mn}\int_{\New^{\mn-\N}}-\sum_{i=1}^{\mn-\N}\log \tilde{\mu}(z_i)~d\tilde{\mu}\vert_{\New}^{\otimes {\mn-\N}}(Z_{\mn-\N})
\end{align*}
Now we examine the three terms in the exponent. $B$ is immediately dealt with by Lemma \ref{energy integrals and local partition functions}. Integrating $A$ and $C$ we find
\begin{equation*}
A=-C\beta\sum_{i \in I_{\partial}} \int_{|z-x_i|\leq \rrc_i} \g( x_i-z)~d\tilde{\mu}\vert_{\New}(z) \geq -C\beta\#I_\partial
\end{equation*}
since the log singularity is integrable in one dimension, and 
\begin{equation*}
C=\int_{\New} \tilde{\mu}\log \tilde{\mu} \geq  -C\tilde{l}
\end{equation*}
from the boundedness of $\tilde{\mu}$. Since $\#I_\partial \lesssim \frac{S(X_n)}{\tilde{l}}\lesssim \frac{\mathcal{C}L}{\tilde{l}}$, we find 
\begin{align*}
&\int_{\New^{\mn-\N}}\exp \left(-\beta C\left(\sum_{(i,j)\in J}\g(x_i-z_j)+\G^{\mathrm{int}}(Z_{\mn-\N},\tilde{\mu},\New)\right)-\sum_{i=1}^{\mn-\N}\log \tilde{\mu}(z_i)\right)~d\tilde{\mu}\vert_{\New}^{\otimes {\mn-\N}}(Z_{\mn-\N}) \\
&\geq (\mn-\N)^{\mn-\N}\exp \left(-C\beta\left(\tilde{l}+\frac{\mathcal{C}L}{\tilde{l}}\right)-C\tilde{l}\right).
\end{align*}
Applying a mean value argument, we then find that for some $X_n^0$ we have
\begin{align*}
\mn^{\mn}\K^{\mathrm{int}}(\Omega)&\geq \int_{\mathcal{G}^\Omega}\exp \left(-\beta\H^{\mathrm{int}}_{\epsilon, \frac{L}{2}}(X_n,\Omega)-C\beta \left(|\mn-n|+\epsilon L\right)\right)~d\rho^{\otimes n}(X_n) \\
&\times  \frac{{\mn \choose \N(X_n^0)}(\mn-\N)^{\mn-\N}}{{n \choose \N(X_n^0)}|\New(X_n^0)|^{n-\N(X_n^0)}} \exp \left(-C\beta\left(\tilde{l}+\frac{\mathcal{C}L}{\tilde{l}}\right)-C\tilde{l}\right).
\end{align*}
Rearranging yields 
\begin{equation*}
\int_{\mathcal{G}^\Omega}\exp(-\beta \H^{\mathrm{int}}_{\epsilon,\frac{L}{2}}(X_n,\Omega))~d\rho^{\otimes n}(X_n)\leq \mn^{\mn}\K^{\mathrm{int}}(\Omega)\exp \left(\beta C\varepsilon_e+\varepsilon_v\right)
\end{equation*}
with 
\begin{equation*}
\varepsilon_e=|\mn-n|+\epsilon L+\tilde{l}+\frac{\mathcal{C}L}{\tilde{l}}
\end{equation*}
and 
\begin{equation*}
\varepsilon_v=C\tilde{l}+\log \left(\frac{{n \choose \N(X_n^0)}|\New(X_n^0)|^{n-\N(X_n^0)}}{{\mn \choose \N(X_n^0)}(\mn-\N)^{\mn-\N}}\right)=C\tilde{l}+\log \left(\frac{n!(\mn-\N)!|\New|^{n-\N}}{\mn!(n-\N)!(\mn-\N)^{\mn-\N}}\right)
\end{equation*}
We need now to make sure the volume error is not too big. Computing directly with Stirling's formula as in \cite[Proposition 4.2]{AS21} we find
\begin{align*}
\varepsilon_v&=C\tilde{l}+n\log n-\mn\log \mn-(n-\N)\log \frac{n-\N}{|\New|}+\frac{1}{2}\log\frac{n(\mn-\N)}{\mn(n-\N)}.
\end{align*}
Let $\alpha=\tilde{\mu}(\New)$ and $\alpha'=|\New|$. Then, observe that $n-\N=\alpha+n-\mn$ since $\alpha=\mn-\N$ and so we can write 
\begin{align*}
\varepsilon_v&=C\tilde{l}+(\mn-n-\alpha)\log \frac{\alpha+n-\mn}{\alpha'}+\frac{1}{2}\log\frac{n\alpha}{\mn(\alpha+n-\mn)} \\
&=C\tilde{l}+(\mn-n-\alpha)\log \frac{\alpha}{\alpha'}+\left(\mn-n-\alpha-\frac{1}{2}\right)\log \left(1+\frac{n-\mn}{\alpha}\right)+\frac{1}{2}\log \frac{n}{\mn},
\end{align*}
as desired.
\end{proof}
Notice that the volume error takes the same exact form as in \cite[Proposition 4.2]{AS21}. Modifying the associated screening errors then, we can control how the interior and exterior volume error contribute to the total error.
\begin{remark}\label{error bound}
The error contributions over $\Omega$ where $n$ points fall and $\Omega^c$ where $N-n$ points fall add up to a well bounded error. To be precise, if $\alpha, \alpha'$ and $\gamma, \gamma'$ are two pairs corresponding as in Proposition \ref{volume of configurations},
\begin{align*}
&\alpha-\alpha'+(\mn-n-\alpha)\log \frac{\alpha}{\alpha'}-(\alpha+n-\mn+\frac{1}{2})\log \left(\frac{n-\mn}{\alpha}+1\right)+\frac{1}{2}\log \frac{\mn}{n} \\
&+\gamma-\gamma'+(n-\mn-\gamma)\log \frac{\gamma}{\gamma'}-(\gamma+\mn-n+\frac{1}{2})\log \left(\frac{\mn-n}{\gamma}+1\right)+\frac{1}{2}\log \frac{N-n}{N-\mn} \lesssim \tilde{l}.
\end{align*}
\end{remark}
\begin{proof}
The argument is almost exactly the same as in Remark 4.4 \cite{AS21}. It is sufficient to show that for fixed $\mn$, $\alpha$, $\alpha'$, $\gamma$ and $\gamma'$ that the above quantity is maximized for $|n-\mn|\lesssim \tilde{l}$, for some uniform constant - then all quantities are products of order $1$ and order $\tilde{l}$ terms, and we are done since the difference terms are order $\tilde{l}$. To do this, we differentiate in $n$ and find that the maximum is achieved for 
\begin{align*}
&-\log \left(1+\frac{n-\mn}{\alpha}\right)+\log \left(1+\frac{\mn-n}{\gamma}\right)+\log \frac{\gamma}{\alpha}=0,
\end{align*}
which in particular tells us that $\left|\log \left(1+\frac{\mn-n}{\gamma}\right)-\log \left(1+\frac{n-\mn}{\alpha}\right)\right|$ is bounded. A routine computation tells us then that $|\mn-n|$ lives at the same scale as $\alpha$ and $\gamma$, which are both uniformly bounded by $\tilde{l}$. This implies the result.
\end{proof}
In particular, the volume errors are controlled by the volume within which we change the configurations. This allows us to deduce a control in exponential moments over most energies.
%
\begin{prop}\label{energy control for good points}
Let $\mathcal{G}$ denote the good event from Proposition \ref{tailored screening}, and denote by $\mathcal{G}_{n}$ the configurations in $\mathcal{G}$ who have $n$ points in $\Omega$. Then, 
\begin{equation*}
\mathbb{E}_{\PNbeta}\left(\exp \left(\frac{\beta}{2}\G^{\omc}(\cdot, \mu)\right)\indic_{\mathcal{G}_{n}}\right)\lesssim_{m,\Lambda} \frac{\K^{\mathrm{int}}\left(\Omega, \frac{\beta}{2}\right)}{\K^{\mathrm{int}}(\Omega,\beta)}\exp \left(\beta \left(\frac{\mathcal{C}}{8}L+\frac{1}{2}C_0n+C|\mn-n|\right)\right)
\end{equation*}
for some constant $C$ dependent only on $m$ and $\Lambda$, where we have indicated the temperature dependence of the partition function as an additional argument.
\end{prop}
\begin{proof}
The proof goes almost exactly as in that of Proposition 4.5 of \cite{AS21}. We first immediately compute using Lemma \ref{superadditivity} and Lemma \ref{relaxation} with $u$ as in (\ref{electrostatic potential}) that 
\begin{align*}
\frac{\beta}{2}\G^{\omc}(\XN,\mu)-\beta \G(\XN,\mu)&\leq \frac{\beta}{2}\E^{\mathrm{int}}(u, \omc)-\beta \E^{\mathrm{int}}(u, \Omega)-\beta\E^{\mathrm{ext}}(u, \Omega^c) \\
&\leq -\frac{\beta}{2}\E^{\mathrm{ext}}(u, \Omega \setminus \omc)-\frac{\beta}{2}\E^{\mathrm{int}}(u, \Omega)-\beta \E^{\mathrm{ext}}(u, \Omega^c) \\
&\leq -\frac{\beta}{2}\E^{\mathrm{int}}(u, \Omega)-\beta \E^{\mathrm{ext}}(u, \Omega^c)+\frac{\beta}{2}C_0n \\
&\leq -\frac{\beta}{2} \H^{\mathrm{int}}_{\epsilon, \frac{L}{2}}(\XN\vert_\Omega, \Omega)-\beta \H^{\mathrm{ext}}_{\epsilon, \frac{3L}{2}}(\XN \vert_{ \Omega^c},\Omega^c)+\frac{\beta}{2}C_0n,
\end{align*}
where the penultimate line follows from Lemma \ref{local energy control} and the final line follows from the minimality of $\H_{\epsilon, h}$ and Proposition \ref{decay control}. Applying Proposition \ref{volume of configurations} we find 
\begin{align*}
&\mathbb{E}_{\PNbeta}\left(\exp \left(\frac{\beta}{2}\G^{\omc}(\cdot, \mu)\right)\indic_{\mathcal{G}_{n}}\right)\\
&\leq \frac{1}{N^N\K^{\mathrm{int}}(\R)}{N \choose n}\int_{\mathcal{G}^\Omega}\exp \left(-\frac{\beta}{2}\H^{\mathrm{int}}_{\epsilon, \frac{L}{2}}(X_n,\Omega)+\frac{\beta}{2}C_0n\right)~d\rho^{\otimes n}(X_n) \\
&\times \int_{(\Omega^c)^{N-n}\cap \mathcal{G}^{\Omega^c}}\exp \left(-\beta \H^{\mathrm{ext}}_{\epsilon, \frac{3L}{2}}(X_{N-n},\Omega^c)\right)~d\rho^{\otimes N-n}(X_{N-n}) \\
&\leq \frac{{N \choose n}}{N^N \K^{\mathrm{int}}(\R, \beta)}n^n(N-n)^{N-n}\K^{\mathrm{int}}\left(\Omega, \frac{\beta}{2}\right)\K^{\mathrm{ext}}(\Omega^c, \beta)\exp \left(\beta \varepsilon_e+\varepsilon_v+\frac{\beta}{2}C_0n\right)
\end{align*}
where in the last line we have indicated the temperature dependence in the partition function to avoid confusion and $\varepsilon_e$ and $\varepsilon_v$ denote the energy and volume errors 
\begin{align*}
\varepsilon_e&=C(\epsilon L+\tilde{l}+|\mn-n|) \\
\varepsilon_v&=C\tilde{l}.
\end{align*}
By subadditivity of the partition functions (Lemma \ref{subadditivity for partition functions}), we obtain 
\begin{equation*}
\mathbb{E}_{\PNbeta}\left(\exp \left(\frac{\beta}{2}\G^{\omc}(\cdot, \mu)\right)\indic_{\mathcal{G}_{n}}\right)\leq \frac{\mn! (N-\mn)!n^n(N-n)^{N-n}}{n!(N-n)!\mn^\mn(N-\mn)^{N-\mn}}\frac{\K^{\mathrm{int}}\left(\Omega, \frac{\beta}{2}\right)}{\K^{\mathrm{int}}(\Omega, \beta)}\exp \left(\beta \varepsilon_e+\varepsilon_v+\frac{\beta}{2}C_0n\right).
\end{equation*}
Stirling's formula allows us to bound 
\begin{equation*}
\mathbb{E}_{\PNbeta}\left(\exp \left(\frac{\beta}{2}\G^{\omc}(\cdot, \mu)\right)\indic_{\mathcal{G}_{n}}\right)\lesssim \frac{\K^{\mathrm{int}}\left(\Omega, \frac{\beta}{2}\right)}{\K^{\mathrm{int}}(\Omega, \beta)}\exp \left(\beta \varepsilon_e+\varepsilon_v+\frac{\beta}{2}C_0n\right);
\end{equation*}
choosing $\epsilon$ so that $C\epsilon+\left(1+\frac{1}{\beta}\right)Cc_1+\frac{C}{C_0}+\frac{C}{c_1L}\leq \frac{\mathcal{C}}{8}$, we obtain the Proposition.
\end{proof}

\subsection{Conclusion}
With the above results in tow, we are in a position to prove Proposition \ref{Main Bootstrap}.
\begin{proof}[Proof of Propoisiton \ref{Main Bootstrap}]
The goal is to first show that 
\begin{equation*}
\mathbb{E}_{\PNbeta}\left(\exp \left(\frac{\beta}{2}\left(\G^{\square_L}(\cdot, \mu)+C_0 \# ( \{\XN\} \cap \square_L)\right)\right)\indic_{\mathcal{G}}\right)\leq \exp \left(\frac{\mathcal{C}}{4}\beta L\right).
\end{equation*}
This is done by summing the previous estimates over $n$. We can almost exactly mirror the arguments in \cite{AS21}, with one difference: on our good event $\mathcal{G}$, point discrepancies are necessarily controlled; namely, we will show that we only need to consider $|\mn-n|\leq KL^{3/5}$ for some constant $K$. 

Let $n-\mn=\int_\Omega \sum \delta_{x_i}-\mu$. Then, using Lemma \ref{discrepancy estimate} we find that if $|n-\mn|\gtrsim L$ that
\begin{equation*}
L^2 \lesssim |\mn-n|^2 \lesssim \int_{\square_{2L}\times [-2L,2L]}|\nabla u_{\rrc}|^2 \lesssim 2L,
\end{equation*}
a contradiction. So, $|n-\mn|\lesssim L$ and Lemma \ref{discrepancy estimate} yields
\begin{equation*}
\frac{|\mn-n|^{5/2}}{\sqrt{L}} \lesssim \int_{\square_{2L}\times [-2L,2L]}|\nabla u_{\rrc}|^2 \lesssim 2L,
\end{equation*}
which tells us that $|\mn-n| \leq KL^{3/5}$ for some constant $K$.
%
%
%
For the remaining small point discrepancies $|\mn-n|\leq K L^{3/5}$ we make use of the control in Proposition \ref{energy control for good points}. First, using that result directly we find 
\begin{align*}
\sum_{|\mn-n|\leq KL^{3/5}}&\mathbb{E}_{\PNbeta}\left(\exp \left(\frac{\beta}{2}\G^{\omc}(\cdot, \mu)\right)\indic_{\mathcal{G}_{n}}\right) \\
&\lesssim \sum_{|\mn-n|\leq KL^{3/5}}\frac{\K^{\mathrm{int}}\left(\Omega, \frac{\beta}{2}\right)}{\K^{\mathrm{int}}(\Omega, \beta)}\exp \left(\beta \left(\frac{\mathcal{C}}{8}L+\frac{1}{2}C_0n+C|\mn-n|\right)\right).
\end{align*}
Now, using Lemma \ref{energy integrals and local partition functions}, we can control the ratio of partition functions uniformly by $e^{CL}$. Controlling $n$ and $|\mn-n|\lesssim L^{3/5}$ and adjusting constants as necessary, we find that
\begin{equation*}
\mathbb{E}_{\PNbeta}\left(\exp \left(\frac{\beta}{2}\left(\G^{\square_L}(\cdot, \mu)+C_0 \# ( \{\XN\} \cap \square_L)\right)\right)\indic_{\mathcal{G}}\right)\leq \exp \left(\frac{\mathcal{C}}{4}\beta L\right).
\end{equation*}
It follows immediately by a Chernoff bound that $\PNbeta(\{\G^{\square_L}(\cdot, \mu)+C_0 \# ( \{\XN\} \cap \square_L) \geq \mathcal{C}L\}\cap \mathcal{G})$ can be controlled by 
\begin{align*}
&\PNbeta(\{\G^{\square_L}(\cdot, \mu)+C_0 \# ( \{\XN\} \cap \square_L) \geq \mathcal{C}L\}\cap \mathcal{G}) \\
&\leq e^{-\frac{\beta}{2}\mathcal{C}L}\Esp \exp \left(\frac{\beta}{2}\left(\G^{\square_L}(\cdot, U)+C_0 \# ( \{\XN\} \cap \square_L)\right)\indic_{\mathcal{G}}\right) \\
& \leq \exp \left(-\frac{\beta}{4}\mathcal{C}L\right).
\end{align*}
This establishes Proposition \ref{Main Bootstrap}, and as discussed earlier in the section, Theorem \ref{Local Law}.
\end{proof}

\section{Uniform Bound on Fluctuations}\label{Section Uniform Bound on Fluctuations}
The goal of this section is to prove both the control necessary for the local laws bootstrap, and a similar control for rescaled test functions. The approach relies on the Laplace transform method introduced in \cite{J98}, coupled with the transport approach of \cite{BLS18} and modified as in \cite{L21}. The computations rely on explicit formulae for the equilibrium measure and transport map that are available in traditional coordinates. Hence, both for computational ease and to match the existing literature for such estimates we revert to traditional coordinates.

Recall that for the local laws in the previous section, we need probabilistic control of the fluctuations of 
\begin{equation*}
\kappa_{a,h}(x')=2\pi \frac{-(a-x')}{(a-x')^2+h^2}
\end{equation*}
and 
\begin{equation*}
 \zeta_{a,h}(x')=-2\pi h\frac{1}{(a-x')^2+h^2}.
\end{equation*}
Changing back to traditional coordinates, we find 
\begin{equation*}
\int_{\R}\zeta_{a,h}(x') \left(\sum_{i=1}^N \delta_{x'_i}-\mu(x')dx'\right)=\int_{\R} \zeta_{a,h}(Nx) \left(\sum_{i=1}^N \delta_{x_i}-N\mu_V(x)dx\right)=\Fluct_N\left(\zeta_{a,h}(Nx)\right).
\end{equation*}
The same computation can be run for the fluctuations of the test function $\kappa_{a,h}$. 
For computational ease, we keep the argument as $Nx$. We prove that these specific fluctuations are bounded in (almost) exponential moments, and simultaneously prove a control on fluctuations of rescaled test functions. In this section and the following, $L$ again denotes a scale in traditional coordinates. 

\begin{theo}\label{Uniform Bound on Fluctuations}
Suppose that $\theta\in C^3$ is a compactly supported test function, and let $\xi_{z,L}$ denote the associated rescaled test function at scale $L>\frac{\omega}{N}$ with $z \in \B$. Let $N_c$ be as in Definition \ref{critical N}. Then, for every $N\geq N_c$ there is an event $\mathcal{G}_N$ such that 
\begin{align*}
&\left|\log \Esp_{\PNbeta}\left[\exp \left(s\Fluct_N(\xi)\right)\indic_{\mathcal{G}_N}\right] \right|= \frac{s^2}{\beta}\|\xi\|_{H^{1/2}}^2 \\
&+O\left(|s|\|\theta\|_{C^3}+|s|\left|1-\frac{1}{\beta}\right|\|\theta\|_{C^3}+\frac{s^2}{\beta}\max \left(\|\theta\|_{C^1},\|\theta\|_{C^3}^2\right)+\frac{|s|^3}{\beta^2 LN}\max\left(\|\theta\|_{C^2}^2,\|\theta\|_{C^3}^3\right)\right)
\end{align*}
and 
\begin{equation*}
\PNbeta( \mathcal{G}_N^c)\leq C_1e^{-C_2\beta LN}
\end{equation*} 
for some fixed constants $C_1$ and $C_2$. If $\xi=\zeta_{a,h}(Nx)$ or $\kappa_{a,h}(Nx)$ for some $a\in 2\Omega$, where $\Omega$ is as in Proposition \ref{Main Bootstrap}, then we can say analogously that
\begin{equation*}
\left|\log \Esp_{\PNbeta}\left[\exp \left(s\Fluct_N(\xi)\right)\indic_{\mathcal{G}'_L}\right]\right|\lesssim \frac{|s|}{h}+\left|1-\frac{1}{\beta}\right|\frac{|s|}{ h}+ \frac{s^2}{\beta h^2}+\frac{|s|^3}{\beta^2 h^4}
\end{equation*}
with $\mathcal{G}_L'$ as in Proposition \ref{Main Bootstrap}.
\end{theo}
When discussing the test functions $\kappa_{a,h}$ and $\zeta_{a,h}$, we will suppose $h$ is exactly as defined in Section \ref{Section Bootstrap on Scales} and will take the assumptions of Proposition \ref{Main Bootstrap}: namely, that local laws already hold at scales $2^kL'$ for $k \geq 1$.

\subsection{Setup and Expansion of Laplace Transform}
In this section we let $\xi$ be an arbitrary test function, which will later be specified as one of our two relevant cases. The main structure makes use of the Laplace transform method of Johansson \cite{J98} and a suitable change of coordinates as in \cite{BLS18}. The idea of the transport is to introduce a change of coordinates that approximately maps the equilibrium measure $\muv$ of a potential $V$ onto that of the equilibrium measure $\muvt$ of the perturbed potential $V_t=V+t\xi$.
\begin{defi}\label{general transport}
Let $U \supset \Sigma_V$ be open and bounded, and let $\psi \in C^1(U)$. Define the associated \textit{transport map}
\begin{equation}\label{transport map}
\phi_t(x):=x+t\psi(x),
\end{equation}
which induces a map on configurations $\Phi_t:\R^N\rightarrow \R^N$ given by $(\Phi_t(\XN))_i=\phi_t(x_i)$. Let $\tilde{\muvt}$ denote the \textit{approximate equilibrium measure}
\begin{equation}\label{approximate equilibrium measure}
\tilde{\muvt}:=\phi_t \#\muv,
\end{equation}
and $\tilde{\zetat}$ the \textit{approximate confinement potential} $\zeta_V \circ \phi_t^{-1}$. 
\end{defi}
For ease of computation in Appendix C, we choose $U$ to be such that we can write 
\begin{equation}\label{potential derivative}
\zeta_V'=M(x)\sigma(x)
\end{equation}
for some positive and sufficiently smooth function $M$ on $U \setminus \Sigma_V$, with $\sigma(x)$ as in (\ref{equilibrium measure}). Under our assumptions, such a set always exists (see \cite[Lemma 3.1]{BLS18}). We will want to restrict to a set $U$ with compact closure, and will want the configuration $\XN$ to live in $U$. Sacrificing a set of exponentially small probability, this is always possible (see \cite[Lemma 2.6]{BLS18}):
\begin{lem}\label{compact point configurations}
Let $U$ be a fixed open neighborhood of $\Sigma_V$. Then,
\begin{equation}\label{points in compact set}
\PNbeta(\XN \subset U^N) \geq 1-e^{-cN},
\end{equation}
with $c$ only dependent on $\beta$ and $U$.
\end{lem}
We would like to select a transport map $\psi$ such that $\tilde{\muvt}$ well approximates the true equilibrium measure $\muvt$, which amounts to finding a $\psi$ such that $\tilde{\muvt}$ satisfies the Euler-Lagrange equation
\begin{equation*}
(h^{\muvt}+V_t)\circ \Phi_t=c
\end{equation*}
at leading order in $t$. This is, as it turns out, equivalent to inverting the so-called ``master operator" (cf. \cite{BG13}, \cite{BG16}):
\begin{defi}\label{master operator}
Let $\psi \in C^1(U)$. Then, the \textit{master operator} $\Xi_V$ acts on $\psi$ by 
\begin{equation}\label{mo}
\Xi_V[\psi](x)=-\psi(x)V'(x)+\int \frac{\psi(x)-\psi(y)}{x-y}~d\muv(y).
\end{equation}
\end{defi}
The following is proven in \cite[Lemma 3.3]{BLS18}. We restate it here for convenience.
\begin{lem}\label{inverting the master operator}
Let 
\begin{equation}\label{definition of transport}
\psi(x)=\begin{cases}
-\frac{1}{\pi^2S(x)}\int_{\Sigma_V}\frac{\xi(y)-\xi(x)}{\sigma(y)(y-x)}~dy & \text{if } x \in \Sigma_V; \\
\frac{\int \frac{\psi(y)\muv(y)}{x-y}~dy+\xi(x)+c_{\xi}}{\int_{\Sigma_V}\frac{\muv(y)}{x-y}~dy-V'(x)} & \text{if }x \in U \setminus \Sigma_V.
\end{cases}
\end{equation}
Then, there is some constant $c_\xi$ such that 
\begin{equation}\label{transport}
\Xi_V[\psi]=\xi+c_\xi
\end{equation}
in $U$. 
\end{lem}
The above $\psi$ is the transport that we will take going forward. With this $\psi$, there is also a natural associated energy difference.
\begin{defi}
Let $\psi$ be as in \ref{transport}. Then, the \textit{energy difference} is defined by
\begin{equation}\label{energy difference}
\tau_t=\mathcal{F}(t,\psi)-\mathcal{F}(0,0)-tc_\xi,
\end{equation}
where $\mathcal{F}(t,\psi)$ is given by 
\begin{equation*}
\int-\log |\phi_t(\cdot)-\phi_t(y)|~d\muv(y)+V_t \circ \phi_t(\cdot)
\end{equation*}
and $\phi_t$ is the transport map associated to $\psi$.
\end{defi}
Expanding the Laplace transform and splitting as in \cite{BLS18}, we can write the exponential moments of the fluctuations as a ratio of partition functions. Relative to the above transport map, this ratio can be expanded further to yield leading order terms that are independent of the configuration, and error terms which are themselves fluctuations. We record this in the following proposition:
\begin{prop}\label{Expansion of Laplace Transform}
Let $\xi$ be a compactly supported measureable test function. Then, for any $\PNbeta$-measurable event $\mathcal{G}$,
\begin{equation}\label{first expansion}
\Esp_{\PNbeta}[\exp(s\Fluct_N(\xi))\indic_{\mathcal{G}}]=\exp\left(-sN\int \xi~d\muv\right)\frac{\ZNbeta^{\Vt, \mathcal{G}}}{\ZNbeta},
\end{equation}
where $t=-\frac{s}{\beta N}$, $V_t=V+t\xi$ and 
\begin{equation*}
\ZNbeta^{\Vt, \mathcal{G}}=\int_{\mathcal{G}}\exp \left(-\beta \HN^{V_t}(\XN)\right)~d\XN.
\end{equation*}
Furthermore, if we let $\psi$ be defined as in (\ref{transport}), we can expand
\begin{align}\label{main expansion}
&\nonumber \Esp_{\PNbeta}[\exp(s\Fluct_N(\xi))\indic_{\mathcal{G}}]\\
&=\exp \left(-\beta N^2 \textsf{Main}_1+N\textsf{Error}_1\right)\Esp_{\PNbeta}\left[\exp \left(-\beta\left(\textsf{Error}_2+\textsf{Error}_3\right)\right)\indic_{\Phi_t^{-1}(\mathcal{G})}\right],
\end{align}
where
\begin{align}
\label{Main1}\textsf{Main}_1&=\mathcal{I}_{\Vt}(\tilde{\muvt})-\mathcal{I}_V(\muv)-t\int \xi~d\muv \\
\label{Error1}\textsf{Error}_1&=\left(1-\beta \right)\int \log \phi_t'(x)~d\muv \\
\label{Error2}\textsf{Error}_2&=\int\int-\log \left|\frac{\phi_t(x)-\phi_t(y)}{x-y}\right|~d\fluct_N(x)d\fluct_N(y) \\
\label{Error3}\textsf{Error}_3&=\Fluct_N \left(\left(1-\frac{1}{\beta}\right)\log \phi_t'+N\tau_t \right)
\end{align}
for any $s$ such that $t=-\frac{s}{\beta N}$ satisfies $\|t\psi'\|_{L^\infty}<\frac{1}{2}$.
\end{prop}
The proof is a computation and an application of the splitting formula Lemma \ref{Splitting Formula}. 
\begin{proof}[Proof of Proposition \ref{Expansion of Laplace Transform}]
We first expand directly, finding with $t=-\frac{s}{\beta N}$ that
\begin{align*}
&\Esp_{\PNbeta}[\exp(s\Fluct_N(\xi))\indic_{\mathcal{G}}] \\
&=\frac{1}{\ZNbeta^V}\int_{\R^N}\exp \left(s\Fluct_N(\xi)-\beta \HNV(\XN)\right)\indic_{\mathcal{G}}~d\XN \\
&=\frac{e^{-sN\int \xi~d\muv}}{\ZNbeta^V}\int_{\R^N}\exp \left(-\beta\left(\frac{1}{2}\sum_{i \ne j}\g(x_i-x_j)+N\sum_{i=1}^N \left(V(x_i)+t\xi(x_i)\right)\right)\right)\indic_{\mathcal{G}}~d\XN,
\end{align*}
which is exactly (\ref{first expansion}). Using Lemma \ref{Splitting Formula}, we can write 
\begin{equation*}
\ZNbeta^V=\exp \left(-\beta N^2\I_V(\muv)\right)\int \exp \left (-\beta \left(\FN(\XN,\muv)+N\sum_{i=1}^N V(x_i)\right)\right)~d\XN
\end{equation*}
and by a change of variables write
\begin{align*}
\ZNbeta^{\Vt, \mathcal{G}}&=\int_{\R^N}\exp \left(-\beta \HN^{\Vt}(\XN)\right)\indic_{\mathcal{G}}~d\XN \\
&=\int_{\R^N}\exp \left(-\beta \HN^{\Vt}(\Phi_t(\XN))+\sum_{i=1}^N \log \phi_t'(x_i)\right)\indic_{\Phi_t^{-1}(\mathcal{G})}~d\XN
\end{align*}
As in Lemma \ref{Splitting Formula}, we can compute
\begin{equation*}
\HNVt(\Phi_t(\XN))=N^2\I_{\Vt}(\tilde{\muvt})+N\sum_{i=1}^N\zeta_V(x_i)+\FN(\Phi_t(\XN),\tilde{\muvt})+\Fluct_N(N\tau_t),
\end{equation*}
where the fluctuation of $\tau_t$ appears from replacing the confinement potential $\zetat$ for $V_t$ with the approximate confinement potential $\tilde{\zetat}$. It follows that 
\begin{align*}
 \ZNbeta^{V_t,\mathcal{G}}=\exp (-\beta N^2\I_{\Vt}(\tilde{\muvt}))\int_{\R^N}\exp \biggl(&-\beta \biggl(\FN(\Phi_t(\XN),\tilde{\muvt})+N\sum_{i=1}^N\zeta_V(x_i) \\
&+\Fluct_N(N\tau_t)-\frac{1}{\beta}\sum_{i=1}^N \log \phi_t'(x_i)\biggr)\biggr)\indic_{\Phi_t^{-1}(\mathcal{G})}~d\XN.
\end{align*}
Inserting the expansions of $\ZNbeta$ and $\ZNbeta^{\Vt,\mathcal{G}}$ into \ref{first expansion} and using the definition of $\textsf{Main}_1$ (\ref{Main1}),  
\begin{align*}
&\Esp_{\PNbeta}[\exp(s\Fluct_N(\xi))\indic_{\mathcal{G}}]=\exp \left(-\beta N^2\textsf{Main}_1\right) \times \\
&\Esp_{\PNbeta}\left[\exp\left(-\beta \left(\FN(\Phi_t(\XN),\tilde{\muvt})-\FN(\XN,\muv)-\frac{1}{\beta}\sum_{i=1}^N \log \phi_t'(x_i)+\Fluct_N(N\tau_t)\right)\right)\indic_{\Phi_t^{-1}(\mathcal{G})}\right].
\end{align*}
The error terms in the expectation can be further simplified. Expanding the difference in next-order energies,
\begin{align*}
&\FN(\Phi_t(\XN), \tilde{\muvt})-\FN(\XN,\muv)-\frac{1}{\beta}\sum_{i=1}^N \log \phi_t'(x_i) \\
&=\int_{\Delta^c}-\log \left|\frac{\phi_t(x)-\phi_t(y)}{x-y}\right|\left(\sum_{i=1}^N \delta_{x_i}-N\muv \right)^2(x,y)-\frac{1}{\beta}\sum_{i=1}^N \log \phi_t'(x_i) \\
&=\int -\log \left|\frac{\phi_t(x)-\phi_t(y)}{x-y}\right|\left(\sum_{i=1}^N \delta_{x_i}-N\muv \right)^2(x,y)+\left(1-\frac{1}{\beta}\right)\sum_{i=1}^N \log \phi_t'(x_i) \\
&=\textsf{Error}_2-\frac{1}{\beta}N\textsf{Error}_1+\left(1-\frac{1}{\beta}\right)\int \log \phi_t'(x)~d\fluct_N(x).
\end{align*}
Substituting this into the expectation and using the definition of $\textsf{Error}_3$ (\ref{Error3}), we obtain (\ref{main expansion}).
\end{proof}
The rest of this section makes use of the specific transport $\psi$ solving (\ref{transport}) to identify $\textsf{Main}_1$ and control each $\textsf{Error}_i$. We will use precise estimates on the transport map, which we prove in Appendix C. As mesoscopic scales, we will need certain assumptions on the test function to obtain the desired estimates; this leads us to introduce the following threshold for $N$. 
\begin{defi}\label{critical N}
Let $\xi_{z,L}$ be a rescaled test function with $\frac{\omega}{N}<L\ll 1$ and $z \in \B$. Let $t \in \R$. We define the value $N_c$ as the smallest number such that for all $N \geq N_c$, 
\begin{itemize}
\item $\supp(\xi_{z,L}) \subset \B$,
\item $\|t\psi'\|_{L^\infty}<\frac{1}{2}$,
\item $L$ is less than one half of the minimum distance between any two endpoints of $\Sigma_V$.
\end{itemize}
When $L$ is a macroscopic constant, we let $N_c=1$.
\end{defi}
As we will see later in this section, in our cases $N_c$ is always finite even for mesoscopic $L$ since we take $t\sim \frac{1}{N}$, $\|\psi'\|_{L^\infty} \lesssim \frac{1}{L}$ and $\frac{\omega}{N}\ll L\ll1$. The final item is primarily to make our computations easier in Appendix C.
 
\begin{lem}[Rescaled Transport Estimates]\label{Rescaled Transport Estimates}
Let $k \geq 0$ and suppose that $\theta \in C^{k+2}$ has compact support. Let $\xi_{z,L}$ denote the associated rescaled test function at scale $L$ with $z \in B$, and let $N \geq N_c$. Then, 
\begin{equation*}
|\psi^{(k)}(x)|\lesssim \begin{cases}
\frac{1}{L^{k}}\|\theta\|_{C^{k+2}} & \text{if }x \in \supp(\xi_{z,L}); \\
L\left(\frac{1}{|x-z|}+\frac{1}{|x-z|^{k+1}}\right)\|\theta\|_{C^{k+2}} & \text{otherwise.}
\end{cases}
\end{equation*}
\end{lem}
Observe that the above estimates hold immediately when $L$ is a fixed macroscopic constant, as one can always bound $\psi$ associated to $\xi$ via $\|\psi\|_{C^k}\leq \|\xi\|_{C^{k+1}}$ in the noncritical case. This justifies our choice $N_c=1$ in the macroscopic regime.

We have analogous estimates for the transport associated to the local laws test functions $\kappa_{a,h}$ and $\zeta_{a,h}$. To keep analogy with the above proposition, we denote the scale $\frac{1}{N}$ by $L$.
\begin{lem}[Local Laws Transport Estimates]\label{Local Laws Transport Estimates}
Let $\psi$ be the transport associated with $\zeta_{a,h}(Nx)$ or $\kappa_{a,h}(Nx)$ for $a \in \Omega$ and $h$. Then, we have for all $n \geq 0$
\begin{equation*}
|\psi^{(n)}(x)|\lesssim \begin{cases}
\frac{1}{h}\frac{1}{(Lh)^n} & \text{if }|x-a|\leq Lh \\
L \left(\frac{1}{|x-a|}+\frac{1}{|x-a|^{n+1}}\right) & \text{if }|x-a|>Lh,
\end{cases}
\end{equation*}
where the constants depend only on $n$ and $\muv$.
\end{lem}

\subsection{Analysis of Main Term}
We first identify the main term asymptotically with the $H^{1/2}$ norm of $\xi$ via the following proposition.
\begin{prop}\label{Main Term Expansion}
Suppose that $\theta\in C^3$ is a compactly supported test function, and let $\xi_{z,L}$ denote the associated rescaled test function at scale $L$ at least microscopic with $z \in \B$. Then, for every $N\geq N_c$,
\begin{equation*}
-\beta N^2\textsf{Main}_1=\frac{s^2}{\beta}\|\theta\|_{H^{1/2}}^2+ O\left(\frac{|s|^3}{\beta ^2 LN}\max\left(\|\theta\|_{C^2}^2,\|\theta\|_{C^3}^3\right)\right).
\end{equation*}
If $\xi=\zeta_{a,h}(Nx)$ or $\xi=\kappa_{a,h}(Nx)$ for $a \in \Omega$ and $h \lesssim N$, we have 
\begin{equation*}
-\beta N^2\textsf{Main}_1=\frac{s^2}{\beta}\|\xi\|_{H^{1/2}}^2+ O\left(\frac{|s|^3}{\beta ^2 h^4}\right).
\end{equation*}
\end{prop}
\begin{proof}
We write $\xi$ generically for our test function, and specialize to the various cases of the proposition. First observe that by the definition of the push-forward,
\begin{equation*}
\I_{\Vt}(\tilde{\muvt})=\int \int -\log |\phi_t(x)-\phi_t(y)|~d\muv(y)d\muv(x)+\int V_t \circ \phi_t~d\muv
\end{equation*}
and so 
\begin{align*}
&\I_{\Vt}(\tilde{\muvt})-\I_V(\muv)-t\int \xi~d\muv \\
&=\int \int -\log \left|\frac{\phi_t(x)-\phi_t(y)}{x-y}\right|~d\muv(y)d\muv(x)+\int (V_t \circ \phi_t-V)~d\muv-t\int \xi~d\muv.
\end{align*}
Taylor expanding the logarithm, we find 
\begin{align*}
&\int \int -\log \left|\frac{\phi_t(x)-\phi_t(y)}{x-y}\right|~d\muv(y)d\muv(x)=-t\int \int \frac{\psi(x)-\psi(y)}{x-y}~d\muv(y)d\muv(x) \\
&+\frac{t^2}{2}\int \int \left(\frac{\psi(x)-\psi(y)}{x-y}\right)^2~d\muv(y)d\muv(x) + O\left(|t|^3 \int \int \left(\frac{\psi(x)-\psi(y)}{x-y}\right)^3~d\muv(y)d\muv(x)\right).
\end{align*}
Taylor expanding the potential term, we have
\begin{align*}
\int (V_t \circ \phi_t-V)~d\muv&-t\int \xi~d\muv
=\int \left(tV'(x)\psi(x)+\frac{t^2}{2}V''(x)\psi(x)^2+t^2 \xi'(x)\psi(x)\right)~d\muv(x) \\
&+O\left(|t|^3\|V'''\|_{L^\infty}\int \psi(x)^3~d\muv(x)+|t|^3\|\xi''\|_{L^\infty}\int \psi(x)^2~d\muv(x)\right).
\end{align*}
Making use of \cite[Lemma 4.2]{BLS18}, we can rewrite the integrals against $\muv$ as 
\begin{equation*}
-t\int \int \frac{\psi(x)-\psi(y)}{x-y}~d\muv(y)d\muv(x)=-t\int V'(x)\psi(x)~d\muv(x)
\end{equation*}
and 
\begin{equation*}
\frac{t^2}{2}\int \int \left(\frac{\psi(x)-\psi(y)}{x-y}\right)^2~d\muv(y)d\muv(x)+\frac{t^2}{2}\int V''(x)\psi(x)^2~d\muv(x)=-\frac{t^2}{2}\int \xi'(x)\psi(x)~d\muv(x).
\end{equation*}
We conclude that 
\begin{align*}
-\beta N^2 \textsf{Main}_1&=-\frac{\beta}{2}N^2t^2 \int \xi'(x)\psi(x)~d\muv(x)+O\left(N^2|t|^3 \int \int \left|\frac{\psi(x)-\psi(y)}{x-y}\right|^3~d\muv(y)d\muv(x)\right) \\
&+\beta O\left(N^2|t|^3\|V'''\|_{L^\infty}\int |\psi(x)|^3~d\muv(x)+N^2|t|^3\|\xi''\|_{L^\infty}\int \psi(x)^2~d\muv(x)\right).
\end{align*} 
Using the estimates on the transport map (Lemma \ref{Rescaled Transport Estimates}), we have immediately that $\int |\psi^3(x)|~d\muv(x) \lesssim L\|\theta\|_{C^2}^3$ and $\int \psi^2(x)~d\muv(x) \lesssim L\|\theta\|_{C^2}^2$ in the rescaled test function case. For the local laws, we similarly obtain $\int |\psi^3(x)|~d\muv(x) \lesssim \frac{1}{Nh^2}$ and $\int \psi^2(x)~d\muv(x) \lesssim \frac{1}{Nh} $. In particular, for a rescaled test function 
\begin{align*}
&O\left(N^2|t|^3\|V'''\|_{L^\infty}\int |\psi(x)|^3~d\muv(x)+N^2|t|^3\|\xi''\|_{L^\infty}\int \psi(x)^2~d\muv(x)\right) \\
&=O\left( \frac{|s|^3}{\beta ^3 LN}\max\left(\|\theta\|_{C^2}^2,\|\theta\|_{C^2}^3\right)\right).
\end{align*}
For $\xi=\kappa_{a,h}(Nx)$ or $\xi=\zeta_{a,h}(Nx)$, we obtain 
\begin{align*}
&O\left(N^2|t|^3\|V'''\|_{L^\infty}\int |\psi(x)|^3~d\muv(x)+N^2|t|^3\|\xi''\|_{L^\infty}\int \psi(x)^2~d\muv(x)\right)=O\left(\frac{|s|^3}{\beta ^3 h^4}\right)
\end{align*}
as $h\lesssim N$. Next, we claim that
\begin{equation}\label{uniform square integral bound}
\int \int \left(\frac{\psi(x)-\psi(y)}{x-y}\right)^2~dxdy \lesssim \|\theta\|_{C^3}^2
\end{equation}
for the rescaled test function, and 
\begin{equation}\label{uniform square integral bound - local laws}
\int \int \left(\frac{\psi(x)-\psi(y)}{x-y}\right)^2~dxdy \lesssim \frac{1}{h^2}
\end{equation}
for $\xi=\kappa_{a,h}(Nx)$ or $\xi=\zeta_{a,h}(Nx)$. Let's start with the rescaled case. We split up the integral into two pieces: one on which $|x-y|<L$, and one on which $|x-y|\geq L$. In the first case, we can use symmetry, a mean value argument and Lemma \ref{Rescaled Transport Estimates} to write 
\begin{align*}
\int_{|x-y|<L}\left(\frac{\psi(x)-\psi(y)}{x-y}\right)^2~dxdy&=2\int_{\substack{|x-y|<L \\ x \leq y}} \left(\frac{\psi(x)-\psi(y)}{x-y}\right)^2~dxdy  \\
&\lesssim \int_{x\leq y <x+L} |\psi'(x)|^2~dxdy 
\lesssim \|\theta\|_{C^3}^2.
\end{align*}
In the second case, we use the decay of $|\psi(x)|$ outside of $\supp(\xi)$ to write 
\begin{align*}
\int_{|x-y|\geq L}\left(\frac{\psi(x)-\psi(y)}{x-y}\right)^2~dxdy&=2\int_{\substack{|x-y|\geq L \\ x \leq y}} \left(\frac{\psi(x)-\psi(y)}{x-y}\right)^2~dxdy \\
&\lesssim \int_{|x-y|\geq L} \frac{|\psi(x)|^2}{|x-y|^2} 
\lesssim \|\theta\|_{C^2}^2.
\end{align*}
This yields (\ref{uniform square integral bound}). Hence, for the case of a rescaled test function we have
\begin{equation*}
O\left(N^2|t|^3 \int \int \left|\frac{\psi(x)-\psi(y)}{x-y}\right|^3~d\muv(y)d\muv(x)\right)=O\left(\frac{|s|^3}{\beta ^3 LN}\|\theta\|_{C^3}^3\right).
\end{equation*}
A similar argument for $\xi=\kappa_{a,h}(Nx)$ or $\xi=\zeta_{a,h}(Nx)$ with Lemma \ref{Local Laws Transport Estimates} yields (\ref{uniform square integral bound - local laws}), and so for $\xi=\kappa_{a,h}(Nx)$ or $\xi=\zeta_{a,h}(Nx)$ we obtain 
\begin{align*}
O\left(N^2|t|^3 \int \int \left|\frac{\psi(x)-\psi(y)}{x-y}\right|^3~d\muv(y)d\muv(x)\right)=O\left(\frac{|s|^3}{\beta ^3 h^4}\right).
\end{align*}
Finally, using our definition of $H^{1/2}$ by harmonic extension and Lemma \ref{inverting the master operator},
\begin{equation*}
\int_{\R} -\xi'\psi~d\muv=\int_{\R}\xi (\psi \muv)'=\int_{\R^2}\tilde{\xi}(\psi \muv)'\delta_{\R}=\frac{1}{2\pi}\int_{\R^2}\tilde{\xi}(-\Delta \tilde{\xi})=\frac{1}{2\pi}\int_{\R^2}|\nabla \tilde{\xi}|^2=2\|\xi\|_{H^{1/2}}^2
\end{equation*}
where $\tilde{\xi}$ denotes the harmonic extension of $\xi$ to the upper half plane in $\R^2$. Thus,
\begin{equation*}
-\frac{\beta}{2}N^2t^2 \int \xi'(x)\psi(x)~d\muv(x)=\frac{s^2}{\beta}\|\xi\|_{H^{1/2}}^2,
\end{equation*}
as desired. Due to translation and scale invariance, we can replace $\xi$ with $\theta$ here for the rescaled test function.
\end{proof}
In the following section, we control the Error terms directly.

\subsection{Error Bounds}
The goal of this section is to show that the remaining error terms are at worst $O(1)$. These error estimates will be improved in the following section. We first consider $\textsf{Error}_1$. 
\begin{prop}\label{easy Error1 bound}
Suppose that $\theta\in C^3$ is a compactly supported test function, and let $\xi_{z,L}$ denote the associated rescaled test function at scale $L$ at least microscopic with $z \in \B$. Then, for every $N\geq N_c$,
\begin{equation*}
\left|\textsf{Error}_1\right|\lesssim |1-\beta||t|\|\theta\|_{C^3}.
\end{equation*}
Suppose now that $\xi =\kappa_{a,h}(Nx)$ or $\xi=\zeta_{a,h}(Nx)$ for $a \in 2\Omega$. Then, for all $N \geq N_c$,
\begin{equation*}
\left|\textsf{Error}_1\right|\lesssim |1-\beta|\frac{|t|}{h}.
\end{equation*}
\end{prop}

\begin{proof}
We again let $\xi$ denote a generic test function, and specialize to the appropriate cases. Since $\|t\psi'\|_{L^\infty}<\frac{1}{2}$, we can actually write
\begin{equation*}
|\log \phi_t'|=|\log(1+t\psi'(x))|\lesssim t|\psi'(x)|.
\end{equation*}
Using this inequality and an $L^\infty$ bound for $\muv$, we find 
\begin{align}\label{L1 bound}
\nonumber \left|\int \log (1+t\psi'(x))~d\muv(x)\right| &\lesssim   \int_J |\log(1+t\psi'(x))|~dx+\int_{J^c}|\log(1+t\psi'(x))|~dx \\
&\lesssim  |t| \int_J |\psi'(x)|~dx+|t|\int_{J^c}|\psi'(x)|~dx.
\end{align}
where $J$ denotes the bulk of the support (either $\supp(\xi)$ in the rescaled case, or an interval of size $Lh$ if $\xi =\kappa_{a,h}(Nx)$ or $\xi=\zeta_{a,h}(Nx)$). In the case where $\xi =\kappa_{a,h}(Nx)$ or $\xi=\zeta_{a,h}(Nx)$, we can estimate using Lemma \ref{Local Laws Transport Estimates}
\begin{equation*}
\int_J |\psi'(x)|~dx \lesssim  \frac{1}{h}
\end{equation*}
and 
\begin{equation*}
\int_{J^c}|\psi'(x)|~dx \lesssim L\int_{J^c} \left(\frac{1}{|x-a|}+\frac{1}{|x-a|^{2}}\right)~dx \lesssim \frac{1}{h}.
\end{equation*}
Substituting this into (\ref{L1 bound}) yields the result. When $\xi=\xi_{z,L}$ is a rescaled test function, we can similarly estimate using Lemma \ref{Rescaled Transport Estimates}
\begin{equation*}
\int_J |\psi'(x)|~dx  \lesssim \|\theta\|_{C^3}
\end{equation*}
and 
\begin{equation*}
\int_{J^c}|\psi'(x)|~dx \lesssim L\int_{J^c} \left(\frac{1}{|x-z|}+\frac{1}{|x-z|^{2}}\right)\|\theta\|_{C^3}~dx \lesssim \|\theta\|_{C^3}.
\end{equation*}
Again, substituting this into (\ref{L1 bound}) bound yields the result.
\end{proof}

We next analyze $\textsf{Error}_3$, controlling the fluctuations of $\tau_t$ and $\log \phi_t'$. To control the fluctuations of $\tau_t$, we will make use of the following lemma. We postpone its proof to Appendix C. 
\begin{lem}[Energy Difference Estimate]\label{Energy Difference Estimate}
Suppose that $\theta\in C^2$ is a compactly supported test function, and $\xi_{z,L}$ is the associated rescaled test function at scale $L$ at least microscopic with $z \in \B$. Then, 
\begin{equation}\label{general tau bound}
|\tau_t(x)|\lesssim \int t^2\left(\frac{\psi(x)-\psi(y)}{x-y}\right)^2~d\muv(y)+t^2\psi(x)^2+|\xi'(\gamma(x))|t^2|\psi(x)|,
\end{equation}
where $\gamma$ is a small perturbation of $x$ defined in the proof. The same result holds for $\xi =\kappa_{a,h}(Nx)$ or $\xi=\zeta_{a,h}(Nx)$ for $a \in 2\Omega$. 
Furthermore, for the rescaled rest function, if $N \geq N_c$
\begin{equation}\label{rescaled tau bound}
|\tau_t(x)| \lesssim \begin{cases}
 \frac{t^2}{L}\max\left(\|\theta\|_{C^1}, \|\theta\|_{C^2}^2\right)
 & \text{if }x \in 2\supp(\xi); \\
 \frac{t^2L\|\theta\|_{C^2}^2}{|x-z|^2} & \text{otherwise}.
 \end{cases}
\end{equation}
If $\xi =\kappa_{a,h}(Nx)$ or $\xi=\zeta_{a,h}(Nx)$ with $N \geq N_c$ we have
\begin{equation}\label{local laws tau bound}
|\tau_t(x)| \lesssim \begin{cases}
 \frac{t^2}{h^3L}
 & \text{if }x \in [a-2Lh, a+2Lh]; \\
 \frac{t^2L}{h|x-a|^2} & \text{otherwise},
 \end{cases}
\end{equation}
where we have again denoted $L=\frac{1}{N}$.
\end{lem}

We will also need the following easy technical lemma, which we make use of throughout this section and the next.
\begin{lem}[Rough $L^1$ Bound]\label{Rough L^1 Bound}.
Let $\xi \in C^1(U)$ be compactly supported and let $A \subset \B$. Let $\XN$ be a configuration such that the local law (\ref{energy bound - local law}) holds on the blown up set $NA$. Then,
\begin{equation}\label{local rough L^1 bound}
\left|\int_{A}\xi(x)~d\fluct_N(x)\right|\lesssim N|A||\xi|_{L^\infty(A)}.
\end{equation}
Furthermore, if $|A|\geq \mathfrak{c}$ for some fixed constant $\mathfrak{c}$, then 
\begin{equation}\label{global rough L^1 bound}
\left|\int_A \xi(x)~d\fluct_N(x)\right|\lesssim_{\mathfrak{c}} N \int_A|\xi(x)|~dx.
\end{equation}
for all $\XN \in \mathcal{G}_\mathsf{M}$, where $\mathcal{G}_{\mathsf{M}}$ is the event on which the macroscopic local law (\ref{energy bound - local law}) holds, and
\begin{equation}
\PNbeta(\mathcal{G}_{\mathsf{M}}^c)\leq C_1e^{-C_2\beta N}  
\end{equation}
for some constants $C_1$, $C_2$ dependent only on $\muv$.
\end{lem}
\begin{proof}
The first item is a direct result of the discrepancy control Lemma \ref{discrepancy estimate}; letting $D_A$ denote $\int_{NA} \fluct_N(x)$, Lemma \ref{discrepancy estimate} yields
\begin{equation*}
D_A^2 \min \left(1, \sqrt{\frac{D_A}{N|A|}}\right)\lesssim N|A|.
\end{equation*}
Rearranging yields $|D_A|\lesssim \max (\sqrt{N|A|}, (N|A|)^{3/5})$, which implies (\ref{local rough L^1 bound}).

For the second item, we split $\B$ into roughly $\frac{1}{\mathfrak{c}}$ pieces of size $\mathfrak{c}$, and use the local law control (\ref{energy bound - local law}) at macroscales on $\mathcal{G}_{\mathsf{M}}$ coupled with Lemma \ref{discrepancy estimate} as previously to show that the number of points in each interval of size $\mathfrak{c}$ is $\lesssim N\mathfrak{c}$. We note that $\PNbeta(\mathcal{G}_{\mathsf{M}}^c)\leq C_1e^{-C_2\beta N}$ from the expansion of partition functions prove in \cite[Theorem 6]{SS15-1} (cf \cite[Lemma 2.3]{BLS18}). This allows us to obtain the control (\ref{global rough L^1 bound}) on $ \mathcal{G}_{\mathsf{M}}$ with constant dependent only on $\mathfrak{c}$, as desired.
\end{proof}
This allows us to obtain the following error estimate.

\begin{prop}\label{easy Error3 bound} 
Suppose that $\theta\in C^3$ is a compactly supported test function, and let $\xi_{z,L}$ denote the associated rescaled test function at scale $L$ at least microscopic with $z \in \B$. Then, for every $N\geq N_c$ and $\XN \in \mathcal{G}_{\mathsf{M}}$,
\begin{equation*}
|\textsf{Error}_3|\lesssim  \left|1-\frac{1}{\beta}\right| |t|N\|\theta\|_{C^3}+N^2t^2 \max\left(\|\theta\|_{C^1}, \|\theta\|_{C^2}^2\right).
\end{equation*}
Suppose that $\xi =\kappa_{a,h}(Nx)$ or $\xi=\zeta_{a,h}(Nx)$ for $a \in 2\Omega$. Then, for all $N \geq N_c$,
\begin{equation*}
|\textsf{Error}_3|\lesssim \left|1-\frac{1}{\beta}\right|\frac{N|t|}{h}+\frac{N^2t^2}{h^2}.
\end{equation*}
\end{prop}
\begin{proof}
We control the fluctuations of $\tau_t$ and $\log \phi_t'$ separately. Let's start with $\tau_t$ in the case where $\xi=\kappa_{a,h}(Nx)$ or $\xi=\zeta_{a,h}(Nx)$. Then, again letting $J$ denote the bulk of the support (an interval of size $2Lh$), we have using Lemmas \ref{Energy Difference Estimate} and \ref{Rough L^1 Bound}
\begin{align*}
\left|N\int_{J} \tau_t(x) \left(\sum_{i=1}^N \delta_{x_i}-N\muv\right)(x)\right|&\lesssim N^2|J|\|\tau_t\|_{L^\infty} \lesssim \frac{N^2t^2}{h^2}.
\end{align*}
Next, outside of $J$, we have 
\begin{align*}
\left|N\int_{J^c} \tau_t(x) \left(\sum_{i=1}^N \delta_{x_i}-N\muv\right)(x)\right| &\lesssim N^2 t^2 \int_{J^c}\frac{L}{h|x-a|^2}~dx \lesssim \frac{N^2t^2}{h^2}
\end{align*}
too. For $\xi=\xi_{z,L}$ a rescaled test function, the same computation with Lemmas \ref{Energy Difference Estimate} and \ref{Rough L^1 Bound} and $J=2\supp(\xi)$ arrives at the bound
\begin{align*}
\left|\Fluct_N(N\tau_t)\right|\lesssim N^2t^2 \max\left(\|\theta\|_{C^1}, \|\theta\|_{C^2}^2\right).
\end{align*}
To control the fluctuations of $\log \phi_t'$, recall from the proof of Proposition \ref{easy Error1 bound} that $\|t\psi'\|_{L^\infty}<\frac{1}{2}$ tells us that $|\log \phi_t'|\lesssim t|\psi'(x)|$. Thus, using Lemma \ref{Rough L^1 Bound} to control the fluctuations we arrive at
\begin{align*}
\left|\int \log (1+t\psi'(x))\left(\sum_{i}\delta_{x_i}-N\muv\right)(x)\right|\lesssim |t|N\|\psi'\|_{L^1}.
\end{align*}
Using Lemmas \ref{Local Laws Transport Estimates} and \ref{Rescaled Transport Estimates}, we immediately conclude that 
\begin{equation*}
\left|\Fluct_N(\log \phi_t')\right|\lesssim \frac{|t|N}{h}
\end{equation*}
for $\xi=\kappa_{a,h}(Nx)$ or $\xi=\zeta_{a,h}(Nx)$, and 
\begin{equation*}
\left|\Fluct_N(\log \phi_t')\right|\lesssim |t|N\|\theta\|_{C^3}
\end{equation*}
for $\xi=\xi_{z,L}$ a rescaled test function. Coupling the above estimates with the definition of $\textsf{Error}_3$ (\ref{Error3}) yields the proposition.
\end{proof}
The estimate for $\textsf{Error}_2$ is more delicate, and requires us to make use of the additional commutator structure observed in \cite{NRS21}. We will appeal to the following local commutator estimate from \cite{RS22}, which we've modified and restated here to fit our purposes. It is a blown-up version of \cite[Theorem 1.1]{RS22}, and provides a localized version of the commutator estimate established in \cite{S20} and \cite{NRS21}.
\begin{lem}[Rosenzweig, Serfaty '22]\label{commutator estimate}
Let $v \in C^1(\Omega)$, with $\Omega \subset \R$ containing a $\frac{1}{N}$ neighborhood of $\supp(v)$. Let $\mu$ be the blown up equilibrium measure of Section \ref{Main Bootstrap}. Then,
\begin{equation*}
\left|\int_{\R^2}\frac{v(x)-v(y)}{x-y}~d\left(\sum_{i=1}^N \delta_{x_i}-N\muv \right)^{\otimes 2}(x,y)\right| \lesssim \|v'\|_{L^\infty(\Omega)}\left(\G^{N\Omega}(\XN,\mu)+C_0 \#(\{X_N\} \cap N \Omega )\right).
\end{equation*}
\end{lem}
Coupling this with Proposition \ref{Main Bootstrap}, we are able to prove the following boundedness with high probability.
\begin{prop}\label{easy Error2 bound - local laws}
Suppose that $\xi=\kappa_{a,h}(Nx)$ or $\xi=\zeta_{a,h}(Nx)$ for some $a \in 2\Omega$. Then, on the event $\mathcal{G}'_L$ of Proposition \ref{Main Bootstrap} we have for all $N \geq N_c$,
\begin{equation*}
\int-\log \left|\frac{\phi_t(x)-\phi_t(y)}{x-y}\right|~d\fluct_N(x)d\fluct_N(y) \lesssim \frac{|t|N}{h}+\frac{t^2N^2}{h^2}.
\end{equation*}
\end{prop}
\begin{proof}
Without loss of generality, let us take $a=0$. We first use a Taylor expansion, writing
\begin{align*}
-\log \left|\frac{\phi_t(x)-\phi_t(y)}{x-y}\right|=-t\frac{\psi(x)-\psi(y)}{x-y}+\frac{t^2}{2(1+\beta(x,y))^2}\left(\frac{\psi(x)-\psi(y)}{x-y}\right)^2,
\end{align*}
with $|\beta(x,y)|< \|t\psi'\|_{L^\infty}<\frac{1}{2}$. Let us first tackle $-t\frac{\psi(x)-\psi(y)}{x-y}$. The idea is to make use of the commutator estimate Lemma \ref{commutator estimate} and the energy controls of Proposition \ref{Main Bootstrap} in such a way that respects the decay of the transport map as was done in the proof of Proposition \ref{tailored screening}. Towards this end, let $\alpha$ be a fixed positive number such that $[-\alpha ,\alpha ]$ is contained in the bulk of $\Sigma_V$, and choose $k_*$ such that $\frac{2^{k_*}L'}{N}\leq \alpha<\frac{2^{k_*+1}L'}{N}$. Notice here that $L'$ is the length scale of $\S 3$. We introduce a partition of unity $\{\zeta_k\}_{k=1}^{k_*+1}$ satisfying 
\begin{align*}
\supp(\zeta_k)& \subset \square_{\frac{2^kL'}{N}}\setminus \square_{\frac{2^{k-1}L'}{N}} \\
|\psi \zeta_k|_{C^m\left( \square_{\frac{2^kL'}{N}}\setminus \square_{\frac{2^{k-1}L'}{N}}\right)}&\lesssim |\psi |_{C^m\left( \square_{\frac{2^kL'}{N}}\setminus \square_{\frac{2^{k-1}L'}{N}}\right)}
\end{align*}
for $2\leq k \leq k_*$, where $\square_T$ denotes the interval $[-T,T]$. We suppose the same holds for $\zeta_1$ with $\square_{\frac{2L'}{N}}\setminus \square_{\frac{L'}{N}}$ replaced by $\square_{\frac{2L'}{N}}$, and for $\zeta_{k_*+1}$ on the rest of $U$. First we observe that
\begin{equation*}
|t|\left|\int_{\R^2}\frac{\psi(x)-\psi(y)}{x-y}~d\left(\sum_{i=1}^N \delta_{x_i}-N\muv \right)^{\otimes 2}(x,y)\right|\leq A+B+C
\end{equation*}
with 
\begin{align*}
A&=|t|\left|\int_{\R^2}\frac{\psi \zeta_1(x)-\psi\zeta_1(y)}{x-y}~d\left(\sum_{i=1}^N \delta_{x_i}-N\muv \right)^{\otimes 2}(x,y)\right| \\
B&=\sum_{k=2}^{k_*}\left|\int_{\R^2}\frac{\psi \zeta_k(x)-\psi\zeta_k(y)}{x-y}~d\left(\sum_{i=1}^N \delta_{x_i}-N\muv \right)^{\otimes 2}(x,y)\right|\\
C&=|t|\left|\int_{\R^2}\frac{\psi \zeta_{k_*+1}(x)-\psi\zeta_{k_*+1}(y)}{x-y}~d\left(\sum_{i=1}^N \delta_{x_i}-N\muv \right)^{\otimes 2}(x,y)\right|.
\end{align*}

\subsection*{Control of A}
Using Lemma \ref{Local Laws Transport Estimates} we have that 
\begin{equation*}
\|(\psi \zeta_1)'\|_{L^\infty\left(\square_{\frac{2L'}{N}}\right)}\leq \|\psi'\|_{L^\infty\left(\square_{\frac{2L'}{N}}\right)} \lesssim \frac{N}{h^2}.
\end{equation*}
Applying Lemma \ref{commutator estimate} and Proposition \ref{Main Bootstrap} then we have on $\mathcal{G}_L'$ that 
\begin{equation}\label{control of A anisotropy}
|A|\lesssim |t| \|\psi'\|_{L^\infty\left(\square_{\frac{2L'}{N}}\right)}\left(\G^{\square_{2L'}}(\XN,\mu)+C_0 \#(\{X_N\} \cap \square_{2L'} )\right)\lesssim \frac{\mathcal{C}L'}{h^2}|t|N.
\end{equation}

\subsection*{Control of B}.
For $2 \leq k \leq k_*$, we have
\begin{equation*}
\|(\psi \zeta_k)'\|_{L^\infty\left( \square_{\frac{2^kL'}{N}}\setminus \square_{\frac{2^{k-1}L'}{N}}\right)}\leq\|\psi'\|_{L^\infty\left( \square_{\frac{2^kL'}{N}}\setminus \square_{\frac{2^{k-1}L'}{N}}\right)} \lesssim \frac{1}{N}\left(\frac{N}{2^kL'}+\frac{N^2}{2^{2k}(L')^2}\right).
\end{equation*}
Then, using Lemma \ref{commutator estimate} and Proposition \ref{Main Bootstrap} we have on $\mathcal{G}_L'$ that 
\begin{align}
\nonumber|B|&\lesssim |t|\sum_{k=2}^{k_*}\|(\psi \zeta_k)'\|_{L^\infty\left( \square_{\frac{2^kL'}{N}}\setminus \square_{\frac{2^{k-1}L'}{N}}\right)}\left(\G^{\square_{2^kL'}}(\XN,\mu)+C_0 \#(\{X_N\} \cap \square_{2^kL'} )\right) \\
\label{control of B anisotropy}&\lesssim |t|\sum_{k=2}^{k_*}\frac{1}{N} \left(\frac{N}{2^kL'}+\frac{N^2}{(L')^22^{2k}}\right)(\mathcal{C}2^{k+1}L') \lesssim \mathcal{C}|t|\log \left(\frac{N}{L'}\right)+\frac{\mathcal{C}|t|N}{L'}
\end{align}
since $k_*\sim \log \left(\frac{N}{L'}\right)$.

\subsection*{Control of C}.
Once more using Lemma \ref{Local Laws Transport Estimates}, we have that 
\begin{equation*}
\|(\psi \zeta_{k_*+1})'\|_{L^\infty\left(\square_{\frac{2^{k_*}L'}{N}}^c\right)}\leq \|\psi'\|_{L^\infty\left(\square_{\frac{2^{k_*}L'}{N}}^c\right)} \lesssim \frac{1}{N}.
\end{equation*}
Thus, using Lemma \ref{commutator estimate} and Lemma \ref{global bound on partition functions} we have
\begin{equation}\label{control of C anisotropy}
|C|\lesssim |t\||(\psi \zeta_{k_*+1})'\|_{L^\infty\left(\square_{\frac{2^{k_*}L'}{N}}^c\right)}\left(\G(\XN,\mu)+C_0 N )\right)\lesssim |t|.
\end{equation}

\subsection*{Conclusion}
Coupling (\ref{control of A anisotropy}), (\ref{control of B anisotropy}), (\ref{control of C anisotropy}), we find 
\begin{equation}\label{control of anisotropy}
|t|\left|\int_{\R^2}\frac{\psi(x)-\psi(y)}{x-y}~d\left(\sum_{i=1}^N \delta_{x_i}-N\muv \right)^{\otimes 2}(x,y)\right|\lesssim |t|\left(\frac{\mathcal{C}L'}{h^2}N+\mathcal{C}\log \left(\frac{N}{L'}\right)+\frac{\mathcal{C}N}{L'}+1\right)
\end{equation}
Recalling that $h=\frac{L'}{K}$, the right hand side above is simply $\frac{|t|N}{h}$. Finally, using (\ref{uniform square integral bound - local laws}) and Lemma \ref{Rough L^1 Bound} on $\mathcal{G}_\mathsf{M}$ we obtain 
\begin{equation}\label{control of anisotropy error}
\left|\int \frac{t^2}{2(1+\beta(x,y))^2}\left(\frac{\psi(x)-\psi(y)}{x-y}\right)^2 ~d\fluct_N(x)d\fluct_N(y)\right|\lesssim \frac{t^2N^2}{h^2}.
\end{equation}
Combining (\ref{control of anisotropy}) and (\ref{control of anisotropy error}) yields the desired estimate.
\end{proof}

\subsection{Proof of Theorem \ref{Uniform Bound on Fluctuations}} 
Coupling Propositions \ref{Main Term Expansion}, \ref{easy Error1 bound}, \ref{easy Error3 bound}, and \ref{easy Error2 bound - local laws} with Proposition \ref{Expansion of Laplace Transform} yields
\begin{equation*}
\Esp_{\PNbeta}\left[\exp \left(s\Fluct_N(\xi)\right)\indic_{\mathcal{G}'_L}\right]= \exp \left(\frac{s^2}{\beta}\|\xi\|_{H^{1/2}}^2+O\left( \frac{|s|}{h}+\left|1-\frac{1}{\beta}\right|\frac{|s|}{ h}+ \frac{s^2}{\beta h^2}+\frac{|s|^3}{\beta^2 h^4}\right)\right)
\end{equation*}
for $\xi=\kappa_{a,h}(Nx)$ or $\xi=\zeta_{a,h}(Nx)$. Computing directly,
\begin{equation*}
\|\xi\|_{H^{1/2}}^2 \lesssim \frac{1}{h^2}.
\end{equation*}
So, 
\begin{equation*}
\left|\log \Esp_{\PNbeta}\left[\exp \left(s\Fluct_N(\xi)\right)\indic_{\mathcal{G}'_L}\right]\right|\lesssim \frac{|s|}{h}+\left|1-\frac{1}{\beta}\right|\frac{|s|}{ h}+ \frac{s^2}{\beta h^2}+\frac{|s|^3}{\beta^2 h^4}
\end{equation*}
which establishes Theorem \ref{Uniform Bound on Fluctuations} for $\xi=\kappa_{a,h}(Nx)$ or $\xi=\zeta_{a,h}(Nx)$.

Having established Theorem \ref{Uniform Bound on Fluctuations} for $\xi=\kappa_{a,h}(Nx)$ or $\xi=\zeta_{a,h}(Nx)$, we are free to run the machinery of Section \ref{Main Bootstrap} to prove the local law of Theorem \ref{Local Law}. With this estimate, we are in a position to finish the proof of Theorem \ref{Uniform Bound on Fluctuations} for any rescaled test function at scale $L>\omega$. The approach is analogous to the proof of Proposition \ref{easy Error2 bound - local laws} and is the content of the following proposition.
 
\begin{prop}\label{easy Error2 bound - rescaled test function}
Suppose that $\theta\in C^3$ is a compactly supported test function, and let $\xi_{z,L}$ denote the associated rescaled test function at scale $L>\frac{\omega}{N}$ with $z \in \B$. Then, for every $N\geq N_c$, there is an event $\mathcal{G}_N$ such that 
\begin{equation*}
\int-\log \left|\frac{\phi_t(x)-\phi_t(y)}{x-y}\right|~d\fluct_N(x)d\fluct_N(y) \lesssim \frac{|s|}{\beta}\|\theta\|_{C^3}+\frac{s^2}{\beta^2}\|\theta\|_{C^3}^2
\end{equation*}
on $\mathcal{G}_N$ and 
\begin{equation*}
\PNbeta(\mathcal{G}_N^c)\leq C_1e^{-C_2 \beta LN}.
\end{equation*}
\end{prop}
\begin{proof}
The proof is exactly as in that of Proposition \ref{easy Error2 bound - local laws}. We first use a Taylor expansion again, writing
\begin{align*}
-\log \left|\frac{\phi_t(x)-\phi_t(y)}{x-y}\right|=-t\frac{\psi(x)-\psi(y)}{x-y}+\frac{t^2}{2(1+\beta(x,y))^2}\left(\frac{\psi(x)-\psi(y)}{x-y}\right)^2,
\end{align*}
with $|\beta(x,y)|< \|t\psi'\|_{L^\infty}<\frac{1}{2}$ and approach $-t\frac{\psi(x)-\psi(y)}{x-y}$ using Lemma \ref{commutator estimate} and Proposition \ref{Main Bootstrap}. Construct the same partition of unity as in the proof of Proposition \ref{easy Error2 bound - local laws}, with $L'=2|\supp(\xi)|N\lesssim LN$ (assuming $z=0$; the approach is analogous for other $z \in \B$). We again obtain
\begin{equation*}
|t|\left|\int_{\R^2}\frac{\psi(x)-\psi(y)}{x-y}~d\left(\sum_{i=1}^N \delta_{x_i}-N\muv \right)^{\otimes 2}(x,y)\right|\leq A+B+C
\end{equation*}
with 
\begin{align*}
A&=|t|\left|\int_{\R^2}\frac{\psi \zeta_1(x)-\psi\zeta_1(y)}{x-y}~d\left(\sum_{i=1}^N \delta_{x_i}-N\muv \right)^{\otimes 2}(x,y)\right| \\
B&=\sum_{k=2}^{k_*}\left|\int_{\R^2}\frac{\psi \zeta_k(x)-\psi\zeta_k(y)}{x-y}~d\left(\sum_{i=1}^N \delta_{x_i}-N\muv \right)^{\otimes 2}(x,y)\right|\\
C&=|t|\left|\int_{\R^2}\frac{\psi \zeta_{k_*+1}(x)-\psi\zeta_{k_*+1}(y)}{x-y}~d\left(\sum_{i=1}^N \delta_{x_i}-N\muv \right)^{\otimes 2}(x,y)\right|.
\end{align*}
and find ourselves needing to control $A$, $B$ and $C$. Using Lemma \ref{Rescaled Transport Estimates}, we find with the same approach as in Proposition \ref{easy Error2 bound - local laws} coupled with Theorem \ref{Local Law} that 
\begin{align*}
|A|&\lesssim |t|N \|\theta\|_{C^3} \\
|B|&\lesssim |t|L\log \left(\frac{1}{L}\right)\|\theta\|_{C^3}\\
|C|&\lesssim |t|LN\|\theta\|_{C^3}
\end{align*}
on event $\mathcal{G}_N$ with 
\begin{equation*}
\PNbeta( \mathcal{G}_N^c)\leq C_1\sum_{k=1}^{k_*+1}e^{-C_2\beta 2^kLN}\leq C_1e^{-C_2\beta LN},
\end{equation*}
up to adjusting the definitions of $C_1$ and $C_2$. This immediately tells us that on $\mathcal{G}_N$ we have
\begin{equation*}
|t|\left|\int_{\R^2}\frac{\psi(x)-\psi(y)}{x-y}~d\left(\sum_{i=1}^N \delta_{x_i}-N\muv \right)^{\otimes 2}(x,y)\right|\leq A+B+C \lesssim |t|N\|\theta\|_{C^3}.
\end{equation*}
Finally, using (\ref{uniform square integral bound}) and Lemma \ref{Rough L^1 Bound} we obtain 
\begin{equation}\label{control of anisotropy error}
\left|\int \frac{t^2}{2(1+\beta(x,y))^2}\left(\frac{\psi(x)-\psi(y)}{x-y}\right)^2 ~d\fluct_N(x)d\fluct_N(y)\right|\lesssim t^2N^2 \|\theta\|_{C^3}^2.
\end{equation}
Combining (\ref{control of anisotropy}) and (\ref{control of anisotropy error}) yields the desired estimate.
\end{proof}
Coupling this with Propositions \ref{Main Term Expansion}, \ref{easy Error1 bound}, \ref{easy Error3 bound}, and Proposition \ref{Expansion of Laplace Transform} yields
\begin{align*}
&\left|\log \Esp_{\PNbeta}\left[\exp \left(s\Fluct_N(\xi)\right)\indic_{\mathcal{G}_N}\right] \right|= \frac{s^2}{\beta}\|\xi\|_{H^{1/2}}^2\\
&+O\left(|s|\|\theta\|_{C^3}+|s|\left|1-\frac{1}{\beta}\right|\|\theta\|_{C^3}+\frac{s^2}{\beta}\max \left(\|\theta\|_{C^1},\|\theta\|_{C^3}^2\right)+\frac{|s|^3}{\beta^2 LN}\max\left(\|\theta\|_{C^2}^2,\|\theta\|_{C^3}^3\right)\right)
\end{align*}
which completes the proof of Theorem \ref{Uniform Bound on Fluctuations}.

\section{CLT Estimate}\label{Section CLT Estimate}
In this section, we upgrade the estimates of the previous section to prove the following CLT for fluctuations of mesoscopic linear statistics, using the fact that the local law of Theorem \ref{Local Law} have now been proven down to the minimal scale. 

\begin{theo}\label{Gaussian asymptotics}
Suppose that $\theta\in C^{13}$ is a compactly supported test function, and let $\xi_{z,L}$ denote the associated rescaled test function at scale $L>\frac{\omega}{N}$ with $z \in \B$. Let $\mathcal{G}_N$ be the good event of Theorem \ref{Uniform Bound on Fluctuations}. Then, for every $N\geq N_c$ there is a good event $\mathcal{G}_N'\subset \mathcal{G}_N$ such that
\begin{align*}
\log \mathbb{E}_{\mathbb{P}_{N,\beta}}[\exp(s\Fluct_N(\xi))\indic_{\mathcal{G}_N'}]&=\frac{s^2}{\beta}\|\xi\|_{H^{1/2}}^2 + s\left(1-\frac{1}{\beta}\right)\int\psi'(x)~d\muv(x) \\ &+O_\beta\left(\frac{|s|}{(LN)^{1/4}}\|\theta\|_{C^{13}}+\frac{s^2}{(LN)^{1/4}}\|\theta\|_{C^5}^2+\frac{|s|^3}{LN}\|\theta\|_{C^3}^3  \right),
\end{align*}
with 
\begin{equation*}
\PNbeta(\mathcal{G}_N \setminus \mathcal{G}_N') \lesssim \frac{1}{(LN)^{1/4}}.
\end{equation*}
\end{theo}
The above estimate yields the desired CLT.
\begin{coro}
Suppose that $\theta\in C^{13}$ is a compactly supported test function, and let $\xi_{z,L}$ denote the associated rescaled test function at scale $\frac{\omega}{N} < L \ll 1$ with $z \in \B$. Then,
\begin{equation*}
\Fluct_N(\xi_{z,L}) \implies \mathcal{N}\left(0,  \frac{2}{\beta}\|\xi\|_{H^{1/2}}^2\right),
\end{equation*}
where the above convergence is in distribution. When $L$ is a fixed macroscopic constant, 
\begin{equation*}
\Fluct_N(\xi_{z,L}) \implies \mathcal{N}\left(\left(1-\frac{1}{\beta}\right)\int\psi'(x)~d\muv(x), \frac{2}{\beta}\|\xi\|_{H^{1/2}}^2\right).
\end{equation*}
\end{coro}
Notice that the macroscopic case is \cite[Theorem 1]{BLS18} in the nonsingular case.
\begin{proof}
Since we can only prove the boundedness of Theorem \ref{Gaussian asymptotics} on a good event, we cannot immediately conclude. However, we can make use of a trick to handle this situation that was used in \cite[Section 6]{L21}. The idea is as follows. Let $X_N$ denote the random variable 
\begin{equation*}
X_N:=\Fluct_N(\xi)\indic_{\mathcal{G}_N'}.
\end{equation*}
When $L\ll 1$, notice that 
\begin{equation*}
\left|\left(1-\frac{1}{\beta}\right)\int\psi'(x)~d\muv(x)\right|\lesssim_\beta(|s|L+|s|L\ln L)\|\theta\|_{C^2} \rightarrow 0
\end{equation*}
as $N \rightarrow \infty$ by (\ref{Error1 asymptotic}) below. Thus, we have that uniformly on compact sets of $s$
\begin{equation*}
\lim_{N \rightarrow \infty}\mathbb{E}_{\mathbb{P}_{N,\beta}}[\exp(s\Fluct_N(\xi))\indic_{\mathcal{G}_N'}]=\exp \left(\frac{s^2}{\beta}\|\xi\|_{H^{1/2}}^2\right).
\end{equation*}
Then, since
\begin{align*}
\Esp_{\PNbeta}\left[\exp \left(s X_N \right)\right]=\Esp_{\PNbeta}[\exp(s\Fluct_N(\xi))\indic_{\mathcal{G}_N'}]+\PNbeta((\mathcal{G}_N')^c),
\end{align*}
we can also conclude that 
\begin{equation*}
\lim_{N \rightarrow \infty}\Esp_{\PNbeta}\left[\exp \left(s X_N \right)\right]=\exp \left(\frac{s^2}{\beta}\|\xi\|_{H^{1/2}}^2\right),
\end{equation*}
and hence that $X_N$ converges in distribution to a centered Gaussian random variable with variance $ \frac{2}{\beta}\|\xi\|_{H^{1/2}}^2$. Since $\PNbeta((\mathcal{G}_N')^c)$ vanishes asymptotically, we can also conclude that $\Fluct_N(\xi_{z,L})$ converges in distribution to the same random variable, as desired. 

When $L$ is a macroscopic constant, we have uniformly on compact sets of $s$ that
\begin{equation*}
\lim_{N \rightarrow \infty}\mathbb{E}_{\mathbb{P}_{N,\beta}}[\exp(s\Fluct_N(\xi))\indic_{\mathcal{G}_N'}]=\exp \left(s\left(1-\frac{1}{\beta}\right)\int\psi'(x)~d\muv(x)+\frac{s^2}{\beta}\|\xi\|_{H^{1/2}}^2\right).
\end{equation*}
Then, exactly as above we can conclude that $X_N$ and $\Fluct_N(\xi_{z,L})$ converge in distribution to 
\begin{equation*}
 \mathcal{N}\left(\left(1-\frac{1}{\beta}\right)\int\psi'(x)~d\muv(x), \frac{2}{\beta}\|\xi\|_{H^{1/2}}^2\right),
\end{equation*}
recovering \cite[Theorem 1]{BLS18} in the nonsingular case.
\end{proof}

We proceed with the proof of Theorem \ref{Gaussian asymptotics}. As a result of Proposition \ref{Expansion of Laplace Transform} and Proposition \ref{Main Term Expansion}, we need only upgrade the controls on the $\textsf{Error}_i$ terms defined in (\ref{Error1}-\ref{Error3}). This will require use of the fluctuation control Lemma \ref{Local Laws Energy Estimate} alongside of Theorem \ref{Local Law}. We restate Lemma \ref{Local Laws Energy Estimate} here in traditional coordinates.


\begin{lem}\label{CLT Energy Estimate}
Let $\xi\in C^1(U)$ and let $\Omega$ contain a $\frac{1}{N}$-neighborhood of the support of $\xi$. Suppose $|\Omega| \in [L,2L]$. Then, 
\begin{align*}
\left|\int \xi(x)\left(\sum_{i=1}^N \delta_{x_i}-N\muv(x)dx\right)\right| &\leq \|\xi'\|_{L^\infty(\Omega)}\left[|\Omega|+N^{-3/4}|\Omega|^{1/4}\|\nabla u_{\rrc}\|_{L^2(N\Omega)\times [-NL,NL]}\right]+\\
&C\left[\sqrt{\mathfrak{h}}\|\xi'\|_{L^2(\Omega)}+\frac{1}{\sqrt{\mathfrak{h}}}\|\xi\|_{L^2(\Omega)}\right]\|\nabla u_{\rrc}\|_{L^2(N\Omega \times [-NL,NL])},
\end{align*}
where $\mathfrak{h} \leq NL-1$ is a free parameter.
\end{lem} 
For ease of notation going forward, we will denote $N \Omega \times [-NL,NL]$ by $\tilde{N\Omega}$.

\subsection{$\textsf{Error}_1$ Estimate}
We first upgrade Proposition \ref{easy Error1 bound} using Lemma \ref{Rescaled Transport Estimates}.
\begin{prop}\label{vanishing Error1 bound}
Suppose that $\theta\in C^3$ is a compactly supported test function, and let $\xi_{z,L}$ denote the associated rescaled test function at scale $L$ at least microscopic with $z \in \B$. Then, for every $N\geq N_c$,
\begin{equation}\label{Error1 expansion}
N\textsf{Error}_1=N(1-\beta)\int t\psi'(x)~d\muv(x)+O_\beta \left(\frac{s^2}{LN}\|\theta\|_{C^3}^2\right).
\end{equation}
When $L \ll 1$,
\begin{equation}\label{Error1 asymptotic}
\left|N \textsf{Error}_1\right|\lesssim_\beta \frac{s^2}{LN}\|\theta\|_{C^3}^2+ (|s|L+|s|L\ln L)\|\theta\|_{C^2}.
\end{equation}
\end{prop}
\begin{proof}
We use a Taylor expansion
\begin{align*}
\log \left|\phi_t'(x)\right|=t\psi'(x)-\frac{t^2}{2(1+\beta(x))^2}\left(\psi'(x)\right)^2,
\end{align*}
with $|\beta(x)|<\|t\psi'\|_{L^\infty}<\frac{1}{2}$ for $N \geq N_c$. We deal with the error term first, observing from Lemma \ref{Rescaled Transport Estimates} that 
\begin{equation*}
\int (\psi'(x))^2~dx\lesssim \frac{1}{L}\|\theta\|_{C^3}^2.
\end{equation*}
This allows us to quickly bound
\begin{equation*}
\left|N \left(1-\beta\right)\int\frac{t^2}{2(1+\beta(x))^2}\left(\psi'(x)\right)^2~d\muv(x)\right|\lesssim |1-\beta |\frac{t^2N}{L}\|\theta\|_{C^3}^2\lesssim \left|\frac{1}{\beta}-\frac{1}{\beta^2}\right| \frac{s^2}{LN}\|\theta\|_{C^3}^2,
\end{equation*}
establishing (\ref{Error1 expansion}). For the main term, we can integrate by parts. Since $\psi\muv(x)$ vanishes on the boundary of $\supp(\muv)$ we can write
\begin{equation*}
t\int \psi'(x)~d\muv=-t\int \psi(x)\muv'(x)~dx.
\end{equation*}
Letting $J$ denote $\supp(\xi)$, we have
\begin{align*}
\left|t\int \psi(x)\muv'(x)~dx\right|& \leq |t|\int_J|\psi(x)\|\muv'(x)|~dx+|t|\int_{J^c}|\psi(x)\|\muv'(x)|~dx \\
&\lesssim (|t|L+|t|L\ln L)\|\theta\|_{C^2},
\end{align*}
where the last line follows from Lemma \ref{Rescaled Transport Estimates} and the fact that $\muv'(x)$ has integrable $1/\sqrt{x}$ singularities at the boundary of $\supp(\muv)$. It follows that 
\begin{equation*}
\left|tN\left(1-\beta\right)\int \psi(x)\muv'(x)~dx\right|\lesssim |1-\beta| |t|LN(1+\ln L)\|\theta\|_{C^2} \lesssim (|s|L+|s|L\ln L)\left|1-\frac{1}{\beta}\right|\|\theta\|_{C^2}.
\end{equation*}
Thus,
\begin{equation*}
\left|N \textsf{Error}_1\right|\lesssim \frac{s^2}{LN}\left|\frac{1}{\beta}-\frac{1}{\beta^2}\right|\|\theta\|_{C^3}^2+ (|s|L+|s|L\ln L)\left|1-\frac{1}{\beta}\right|\|\theta\|_{C^2}.
\end{equation*}
\end{proof}

\subsection{$\textsf{Error}_3$ Estimate}
Careful consideration of Lemma \ref{CLT Energy Estimate} and Theorem \ref{Local Law} also allows us to quickly upgrade Proposition \ref{easy Error3 bound}. We will show that the fluctuations of $\tau_t$ and $\log \phi_t'$ vanish uniformly on a large set. In addition to Lemmas \ref{Rescaled Transport Estimates} and \ref{Energy Difference Estimate}, we will also need the following control on $\tau_t'$ that holds for mesoscopic test functions. For a proof, see Appendix C.

\begin{lem}\label{Energy Difference Derivative Estimate}
Suppose that $\theta\in C^4$ is a compactly supported test function, and let $\xi_{z,L}$ denote the associated rescaled test function at scale $L$ at least microscopic with $z \in \B$. Then, for every $N\geq N_c$,
\begin{equation*}
\|\tau_t'\|_{L^\infty}\lesssim \frac{t^2\|\theta\|_{C^4}^2}{L^2}.
\end{equation*}
\end{lem}

Using this, we obtain the following.
\begin{prop}\label{vanishing Error3 bound}
Suppose that $\theta\in C^4$ is a compactly supported test function, and let $\xi_{z,L}$ denote the associated rescaled test function at scale $L>\frac{\omega}{N}$ with $z \in \B$. Then, for every $N\geq N_c$ there is an event $\mathcal{G}_N'\subset \mathcal{G}_{\mathsf{M}}$ such that 
\begin{equation}\label{vanishing Error3}
\left|\textsf{Error}_3\right|\lesssim \left(\frac{|s|}{\beta}\|\theta\|_{C^4}\frac{1}{(LN)^{1/4}}+\frac{s^2}{\beta^2 (LN)^{1/4}}\|\theta\|_{C^4}^2\right)\left(1+\left|1-\frac{1}{\beta}\right|\right)
\end{equation}
on $\mathcal{G}_N'$. If $L<N^{-1/5}$ asymptotically then 
\begin{equation*}
\PNbeta(\mathcal{G}_{\mathsf{M}} \setminus \mathcal{G}_N')\leq C_1e^{-C_2\beta (LN)^{5/4}}
\end{equation*} 
for some constants $C_1$ and $C_2$, independent of $N$. Otherwise, 
\begin{equation*}
\PNbeta(\mathcal{G}_{\mathsf{M}}\setminus \mathcal{G}_N')\leq C_1e^{-C_2\beta N}.
\end{equation*}
\end{prop}
\begin{proof}
Recall that 
\begin{equation*}
\textsf{Error}_3=\Fluct_N \left(\left(1-\frac{1}{\beta} \right)\log \phi_t'+N\tau_t \right).
\end{equation*}
We consider the fluctuations of $\log \phi_t'$ and $\tau_t$ separately.
\subsection*{$\Fluct_N(\log \phi_t')$ control}
Let's first consider the fluctuations of $\log \phi_t'$. Since we want to use an energy bound on the fluctuation measure, which is not a positive measure, we can't just immediately substitute $|\log \phi_t'|\lesssim |t\psi'|$. Instead, we use a Taylor expansion
\begin{equation*}
\log(1+t\psi'(x))=t\psi'(x)-\frac{1}{2(1+\beta(x))^2}t^2\psi'(x)^2
\end{equation*}
where $|\beta(x)|<\|t\psi'\|_{L^\infty}<\frac{1}{2}$, and estimate the error term.  Let's focus on the main term first. We split the above into the contribution from $I=\supp\left(\xi_{z,\frac{L}{\epsilon}}\right)$ with $\epsilon\rightarrow 0$ to be determined. Suppose first that $\epsilon$ is such that $\frac{L}{\epsilon}\rightarrow 0$, so that $I$ remains in the bulk of $\Sigma_V$. We will estimate on $I$ using Lemmas \ref{CLT Energy Estimate} and \ref{Rescaled Transport Estimates}, and Theorem \ref{Local Law}; the decay off of $I$ we will estimate with a simple $L^1$ bound. On $I$, we have that off of a bad event in $\mathcal{G}_{\mathsf{M}}$ of probability $\leq C_1e^{-\frac{C_2\beta LN}{\epsilon}}$
\begin{align*}
\left|\int_I t\psi'(x)\left(\sum_{i}\delta_{x_i}-N\muv\right)(x)\right| &\lesssim \frac{|s|}{\beta N}\|\psi''\|_{L^\infty(I)}\left(|I|+N^{-3/4}|I|^{1/4}\|\nabla u_{\rrc}\|_{L^2\left(\tilde{NI}\right)}\right) \\
&+\frac{|s|}{\beta N}\left(\sqrt{\mathfrak{h}}\|\psi''\|_{L^2\left(I\right)}+\sqrt{\frac{1}{\mathfrak{h}}}\|\psi'\|_{L^2\left(I\right)}\right)\|\nabla u_{\rrc}\|_{L^2\left(\tilde{NI}\right)}\\
&\lesssim \frac{|s|}{\beta}\left(\frac{1}{LN}\frac{1}{\epsilon}+\frac{1}{(LN)^{5/4}}\frac{1}{\epsilon^{3/4}}+\frac{1}{\sqrt{LN}}\frac{1}{\epsilon}\right)\|\theta\|_{C^4},
\end{align*}
with $\mathfrak{h}=L$. Choosing $\epsilon=\frac{1}{(LN)^{1/4}}$, we find 
\begin{align*}
&\left|\int_I t\psi'(x)\left(\sum_{i}\delta_{x_i}-N\muv\right)(x)\right|\\
&\lesssim \frac{|s|}{\beta}\left(\frac{1}{LN}(LN)^{1/4}+\frac{1}{(LN)^{5/4}}(LN)^{5/16}+\frac{1}{\sqrt{LN}}(LN)^{1/4}\right) \|\theta\|_{C^4} \\
&=\frac{|s|}{\beta (LN)^{1/4}}\|\theta\|_{C^4}
\end{align*}
On $I^c$, we use a simple $L^1$ bound with Lemma \ref{Rough L^1 Bound}:
\begin{align*}
\left|\int_{I^c} t\psi'(x)\left(\sum_{i}\delta_{x_i}-N\muv\right)(x)\right|&\lesssim |t|N\|\theta\|_{C^3} \int_{I^c} L\left(\frac{1}{|x-z|}+\frac{1}{|x-z|^2}\right)~dx \\
&\lesssim\frac{|s|}{\beta}\|\theta\|_{C^3}\epsilon=\frac{|s|}{\beta}\|\theta\|_{C^3}\frac{1}{(LN)^{1/4}}
\end{align*}
since both $L$ and $\epsilon$ tend to zero. We bounded the log using that $L/\epsilon$ can be made less than $1$, and hence $\epsilon/L$ greater than one. This tells us that on the good event $\mathcal{G}'_N$, 
\begin{equation*}
\left|\int t\psi'(x)\left(\sum_{i}\delta_{x_i}-N\muv\right)(x)\right|=\frac{|s|}{\beta}\|\theta\|_{C^4}\frac{1}{(LN)^{1/4}}
\end{equation*}
which is $o_N(1)$, as desired.

It is of course possible at some mesoscales that with our above choice of $\epsilon$, $\frac{L}{\epsilon}$ does not tend to zero. This happens precisely when $(LN)^{5/4}\geq N$, or $L \geq N^{-1/5}$. In that case, simply choose the boundary of $I$ to have radius $\max \left(\frac{L}{\epsilon},\alpha \right)$ where $[z-\alpha,z+\alpha]$ is contained in the bulk $\B$, on which Theorem \ref{Local Law} holds and repeat the above argument. In that case, since $I$ is of order $1$, $\PNbeta(\mathcal{G}_{\mathsf{M}} \setminus \mathcal{G}'_N)\leq C_1e^{-C_2\beta N}$.

It remains to show that the error term also vanishes asymptotically. This can again be done with a simple $L^1$ bound. Once more recalling that $J=\supp(\xi)$ lives at scale $L$, we find with Lemmas \ref{Rescaled Transport Estimates} and \ref{Rough L^1 Bound}
\begin{align*}
\left|\int -\frac{1}{2(1+\beta(x))^2}t^2\psi'(x)^2\left(\sum_{i}\delta_{x_i}-N\muv\right)(x)\right|&\lesssim t^2N \int_J |\psi'(x)|^2~dx+t^2N \int_{J^c}|\psi'(x)|^2~dx \\
&\lesssim\frac{s^2}{\beta^2(LN)}\|\theta\|_{C^3}^2
\end{align*}
Hence, 
\begin{equation}\label{vanishing log phi}
\left|\int \log \phi_t'(x)~d\fluct_N(x)\right|\lesssim\frac{|s|}{\beta}\|\theta\|_{C^4}\left(L+\frac{1}{(LN)^{1/4}}\right)+\frac{s^2}{\beta^2(LN)}\|\theta\|_{C^3}^2.
\end{equation}

\subsection*{$\Fluct_N(N\tau_t)$ control}
We can run the same argument for $\tau_t$ on the same event $\mathcal{G}'_N$ defined in the previous subsection, showing that
\begin{equation}\label{vanishing tau}
N\int \tau_t (x) \left(\sum_{i=1}^N \delta_{x_i}-N\muv\right)(x)= \frac{s^2}{\beta^2}o_N(1)
\end{equation}
uniformly. We estimate $\tau_t$ at scale $L/\epsilon \rightarrow 0$ with $\epsilon=\frac{1}{(LN)^{1/4}}$ using Lemmas \ref{CLT Energy Estimate}, \ref{Energy Difference Estimate}, \ref{Energy Difference Derivative Estimate} and the control of Theorem \ref{Local Law}. Setting $I:=\supp \left(\xi_{z,\frac{2L}{\epsilon}}\right)$, we find on $\mathcal{G}_N$ that
\begin{align*}
\left|N\int_I \tau_t  \left(\sum_{i=1}^N \delta_{x_i}-N\muv\right)(x)\right|&\lesssim N\|\tau_t'\|_{L^\infty \left(I\right)}\left(\frac{L}{\epsilon}+N^{-3/4}\left(\frac{L}{\epsilon}\right)^{1/4}\|\nabla u_{\rrc}\|_{L^2\left(\tilde{NI}\right)}\right) \\
&+N\left(\sqrt{\mathfrak{h}}\|\tau_t'\|_{L^2\left(I\right)}+\sqrt{\frac{1}{\mathfrak{h}}}\|\tau_t\|_{L^2\left(I\right)}\right)\|\nabla u_{\rrc}\|_{L^2\left(\tilde{NI}\right)} \\
&\lesssim\frac{s^2}{\beta^2}\left(\frac{1}{(NL)^{3/4}}+\frac{1}{(NL)^{17/16}}+\frac{1}{(NL)^{1/4}}\right)\|\theta\|_{C^4}^2 \\
&\lesssim\frac{s^2}{\beta^2(LN)^{1/4}}\|\theta\|_{C^4}^2
\end{align*}
choosing $\mathfrak{h}=L$. 
Outside of $I$, we use a simple $L^1$ bound and Lemmas \ref{Energy Difference Estimate} and \ref{Rough L^1 Bound} to find 
\begin{align*}
\left|N\int_{I^c} \tau_t (x) \left(\sum_{i=1}^N \delta_{x_i}-N\muv\right)(x)\right| \lesssim N^2t^2 L\|\theta\|_{C^2}^2\int_{I^c}\frac{1}{|x|^2}~dx\lesssim \frac{s^2}{\beta^2}\frac{\|\theta\|_{C^4}^2}{(LN)^{1/4}}.
\end{align*}
This establishes (\ref{vanishing Error3}) on $\mathcal{G}'_N$. It is again of course possible at some mesoscales that with our above choice of $\epsilon$, $\frac{L}{\epsilon}$ does not tend to zero. In that case, simply choose the boundary of $I$ again to be $\max \left(\frac{L}{\epsilon},\alpha \right)$ where $[-\alpha,\alpha]$ denotes the bulk on which Theorem \ref{Local Law} holds and repeat the above argument, with the same modifications to $\mathcal{G}'_N$ as in the control on $\log \phi_t'$.
\end{proof}

\subsection{$\textsf{Error}_2$ Estimate}
Finally, we upgrade the estimate of $\textsf{Error}_2$. Much like in Section \ref{Uniform Bound on Fluctuations}, we need a stronger energy control than just Lemma \ref{CLT Energy Estimate}. We use a Fourier transform trick inspired by Appendix A of \cite{BLS18}. The proof uses the following Corollary to Theorem \ref{Uniform Bound on Fluctuations}:
\begin{coro}\label{pseudo-Gaussian moments}
Suppose that $\theta\in C^4$ is a compactly supported test function, and let $\xi_{z,L}$ denote the associated rescaled test function at scale $L>\frac{\omega}{N}$ with $z \in \B$. Let $\mathcal{G}_N$ be the good event of Theorem \ref{Uniform Bound on Fluctuations}. Then, for all $k \geq 1$ there is some constant $C>0$ dependent only on $k$ such that
\begin{align*}
\Esp \left|\Fluct_N(\xi)\right|^k\indic_{\mathcal{G}_N} &\lesssim \left(\|\theta\|_{C^3}\left(1+\left|1-\frac{1}{\beta}\right|\right)\right)^k+\left(\frac{\|\xi\|_{H^{1/2}}^2+\max(\|\theta\|_{C^1}, \|\theta\|_{C^3}^2)}{\beta}\right)^{\frac{k}{2}} \\
&+\left(\frac{\max(\|\theta\|_{C^2}^2, \|\theta\|_{C^3}^3)}{LN\beta^2}\right)^{\frac{k}{3}}.
\end{align*}
\end{coro}
We have the following proposition.
\begin{prop}\label{vanishing Error2 bound}
Suppose that $\theta\in C^{13}$ is a compactly supported test function, and let $\xi_{z,L}$ denote the associated rescaled test function at scale $L>\frac{\omega}{N}$ with $z \in \B$. Then, for every $N\geq N_c$ there is a good event $\mathcal{G}_N' \subset \mathcal{G}_N$ such that 
\begin{equation*}
 \left|\textsf{Error}_2\right|\lesssim\frac{|s|}{(LN)^{1/4}} \|\theta\|_{C^{13}}+ \frac{s^2}{LN}\|\theta\|_{C^5}^2+\frac{|s|^3}{LN}\|\theta\|_{C^3}^3
 \end{equation*}
on $\mathcal{G}_N'$ and 
\begin{equation*}
\PNbeta(\mathcal{G}_N\setminus \mathcal{G}_N') \lesssim\frac{1}{(LN)^{1/4}}.
\end{equation*}
\end{prop}
\begin{proof}
Recall that $|\supp(\xi)|\sim L$; for computational ease throughout the proof then, we redefine $L=|\supp(\xi)|$. First, a Taylor expansion yields
\begin{align*}
&-\log \left|\frac{\phi_t(x)-\phi_t(y)}{x-y}\right|\\
&=-t\frac{\psi(x)-\psi(y)}{x-y}+\frac{t^2}{2}\left(\frac{\psi(x)-\psi(y)}{x-y}\right)^2-\frac{t^3}{3(1+\beta(x,y))^3}\left(\frac{\psi(x)-\psi(y)}{x-y}\right)^3
\end{align*}
with $|\beta(x,y)|<\left|t\frac{\psi(x)-\psi(y)}{x-y}\right|\leq |t\psi'|<\frac{1}{2}$. We need to expand to third order since as we saw in Proposition \ref{easy Error2 bound - rescaled test function}, the second order term does not clearly vanish uniformly. However, the integral of the third order term can be approximated uniformly, using Lemmas \ref{Rescaled Transport Estimates}, \ref{Rough L^1 Bound} and \ref{uniform square integral bound}:
\begin{align*}
&\int -\frac{t^3}{3(1+\beta(x,y))^3}\left(\frac{\psi(x)-\psi(y)}{x-y}\right)^3~d\fluct_N(x)d\fluct_N(y) \\
&\lesssim \frac{|s|^3}{N^3}\|\psi'\|_{L^\infty}N^2 \int \left(\frac{\psi(x)-\psi(y)}{x-y}\right)^2~dxdy\lesssim \frac{|s|^3}{LN}\|\theta\|_{C^3}^3.
\end{align*}
We need to upgrade both our analysis of the first and second order terms. In order to do this, we will need to make use of the decay of $\psi$. Without loss of generality, assume that $z=0$; the argument is the same for other $z \in \B$, but more notationally tricky to write out. Let us recall the partition of unity $\{\zeta_k\}_{k=1}^{k_*+1}$ introduced in the proofs of Propositions \ref{easy Error2 bound - local laws} and \ref{easy Error2 bound - rescaled test function}. Let $\alpha$ be a fixed positive number such that $[-\alpha ,\alpha ]$ is contained in the bulk of $\Sigma_V$, and choose $k_*$ such that $\frac{2^{k_*}LN}{N}\leq \alpha<\frac{2^{k_*+1}LN}{N}$. We introduce a partition of unity $\{\zeta_k\}_{k=1}^{k_*+1}$ satisfying 
\begin{align*}
\supp(\zeta_k)& \subset \square_{2^kL}\setminus \square_{2^{k-1}L} \\
|\psi \zeta_k|_{C^m\left( \square_{2^kL}\setminus \square_{2^{k-1}L}\right)}&\lesssim |\psi \zeta_k|_{C^m\left( \square_{2^kL}\setminus \square_{2^{k-1}L}\right)}
\end{align*}
for $2\leq k \leq k_*$. We suppose the same holds for $\zeta_1$ with $\square_{2L}\setminus \square_{L}$ replaced by $\square_{2L}$, and for $\zeta_{k_*+1}$ on the rest of $U$.

Let's start with the anisotropy term. First, we write 
\begin{equation*}
\frac{\psi(x)-\psi(y)}{x-y}=\sum_{i=1}^{k_*+1} \frac{\psi \zeta_i(x)-\psi \zeta_i(y)}{x-y},
\end{equation*} 
For notational ease, we just write $\psi \zeta_i$ as $\psi_i$ going forward. Introducing a cutoff function $\chi(x)$ supported on $I_i \pm 2^{i}L$ and identically one on $I_i \pm 2^{i-1}L$ with $C^k$ norm controlled by $\frac{1}{(2^iL)^k}$ (for instance, by rescaling a compactly supported bump), we write with symmetry that
\begin{align*}
\int -t \frac{\psi_i(x)-\psi_i(y)}{x-y}~d\fluct_N(y)d\fluct_N(x)&=\int -t \frac{\psi_i(x)-\psi_i(y)}{x-y}\chi(x)\chi(y)~d\fluct_N(y)d\fluct_N(x) \\
&+2\int -t \frac{\psi_i(x)-\psi_i(y)}{x-y}(1-\chi(y))\chi(y)~d\fluct_N(y)d\fluct_N(x),
\end{align*}
where we dropped the term with $(1-\chi(x))(1-\chi(y))$ because both $\psi_i(x)$ and $\psi_i(y)$ vanish outside of $I_i$.

\subsection*{Anisotropy Control $1 \leq i \leq k_*$: Main Term}
For $\psi \zeta_i$ with $1 \leq i \leq k_*$, we will make use of the Fourier transform trick introduced in \cite[Appendix A]{BLS18}. Notice by Fourier inversion that 
\begin{align*}
\frac{\psi_i(x)-\psi_i(y)}{x-y}=\int_0^1 \psi_i'(sx+(1-s)y)~ds&
=c\int _0^1 \int_{\R}\mathcal{F}(\psi_i')(\lambda)e^{i\lambda sx}e^{i\lambda(1-s)y}~d\lambda ds
\end{align*}
and so we can control $\int -t \frac{\psi_i(x)-\psi_i(y)}{x-y}\chi(x)\chi(y)~d\fluct_N(y)d\fluct_N(x)$ using Fubini by 
\begin{align*}
&\int -t\frac{\psi_i(x)-\psi_i(y)}{x-y}\chi(x)\chi(y)~d\fluct_N(y)d\fluct_N(x) \\
&=-tc\int \int_0^1 \lambda \mathcal{F}(\psi_i)(\lambda)\left[\Fluct_N(e^{is\lambda \cdot}\chi(\cdot))\Fluct_N(e^{i(1-s)\lambda \cdot}\chi(\cdot))\right]~dsd\lambda.
\end{align*}
It follows that on the good event $\mathcal{G}_N$ of Theorem \ref{Uniform Bound on Fluctuations}
\begin{align*}
&\Esp_{\PNbeta}\left[\left|\left[-t\int \int\frac{\psi_i(x)-\psi_i(y)}{x-y}\chi(x)\chi(y)~d\fluct_N(x)d\fluct_N(y)\right]\right|\indic_{\mathcal{G}_N}\right] \\
& \lesssim \frac{|s|}{N}\int \int_0^1 |\lambda \mathcal{F}(\psi_i)(\lambda)|\Esp_{\PNbeta}\left|\left[\Fluct_N(e^{is\lambda \cdot}\chi(\cdot))\Fluct_N(e^{i(1-s)\lambda \cdot}\chi(\cdot))\indic_{\mathcal{G}_N}\right]\right|~dsd\lambda.
\end{align*}
Via Cauchy-Schwarz, we control 
\begin{align*}
&\Esp_{\PNbeta}|\left|\left[\Fluct_N(e^{is\lambda \cdot}\chi(\cdot))\Fluct_N(e^{i(1-s)\lambda \cdot}\chi(\cdot))\indic_{\mathcal{G}_N}\right]\right| \\
& \leq \sqrt{\Esp_{\PNbeta}\left[\left(\Fluct_N(e^{is\lambda \cdot}\chi(\cdot))\right)^2\indic_{\mathcal{G}_N}\right]}\sqrt{\Esp_{\PNbeta}\left[\left(\Fluct_N(e^{i(1-s)\lambda \cdot}\chi(\cdot))\right)^2\indic_{\mathcal{G}_N}\right]}.
\end{align*}
We now observe that since $\chi$ is a rescaled version of a compactly supported test function with support contained in the bulk, $e^{is\lambda \cdot}\chi(\cdot)$ is a test function that falls under the assumptions of Theorem \ref{Uniform Bound on Fluctuations}. Now, Theorem \ref{Uniform Bound on Fluctuations} tells us that on $\mathcal{G}_N$, these fluctuations are morally sub-Gaussian, and so by Corollary \ref{pseudo-Gaussian moments}
\begin{align*}
\Esp \left|\Fluct_N(\xi)\right|^k\indic_{\mathcal{G}_N} &\lesssim \left(\|\theta\|_{C^3}\left(1+\left|1-\frac{1}{\beta}\right|\right)\right)^k+\left(\frac{\|\xi\|_{H^{1/2}}^2+\max(\|\theta\|_{C^1}, \|\theta\|_{C^3}^2)}{\beta}\right)^{\frac{k}{2}} \\
&+\left(\frac{\max(\|\theta\|_{C^2}^2, \|\theta\|_{C^3}^3)}{LN\beta^2}\right)^{\frac{k}{3}}.
\end{align*}
We will prove in Appendix C that
\begin{equation}\label{first anisotropy function}
\left|e^{im \cdot}\chi(\cdot)\right|_{H^{1/2}}^2\leq C\sqrt{1+m^2(2^iL)^2}.
\end{equation}
Furthermore, $e^{im \cdot}\chi(\cdot)$ is a rescaled version of $\vartheta=e^{i2^iLm \cdot}\varphi(\cdot)$ where $\varphi$ is the test function defining $\chi$, independent of $\lambda$. Thus, $\|\vartheta\|_{C^k}\lesssim 2^{ki}L^km^k$. Inserting these into the above estimate and absorbing the $\beta$ dependence in the constant $\beta$ regime, we find 
\begin{equation*}
\Esp \left|\Fluct_N(e^{im \cdot}\chi(\cdot))\right|^2\indic_{\mathcal{G}_N} \lesssim_\beta 1+2^{6i}L^6m^6
\end{equation*}
and so
\begin{align*}
&\Esp_{\PNbeta}|\left|\left[\Fluct_N(e^{is\lambda \cdot}\chi(\cdot))\Fluct_N(e^{i(1-s)\lambda \cdot}\chi(\cdot))\indic_{\mathcal{G}_N}\right]\right| \\
& \lesssim_\beta \sqrt{(1+\sqrt{|s|^6\lambda^6(2^iL)^6})(1+\sqrt{|(1-s)|^6\lambda^6(2^iL)^6})} \lesssim_\beta  1+|\lambda|^62^{6i}L^6.
\end{align*}
It follows that 
\begin{align*}
&\Esp_{\PNbeta}\left[\left|\left[-t\int \int\frac{\psi_i(x)-\psi_i(y)}{x-y}\chi(x)\chi(y)~d\fluct_N(x)d\fluct_N(y)\right]\right|\indic_{\mathcal{G}_N}\right] \\
& \lesssim_\beta \frac{|s|}{N}\int \int_0^1 |\lambda \mathcal{F}(\psi_i)(\lambda)|\Esp_{\PNbeta}\left|\left[\Fluct_N(e^{is\lambda \cdot}\chi(\cdot))\Fluct_N(e^{i(1-s)\lambda \cdot}\chi(\cdot))\indic_{\mathcal{G}_N}\right]\right|~dsd\lambda \\
&\lesssim_\beta \frac{|s|}{N}\int |\lambda \mathcal{F}(\psi_i)(\lambda)|~d\lambda+\frac{|s|2^{6i}L^6}{N}\int \lambda^{7 }| \mathcal{F}(\psi_i)(\lambda)|~d\lambda \\
&=\frac{|s|}{N}\|\mathcal{F}(\psi_i')\|_{L^1}+\frac{|s|2^{6i}L^6}{N}\|\mathcal{F}(\psi_i^{(7)})\|_{L^1}.
\end{align*}
It remains to control the Fourier transforms of the derivatives of $\psi_i$ in $L^1$, which we can approach via the decay estimates of Lemma \ref{Rescaled Transport Estimates}. For $\psi_i'$, we use a simple $\|\mathcal{F}(\psi_i')\|_{L^\infty}\leq \|\psi_i'\|_{L^1}\lesssim \|\theta\|_{C^3}$ for $|\lambda|\leq \frac{1}{2^iL}$, and then use the decay estimate $|\mathcal{F}(\psi_i')(\lambda)|\leq \frac{\|\psi_i'''\|_{L^1}}{|\lambda|^2}\lesssim \frac{1}{(2^iL)^2|\lambda|^2}\|\theta\|_{C^5}$:
\begin{align*}
\int |\mathcal{F}(\psi_i')(\lambda)|~d\lambda\lesssim \int_{|\lambda|\leq \frac{1}{2^iL}}\|\theta\|_{C^3}+\int_{|\lambda|>\frac{1}{2^iL}}\frac{\|\theta\|_{C^5}}{(2^iL)^2|\lambda|^2}~d\lambda\lesssim \frac{1}{2^iL}\|\theta\|_{C^5}.
\end{align*}
For $\psi_i^{(7)}$, we similarly bound $\|\mathcal{F}(\psi_i^{(7)})\|_{L^\infty}\leq \|\psi_i^{(7)}\|_{L^1}\lesssim \frac{L}{(2^iL)^{7}}\|\theta\|_{C^{9}}$ for $|\lambda|\leq \frac{1}{(2^iL)}$, and then use the decay estimate $|\mathcal{F}(\psi_i^{(7)})(\lambda)|\leq \frac{\|\psi_i^{(9)}\|_{L^1}}{|\lambda|^{2}}\lesssim \frac{L}{(2^iL)^{9}|\lambda|^{2}}\|\theta\|_{C^{11}}$:
\begin{align*}
\int |\mathcal{F}(\psi_i^{(7)})(\lambda)|~d\lambda\lesssim \int_{|\lambda|\leq \frac{1}{(2^iL)}}\frac{L}{(2^iL)^{7}}\|\theta\|_{C^{9}}~d\lambda+\int_{|\lambda|>\frac{1}{(2^iL)}}\frac{L}{(2^iL)^{9}|\lambda|^{2}}\|\theta\|_{C^{11}}~d\lambda\lesssim \frac{L}{(2^iL)^{8}}\|\theta\|_{C^{11}}.
\end{align*}
Hence, 
\begin{equation}\label{anisotropy main terms}
\mathbb{E}_{\PNbeta}\left|-t\int \int\frac{\psi_i(x)-\psi_i(y)}{x-y}\chi(x)\chi(y)~d\fluct_N(x)d\fluct_N(y)\indic_{\mathcal{G}_N}\right|\lesssim_\beta \frac{|s|}{2^{i}LN}\|\theta\|_{C^{11}}.
\end{equation}
\subsection*{Anisotropy Control $1 \leq i \leq k_*$: Cross Terms}
We turn our attention the cross terms. A similar argument will work, but we will have to be a bit more careful in our estimates. Observe that $\psi_i(y)$ vanishes on the support of $1-\chi(y)$, so we can actually write the cross term as 
\begin{equation*}
-2t \int \frac{\psi_i(x)}{x-y}\chi(x)(1-\chi(y))~d\fluct_N(y)d\fluct_N(x),
\end{equation*}
which we can then write using Fourier inversion as 
\begin{equation*}
ct\int \frac{\chi(x)(1-\chi(y)}{x-y}\int \mathcal{F}(\psi_i)(\lambda)~d\lambda \fluct_N(y)d\fluct_N(x)=ct\int \mathcal{F}(\psi_i)(\lambda)\Fluct_N \left[e^{ix \lambda}\chi(x)g(x)\right]~d\lambda
\end{equation*}
for some constant $c$, where we have defined $g(x):=\int \frac{1-\chi(y)}{x-y}~d\fluct_N(y)$. Thus, we have immediately that on $\mathcal{G}_N$ \begin{align*}
&\Esp_{\PNbeta}\left[\left|\left[-t\int \int\frac{\psi_i(x)-\psi_i(y)}{x-y}\chi(x)(1-\chi(y))~d\fluct_N(x)d\fluct_N(y)\right]\right|\indic_{\mathcal{G}_N}\right] \\
& \lesssim \frac{|s|}{N}\int  | \mathcal{F}(\psi_i)(\lambda)|\Esp_{\PNbeta}\left|\left[\Fluct_N(e^{i\lambda \cdot}\chi(\cdot)g(\cdot)\indic_{\mathcal{G}_N}\right]\right|~d\lambda.
\end{align*}
Notice that by choice of $\chi$, $\xi(x)=e^{i\lambda x}\chi(x)g(x)$ is a rescaled version of a compactly supported test function that falls under the assumptions of Theorem \ref{Uniform Bound on Fluctuations}. It is a rescaled version of $\vartheta=e^{i2^iLm \cdot}\varphi(\cdot)h(\cdot)$ where $\varphi$ is the test function defining $\chi$, independent of $\lambda$, and one can show from the computations in the proof of (\ref{second anisotropy function}) in Appendix C that $\|h\|_{C^k}\lesssim \sqrt{\frac{N}{2^iL}}$. Thus, $\|\vartheta\|_{C^k}\lesssim \sqrt{\frac{N}{2^iL}}\left(1+ 2^{ki}L^km^k\right)$. 
Now, we claim that 
\begin{equation}\label{second anisotropy function}
\|\xi\|_{H^{1/2}}^2 \leq CN \sqrt{\lambda^2+\frac{1}{(2^iL)^2}}.
\end{equation}
We prove this in Appendix C. Applying Corollary \ref{pseudo-Gaussian moments} and absorbing the $\beta$ dependence in the constant $\beta$ regime then, we find 
\begin{equation*}
\Esp \left|\Fluct_N(e^{im \cdot}\chi(\cdot)g(\cdot))\right|\indic_{\mathcal{G}_N} \lesssim_\beta \sqrt{\frac{N}{2^iL}}\left(1+2^{3i}L^3m^3\right).
\end{equation*}
Inserting this estimate into our previous expansion, we find 
\begin{align*}
&\Esp_{\PNbeta}\left[\left|\left[-2t\int \int\frac{\psi_i(x)-\psi_i(y)}{x-y}\chi(x)(1-\chi(y))~d\fluct_N(x)d\fluct_N(y)\right]\right|\indic_{\mathcal{G}_N}\right] \\
 &\lesssim_\beta \frac{|s|}{N}\int  | \mathcal{F}(\psi_i)(\lambda)|\sqrt{\frac{N}{2^iL}}\left(1+2^{3i}L^3\lambda^3\right)~d\lambda \\
 &\lesssim_\beta \frac{|s|}{\sqrt{2^iLN}}\int |\mathcal{F}(\psi_i)(\lambda)|~d\lambda+\frac{|s|(2^iL)^3}{\sqrt{2^iLN}}\int  |\mathcal{F}(\psi_i^{(3)})(\lambda)|~d\lambda.
 \end{align*}
Estimating as before using Lemma \ref{Rescaled Transport Estimates}, $\|\mathcal{F}(\psi_i)\|_{L^1}\lesssim \frac{L}{2^iL}\|\theta\|_{C^4}$ and $\|\mathcal{F}(\psi_i^{(3)})\|_{L^1} \lesssim \frac{L}{(2^iL)^4}\|\theta\|_{C^7}$, so we conclude that 
 \begin{equation}\label{anisotropy cross terms}
 \Esp_{\PNbeta}\left[\left|\left[-t\int \int\frac{\psi_i(x)-\psi_i(y)}{x-y}\chi(x)(1-\chi(y))~d\fluct_N(x)d\fluct_N(y)\right]\right|\indic_{\mathcal{G}_N}\right] \lesssim_\beta \frac{|s|}{\sqrt{2^iLN}}\|\theta\|_{C^7}.
 \end{equation}

 \subsection*{Anisotropy Control $i=k_*+1$ and Conclusion}
 For $i =k_*+1$, we again observe 
 \begin{align*}
&\Esp_{\PNbeta}\left[\left|\left[-t\int \int\frac{\psi_i(x)-\psi_i(y)}{x-y}\chi(x)\chi(y)~d\fluct_N(x)d\fluct_N(y)\right]\right|\indic_{\mathcal{G}_N}\right] \\
& \lesssim \frac{|s|}{N}\int \int_0^1 |\lambda \mathcal{F}(\psi_i)(\lambda) |\sqrt{\Esp_{\PNbeta}\left[\left(\Fluct_N(e^{is\lambda \cdot}\chi(\cdot))\right)^2\right]}\sqrt{\Esp_{\PNbeta}\left[\left(\Fluct_N(e^{i(1-s)\lambda \cdot}\chi(\cdot))\right)^2\right]}~dsd\lambda
\end{align*}
but now cannot estimate using Theorem \ref{Uniform Bound on Fluctuations} since $e^{im \cdot}$ is not supported in the bulk. Instead we appeal to the quantitative macroscopic CLT as in \cite[(5.12)]{BLS18} and control (using a rough bound on their error terms)
\begin{equation*}
\Esp_{\PNbeta}\exp \left(s\Fluct_N(\xi)\right)\lesssim s\|\xi\|_{C^2}^2+s^2\|\xi\|_{C^4}^2+s^3\|\xi\|_{C^2}^3
\end{equation*}
with $\xi(\cdot)=e^{im \cdot}$, and thus control via (\ref{moment bound})
\begin{equation*}
\Esp_{\PNbeta}\Fluct_N(\xi)^2\lesssim \|\xi\|_{C^2}^4+\|\xi\|_{C^4}^2+\|\xi\|_{C^2}^2 \lesssim 1+m^8.
\end{equation*}
As a result, we control
 \begin{align*}
&\Esp_{\PNbeta}\left[\left|\left[-t\int \int\frac{\psi_i(x)-\psi_i(y)}{x-y}\chi(x)\chi(y)~d\fluct_N(x)d\fluct_N(y)\right]\right|\indic_{\mathcal{G}_N}\right] \\
& \lesssim \frac{|s|}{N} \int |\lambda \mathcal{F}(\psi_i)(\lambda) |(1+|\lambda|^8)~d\lambda \\
&\lesssim \frac{|s|}{N}\left(\|\mathcal{F}(\psi_i')\|_{L^1}+\|\mathcal{F}(\psi_i^{(9)})\|_{L^1}\right).
\end{align*}
Estimating as before using Lemma \ref{Rescaled Transport Estimates}, we bound $\|\mathcal{F}(\psi_i')\|_{L^1} \lesssim L\|\theta\|_{C^5}$ and $\|\mathcal{F}(\psi_i^{(9)})\|_{L^1} \lesssim L\|\theta\|_{C^{13}}$ for $i=k_*+1$. Hence, 
\begin{equation}\label{anisotropy last term}
\Esp_{\PNbeta}\left[\left|\left[-t\int \int\frac{\psi_{k_*+1}(x)-\psi_{k_*+1}(y)}{x-y}\chi(x)\chi(y)~d\fluct_N(x)d\fluct_N(y)\right]\right|\indic_{\mathcal{G}_N}\right] \lesssim \frac{|s|L}{N}\|\theta\|_{C^{13}}.
\end{equation}
Summing over $i$ and using (\ref{anisotropy main terms}), (\ref{anisotropy cross terms}) and (\ref{anisotropy last term}), we obtain 
 \begin{align*}
  &\Esp_{\PNbeta}\left|\left[-t\sum_{i=0}^{k_*+1}\int \int\frac{\psi\zeta_i(x)-\psi\zeta_i(y)}{x-y}~d\fluct_N(x)d\fluct_N(y)\right]\right|\indic_{\mathcal{G}_N} \\
  &\lesssim_\beta \left(\sum_{i=0}^{k^*}\frac{|s|}{(2^{i})LN}+\frac{|s|}{\sqrt{2^{i}LN}}+|s|L\right)\|\theta\|_{C^{13}}
  \lesssim_\beta \frac{|s|\|\theta\|_{C^{13}}}{\sqrt{LN}}
\end{align*}
uniformly in $N$. Via a Markov inequality, we have then outside of a new set $\mathcal{G}_N'$ with $\PNbeta(\mathcal{G}_N\setminus \mathcal{G}'_N) \lesssim\frac{1}{(LN)^{1/4}}$ independent of $s$ that 
 \begin{equation*}
 \left|-t\int \int\frac{\psi(x)-\psi(y)}{x-y}~d\fluct_N(x)d\fluct_N(y)\right| \lesssim_\beta  |s|\|\theta\|_{C^{11}}\frac{1}{(LN)^{1/4}}.
 \end{equation*}
 
\subsection*{Control of Second Order Error Term}
We turn now to the second order error term, which can be dealt with via a clever application of Lemma \ref{CLT Energy Estimate} exactly as we will in the proof of (\ref{second anisotropy function}) in Appendix C. Let $\{\zeta_i\}$ be the same partition of unity as above. Then, to analyze 
\begin{equation*}
\int \frac{t^2}{2}\left(\frac{\psi \zeta_i(x)-\psi\zeta_i(y)}{x-y}\right)^2~d\fluct_N(y)d\fluct_N(x)
\end{equation*}
for $1 \leq i \leq k_*$ we introduce the same cutoff function $\chi$ supported on $I_i \pm 2^{i}L$ and identically one on $I_i \pm 2^{i-1}L$ with $C^k$ norm controlled by $\frac{1}{(2^iL)^k}$ (for instance, by rescaling a compactly supported bump). We write with symmetry that
\begin{align*}
&\int \frac{t^2}{2} \left(\frac{\psi_i(x)-\psi_i(y)}{x-y}\right)^2~d\fluct_N(y)d\fluct_N(x) \\
&=\int \frac{t^2}{2}\left( \frac{\psi_i(x)-\psi_i(y)}{x-y}\right)^2\chi(x)\chi(y)~d\fluct_N(y)d\fluct_N(x) \\
&+\int t^2\left( \frac{\psi_i(x)-\psi_i(y)}{x-y}\right)^2\chi(x)(1-\chi(y))~d\fluct_N(y)d\fluct_N(x),
\end{align*}
where we have replaced $\psi \zeta_i$ with $\psi_i$ for notational ease, as before. Let's start with 
\begin{equation*}
\int \frac{t^2}{2} \left(\frac{\psi_i(x)-\psi_i(y)}{x-y}\right)^2\chi(x)\chi(y)~d\fluct_N(y)d\fluct_N(x).
\end{equation*}
We write $h(x)=\int  \left(\frac{\psi_i(x)-\psi_i(y)}{x-y}\right)^2\chi(y)~d\fluct_N(y)$, which we can bound on $\mathcal{G}_N$ using Lemma \ref{CLT Energy Estimate}, Lemma \ref{Rescaled Transport Estimates} and Theorem \ref{Local Law} by
\begin{align*}
&|h(x)|=\left|\int  \left(\frac{\psi_i(x)-\psi_i(y)}{x-y}\right)^2\chi(y)~d\fluct_N(y)\right|  \\
&\lesssim \left(\frac{1}{(2^iL)^3}\left(2^iL+N^{-3/4}(2^iL)^{1/4}\sqrt{2^iLN}\right)+\left[\sqrt{\mathfrak{h}}\frac{1}{(2^iL)^{5/2}}+\sqrt{\frac{1}{\mathfrak{h}}}\frac{1}{(2^iL)^{3/2}}\right]\sqrt{2^iLN} \right)\|\theta\|_{C^4}^2\\
&\lesssim\left( \frac{1}{(2^iL)^2}+\frac{1}{(2^iL)^{9/4}N^{1/4}}+\frac{\sqrt{N}}{(2^iL)^{3/2}}\right)\|\theta\|_{C^4}^2 \lesssim \frac{\sqrt{N}}{(2^iL)^{3/2}}\|\theta\|_{C^4}^2.
\end{align*}
We have used mean value estimates and $x \in I_i \pm 2^{i}L$ freely, and chose $\mathfrak{h}=2^iL$. The last line follows from $\frac{1}{2^iLN}\lesssim 1$, which allows us to identify $\frac{\sqrt{N}}{(2^iL)^{3/2}}$ as the dominant term. 

Similarly, observe that $h'(x)=\int  2\left(\frac{\psi_i(x)-\psi_i(y)}{x-y}\right)\left(\frac{\psi_i(x)-\psi_i(y)-\psi_i'(y)(x-y)}{(x-y)^2}\right)\chi(y)~d\fluct_N(y)$ and thus we obtain again via Lemma \ref{CLT Energy Estimate}, Lemma \ref{Rescaled Transport Estimates} and Theorem \ref{Local Law} that on the event $\mathcal{G}_N$
\begin{align*}
|h'(x)|&\lesssim \left(\frac{1}{(2^iL)^4}\left(2^iL+N^{-3/4}(2^iL)^{1/4}\sqrt{2^iLN}\right)+\left[\sqrt{\mathfrak{h}}\frac{1}{(2^iL)^{7/2}}+\sqrt{\frac{1}{\mathfrak{h}}}\frac{1}{(2^iL)^{5/2}}\right]\sqrt{2^iLN}\right)\|\theta\|_{C^5}^2 \\
&\lesssim \left(\frac{1}{(2^iL)^3}+\frac{1}{(2^iL)^{13/4}N^{1/4}}+\frac{\sqrt{N}}{(2^iL)^{5/2}} \right)\|\theta\|_{C^5}^2\lesssim \frac{\sqrt{N}}{(2^iL)^{5/2}}\|\theta\|_{C^5}^2
\end{align*}
where we have used the energy control that we have on the good event $\mathcal{G}_N$ that we have restricted ourselves to in this proposition.
We have used mean value estimates and $x \in I_i \pm 2^{i}L$ freely, and chose $\mathfrak{h}=2^iL$. The last line follows from $\frac{1}{2^iLN}\lesssim 1$, which allows us to again identify $\frac{\sqrt{N}}{(2^iL)^{5/2}}$ as the dominant term. Then, we have immediately from Lemma \ref{CLT Energy Estimate} that 
\begin{align*}
&\left|\int \frac{t^2}{2} \left(\frac{\psi_i(x)-\psi_i(y)}{x-y}\right)^2\chi(x)\chi(y)~d\fluct_N(y)d\fluct_N(x)\right| \\
&\lesssim \frac{s^2}{N^2} \| (h\chi)'(x)\|_{L^\infty(I_i)}\left(|I_i|+N^{-3/4}|I_i|^{1/4}\|\nabla u_{\rrc}\|_{L^2(\tilde{NI_i} )}\right) \\
&+\frac{s^2}{N^2}\left[\sqrt{\mathfrak{h}}\|(h\chi)'(x)\|_{L^2(I_i)}+\sqrt{\frac{1}{\mathfrak{h}}}\| h\chi(x)\|_{L^2(I_i)}\right]\|\nabla u_{\rrc}\|_{L^2(\tilde{NI_i} )} \\
&\lesssim \left(\frac{s^2}{N^2}\frac{\sqrt{N}}{(2^iL)^{5/2}}\left(2^iL+N^{-3/4}(2^iL)^{1/4}\sqrt{2^iLN}\right) +\frac{s^2}{N^2}\left[\sqrt{\mathfrak{h}}\frac{\sqrt{N}}{(2^iL)^2}+\sqrt{\frac{1}{\mathfrak{h}}}\frac{\sqrt{N}}{2^iL}\right]\sqrt{2^iLN}\right)\|\theta\|_{C^5}^2 \\
&\lesssim\left( \frac{s^2}{(2^iLN)^{3/2}}+\frac{s^2}{(2^iLN)^{7/4}}+\frac{s^2}{2^iLN}\right)\|\theta\|_{C^5}^2\lesssim \frac{s^2}{2^iLN}\|\theta\|_{C^5}^2.
\end{align*}
Next, we look at the cross terms
\begin{align*}
\int t^2\left( \frac{\psi_i(x)-\psi_i(y)}{x-y}\right)^2\chi(x)(1-\chi(y))~d\fluct_N(y)d\fluct_N(x)=\int t^2 \psi_i(x)^2 g(x)\chi(x)~d\fluct_N(x),
\end{align*}
with $g(x)=\int \frac{1-\chi(y)}{(x-y)^2}~d\fluct_N(y)$. The equality holds since $\psi_i(y)$ vanishes on the support of $1-\chi(y)$ by construction. Dealing with $g$ again exactly as was done in (\ref{second anisotropy function}), we estimate 
\begin{equation*}
|g(x)|\lesssim \frac{\sqrt{N}}{(2^iL)^{3/2}} \hspace{3mm} \text{ and }\hspace{3mm} |g'(x)|\lesssim \frac{\sqrt{N}}{(2^iL)^{5/2}}.
\end{equation*}
With these estimates in tow, we can control
\begin{align*}
&\left|\int t^2 \psi_i(x)^2 g(x)\chi(x)~d\fluct_N(x)\right|\\
&\lesssim \frac{s^2}{N^2} \| (\psi_i^2g\chi)'(x)\|_{L^\infty(I_i)}\left(|I_i|+N^{-3/4}|I_i|^{1/4}\|\nabla u_{\rrc}\|_{L^2(\tilde{NI_i} )}\right) \\
&+\frac{s^2}{N^2}\left[\sqrt{\mathfrak{h}}\|(\psi_i^2g\chi)'(x)\|_{L^2(I_i)}+\sqrt{\frac{1}{\mathfrak{h}}}\| \psi_i^2g\chi(x)\|_{L^2(I_i)}\right]\|\nabla u_{\rrc}\|_{L^2(\tilde{NI_i} )} \\
&\lesssim \left(\frac{s^2}{N^2}\frac{L^2\sqrt{N}}{(2^iL)^{9/2}}(2^iL+N^{-3/4}(2^iL)^{1/4}\sqrt{2^iLN})+\frac{s^2}{N^2}\left[\sqrt{\mathfrak{h}}\frac{L^2\sqrt{N}}{(2^iL)^4}+\sqrt{\frac{1}{\mathfrak{h}}}\frac{L^2\sqrt{N}}{(2^iL)^3}\right]\sqrt{2^iLN}\right)\|\theta\|_{C^3}^2 \\
&\lesssim \left(\frac{s^2}{2^{2i}(2^iLN)^{3/2}}+ \frac{s^2}{2^{2i}(2^iLN)^{7/4}}+ \frac{s^2}{2^{2i}(2^iLN)}\right)\|\theta\|_{C^3}^2\lesssim  \frac{s^2}{2^{2i}(2^iLN)}\|\theta\|_{C^3}^2
\end{align*}
choosing $\mathfrak{h}=2^iL$. Summing over $i$, we find 
\begin{align*}
\left|\sum_{i=1}^{k_*}\int \frac{t^2}{2}\left(\frac{\psi \zeta_i(x)-\psi\zeta_i(y)}{x-y}\right)^2~d\fluct_N(y)d\fluct_N(x)\right| \lesssim \frac{s^2}{LN}\|\theta\|_{C^5}^2.
\end{align*}
Finally, for $i=k_*+1$, since $I_{k_*+1}$ lives at order $1$ we refrain from introducing a cutoff $\chi$ and content ourselves with macroscopic estimates. The same procedure then works as was used to control
\begin{equation*}
\int \frac{t^2}{2} \left(\frac{\psi_i(x)-\psi_i(y)}{x-y}\right)^2\chi(x)\chi(y)~d\fluct_N(y)d\fluct_N(x)
\end{equation*}
for $i \leq k_*$, and we find 
\begin{equation*}
\int \frac{t^2}{2} \left(\frac{\psi_i(x)-\psi_i(y)}{x-y}\right)^2\chi(x)\chi(y)~d\fluct_N(y)d\fluct_N(x) \lesssim \frac{s^2\|\theta\|_{C^5}^2}{N}
\end{equation*}
for $i=k_*+1$. Hence
\begin{align*}
\left|\sum_{i=1}^{k_*+1}\int \frac{t^2}{2}\left(\frac{\psi \zeta_i(x)-\psi\zeta_i(y)}{x-y}\right)^2~d\fluct_N(y)d\fluct_N(x)\right| \lesssim \frac{s^2}{LN}\|\theta\|_{C^5}^2,
\end{align*}
concluding the proof.
 \end{proof}
 \subsection{Conclusion}
 Coupling Propositions \ref{vanishing Error1 bound}, \ref{vanishing Error3 bound} and \ref{vanishing Error2 bound} and intersecting the $\mathcal{G}_N'$ of Propositions \ref{vanishing Error3 bound} and \ref{vanishing Error2 bound}, we find that on our new good event $\mathcal{G}_N'\subset \mathcal{G}_N$ (sacrificing some optimality in the error terms for simplicity)
\begin{align*}
\log \mathbb{E}_{\mathbb{P}_{N,\beta}}[\exp(s\Fluct_N(\xi))\indic_{\mathcal{G}_N'}]&=\frac{s^2}{\beta}\|\xi\|_{H^{1/2}}^2 + s\left(1-\frac{1}{\beta}\right)\int\psi'(x)~d\muv(x) \\ &+O_\beta\left(\frac{|s|}{(LN)^{1/4}}\|\theta\|_{C^{13}}+\frac{s^2}{(LN)^{1/4}}\|\theta\|_{C^5}^2+\frac{|s|^3}{LN}\|\theta\|_{C^3}^3  \right).
\end{align*}
This concludes the proof of Theorem \ref{Gaussian asymptotics}.


\renewcommand\thesection{\Alph{section}}
\setcounter{section}{1}
\setcounter{prop}{0}
\setcounter{equation}{0}
\section*{Appendix A: Auxiliary Computations for Local Laws}\label{Appendix Energy Estimates}
The following is an adaptation of the fluctuation bound of \cite[Lemma B.5]{AS21}. 
\begin{lem}\label{Local Laws Energy Estimate}
Let $\zeta$ be a sufficiently smooth test function, and $\Omega$ a bounded region containing a $1$-neighborhood of the support of $\zeta$. Let $\tilde{\Omega}$ be as in (\ref{extension set}) and let $u$ solve
\begin{equation*}
-\Delta u=2\pi\left(\sum_{i=1}^N \delta_{x_i}-\mu\delta_\R\right)
\end{equation*}
in $\tilde{\Omega}$. Let $\rrc$ be as in (\ref{local minimal distance}).Then,
\begin{align*}
\left|\int \zeta~d\left(\sum_{i=1}^N \delta_{x_i}-\mu \right)\right| \leq& \|\zeta'\|_{L^\infty(\Omega)}(|\Omega|+|\Omega|^{1/4}\|\nabla u_{\rrc}\|_{L^2(\tilde{\Omega})}) \\
&+C\left(\sqrt{\mathfrak{h}}\|\zeta'\|_{L^2(\Omega)}+\frac{1}{\sqrt{\mathfrak{h}}}\|\zeta\|_{L^2(\Omega)}\right)\|\nabla u_{\rrc}\|_{L^2(\tilde{\Omega})},
\end{align*}
for arbitrary $L-1>\mathfrak{h}>0$.
\end{lem}
\begin{proof}
Let $\widetilde{\chi}(x,y):=\zeta(x)\chi(y)$, where $\chi(y)$ is a function that is identically $1$ for $|y|\leq 1$ and vanishes outsides of $[-(\mathfrak{h}+1), \mathfrak{h}+1]$, whose derivative is bounded by $\frac{1}{\mathfrak{h}}$. Then
\begin{equation*}
|\nabla \widetilde{\chi}|\leq \frac{|\zeta(x)|}{\mathfrak{h}}+|\zeta'(x)|.
\end{equation*}
We estimate
\begin{align*}
\left|\int \zeta~d\left(\sum_{i=1}^N \delta_{x_i}-\mu \right)\right|
&\leq \left|\int \widetilde{\chi} ~d\left(\sum \delta_{(x_i,0)}^{(\rrc_i)}-\mu\delta_{\mathbb{R}}\right)\right|+\left|\int \widetilde{\chi} ~d\left(\sum \delta_{(x_i,0)}-\delta_{(x_i,0)}^{(\rrc_i)}\right)\right|
\end{align*}
The first integral we integrate by parts, and bound 
\begin{align*}
 \left|\int \widetilde{\chi} ~d\left(\sum \delta_{(x_i,0)}^{(\rrc_i)}-\mu\delta_{\mathbb{R}}\right)\right|&=\left| \int_{\Omega \times [-(\mathfrak{h}+1),\mathfrak{h}+1]}\nabla \widetilde{\chi} \cdot \nabla u_{\rrc}\right| \\
 &\leq \|\nabla \widetilde{\chi}\|_{L^2(\Omega \times [-(\mathfrak{h}+1),\mathfrak{h}+1])}\|\nabla u_{\rrc}\|_{L^2(\tilde{\Omega})}
\end{align*}
To bound the gradient we have 
\begin{align*}
 \|\nabla \widetilde{\chi}\|_{L^2(\Omega \times [-(\mathfrak{h}+1),\mathfrak{h}+1])}^2=\int_{\Omega}\int_{-(\mathfrak{h}+1)}^{\mathfrak{h}+1} |\nabla \widetilde{\chi}|^2~dydx&=\int_{\Omega}\int_{-(\mathfrak{h}+1)}^{\mathfrak{h}+1} |\zeta'(x)|^2+\frac{|\zeta(x)|^2}{\mathfrak{h}^2}~dy dx \\
 &\leq \mathfrak{h} \|\zeta'\|_{L^2(\Omega)}^2+\frac{1}{\mathfrak{h}}\|\zeta\|_{L^2(\Omega)}^2.
\end{align*}
Thus 
\begin{equation*}
\left|\int \widetilde{\chi} ~d\left(\sum \delta_{(x_i,0)}^{(\rrc_i)}-\mu\delta_{\mathbb{R}}\right)\right|\leq C (\sqrt{\mathfrak{h}} \|\zeta'\|_{L^2(\Omega)}+\frac{1}{\sqrt{h}}\|\zeta\|_{L^2(\Omega)})\|\nabla u_{\rrc}\|_{L^2(\tilde{\Omega})}.
\end{equation*}
We turn now to the second integral, and notice that since $\rrc_i \leq \frac{1}{4}$
\begin{align*}
\left|\int \widetilde{\chi} ~d\left(\sum \delta_{(x_i,0)}-\delta_{(x_i,0)}^{(\rrc_i)}\right)\right|&=\left|\int_{\Omega \times [-1,1]} \zeta(x) ~d\left(\sum \delta_{(x_i,0)}-\delta_{(x_i,0)}^{(\rrc_i)}\right)\right| \\
&\leq \sum_{i:B(x_i, \rrc_i)\cap \Omega \ne \emptyset} \dashint_{B(x_i, \rrc_i)}|\zeta(x_i)-\zeta(x)| \\
&\leq \|\zeta'\|_{L^\infty(\Omega)}\max{\rrc_i}\#I \leq \#I  \|\zeta'\|_{L^\infty(\Omega)}
\end{align*} 
where $\#I$ is the number of balls $B(x_i, \rrc_i)$ intersecting $\Omega$. We aim to bound $\#I$ as in \cite[Lemma B.5]{AS21}. In this vein, let $\xi$ be a test function that is identically $1$ on a $\frac{1}{2}$-neighborhood of the support of $\zeta$, and vanishes outside of a $1$-neighborhood of the support of $\zeta$. Let $\chi$ be as defined above. Then, with $\widetilde{\xi}=\xi(x)\chi(y)$ we obtain
\begin{align*}
\int |\nabla \widetilde{\xi}|^2&\leq \int_{\Omega} \int_{-(\mathfrak{h}+1)}^{\mathfrak{h}+1}\left( \frac{|\xi(x)|^2}{\mathfrak{h}^2}+|\xi'(x)|^2\right)~dydx \\
&=\int_{1+\supp(\zeta)}\int_{-(\mathfrak{h}+1)}^{\mathfrak{h}+1}\frac{|\xi(x)|^2}{\mathfrak{h}^2}~dydx + \int_{(1+\supp(\zeta))\setminus(\frac{1}{2}+\supp(\zeta))}\int_{-(\mathfrak{h}+1)}^{\mathfrak{h}+1}|\xi'(x)|^2~dydx \\
&\leq \int_{1+\supp(\zeta)} \frac{1}{\mathfrak{h}}~dx+  \int_{(1+\supp(\zeta))\setminus(\frac{1}{2}+\supp(\zeta))}\mathfrak{h}|\xi'(x)|^2 \\
&\leq C\mathfrak{h}+C\frac{|\Omega|}{\mathfrak{h}}.
\end{align*}
This is optimized by choosing $\mathfrak{h}=\sqrt{|\Omega|}$ and we obtain 
\begin{equation*}
\int |\nabla \widetilde{\xi}|^2 \leq C\sqrt{|\Omega|}.
\end{equation*}
Bounding as we did with $\widetilde{\chi}$ then, we obtain 
\begin{equation*}
\left|\int \widetilde{\xi} ~d\left(\sum \delta_{(x_i,0)}^{(\rrc_i)}-\mu\delta_{\mathbb{R}}\right)\right|\leq C|\Omega|^{1/4}\|\nabla u_{\rrc}\|_{L^2(\tilde{\Omega})}.
\end{equation*}
Now, by construction,
\begin{equation*}
\left|\int \widetilde{\xi} ~d\left(\sum \delta_{(x_i,0)}^{\rrc_i}-\mu\delta_{\mathbb{R}}\right)\right|=\left|\#I - \int \widetilde{\xi}~d\mu\delta_{\R}\right| \geq \#I-\left|\int_{\Omega}\xi ~d\mu \right| \geq \#I -C|\Omega|,
\end{equation*}
and so 
\begin{equation*}
\#I \leq C|\Omega|+C|\Omega|^{1/4}\|\nabla u_{\rrc}\|_{L^2(\tilde{\Omega})}.
\end{equation*}
We conclude that 
\begin{align*}
\left|\int \zeta~d\left(\sum_{i=1}^N \delta_{x_i}-\mu \right)\right| \leq& \|\zeta'\|_{L^\infty(\Omega)}(|\Omega|+|\Omega|^{1/4}\|\nabla u_{\rrc}\|_{L^2(\tilde{\Omega})}) \\
&+C\left(\sqrt{\mathfrak{h}}\|\zeta'\|_{L^2(\Omega)}+\frac{1}{\sqrt{\mathfrak{h}}}\|\zeta\|_{L^2(\Omega)}\right)\|\nabla u_{\rrc}\|_{L^2(\tilde{\Omega})}.
\end{align*}

\end{proof}

We will also need to make use of a point discrepancy estimate, which is an improved version of \cite[Lemma 2.2]{PS17}. We state it here for ease of reference, and improve the bound slightly by modifying its argument in the spirit of \cite[Lemma B.4]{AS21} and the above Lemma \ref{Local Laws Energy Estimate}.
\begin{lem}\label{discrepancy estimate}
Let $\Omega=B_L(a):=\left[a-\frac{L}{2},a+\frac{L}{2}\right]$ be a bounded interval in $\R$, and set
\begin{equation*}
D_\Omega=\int_{\Omega} \sum_{i=1}^N\delta_{x_i}-\mu.
\end{equation*}
Suppose that $m \leq \mu \leq \Lambda$ in $B_{2L(a)}$ and $L>2$. Then, with $u$ solving 
\begin{equation*}
-\Delta u=2\pi\left(\sum_{i=1}^N \delta_{(x_i,0)}-\mu\delta_\R\right)
\end{equation*}
in $\tilde{B}_{2L}(a)$ (\ref{extension set}) and $\rrc$ as defined in (\ref{local minimal distance}), either $D_\Omega \leq 4\Lambda$ or 
\begin{equation}\label{discrepancy bound}
\int_{\tilde{B}_{2L}(a)} |\nabla u_{\rrc}|^2 \gtrsim  D_\Omega^2 \min \left(1,\sqrt{ \frac{|D_\Omega|}{|\Omega|}}\right).
\end{equation}
\end{lem}
\begin{proof}
First, suppose that $D>0$. Let $\xi$ denote the test function which is identically $1$ on a $\frac{1}{2}$-neighborhood of $\Omega$, and vanishes outside of a $(\frac{1}{2}+\alpha)$-neighborhood of $\Omega$ (with $\alpha$ to be determined). Let $\chi$ and $\tilde{\xi}$ be as in the proof of Lemma \ref{Local Laws Energy Estimate}. First, we observe that since $\rrc_i <\frac{1}{2}$
\begin{align*}
D_\Omega\leq \int \xi \sum_{i=1}^N \delta_{x_i}^{(\rrc_i)}-\int_{\Omega}\xi \mu&=\int \tilde{\xi}\left(\sum_{i=1}^N \delta_{(x_i,0)}^{(\rrc_i)}-\mu\delta_\R\right)+\int_{B_{L+1+2\alpha}(a)\setminus B_L(a)} \xi\mu \\
&\leq \frac{1}{2\pi}\int \nabla \tilde{\xi}\cdot \nabla u_{\rrc}+\Lambda(1+2\alpha) \\
&\leq \frac{1}{2\pi}\|\nabla \tilde{\xi}\|_{L^2}\|\nabla u_{\rrc}\|_{L^2(\tilde{B}_{L+1+2\alpha}(a))}+\Lambda(1+2\alpha).
\end{align*}
So long as $D_\Omega > 2\Lambda$, we can choose $\alpha>0$ such that $\Lambda(1+2\alpha) \leq \frac{D_\Omega}{2}$. Setting $0<\alpha \leq \frac{D_\Omega-2\Lambda}{2\Lambda}$ to be determined later and rearranging, we have
\begin{equation*}
D_\Omega \leq \frac{1}{\pi}\|\nabla \tilde{\xi}\|_{L^2}\|\nabla u_{\rrc}\|_{L^2(\tilde{B}_{L+1+2\alpha}(a))}.
\end{equation*}
It remains to estimate $\|\nabla \tilde{\xi}\|_{L^2}$. We find 
\begin{align*}
\int |\nabla \tilde{\xi}|^2&=\int \xi'(x)^2\chi(y)^2~dxdy+\int \xi(x)^2 \chi'(y)^2~dy \\
&\lesssim \frac{\alpha}{\alpha^2}\mathfrak{h}+(L+1+2\alpha)\frac{\mathfrak{h}}{\mathfrak{h}^2} \\
&\lesssim \sqrt{\frac{(L+1+2\alpha)}{\alpha}}
\end{align*}
by optimally setting $\mathfrak{h}=\sqrt{\alpha(L+1+2\alpha)}$. Now, choose
\begin{equation*}
\alpha:=\min \left(\frac{D_\Omega-2\Lambda}{2\Lambda}, \frac{L-1}{2}\right).
\end{equation*}
Then, if $D_\Omega \geq 4\Lambda$,
\begin{equation*}
\int |\nabla \tilde{\xi}|^2\lesssim\sqrt{2+\frac{L+1}{\min \left(\frac{D_\Omega-2\Lambda}{2\Lambda}, \frac{L-1}{2}\right)}}\lesssim \sqrt{\max \left(1, \frac{L}{|D_\Omega|}\right)}.
\end{equation*}
Rearranging yields (\ref{discrepancy bound}).

The proof for $D<0$ is analogous. Let $\xi$ denote the test function which is identically $1$ on $B_{L-2\alpha}(a)$ and vanishes outside of $\Omega$ (with $\alpha \leq \frac{|\Omega|}{4}$ to be determined). Let $\chi$ and $\tilde{\xi}$ be as in the proof of Lemma \ref{Local Laws Energy Estimate}. First, we observe that  
\begin{align*}
D_\Omega\geq \int \xi \sum_{i=1}^N \delta_{x_i}^{(\rrc_i)}-\int_{\Omega} \mu&=\int \tilde{\xi}\left(\sum_{i=1}^N \delta_{(x_i,0)}^{(\rrc_i)}-\mu\delta_\R\right)-\int_{B_{L}(a)\setminus B_{L-2\alpha}(a)} (1-\xi)\mu \\
&\geq \frac{1}{2\pi}\int \nabla \tilde{\xi}\cdot \nabla u_{\rrc}-2\Lambda \alpha \\
&\geq \frac{-1}{2\pi}\|\nabla \tilde{\xi}\|_{L^2}\|\nabla u_{\rrc}\|_{L^2(\tilde{B}_{L}(a))}-2\Lambda \alpha.
\end{align*}
Setting $0<\alpha<\frac{-D_\Omega}{4\Lambda}$ to be determined later, we find 
\begin{equation*}
D_\Omega \geq \frac{-1}{\pi}\|\nabla \tilde{\xi}\|_{L^2}\|\nabla u_{\rrc}\|_{L^2(\tilde{B}_{L}(a))}.
\end{equation*}
Estimating $\|\nabla \tilde{\xi}\|_{L^2}$ as above, we find 
\begin{equation*}
\int |\nabla \tilde{\xi}|^2=\int \xi'(x)^2\chi(y)^2~dxdy+\int \xi(x)^2 \chi'(y)^2~dy \lesssim \frac{\alpha}{\alpha^2}h+\frac{|\Omega|}{h^2}h \lesssim \sqrt{\frac{|\Omega|}{\alpha}}
\end{equation*}
choosing $h=\sqrt{\alpha |\Omega|}$. Choosing $\alpha=\min \left(\frac{-D_\Omega}{4\Lambda}, \frac{|\Omega|}{4}\right)$ and rearranging yields the result.
\end{proof}
We will also need to make use of a monotonicity in the truncation parameter. This is \cite[Lemma 2.2]{RS22}, with the notation modified for our purposes. A similar result appeared in \cite{PS17}. It is \cite[Lemma B.1]{AS21}, applied to $u$ solving 
\begin{equation*}
-\Delta u=2\pi\sum_{i=1}^N \left(\delta_{(x_i,0)}-\mu \delta_{\R}\right) \hspace{3mm} \text{in }U \times [-h,h].
\end{equation*}
\begin{lem}[Rosenzweig, Serfaty '22]\label{Monotonicity}
Suppose that $u$ solves
\begin{equation*}
-\Delta u=2\pi\sum_{i=1}^N \left(\delta_{(x_i,0)}-\mu \delta_{\R}\right) \hspace{3mm} \text{in }U \times [-h,h]
\end{equation*}
and that $u_{\vec{\alpha}}$ and $u_{\vec{\eta}}$ are two truncations as defined in $\S2$ with $\alpha_i \leq \eta_i \leq h$ for all $i$. Let $I_N=\{i: \alpha_i \ne \eta_i\}$, and suppose that for all $i \in I_N$, $B(x_i, \eta_i) \subset U \times [-h,h]$. Let
\begin{equation*}
\G^{U}(\XN,\mu, \vec{\eta}):=\frac{1}{4\pi}\left(\int_{U \times [-h,h]}|\nabla u_{\vec{\eta}}|^2-2\pi\sum_{x_i \in \Omega} \g(\eta_i)\right)-\sum_{x_i \in U} \int_{U}\f_{\eta_i}(x-x_i)d\mu(x).
\end{equation*}
Then, $\G^U(\XN, \mu, \vec{\eta}) \leq \G^U(\XN, \mu, \vec{\alpha})$.
\end{lem}
We also have local energy control. This is \cite[Proposition 2.3]{RS22}, with the notation modified for our purposes. Since our integrals are only over $|y|\leq L$, we modify the result appropriately. It is analogous to the result for the Coulomb gas \cite[Lemma B.2]{AS21}:
\begin{lem}[Rosenzweig, Serfaty '22]\label{local energy control}
Let $\Omega$ be $\R$ or a closed and bounded interval, or the complement of such a set. Let $\tilde{\Omega}$ be as in (\ref{extension set}). Then, there is some $C>0$, only depending on $\|\mu\|_{L^\infty}$, such that
\begin{equation*}
\int_{\tilde{\Omega}}|\nabla u_{\rrc}|^2 \leq 8\pi\G^\Omega(X_N,\mu,U)+C \#(\{X_N\} \cap \Omega)
\end{equation*}
and
\begin{equation*}
\sum_{i:x_i \in \Omega}\g(\rrc_i)\leq 2\G^\Omega(X_N,\mu,U)+C \#(\{X_N\} \cap \Omega)
\end{equation*}
\end{lem}

\setcounter{section}{2}
\setcounter{prop}{0}
\setcounter{equation}{0}
\section*{Appendix B: Screening}\label{Appendix Screening}
This section aims to prove the so-called "screening" result of Proposition \ref{screening result}. The method essentially adapts and optimizes the approach of \cite{PS17} to the case of a one-dimensional log gas, in an analogous way to the optimization of \cite{AS21} for Coulomb gases. The approach to inner screening is a novel adaptation for the one-dimensional log gas. A similar approach can also be used to optimize \cite{PS17} for higher dimensional Riesz gases. We focus on the proof of outer screening; we will discuss the main idea of inner screening and what computational changes are necessary at the end of this section. From this point on, we denote $E_{\rrc}=\nabla w_{\rrc}$ to emphasize that we are considering the electric field. 
\setcounter{subsection}{0}
\subsection{The Setup}
The first part of the proof consists of using the energy bounds to find a good boundary outside of which to construct the screened configuration. Using a mean value argument as in \cite{PS17} and \cite{AS21}, we can find a $T \in [L-\tilde{l},L-2\tilde{l}]$ such that 
\begin{align*}
    M:=\int_{(\square_{T+4}\setminus \square_{T-4})\times[-h_2,h_1]}|E_{\rrc}|^2 &\leq \frac{S(X_n)}{\tilde{l}} \\
    \int_{\partial \square_T \times [-h_2,h_1]}|E_{\rrc}|^2 &\lesssim M 
\end{align*}
We then take $\Gamma:=\partial \square_T$,which in one dimension is simply two points on the axis. These points enclose an interval $\Old$, in which we will keep $X_n$ and the associated electric field $E$ unchanged. The modifications to the electric field will take place in the set $\New:=\square_L \setminus \Old$. We also let $M_0^+$ and $M_0^-$ be averaged electric fields on the top and bottom of our rectangular region:
\begin{align}
    M_0^+:=\frac{1}{|\New|}\int_{\Old \times \{h_1\}}E_{\rrc}\cdot \hat{n} \label{definition of M} \\
    M_0^-:=\frac{1}{|\New|}\int_{\Old \times \{-h_2\}}E_{\rrc}\cdot \hat{n} \label{definition of M - neg}
\end{align}
Often we will only care about the sum of these two, which we denote $M_0$.

With this region in tow, we partition our space (as in \cite{PS17}) into the following subregions, and solve elliptic problems in each:
\begin{enumerate}
    \item $D_0:= \Old \times [-h_2, h_1]$
    \item $D_\partial:=(\square_L\setminus \Old)\times[-h_2, h_1]$
    \item $D_1:=(\square_L \times [-L, L])\setminus (D_0 \cup D_\partial).$
\end{enumerate}
We partition $\square_L \setminus \Old$ into little intervals $H_k$ with sidelengths at scale $l$, in $\left[\frac{l}{C}, lC\right]$. Let $\tilde{H}_k$ denote the rectangles $H_k \times [-h_2, h_1]$.

Observe that the delineation of our points into old and new sets might intersect some of the ``smeared" points; these smeared regions will have to be modified appropriately. We let $I_\partial$ denote the set of charges that are smeared by the boundary $\Gamma$, i.e.
\begin{equation*}
    I_\partial:=\{i: B(x_i, \rrc_i) \cap \Gamma \ne \emptyset \}.
\end{equation*}
Set
\begin{equation*}
    n_k:=2\pi \int_{\tilde{H}_k} \sum_{i \in I_\partial}\delta_{x_i}^{\rrc_i}
\end{equation*}
to be the amount of smear in a region $\tilde{H}_k$.

We let $\N$ denote the number of smeared charges and the number of charges we want wholly unchanged in $\Old$, i.e.
\begin{equation*}
    \N:=\# I_\partial + \# (\{i:x_i \in \Old\} \setminus I_\partial).
\end{equation*}
The goal will be to place $\mn-\N$ sampled points in $\New:=\square_L\setminus \Old$.

For each constant $k$, choose constants $m_k$ such that 
\begin{equation*}
    m_k |H_k|=\int_{\partial D_0 \cap \partial \tilde{H}_k}E_{\rrc}\cdot \hat{n}+M_0|H_k|-n_k.
\end{equation*}
If $m_k$ is small enough, namely $|m_k|\leq \frac{1}{2}m$, then we can guarantee $\int_{H_k}\mu+m_k \in \mathbb{N}$. Define 
\begin{equation*}
    \tilde{\mu}:=\mu+\sum_{k}\mathbbm{1}_{H_k}m_k.
\end{equation*}
Now, 
\begin{equation*}
    \frac{1}{2\pi}\int_{\partial D_0}E_{\rrc}\cdot \hat{n}=\int_\Old d\mu-\N+\frac{1}{2\pi}\sum_{k}n_k.
\end{equation*}
Hence, by construction
\begin{align*}
    \tilde{\mu}(\New)&=\mu(\New)+\sum_k m_k|H_k| \\
    &=\mn-\mu(\Old)+\frac{1}{2\pi}\int_{\partial D_0}E_{\rrc}\cdot\hat{n}-\frac{1}{2\pi}\sum_k n_k \\
    &=\mn-n_\Old.
\end{align*}
With all of these quantities defined, we are in a position to construct a new screened field outside of $D_0$.

\subsection{Defining the Electric Field}
We define the screened electric field in each of the different subregions. 

First we have $E_1$, which completes the smeared charges. Set
\begin{equation*}
    E_1:=\sum_k \mathbbm{1}_{\tilde{H}_k}\nabla h_{1,k},
\end{equation*}
where $h_{1,k}$ solves 
\begin{equation*}
    \begin{cases}
    -\Delta h_{1,k}=2\pi \sum_{i \in I_\partial}\delta_{x_i}^{\rrc_i} & \text{in } \tilde{H}_k \\
    \partial_{n}h_{1,k}=0 & \text{on }\partial \tilde{H}_k \setminus \partial D_0 \\
    \partial_n h_{1,k}=-\frac{n_k}{|F_k|} & \text{on } F_k,
    \end{cases}
\end{equation*}
where $F_k$ is the face of $\partial \tilde{H}_k$ touching $\partial D_0$, if it exists.

$E_2$ balances the top region, $D_1$. We define
\begin{equation*}
    E_2:=\sum_k \mathbbm{1}_{\tilde{H}_k}\nabla h_{2,k},
\end{equation*}
where $h_{2,k}$ solves
\begin{equation*}
    \begin{cases}
    -\Delta h_{2,k}=m_k\delta_{\R} & \text{in } \tilde{H}_k \\
    \partial_n h_{2,k}=-M_0^+ & \text{on }H_k \times \{h_1 \} \\
    \partial_n h_{2,k}=-M_0^- & \text{on }H_k \times \{-h_2\} \\
    \partial_n h_{2,k}=g_k & \text{on the rest of }\partial \tilde{H}_k,
    \end{cases}
\end{equation*}
where $g_k\equiv 0$ if $H_k$ doesn't touch $\Gamma$, and $g_k=-E_{\rrc}\cdot \hat{n}+\frac{n_k}{|F_k|}$ otherwise, with $\hat{n}$ throughout the outward normal from $D_0$.

$E_3$ gives us the sampled configuration $Z_{\mn-n_\Old}$ in $\New$. Define 
\begin{equation*}
    E_3=\nabla h_3,
\end{equation*}
where $h_3$ solves the Neumann problem
\begin{equation*}
    \begin{cases}
    -\Delta h_3=2\pi \left(\sum_{j=1}^{\mn-\N}\delta_{z_j}-\tilde{\mu}\delta_\R\right) &\text{in }D_\partial \\
    \partial_n h_3=0 & \text{on }\partial D_\partial.
    \end{cases}
\end{equation*}

Finally, $E_4$ gives us the screened electric field in $D_1$. We define
\begin{equation*}
    E_4:=\nabla h_4,
\end{equation*}
where $h_4$ solves
\begin{equation*}
    \begin{cases}
    -\Delta h_4=0 &\text{in }D_1 \\
    \partial_n h_4 =-\phi & \text{on }\partial D_1,
    \end{cases}
\end{equation*}
where $\phi:=\mathbf{1}_{\partial D_1 \cap \partial D_0}E\cdot \hat{n}-\mathbf{1}_{\partial D_1 \cap \partial D_\partial \cap \{y>0\}}M_0^+-\mathbf{1}_{\partial D_1 \cap \partial D_\partial \cap \{y<0\}}M_0^-.$

Now, set $\Escr_{\rrc}:=(E_1+E_2+E_3)\mathbf{1}_{D_\partial}+E_4\mathbf{1}_{D_1}+E_{\rrc}\mathbf{1}_{D_0}$ and add back in the truncations
\begin{equation*}
    \Escr:=\Escr_{\rrc}+\sum_{i=1}^{\mn}\nabla f_{\overline{\mathsf r}_i}(x-y_i),
\end{equation*}
where $Y_{\mn}=(\{X_n\} \cap \Old) \cup Z_{\mn-\N}$, and $\overline{\mathsf r}$ are the (possibly changed) minimal distances for the new configuration $Y_{\mn}$. Due to the Neumann condition, no divergence is created across boundaries when we set $\Escr$ to vanish outisde of our region. By definition then, we have
\begin{equation*}
    \begin{cases}
    -\text{div}\Escr=2\pi \left(\sum_{i \in Y_{\mn}}\delta_{y_i}-\mu \right) &\text{in } \square_L\times [-L,L] \\
    \Escr \cdot \hat{n}=0 &\text{on }\partial (\square_L \times[-L,L).
    \end{cases}
\end{equation*}

\subsection{Estimating Constants}
Instead of estimating $M_0^+$ and $M_0^-$ using Cauchy-Schwarz immediately as in \cite{PS17}, we instead carry these constants through our calculations. This will allows us to be as precise as possible in our estimates of $M_0$. As we discussed above, the screening process requires that $|m_k|<\frac{m}{2}$. In order to get a strict bound on $\left|\frac{\mu-\tilde{\mu}}{\tilde{\mu}}\right|_{L^\infty(\New)}$, we will actually seek $|m_k|<\frac{m}{3}$. 

First observe as in \cite{PS17} that $\sum n_k^2 \leq CM$ and $\#I_\partial \leq CM$. Then, $n_k$ is a nonnegative integer so $n_k ^2\leq \sum _k n_k ^2$ and thus $n_k \leq C\sqrt{M}$. $|H_k| \sim l$, so we have
\begin{align*}
    |m_k| &\lesssim \frac{1}{|H_k|}\int_{\partial D_0 \cap \partial \tilde{H}_k}|E_{\rrc}\cdot \hat{n}|+(M_0^++M_0^-)+\frac{n_k}{|H_k|} \\
    &\lesssim M_0+ \frac{1}{l} \int_{\partial D_0 \cap \partial \tilde{H}_k}|E_{\rrc}\cdot \hat{n}|+\frac{M^{1/2}}{l} \\
    &\lesssim M_0+\frac{h^{1/2}M^{1/2}}{l}+\frac{M^{1/2}}{l} \\
    &\lesssim M_0+\frac{h^{1/2}M^{1/2}}{l}<\frac{m}{3},
\end{align*}
since $h>1$. Squaring, rearranging and using the bounds on $M$ we obtain the condition
\begin{equation}\label{definition of little c}
  M_0^2+\frac{hS(X_n)}{\tilde{l}l^2}\leq \mathsf{c}
\end{equation}
for some fixed constant $\mathsf{c}$.
Equivalently, we could phrase this as 
\begin{equation*}
\max\left(M_0^2, \frac{hS(X_n)}{\tilde{l}l^2}\right)\leq \mathsf{c}.
\end{equation*}

Notice that this condition yields a nice $L^\infty$ bound
\begin{equation*}
\left|\frac{\mu-\tilde{\mu}}{\tilde{\mu}}\right|_{L^\infty(H_k)}=\left|\frac{m_k}{\mu+m_k}\right|_{L^\infty(H_k)} \leq \frac{\frac{m}{3}}{m-\frac{m}{3}}=\frac{1}{2}.
\end{equation*}
Since $\tilde{\mu}$ is defined separately on the $H_k$, we then also have
\begin{equation*}
\left|\frac{\mu-\tilde{\mu}}{\tilde{\mu}}\right|_{L^\infty(\New)}\leq\frac{1}{2}.
\end{equation*}

To get the $L^1$ bound on $\mu$ vs. $\tilde{\mu}$, notice that as in \cite{AS21}, with our choice of boundary,
\begin{align*}
\left|\int_{\New}\mu-\tilde{\mu}\right| \lesssim \sum_k |m_k|l \lesssim \tilde{l}M_0+\int_{\Gamma \times [-h_2,h_1]}|E_{\rrc}|+\sum_k n_k^2 &\lesssim\tilde{l}M_0+ h^{1/2}M^{1/2}+M\\
&\lesssim h+\frac{S(X_n)}{\tilde{l}}+\tilde{l}M_0 \\
&\lesssim h+\frac{S(X_n)}{\tilde{l}},
\end{align*}
since screenability implies a uniform bound on $M_0^2$ and hence a uniform upper bound on $M_0$. We've also used the estimate $\sum_k n_k^2 \lesssim M$ from above.

Finally, to get the $L^2$ bound on $\mu$ vs $\tilde{\mu}$ we use Cauchy-Schwarz on the integral and obtain
\begin{align*}
\int_{\New}(\mu-\tilde{\mu})^2 \lesssim l\sum_k m_k^2 &\lesssim \frac{h}{l}\int_{\Gamma \times [-h_2,h_1]}|E_{\rrc}|^2+\frac{1}{l}\sum_kn_k^2+lM_0^2 \\
&\lesssim \frac{hS(X_n)}{l\tilde{l}}+lM_0^2 \\
&\lesssim l,
\end{align*}
where the last line follows from the screenability condition. We now use these estimates and typical elliptic estimates to control the screened field.

\subsection{Estimating the screened field}
We first estimate $E_1$. Since $\tilde{H}_k$ is a rectangle in $\R^2$, we can use the exact form of \cite[Lemma 6.6]{PS17} to control
\begin{equation}
    \int_{\tilde{H_k}}|\nabla h_{1,k}|^2 \lesssim n_k^2
    \end{equation}
since $\rrc_i$ is bounded below by $\frac{1}{4}$.
Hence,
\begin{equation*}
    \int_{D_\partial}|E_1|^2 \lesssim \sum_{k}n_k^2\lesssim M.
\end{equation*}
We next turn to $E_2$, which requires us to estimate the $h_{2,k}$. Notice that $g_k$ and $m_k$ are defined in such a way that 
\begin{equation}\label{main relation}
\int_{F_k}g_k+m_k|H_k|-M_0|H_k|=0,
\end{equation}
so that the above equation has a unique mean zero solution. Since $h \sim l$, we can apply \cite[Lemma 6.4]{PS17} directly on each $\tilde{H}_k$ and sum to bound
\begin{equation*}
\int_{D_\partial}|E_2|^2\lesssim l \int_{F_k^+}|g_k|^2+l\int_{F_k^-}|g_k|^2+\tilde{l}lM_0^2,
\end{equation*}
where we have denoted the $F_k$ on the left side of $\Old$ by $F_k^-$ and the one on the right by $F_k^+$. We can control this further by observing 
\begin{equation*}
 \int_{F_k^+}|g_k|^2+\int_{F_k^-}|g_k|^2\leq \int_{F_k^+}|E_{\rrc}|^2+\int_{F_k^-}|E_{\rrc}|^2+n_k^2 \left(\frac{1}{|F_k^+|}+\frac{1}{|F_k^-|}\right) \lesssim \int_{\partial \square_T}|E_{\rrc}|^2+M \lesssim M,
\end{equation*}
and so 
\begin{equation*}
\int_{D_\partial}|E_2|^2\lesssim lM+\tilde{l}lM_0^2.
\end{equation*}

We turn to bounding $E_3$ and $E_4$. For $E_3$, we obtain from the definition of $\G^{\mathrm{int}}$ and the fact that $\f_\eta$ is uniformly bounded in $L^1$ for small $\eta$ (see the computations of $\S 6$ in \cite{PS17}) that 
\begin{equation*}
    \int_{D_\partial}|\nabla h_3|^2 \leq 4\pi \left(\G^{\text{int}}(Z_{\mn-n_\Old}, \tilde{\mu}, \New)-\sum_{i=1}^{\mn-n_\Old}\GG(z_j)\right)+2\pi\sum_{i=1}^{\mn-n_\Old}\g(\rrc_j)+C(\mn-n_\Old).
\end{equation*}

Finally, to bound the top field $E_4$, we use \cite[Lemma 6.4]{PS17} at scale $L$ to obtain
\begin{align*}
    \int_{D_1}|E_4|^2 \lesssim L \int_{\partial D_1}|\phi|^2&\leq L|\New|M_0^2+ L\int_{\square_{T+1}\times \{-h_2,h_1\}}|\nabla w|^2 \\
    &\leq L|\New|M_0^2+Le(X_n) \\
    &\lesssim L\tilde{l}M_0^2+Le(X_n).
\end{align*}
It remains to put it all together and obtain the requisite screening estimate. 

We kept the original electric field fixed in $D_0$, so combining the above estimates allows us to write
\begin{align*}
\int_{\square_L \times [-L,L]}|\Escr_{\rrc}|^2 &\leq \int_{D_0}|\nabla w_{\rrc}|^2+ClM+C\tilde{l}LM_0^2+CLe(X_n)+C(\mn-n_\Old)+\\
&C \left(4\pi \left(\G^{\text{int}}(Z_{\mn-n_\Old}, \tilde{\mu}, \New)-\sum_{i=1}^{\mn-n_\Old}\GG(z_j)\right)+2\pi\sum_{i=1}^{\mn-n_\Old}\g(\rrc_j)\right).
\end{align*}
Using Lemma \ref{projection} we can replace the screened electric field with the gradient defining $\G^{\mathrm{int}}$. Using the definition of $\G^{\mathrm{int}}$ then and replacing $\rrc$ with $\overline{\mathsf{r}}$ using Lemma \ref{Monotonicity}, we can rewrite the above estimate as 
\begin{align*}
\G^{\text{int}}(Y_\mn, \square_L)&\leq \frac{1}{4\pi}\int_{\Old \times [-L,L]}|\nabla w_{\rrc}|^2-\frac{1}{2}\sum_{i=1}^{\mn}\g(\rrc_i)-\sum_{i=1}^{\mn}\int_{\R}f_{\rrc_i}(y-y_i)~d\mu(y)+C\sum_{i,j}\g(x_i-z_j)\\
&+\sum_{j=1}^{\mn-n_\Old}\GG(z_j) +ClM+C\tilde{l}LM_0^2+CLe(X_n)+C(\mn-n_\Old)+ \\
&C \left( \left(\G^{\text{int}}(Z_{\mn-n_\Old}, \tilde{\mu}, \New)-\sum_{i=1}^{\mn-n_\Old}\GG(z_j)\right)+\sum_{i=1}^{\mn-n_\Old}\g(\rrc_j)\right),
\end{align*}
where we have used that $\GG=0$ on $\square_L \setminus \Old$. Replacing the definition of $\E^{\text{int}}(w, \square_L)$ in the above expression, we find
with a uniform bound on $f_\eta$ in $L^1$ for small $\eta$ that
\begin{align*}
&\G^{\text{int}}(Y_{\mn}, \square_L)-\E^{\text{int}}(w,\square_L)\leq -\frac{1}{4\pi}\int_{\New \times [-L,L]}|\nabla w_{\rrc}|^2+\frac{1}{2}\sum_{\{i \in \{1,\dots,n\}:x_i \notin \Old\}}\g(\rrc_i)+C\sum_{j=1}^{\mn-n_\Old}\g(\rrc_j)+ClM\\
&+C\tilde{l}LM_0^2+CLe(X_n)+C\G^{\text{int}}(Z_{\mn-n_\Old}, \tilde{\mu}, \New)+C\sum_{i,j}\g(x_i-z_j)+C|n-\mn|+C(\mn-n_\Old).
\end{align*}
Next, we would like to control $\frac{1}{2}\sum_{\{i \in \{1,\dots,n\}:x_i \notin \Old\}}\g(\rrc_i) -\frac{1}{4\pi}\int_{\New \times [-L,L]}|\nabla w_{\rrc}|^2$ by the number of points not in $\Old$, but the possible blowup of $g(\rrc_i)$ presents an issue. We adjust the truncation parameter and apply Lemma \ref{Monotonicity}, but need to shrink $\Old$ a tad in order to guarantee that it does not intersect $B(x_i, \frac{1}{4})$ for all $x_i \notin \Old$. To do this, we simply observe that $\square_{T-4}\subset \Old$ and write 
\begin{align*}
&\frac{1}{2}\sum_{\{i \in \{1,\dots,n\}:x_i \notin \Old\}}\g(\rrc_i) -\frac{1}{4\pi}\int_{\New \times [-L,L]}|\nabla w_{\rrc}|^2 \\
&=\frac{1}{2}\sum_{\{i \in \{1,\dots,n\}:x_i \notin \Old\}}\g(\rrc_i) -\frac{1}{4\pi}\int_{(\square_L \setminus \square_{T-4})\times [-L,L]}|\nabla w_{\rrc}|^2+\frac{1}{4\pi}\int_{( \Old \setminus\square_{T-4})\times [-L,L]}|\nabla w_{\rrc}|^2 \\
&\leq \frac{1}{2}\sum_{\{i \in \{1,\dots,n\}:x_i \notin \square_{T-4}\}}\g(\tilde{\rrc}_i) -\frac{1}{4\pi}\int_{(\square_L \setminus \square_{T-4})\times [-L,L]}|\nabla w_{\tilde{\rrc}}|^2+\frac{1}{4\pi}\int_{(\Old \setminus \square_{T-4} )\times [-L,L]}|\nabla w_{\rrc}|^2 \\
&+\sum_{\{i: x_i \notin \square_{T-4}\}} \int_{\square_L \setminus \square_{T-4}} (\f_{\tilde{\rrc}_i}-\f_{\rrc_i})(x-x_i)~d\mu -\sum_{\{i: x_i \in \Old\setminus \square_{T-4}\}}\g(\rrc_i) \\
& \leq C(n-\N)+CM,
\end{align*}
where $\tilde{\rrc}_i$ is defined to be $\frac{1}{4}$ for $x_i \notin \Old$ and is kept fixed otherwise. This allows us to cancel all contributions of $\f$ and $g$ for $x_i \in \Old \setminus \square_{T-4}$, and bound the remaining contributions of $\f$ and $\g$ by $C(n-\N)$ since $\tilde{\rrc}_i$ is bounded below for such $i$ and $\f_\eta$ is again controlled uniformly in $L^1$ for small $\eta$. We have bounded negative contributions of $L^2$ energy by zero, and $\frac{1}{4\pi}\int_{(\Old \setminus \square_{T-4})\times [-L,L]}|\nabla w_{\rrc}|^2 \leq CM$. 
Finally, using Lemma \ref{local energy control} we can bound 
\begin{align*}
\sum_{j=1}^{\mn-n_\Old}\g(\rrc_j)&\leq C\left(\G^{\text{int}}(Z_{\mn-n_\Old}, \tilde{\mu}, \New)+\n-\N\right) \\
&\leq C\left(\G^{\text{int}}(Z_{\mn-n_\Old}, \tilde{\mu}, \New)+\tilde{l}\right),
\end{align*}
controlling $\n-\N=\tilde{\mu}(\New)\lesssim \tilde{l}$ by our $L^\infty$ control on $\tilde{\mu}$, completing the argument.

\subsection{Inner Screening}
Let us first comment on the changes to the setup that are required for screening in $(\square_L \times [-L,L])^c$. First, we choose our good boundary $\Gamma$ exterior to $\square_L$; namely, we find $T \in [L+\tilde{l}, L+2\tilde{l}]$ such that
\begin{align*}
    M:=\int_{(\square_{T+4}\setminus \square_{T-4})\times[-h_2,h_1]}|E_{\rrc}|^2 &\leq \frac{S(X_n)}{\tilde{l}} \\
    \int_{\partial \square_T \times [-h_2,h_1]}|E_{\rrc}|^2 &\lesssim M 
\end{align*}
and again take $\Gamma=\partial \square_T$. We will leave the configuration unchanged in $\Old=\square_T^c$, and only place new points in $\New=\square_L^c \setminus \square_T^c$. We now partition space so that we only change the field near $\partial (\square_L \times [-L,L])$. Namely, we define
\begin{enumerate}
    \item $D_0:= (\square_T \times [-h_2,h_1])^c$
    \item $D_\partial:=\New \times[-L, L]$
    \item $D_1:=\square_T \times ([-h_2, h_1]\setminus [-L,L]).$
\end{enumerate}
We partition $\New$ into little intervals $H_k$ with sidelengths at scale $l$, in $\left[\frac{l}{C}, lC\right]$, and let $\tilde{H}_k$ denote the rectangles $H_k \times [-L, L]$. Then, we set 
\begin{align*}
    M_0^+:=\frac{1}{|\New|}\int_{(\partial D_1 \cap \{y>0\}) \setminus (\square_T \times \{L\})}E_{\rrc}\cdot \hat{n}  \\
    M_0^+:=\frac{1}{|\New|}\int_{(\partial D_1 \cap \{y<0\}) \setminus (\square_T \times \{-L\})}E_{\rrc}\cdot \hat{n}
\end{align*}
and denote their sum $M_0$. $n_k$, $I_\partial$, $n_\Old$, $m_k$ and $\tilde{\mu}$ are all then defined analogously to the outer screening, as is the screenability condition. $E_1$, $E_2$, $E_3$ and $E_4$ are defined in exactly the same manner as in the outer screening. Setting $\Escr_{\rrc}:=(E_1+E_2+E_3)\mathbf{1}_{D_\partial}+E_4\mathbf{1}_{D_1}+E_{\rrc}\mathbf{1}_{D_0}$ and adding back the truncations we have
\begin{equation*}
    \Escr:=\Escr_{\rrc}+\sum_{i=1}^{\mn}\nabla f_{\overline{\mathsf r}_i}(x-y_i),
\end{equation*}
where $Y_{\mn}=(\{X_n\} \cap \Old) \cup Z_{\mn-\N}$, and $\overline{\mathsf r}$ are the (possibly changed) minimal distances for the new configuration $Y_{\mn}$. Due to the Neumann condition, no divergence is created across boundaries when we set $\Escr$ to vanish outisde of our region. By definition then, we have
\begin{equation*}
    \begin{cases}
    -\text{div}\Escr=2\pi \left(\sum_{i \in Y_{\mn}}\delta_{y_i}-\mu \right) &\text{in } (\square_L \times [-L,L])^c \\
    \Escr \cdot \hat{n}=0 &\text{on }\partial (\square_L \times [-L,L]).
    \end{cases}
\end{equation*}
Since the geometry of $D_0$, $D_\partial$ and $D_1$ are unchanged and the equations the same as with outer screening, all of our estimates almost carry through as with outer screening. The only change comes in estimating $E_4$, since we now have contributions from the vertical sides of the boundary. However, these are estimated by $LM \sim l \frac{S(X_n)}{\tilde{l}}$, so the error terms remain unchanged. Thus,
\begin{align*}
\int_{(\square_L \times [-L,L])^c}|\Escr_{\rrc}|^2 &\leq \int_{D_0}|\nabla w_{\rrc}|^2+ClM+C\tilde{l}LM_0^2+CLe(X_n)+C(\mn-n_\Old)+\\
&C \left(4\pi \left(\G^{\text{int}}(Z_{\mn-n_\Old}, \tilde{\mu}, \New)-\sum_{i=1}^{\mn-n_\Old}\GG(z_j)\right)+2\pi\sum_{i=1}^{\mn-n_\Old}\g(\rrc_j)\right).
\end{align*}
Using Lemma \ref{projection} we can replace the screened electric field with the gradient defining $\G$ as before to find
\begin{align*}
\G^{\text{ext}}(Y_\mn, \square_L^c)&\leq \frac{1}{4\pi}\int_{D_0}|\nabla w_{\rrc}|^2-\frac{1}{2}\sum_{i=1}^{\mn}\g(\rrc_i)-\sum_{i=1}^{\mn}\int_{\R }f_{\rrc_i}(y-y_i)~d\mu(y)+C\sum_{i,j}\g(x_i-z_j)\\
&+\sum_{j=1}^{\mn-n_\Old}\GG(z_j) +ClM+C\tilde{l}LM_0^2+CLe(X_n)+C(\mn-n_\Old)+ \\
&C \left( \left(\G^{\text{int}}(Z_{\mn-n_\Old}, \tilde{\mu}, \New)-\sum_{i=1}^{\mn-n_\Old}\GG(z_j)\right)+\sum_{i=1}^{\mn-n_\Old}\g(\rrc_j)\right),
\end{align*}
Replacing the definition of $\E^{\text{ext}}(w, \square_L^c)$ in the above expression, we find
with the same uniform bound on $f_\eta$ in $L^1$ for small $\eta$ that
\begin{align*}
&\G^{\text{ext}}(Y_{\mn}, \square_L^c)-\E^{\text{ext}}(w,\square_L^c)\leq -\frac{1}{4\pi}\int_{D_\partial \cup D_1}|\nabla w_{\rrc}|^2+\frac{1}{2}\sum_{\{i \in \{1,\dots,n\}:x_i \notin \Old\}}\g(\rrc_i)+C\sum_{j=1}^{\mn-n_\Old}\g(\rrc_j)+ClM \\
&+C\tilde{l}LM_0^2+CLe(X_n)+C\G^{\text{int}}(Z_{\mn-n_\Old}, \tilde{\mu}, \New)+C\sum_{i,j}\g(x_i-z_j)+C|n-\mn|+C(\mn-n_\Old).
\end{align*}
Next, we would like to control $\frac{1}{2}\sum_{\{i \in \{1,\dots,n\}:x_i \notin \Old\}}\g(\rrc_i) -\frac{1}{4\pi}\int_{D_\partial \cup D_1}|\nabla w_{\rrc}|^2$ by the number of points not in $\Old$, but the possible blowup of $g(\rrc_i)$ presents an issue. We adjust the truncation parameter and apply Lemma \ref{Monotonicity}, but again need to shrink $\Old$ a tad in order to guarantee that it does not intersect $B(x_i, \frac{1}{4})$ for all $x_i \notin \Old$. To do this, we simply observe that $\square_{T+4}^c\subset \Old$ and write 
\begin{align*}
&\frac{1}{2}\sum_{\{i \in \{1,\dots,n\}:x_i \notin \Old\}}\g(\rrc_i) -\frac{1}{4\pi}\int_{D_\partial \cup D_1}|\nabla w_{\rrc}|^2 \leq \frac{1}{2}\sum_{\{i \in \{1,\dots,n\}:x_i \notin \Old\}}\g(\rrc_i) -\frac{1}{4\pi}\int_{D_\partial }|\nabla w_{\rrc}|^2 \\
&=\frac{1}{2}\sum_{\{i \in \{1,\dots,n\}:x_i \notin \Old\}}\g(\rrc_i) -\frac{1}{4\pi}\int_{(\square_L^c \setminus \square_{T+4}^c)\times [-L,L]}|\nabla w_{\rrc}|^2+\frac{1}{4\pi}\int_{( \square_{T+4})\setminus \New\times [-L,L]}|\nabla w_{\rrc}|^2 \\
&\leq \frac{1}{2}\sum_{\{i \in \{1,\dots,n\}:x_i \notin \square_{T+4}^c\}}\g(\tilde{\rrc}_i) -\frac{1}{4\pi}\int_{(\square_L^c \setminus \square_{T+4}^c)\times [-L,L]}|\nabla w_{\tilde{\rrc}}|^2+\frac{1}{4\pi}\int_{(  \square_{T+4}\setminus \New)\times [-L,L]}|\nabla w_{\rrc}|^2 \\
&+\sum_{\{i: x_i \notin \square_{T+4}^c\}} \int_{\square_L^c \setminus \square_{T+4}^c} (\f_{\tilde{\rrc}_i}-\f_{\rrc_i})(x-x_i)~d\mu -\sum_{\{i: x_i \in \Old\setminus \square_{T+4}^c\}}\g(\rrc_i) \\
& \leq C(n-\N)+CM,
\end{align*}
where $\tilde{\rrc}_i$ is defined to be $\frac{1}{4}$ for $x_i \notin \Old$ and is kept fixed otherwise. This allows us to cancel all contributions of $\f$ and $g$ for $x_i \in \Old \setminus \square_{T+4}^c$, and bound the remaining contributions of $\f$ and $\g$ by $C(n-\N)$ since $\tilde{\rrc}_i$ is bounded below for such $i$ and $\f_\eta$ is again controlled uniformly in $L^1$ for small $\eta$. We have bounded negative contributions of $L^2$ energy by zero, and $\frac{1}{4\pi}\int_{(\square_{T+4} \setminus \New)\times [-L,L]}|\nabla w_{\rrc}|^2 \leq CM$. 
Finally, using Lemma \ref{local energy control} we can bound 
\begin{align*}
\sum_{j=1}^{\mn-n_\Old}\g(\rrc_j)&\leq C\left(\G^{\text{int}}(Z_{\mn-n_\Old}, \tilde{\mu}, \New)+\n-\N\right) \\
&\leq C\left(\G^{\text{int}}(Z_{\mn-n_\Old}, \tilde{\mu}, \New)+\tilde{l}\right),
\end{align*}
controlling $\n-\N=\tilde{\mu}(\New)\lesssim \tilde{l}$ by our $L^\infty$ control on $\tilde{\mu}$, completing the argument.

\setcounter{section}{3}
\setcounter{subsection}{0}
\setcounter{prop}{0}
\setcounter{equation}{0}

\section*{Appendix C: Transport Estimates}\label{Appendix Transport Estimates}
For direct estimates, we use the formula from Lemma \ref{inverting the master operator}
\begin{equation*}
\psi(x)=\begin{cases}
-\frac{1}{\pi^2S(x)}\int_{\Sigma_V}\frac{\xi(y)-\xi(x)}{\sigma(y)(y-x)}~dy & \text{if } x \in \Sigma_V; \\
\frac{\int \frac{\psi(y)\muv(y)}{x-y}~dy+\xi(x)+c_{\xi}}{\int_{\Sigma_V}\frac{\muv(y)}{x-y}~dy-V'(x)} & \text{if }x \in U \setminus \Sigma_V.
\end{cases}
\end{equation*}
For a detailed proof we refer the reader to Appendix B of \cite{BLS18}, which utilizes formulae and theory from \cite{M92}.
We start by considering the transport for a rescaled test function.
\subsection{Rescaled Test Functions}
Without loss of generality, assume $\xi(\cdot)=\xi_{z,L}(\cdot)$ with $z=0$; the proof is exactly analogous for other values of $z \in \B$. We also write $\supp(\theta)\subset[-c,d]$ and assume it is nice and smooth enough on that interval. Then, $\supp(\xi)\subset[-cL, dL]$. We assume also that $N$ is large enough such that $[-cL,dL]\subset \frac{1}{4}[-a,b]$, where $0 \in [-a,b]$; if this is not the case for $N_c$, we can simply increase it since $L \downarrow 0$ (this is purely for ease of writing the computations and is of course not necessary).

\begin{proof} [Proof of Lemma \ref{Rescaled Transport Estimates}]
We start by estimating the transport function, i.e. $k=0$. First, for $x \in \supp(\xi)$ we will use that $\sigma(y)$ is bounded below by a constant independent of $L$ in $2\supp(\xi)$. Furthermore, outside of $[-a,b]$, $\frac{|\sigma(y)|}{|x-y|}$ is also bounded independently of $L$. We find 
\begin{align*}
|\psi(x)|&\lesssim \int_{2\supp(\xi)}\frac{dy}{\sigma(y)}|\xi'|_{L^\infty}+\left|\int_{\Sigma_V \setminus 2 \supp(\xi)}\frac{dy}{\sigma(y)}\frac{\xi(x)}{y-x}\right| \\
&\lesssim \frac{L}{L}|\theta'|_{L^\infty}+\|\theta\|_{L^\infty}+\left|\int_{[-a,b] \setminus 2 \supp(\xi)}\frac{dy}{\sigma(y)}\frac{\xi(x)}{y-x}\right|.
\end{align*}
We only need to estimate 
\begin{equation*}
\left|\int_{(-a,b) \setminus 2 \supp(\xi))}\frac{dy}{\sigma(y)}\frac{\xi(x)}{y-x}\right|\leq \|\theta\|_{L^\infty}\left|\int_{-a}^{-2cL}\frac{dy}{\sigma(y)(y-x)}+\int_{2dL}^b \frac{dy}{\sigma(y)(y-x)}\right|.
\end{equation*} 
Since the singularities near $-a$ and $b$ are integrable, we can bound those by a constant, and we're left needing to bound
\begin{equation*}
\left|\int_{-a/2}^{-2cL}\frac{dy}{\sigma(y)(y-x)}+\int_{2dL}^{b/2} \frac{dy}{\sigma(y)(y-x)}\right|.
\end{equation*}
$\frac{1}{\sigma(y)}$ and its derivatives are nicely bounded on this domain of integration, so we call $1/\sigma=f$ and treat the above like a Cauchy principal value:
\begin{align*}
\left|\int_{-a/2}^{-2cL}\frac{dy}{\sigma(y)(y-x)}+\int_{2dL}^{b/2} \frac{dy}{\sigma(y)(y-x)}\right|
&=\left|\int_{2dL-x}^{\frac{b}{2}-x}\frac{f(x+u)}{u}~du-\int_{2cL+x}^{\frac{a}{2}+x}\frac{f(x-u)}{u}~du\right|.
\end{align*}
Recall that $x \in (-cL, dL)$, so that all of the above limits of integration are positive. To keep things clean, we write $\min(2dL-x, 2cL+x):=\alpha_1$, $\max(2dL-x, 2cL+x):=\alpha_2$, $\min\left(\frac{b}{2}-x, \frac{a}{2}+x\right):=\beta_1$ and $\max\left(\frac{b}{2}-x, \frac{a}{2}+x\right):=\beta_2$. Then we can write (dropping all constants independent of $L$)
\begin{align*}
\left|\int_{2dL-x}^{\frac{b}{2}-x}\frac{f(x+u)}{u}~du-\int_{2cL+x}^{\frac{a}{2}+x}\frac{f(x-u)}{u}~du\right|& \leq \|f\|_{L^\infty}\int_{\alpha_1}^{\alpha_2}\frac{du}{u}+\left|\int_{\alpha_2}^{\beta_1}\frac{f(x+u)-f(x-u)}{u}~du\right| \\
&\leq  \|f\|_{L^\infty}\ln \frac{\alpha_2}{\alpha_1}+2\|f'\|_{L^\infty}(\beta_1-\alpha_2)+\|f\|_{L^\infty}\ln \frac{\beta_2}{\beta_1} \\
&\lesssim 1+\ln \frac{\alpha_2}{\alpha_1}.
\end{align*}
Finally, $\alpha_2 \leq \max(2d+c, 2c+d)L$ and $\alpha_1 \geq \min(c,d)L$ since $x \in (-cL,dL)$, so 
\begin{equation*}
\ln \frac{\alpha_2}{\alpha_1} \leq \ln \frac{\max(2d+c, 2c+d)L}{\min(c,d)L} = \ln \frac{\max(2d+c, 2c+d)}{\min(c,d)},
\end{equation*}
which is a constant only dependent on $\supp(\theta)$. Thus, 
\begin{equation*}
\left|\int_{(-a,b) \setminus 2 \supp(\xi)L)}\frac{dy}{\sigma(y)}\frac{\xi(x)}{y-x}\right| \lesssim \|\theta\|_{L^\infty},
\end{equation*}
completing the estimate for $|\psi(x)|$ in $\supp(\xi)$. For $x \in 2\supp(\xi)\setminus \supp(\xi)$ we can just use a mean value argument and see that, since $\sigma(y)$ is still bounded below independently of $L$,
\begin{equation*}
|\psi(x)|=\left|\frac{-1}{2\pi^2S(x)}\int_{\supp(\xi)}\frac{\xi(y)-0}{\sigma(y)(y-x)}~dy \right| \lesssim L \|\xi'\|_{L^\infty}\sim \|\theta'\|_{L^\infty} \lesssim \frac{L}{|x|}\|\theta'\|_{L^\infty}
\end{equation*}
since we chose $|x|\sim L$. Finally, for $x \in \supp(\muv)\setminus 2\supp(\xi)$, we have $|y-x| \gtrsim |x|$ in the domain of integration, which is just $\supp(\xi)$. Thus
\begin{equation*}
|\psi(x)|=\left|\frac{-1}{2\pi^2S(x)}\int_{\supp(\xi)}\frac{\xi(y)-0}{\sigma(y)(y-x)}~dy \right| \lesssim \frac{L}{|x|}\|\theta\|_{L^\infty},
\end{equation*}
completing the estimates for $k=0$.

Next, we approach the derivatives. First, we run estimates for $x \in \supp(\xi)$. $\psi$ is the product of $\frac{1}{S}$ and the integral 
\begin{equation*}
\int_{\Sigma_V}\frac{\xi(y)-\xi(x)}{\sigma(y)(y-x)}~dy,
\end{equation*}
so all terms of the product rule expansion are derivatives of $1/S$ and terms of the form $\frac{d^n}{dx^n}\int_{\Sigma_V}\frac{\xi(y)-\xi(x)}{\sigma(y)(y-x)}~dy$ We start by looking at $\frac{-1}{2\pi^2S(x)}\frac{d^k}{dx^k}\int_{\Sigma_V}\frac{\xi(y)-\xi(x)}{\sigma(y)(y-x)}~dy$. Controlling $S(x)$ from below, we just need to focus on the integral. Bringing the derivatives inside assuming sufficient regularity, we need to differentiate $\frac{\xi(y)-\xi(x)}{y-x}$ $k$ times, which a computation shows is given by 
\begin{equation*}
\frac{k!(\xi(y)-P_k^x(y))}{(y-x)^{k+1}},
\end{equation*}
where $P_k^x(y)$ is the $k$th Taylor polynomial about $x$, i.e. $\xi(x)+\xi'(x)(y-x)+\cdots+\frac{\xi^{(k)}(x)}{k!}(y-x)^k$. 
Now with this formula we can estimate using Taylor's theorem and find 
\begin{align*}
&\left|\int_{\Sigma_V}\frac{1}{\sigma(y)}\frac{d^k}{dx^k}\frac{\xi(y)-\xi(x)}{(y-x)}~dy\right|= \left|\int_{\Sigma_V}\frac{k!}{\sigma(y)}\frac{\xi(y)-P_k^x(y)}{(y-x)^{k+1}}~dy\right|\\
&\lesssim \int_{2\supp(\xi)}\frac{1}{\sigma(y)}\left|\frac{\xi(y)-P_k^x(y)}{(y-x)^{k+1}}\right|~dy+\left|\int_{\Sigma_V\setminus 2\supp(\xi)}\frac{k!}{\sigma(y)}\frac{\xi(y)-P_k^x(y)}{(y-x)^{k+1}}~dy\right| \\
&\lesssim \frac{1}{L^k}\|\theta^{(k+1)}\|_{L^\infty}+\left|\int_{\Sigma_V\setminus 2\supp(\xi)}\frac{k!}{\sigma(y)}\sum_{n=0}^k\frac{-\xi^{(n)}(x)}{n!}\frac{1}{(y-x)^{k+1-n}}~dy\right|.
\end{align*}
Everything is nicely bounded outside of $(-a,b)$ and the singularity of $1/\sigma(y)$ is again integrable at $-a$ and $b$, so we really only need to consider 
\begin{align*}
\left|\int_{(-a/2,b/2)\setminus 2\supp(\xi)}\frac{k!}{\sigma(y)}\sum_{n=0}^k\frac{-\xi^{(n)}(x)}{n!(y-x)^{k+1-n}}~dy\right|&\lesssim \sum_{n=0}^{k-1}\|\xi^{(n)}\|_{L^\infty}\int_{(-a/2,b/2)\setminus 2\supp(\xi)}\frac{dy}{|y-x|^{k+1-n}} \\
&+\|\xi^{(k)}\|_{L^\infty} \left|\int_{(-a/2,b/2)\setminus 2\supp(\xi)}\frac{dy}{\sigma(y)(y-x)}\right|
\end{align*}
Notice that by our above work, the last integral is bounded! Hence we bound
\begin{align*}
&\sum_{n=0}^{k-1}\|\xi^{(n)}\|_{L^\infty}\int_{(-a/2,b/2)\setminus 2\supp(\xi)}\frac{dy}{|y-x|^{k+1-n}}+\|\xi^{(k)}\|_{L^\infty} \left|\int_{(-a/2,b/2)\setminus 2\supp(\xi)}\frac{dy}{\sigma(y)(y-x)}\right| \\
& \lesssim \sum_{n=0}^{k-1}\frac{\|\theta^{(n)}\|_{L^\infty}}{L^n}\frac{1}{L^{k-n}}+\frac{\|\theta^{(k)}\|_{L^\infty}}{L^k}\lesssim \frac{\|\theta\|_{C^{k+1}}}{L^k},
\end{align*}
as desired. Notice that the same argument controls any term which is a product of a derivative of $1/S$ and $\frac{d^m}{dx^m}\int_{\Sigma_V}\frac{\xi(y)-\xi(x)}{\sigma(y)(y-x)}~dy$ for $m<k$ by $\frac{\|\theta\|_{C^{m+1}}}{L^m}$, so the requisite estimate is established.

Next, we turn our attention to the decay from $\supp(\xi)$. Using the Leibniz rule on the definition of $\psi$, we have 
\begin{align*}
|\psi^{(k)}(x)|&\lesssim \sum_{n=0}^k {k \choose n}\frac{d^{k-n}}{dx^{k-n}}\left(\frac{-1}{2\pi^2S(x)}\right)\int_{\supp(\xi)}\left|\frac{\xi(y)}{\sigma(y)}\frac{d^n}{dx^n}\frac{1}{y-x}~dy\right|\\
&\lesssim \sum_{n=0}^k \int_{\supp(\xi)}\left|\frac{\xi(y)}{(y-x)^{n+1}}~dy\right|.
\end{align*}
Now, for $x$ outside of $2\supp(\xi)$, in the integrand we have $|y-x|\gtrsim |x|$ and hence 
\begin{equation*}
|\psi^{(k)}(x)|\lesssim \sum_{n=0}^k L\|\xi\|_{L^\infty}\frac{1}{|x|^{n+1}}=L\|\theta\|_{L^\infty} \sum_{n=0}^k \frac{1}{|x|^{n+1}} \lesssim L\|\theta\|_{L^\infty}\left(\frac{1}{|x|}+\frac{1}{|x|^{k+1}}\right).
\end{equation*}
The above estimate won't work in $2\supp(\xi)\setminus \supp(\xi)$, but that's okay because we can use the Taylor approximation approach if we assume that $\xi$ is $C^k$ past the boundary (or that the $k$ derivatives vanish at the boundary and take a $C^k$ extension). In that case, we can take the previous bounds and write
\begin{align*}
|\psi^{(k)}(x)|
 \lesssim \sum_{n=0}^k \int_{\supp(\xi)}\left|\frac{\xi(y)-P_n^x(y)}{(y-x)^{n+1}}~dy\right| \lesssim \sum_{n=0}^k L\|\xi^{(n+1)}\|_{L^\infty}\lesssim \|\theta\|_{C^{k+1}}\left(1+\frac{1}{L^k}\right),
\end{align*}
which, with $L \sim |x|$, is exactly $\sim L\|\theta\|_{C^{k+1}} \left(\frac{1}{|x|}+\frac{1}{|x|^{k+1}}\right)$, as desired.

We now look at extending these estimates outside of $\Sigma_V$. Our goal will be to show that $|\psi^{(n)}(x)|\lesssim L\|\theta\|_{C^{n+2}}$; then for computational ease we can certainly write the estimate as on $\Sigma_V$ since $|x|\gtrsim 1$, which  handles our purpose. 

Recall from Lemma 3.1 and 3.3 in \cite{BLS18} that $\psi$ is given on $U\setminus \Sigma_V$ by 
\begin{equation*}
\psi(x)=\frac{\int \frac{\psi(y)\muv(y)}{x-y}~dy +\frac{\xi(x)}{2}+c_\xi}{\int \frac{\muv(y)}{x-y}~dy-V'(x)}=\frac{\int \frac{\psi(y)\muv(y)}{x-y}~dy +\xi(x)+c_\xi}{-\zeta_V'(x)},
\end{equation*}
where $\zeta_V'(x)=S(x)\sigma(x)$ for some sufficiently smooth $S$. Using this formula directly presents a problem for estimating derivatives of $\psi$. The $\sigma(x)$ in the denominator has a $1/\sqrt{x}$ singularity, whose derivatives are not integrable - successive differentiations will introduce a term that blows up with powers of $1/\dist(x,\Sigma_V)$ as $x \rightarrow \partial \Sigma_V$.  To deal with this, we rewrite the formula for the transport as in \cite[Appendix B]{BLS18} to respect boundary regularity. Namely, after choosing a suitably regular $C^m$ extension $\tilde{\psi}$ we can write 
\begin{equation*}
\psi(x)=\tilde{\psi}(x)+\frac{\xi(x)+c_\xi-\Xi_V[\tilde{\psi}](x)}{-\zeta_V'(x)}=\frac{h(x)}{\sqrt{x-\beta}}\left(\Xi_V[\tilde{\psi}(x)]-c_\xi-\xi(x)\right)
\end{equation*}
with $h(x)$ sufficiently smooth and $x$ near the boundary point $\beta$ of $[\alpha,\beta]\subset \Sigma_V$. Without loss of generality, we have assumed $x>\beta$; the computation is analogous for $x<\alpha$. Notice also that we're welcome to choose $\|\tilde{\psi}(x)\|_{C^m\left(U \setminus \left[-\frac{1}{4}a,\frac{1}{4}b\right]\right)}\lesssim \|\psi(x)\|_{C^m\left(\Sigma_V \setminus\left[-\frac{1}{4}a,\frac{1}{4}b\right]\right)}\lesssim L\|\theta\|_{C^{m+1}}$. Differentiating this formula, we find 
\begin{align*}
|\psi^{(n)}(x)|&=\left|\tilde{\psi}^{(n)}(x)+\sum_{k=0}^n {n \choose k}\frac{d^{n-k}}{dx^{n-k}}\left(\frac{h(x)}{\sqrt{x-\beta}}\right)\frac{d^k}{dx^k}\left(\Xi_V[\tilde{\psi}](x)-c_\xi-\frac{\xi(x)}{2}\right)\right| \\
&\lesssim |\tilde{\psi}^{(n)}(x)|+\sum_{k=0}^n (x-\beta)^{-\frac{1}{2}-(n-k)}\left|\left(\Xi_V[\tilde{\psi}]\right)^{n+1}(z_k)\right||x-\beta|^{n-k+1} \\
&= |\tilde{\psi}^{(n)}(x)|+\sum_{k=0}^n\sqrt{x-\beta}\left|\left(\Xi_V[\tilde{\psi}]\right)^{n+1}(z_k)\right|
\end{align*}
for some collection $z_k \in [\beta,x]$, where in the penultimate line above we have used a Taylor expansion of the $k$th derivative of $\Xi_V[\tilde{\psi}]-c_\xi-\frac{\xi}{2}$ to $(n-k)+1$ order about $\beta$, removing all terms besides the error term using that $\Xi_V[\tilde{\psi}]-c_\xi-\frac{\xi(x)}{2}$ and its derivatives vanish at arbitrary order on $\Sigma_V$.  We have also removed the derivatives of $\xi(x)$ since it vanishes outside of $\left[-\frac{1}{4}a,\frac{1}{4}b\right]\subset \Sigma_V$. We turn now to estimating $\Xi_V[\tilde{\psi}]^{n+1}(z)$ for $z \in [\beta,x]$. Recalling the definition 
\begin{equation*}
\Xi_V[\tilde{\psi}](z)=-\frac{1}{2}\tilde{\psi}(z)V'(z)+\int \frac{\tilde{\psi}(z)-\tilde{\psi}(y)}{z-y}~d\muv(y),
\end{equation*}
we need only estimate derivatives of the integral term since $U$ has compact closure and so $V'(z)$ and its derivatives are controlled. Computing explicitly, we obtain 
\begin{equation*}
\frac{d^m}{dz^m}\int \frac{\tilde{\psi}(z)-\tilde{\psi}(y)}{z-y}~d\muv(y)=\int \frac{m!(\psi(y)-P_m^z(y))}{(y-z)^{m+1}}~d\muv(y),
\end{equation*}
where we have used $\tilde{\psi}=\psi$ on $\Sigma_V$ and recall that $P_m^z(y)=\tilde{\psi}(z)+\tilde{\psi}'(z)(y-z)+\cdots+\frac{\tilde{\psi}^{(m)}(z)}{m!}(y-z)^m$. Let us first consider the integral over intervals $[e,f]$, where $[e,f]$ is neither $[\alpha,\beta]$ nor the interval $[-a,b]$ containing zero (which, of course may be the same). On these, we can simply bound $|y-z|$ from below uniformly and $|\psi(y)|\lesssim \frac{L\|\theta\|_{C^1}}{|y|}\lesssim L\|\theta\|_{C^1}$ to obtain 
\begin{equation*}
\left|\int_e^f \frac{m!(\psi(y)-P_m^z(y))}{(y-z)^{m+1}}~d\muv(y)\right|\lesssim L\|\theta\|_{C^1}+\sum_{k=0}^m|\tilde{\psi}^{(k)}(z)|.
\end{equation*}
Now, for the remaining intervals we first assume $[\alpha,\beta]\ne [-a,b]$. On $[-a,b]$, we can bound all terms besides $\psi(y)$ in the same way to obtain 
\begin{equation*}
\left|\int_{-a}^b \frac{m!P_m^z(y)}{(y-z)^{m+1}}~d\muv(y)\right|\lesssim \sum_{k=0}^m|\tilde{\psi}^{(k)}(z)|.
\end{equation*}
The $\psi(y)$ term is a bit trickier, since the $\frac{1}{|y|}$ decay from the bulk a priori will give us a pesky logarithm. However, we are free to use the previous definition of $\psi$ on $\Sigma_V$. First, we observe that on $2\supp(\xi)$ that 
\begin{equation*}
\int_{2\supp(\xi)}\frac{|\psi(y)\muv(y)|}{|y-z|^{m+1}}~dy \lesssim \frac{L\|\theta\|_{C^1}}{|z|^{m+1}} \lesssim L\|\theta\|_{C^1},
\end{equation*}
simply from a uniform bound on $\psi$, and using $|z-y|\gtrsim |z|$ and that the domain of integration is of size $L$. Away from this interval, we use the definition of $\psi$ in $\Sigma_V$ to write 
\begin{align*}
\int_{[-a,b]\setminus2\supp(\xi)}\frac{\psi(y)\muv(y)}{(y-z)^{m+1}}~dy&=\int_{[-a,b]\setminus 2\supp(\xi)}\frac{S(y)\sigma(y)}{(y-z)^{m+1}}\frac{(-1)}{2\pi^2S(y)}\int_{\Sigma_V}\frac{\xi(w)-\xi(y)}{\sigma(w)(w-y)}~dwdy \\
&= \frac{(-1)}{2\pi^2}\int_{[-a,b]\setminus 2\supp(\xi)}\int_{\Sigma_V}\frac{\sigma(y)}{\sigma(w)}\frac{1}{(y-z)^{m+1}}\frac{\xi(w)-\xi(y)}{w-y}~dwdy.
\end{align*}
Notice that by construction, $\xi(y)$ is never supported on the domain of integration and $\xi(w)$ is only supported on $\supp(\xi)$. Thus, we obtain
\begin{align*}
\int_{[-a,b]\setminus2\supp(\xi)}\frac{\psi(y)\muv(y)}{(y-z)^{m+1}}~dy&= \frac{(-1)}{2\pi^2}\int_{[-a,b]\setminus 2\supp(\xi)}\int_{\Sigma_V}\frac{\sigma(y)}{\sigma(w)}\frac{1}{(y-z)^{m+1}}\frac{\xi(w)}{w-y}~dwdy \\
&= \frac{(-1)}{2\pi^2}\int_{[-a,b]\setminus 2\supp(\xi)}\int_{\supp(\xi)}\frac{\sigma(y)}{\sigma(w)}\frac{1}{(y-z)^{m+1}}\frac{\xi(w)}{w-y}~dwdy.
\end{align*} 
Everything is integrable by construction, so we can use Fubini and write
\begin{align*}
\left|\int_{[-a,b]\setminus2\supp(\xi)}\frac{\psi(y)\muv(y)}{(y-z)^{m+1}}~dy\right| &=\left|  \frac{(-1)}{2\pi^2}\int_{\supp(\xi)}\int_{[-a,b]\setminus 2\supp(\xi)}\frac{\sigma(y)}{\sigma(w)}\frac{1}{(y-z)^{m+1}}\frac{\xi(w)}{w-y}~dydw\right| \\
&\lesssim \int_{\supp(\xi)} \frac{|\xi(w)|}{\sigma(w)}\left|\int_{[-a,b]\setminus 2\supp(\xi)}\frac{\sigma(y)}{(y-z)^{m+1}}\frac{1}{w-y}~dy\right|dw.
\end{align*}
$\frac{\sigma(y)}{(y-z)^{m+1}}$ and its derivatives are nicely bounded on this domain of integration, so we set this function equal to $f$ and notice that we are trying to estimate
\begin{equation*}
\left|\int_{[-a,b]\setminus 2\supp(\xi)}f(y)\frac{1}{w-y}~dy\right|,
\end{equation*}
which by the same Cauchy principal value estimates as in the proof of our estimates on $\Sigma_V$ is simply $\lesssim 1$. Bounding $\sigma(w)$ from below then, we immediately obtain 
\begin{equation*}
\left|\int_{[-a,b]\setminus2\supp(\xi)}\frac{\psi(y)\muv(y)}{(y-z)^{m+1}}~dy\right| \lesssim L\|\theta\|_{C^1}.
\end{equation*}
We conclude that 
\begin{equation*}
\left|\int_{-a}^b\frac{\psi(y)\muv(y)}{(y-z)^{m+1}}~dy\right| \lesssim L\|\theta\|_{C^1}
\end{equation*}
and hence
\begin{equation*}
\left|\int_{-a}^b \frac{m!(\psi(y)-P_m^z(y))}{(y-z)^{m+1}}~d\muv(y)\right|\lesssim L\|\theta\|_{C^1}+\sum_{k=0}^m|\tilde{\psi}^{(k)}(z)|.
\end{equation*}
Finally, we look at $[\alpha,\beta]$. We cannot bound $|y-z|$ uniformly from below (this was the problem with using the original formula directly!) but using Taylor's theorem we can write 
\begin{align*}
\left|\int_{\alpha}^\beta \frac{m!(\psi(y)-P_m^z(y))}{(y-z)^{m+1}}~d\muv(y)\right|&
\lesssim \int_\alpha^\beta \frac{\|\tilde{\psi}^m\|_{L^\infty(\alpha,z)}|y-z|^m}{|y-z|^{m+1}}~d\muv(y) \lesssim \|\tilde{\psi}^m\|_{L^\infty(\alpha,z)}.
\end{align*}
If $[\alpha,\beta]=[-a,b]$, then we can simply split the interval into two pieces; one near zero that we control as $[-a,b]$ above, and one near $b$ that we can control as $[\alpha,\beta]$ above. Of course, we could also do the same if $[\alpha,\beta]$ is not $[-a,b]$; take $[\alpha,\beta]$ and truncate it into $\left[\alpha, \frac{\beta}{2}\right]$ which is dealt with as $[e,f]$, and $\left[\frac{\beta}{2},\beta\right]$, which is dealt with as $[\alpha,\beta]$. This allows us to more easily rewrite the above estimate as 
\begin{equation*}
\left|\frac{d^m}{dz^m}\Xi_V[\tilde{\psi}](z)\right|\lesssim L\|\theta\|_{C^1}+\sum_{k=0}^m|\tilde{\psi}^{(k)}(z)|+\|\tilde{\psi}^m\|_{L^\infty\left(U \setminus \left[-\frac{1}{4}a, \frac{1}{4}b\right]\right)} \lesssim L\|\theta\|_{C^{m+1}}
\end{equation*}
by construction of $\tilde{\psi}$. Thus
\begin{align*}
|\psi^{(n)}(x)|&= |\tilde{\psi}^{(n)}(x)|+\sum_{k=0}^n\sqrt{x-\beta}\left|\left(\Xi_V[\tilde{\psi}]\right)^{n+1}(z_k)\right| \\
&\lesssim L\|\theta\|_{C^{n+1}}+L\|\theta\|_{C^{n+2}}\sqrt{x-\beta} \lesssim L\|\theta\|_{C^{n+2}}
\end{align*}
which completes the argument.
\end{proof}

We now turn our attention to the test functions of interest to the local laws.
\subsection{Local Laws Test Functions}
\begin{proof}[Proof of Lemma \ref{Local Laws Transport Estimates}]
First, we observe that it suffices to show the requisite estimates for $\zeta_{0,h}(Nx)$ and $\kappa_{0,h}(Nx)$ up to adjusting constants. In what follows, we write $\zeta_{0,h}=\zeta_L$ and $\kappa_{0,h}=\kappa_L$, to indicate more explicitly the analogy with the previous subsection.

Let us start with the transport for $\zeta_L$. Assuming $S(x)$ is sufficiently smooth and bounded away from zero, we need only estimate the following quantities for $0\leq k \leq n$:
\begin{equation*}
\int_{\Sigma_V}\frac{dy}{\sigma(y)}\frac{d^k}{dx^k}\left(\frac{\zeta_L(y)-\zeta_L(x)}{y-x}\right).
\end{equation*}
Computing the difference quotient exactly, we have
\begin{align}
\label{zeta derivative expansion}&\int_{\Sigma_V}\frac{dy}{\sigma(y)}\frac{d^k}{dx^k}\left(\frac{\zeta_L(y)-\zeta_L(x)}{y-x}\right) \\
\nonumber&=\frac{d^k}{dx^k}\left[\frac{x}{x^2+(Lh)^2}\right]\int_{\Sigma_V}\frac{dy}{\sigma(y)}\frac{2\pi hL^2}{y^2+(Lh)^2}+\frac{d^k}{dx^k}\left[\frac{1}{x^2+(Lh)^2}\right]\int_{\Sigma_V}\frac{dy}{\sigma(y)}\frac{2\pi hL^2y}{y^2+(Lh)^2}.
\end{align}
We first estimate the integrals, since those estimates are independent of $k$. First, we consider $\int_{\Sigma_V}\frac{dy}{\sigma(y)}\frac{2\pi hL^2}{y^2+(Lh)^2}$. We split the integral into three separate kinds of parts, as usual. First, considering the piece between $-Lh$ and $Lh$ we use that $\sigma(y)$ is nicely bounded below to estimate
\begin{equation*}
\int_{-Lh}^{Lh}\frac{dy}{\sigma(y)}\frac{2\pi hL^2}{y^2+(Lh)^2} \lesssim hL^2 \frac{Lh}{(Lh)^2}=L.
\end{equation*}
Next, we use that $\sigma(y)$ is nicely bounded away from the endpoints of $\Sigma_V$ and that $\frac{1}{\sigma(y)}$ has integrable singularities at those endpoints to bound the rest of the integral on $(-a,b)$ as 
\begin{align*}
\int_{(-a,b)\setminus(-Lh,Lh)}\frac{dy}{\sigma(y)}\frac{2\pi hL^2}{y^2+(Lh)^2}&\lesssim hL^2 \int_{(-a,b)\setminus \left(-\frac{a}{2},\frac{b}{2}\right)}\frac{dy}{\sigma(y)y^2}+hL^2\int_{\left(-\frac{a}{2},\frac{b}{2}\right)\setminus (-Lh,Lh)}\frac{dy}{y^2} \\
&\lesssim hL^2+L.
\end{align*}
By construction, $hL \lesssim 1$, so the entire above quantity we estimate by $\lesssim L$. Finally, for any interval of $\Sigma_V$ $(e,f)$ that does not contain zero, we use that $y$ is nicely bounded and the singularity of $1/\sigma(y)$ is nicely integrable to find 
\begin{equation*}
\int_e^f \frac{dy}{\sigma(y)}\frac{2\pi hL^2}{y^2+(Lh)^2}\lesssim hL^2 \int_e^f\frac{dy}{\sigma(y)}\lesssim hL^2 \lesssim L.
\end{equation*}
Thus, we conclude 
\begin{equation}\label{first integral}
\int_{\Sigma_V}\frac{dy}{\sigma(y)}\frac{2\pi hL^2}{y^2+(Lh)^2}\lesssim L.
\end{equation}
Next, we consider $\int_{\Sigma_V}\frac{dy}{\sigma(y)}\frac{2\pi hL^2y}{y^2+(Lh)^2}$. Again, first considering the piece between $-Lh$ and $Lh$ we use that $\sigma(y)$ is nicely bounded below to estimate 
\begin{equation*}
\left|\int_{-Lh}^{Lh}\frac{dy}{\sigma(y)}\frac{2\pi hL^2y}{y^2+(Lh)^2}\right| \lesssim hL^2\frac{Lh}{(Lh)^2}Lh=\frac{L^4h^3}{L^2h^2}=hL^2.
\end{equation*}
On $\left(-\frac{a}{2},\frac{b}{2}\right)\setminus (-Lh, Lh)$, we need to use a Cauchy principal value argument to avoid a logarithmic error. We find, with $f=1/\sigma$ sufficiently smooth and bounded on the interval,
\begin{align*}
&\left|\int_{\left(-\frac{a}{2},\frac{b}{2}\right)\setminus (-Lh, Lh)}\frac{dy}{\sigma(y)}\frac{2\pi hL^2y}{y^2+(Lh)^2}\right| \\
&\lesssim hL^2 \left|\int_{-\frac{a}{2}}^{-Lh}\frac{f(y)y}{y^2+(Lh)^2}~dy+\int_{Lh}^{\frac{b}{2}}\frac{f(y)y}{y^2+(Lh)^2}~dy\right| \\
&\lesssim hL^2 \int_{Lh}^{\frac{1}{2}\min(a,b)}\left|\frac{f(y)-f(-y)}{y}\right|~dy+hL^2 \int_{\frac{1}{2}\min(a,b)}^{\frac{1}{2}\max(a,b)}\frac{|f|_{L^\infty \left(\frac{1}{2}\min(a,b), \frac{1}{2}\max(a,b)\right)}}{|y|}~dy \\
&\lesssim hL^2.
\end{align*}
On the rest of $(-a,b)$, $y$ is bounded from below and the singularity in $1/\sigma$ is integrable at the boundaries, so we find 
\begin{equation*}
\left|\int_{(-a,b)\setminus\left(-\frac{a}{2},\frac{b}{2}\right)}\frac{dy}{\sigma(y)}\frac{2\pi hL^2y}{y^2+(Lh)^2}\right|\lesssim hL^2 \int_{-a}^b \frac{dy}{\sigma(y)}\frac{1}{|y|} \lesssim hL^2 \int_{-a}^b \frac{dy}{\sigma(y)}\lesssim hL^2.
\end{equation*}
Finally, on any interval $(e,f)$ of $\Sigma_V$ that does not contain zero, we use that the singularity in $1/\sigma$ is integrable at the endpoints and that $y$ is bounded below to estimate 
\begin{equation*}
\left|\int_e^f \frac{dy}{\sigma(y)}\frac{2\pi hL^2 y}{y^2+(Lh)^2}\right|\lesssim hL^2 \int_e^f \frac{dy}{\sigma(y)}\frac{1}{|y|}\lesssim hL^2 \int_e^f \frac{dy}{\sigma(y)}\lesssim hL^2.
\end{equation*}
Hence, we conclude 
\begin{equation}\label{second integral}
\left|\int_{\Sigma_V}\frac{dy}{\sigma(y)}\frac{2\pi hL^2 y}{y^2+(Lh)^2}\right| \lesssim hL^2.
\end{equation}
Inserting (\ref{first integral}) and (\ref{second integral}) into (\ref{zeta derivative expansion}), we find
\begin{align}
\label{zeta derivative bound}\left|\int_{\Sigma_V}\frac{dy}{\sigma(y)}\frac{d^k}{dx^k}\left(\frac{\zeta_L(y)-\zeta_L(x)}{y-x}\right)\right| 
&\lesssim  \left|\frac{d^k}{dx^k}\left[\frac{x}{x^2+(Lh)^2}\right]\right|L+\left|\frac{d^k}{dx^k}\left[\frac{1}{x^2+(Lh)^2}\right]\right|hL^2
\end{align}
and we need only estimate the derivatives.

To do this, we use a recursive formula for the derivatives. Computing using induction and the quotient rule, one can verify that the derivatives satisfy
\begin{align}
\label{first derivative formula}\frac{d^k}{dx^k}\left[\frac{1}{x^2+(Lh)^2}\right]&=\frac{p_k^1(x)}{(x^2+(Lh)^2)^{k+1}} \\
\label{second derivative formula}\frac{d^k}{dx^k}\left[\frac{x}{x^2+(Lh)^2}\right]&=\frac{p_k^2(x)}{(x^2+(Lh)^2)^{k+1}},
\end{align}
where $p_k^1$ and $p_k^2$ both satisfy the recursive formula 
\begin{equation*}
p_k^i(x)=-2kxp_{k-1}^i(x)+(x^2+(Lh)^2)(p_{k-1}^i)'(x)
\end{equation*}
for $k \geq 1$, with $p_0^1(x)=1$ and $p_0^2(x)=x$. 
Since at each iteration of the recursion the degree of the polynomial is increased by $1$, it is easy to see that $\deg(p_k^1)\leq k$ and $\deg(p_k^2)\leq k+1$. We can actually say a bit more - namely, that the polynomials are nice polynomials in both $x$ and $Lh$. Using induction again, we can show that 
\begin{align*}
p_n^1(x)&=\sum_{\substack{k \text{ even} \\k \leq n}} c_{n,k}(Lh)^kx^{n-k} \\
p_n^2(x)&=\sum_{\substack{k \text{ even} \\k \leq n+1}} d_{n,k}(Lh)^kx^{n+1-k}.
\end{align*}
and hence
\begin{align}
\label{first derivative bound}|p_n^1(x)|&\lesssim (Lh)^n+|x|^n \\
\label{second derivative bound}|p_n^2(x)|&\lesssim (Lh)^{n+1}+|x|^{n+1}.
\end{align}
This is fabulous, because now we can conclude. Substituting (\ref{first derivative formula}) and (\ref{second derivative formula}) into (\ref{zeta derivative bound}) and using (\ref{first derivative bound}) and (\ref{second derivative bound}), we find
\begin{align*}
\left|\int_{\Sigma_V}\frac{dy}{\sigma(y)}\frac{d^k}{dx^k}\left(\frac{\zeta_L(y)-\zeta_L(x)}{y-x}\right)\right| 
&\lesssim L \frac{(Lh)^{k+1}+|x|^{k+1}}{(x^2+(Lh)^2)^{k+1}}+hL^2\frac{(Lh)^k+|x|^k}{(x^2+(Lh)^2)^{k+1}}.
\end{align*}
For $|x|\leq Lh$, this yields 
\begin{align*}
\left|\int_{\Sigma_V}\frac{dy}{\sigma(y)}\frac{d^k}{dx^k}\left(\frac{\zeta_L(y)-\zeta_L(x)}{y-x}\right)\right|&
\lesssim\frac{1}{h}\frac{1}{(Lh)^k},
\end{align*}
and for $|x|>Lh$, this yields
\begin{align*}
\left|\int_{\Sigma_V}\frac{dy}{\sigma(y)}\frac{d^k}{dx^k}\left(\frac{\zeta_L(y)-\zeta_L(x)}{y-x}\right)\right|&
\lesssim \frac{L}{|x|^{k+1}}.
\end{align*}
Thus, we have for all $k \geq 0$,
\begin{equation*}
\left|\int_{\Sigma_V}\frac{dy}{\sigma(y)}\frac{d^k}{dx^k}\left(\frac{\zeta_L(y)-\zeta_L(x)}{y-x}\right)\right|\lesssim \begin{cases}
\frac{1}{h}\frac{1}{(Lh)^k} & \text{if }|x|\leq Lh, \\
\frac{L}{|x|^{k+1}} & \text{if }|x|>Lh.
\end{cases}
\end{equation*}
This immediately yields for $|x|\leq Lh$
\begin{equation*}
|\psi^{(n)}(x)|\lesssim \sum_{k=0}^n \left|\int_{\Sigma_V}\frac{dy}{\sigma(y)}\frac{d^k}{dx^k}\left(\frac{\zeta_L(y)-\zeta_L(x)}{y-x}\right)\right| \lesssim \sum_{k=0}^n \frac{1}{h}\frac{1}{(Lh)^k} \lesssim \frac{1}{h}\frac{1}{(Lh)^n}
\end{equation*}
since $Lh$ is small, and for $|x|>Lh$
\begin{equation*}
|\psi^{(n)}(x)|\lesssim \sum_{k=0}^n \left|\int_{\Sigma_V}\frac{dy}{\sigma(y)}\frac{d^k}{dx^k}\left(\frac{\zeta_L(y)-\zeta_L(x)}{y-x}\right)\right| \lesssim \sum_{k=0}^n \frac{L}{|x|^{k+1}} \lesssim L \left(\frac{1}{|x|}+\frac{1}{|x|^{n+1}}\right),
\end{equation*}
where the dominant term depends on whether $|x|$ is larger than one or not.

Next we turn to the transport for $\kappa_L$, and will observe that our estimates above transfer almost exactly. Again assuming $S(x)$ is sufficiently smooth and bounded away from zero, we need only estimate the following quantities for $0\leq k \leq n$:
\begin{equation*}
\int_{\Sigma_V}\frac{dy}{\sigma(y)}\frac{d^k}{dx^k}\left(\frac{\kappa_L(y)-\kappa_L(x)}{y-x}\right).
\end{equation*}
Computing the difference quotient exactly, we have
\begin{align*}
&\int_{\Sigma_V}\frac{dy}{\sigma(y)}\frac{d^k}{dx^k}\left(\frac{\kappa_L(y)-\kappa_L(x)}{y-x}\right) \\
&=\frac{d^k}{dx^k}\left[\frac{1}{x^2+(Lh)^2}\right]\int_{\Sigma_V}\frac{dy}{\sigma(y)}\frac{2\pi h^2L^3}{y^2+(Lh)^2}-\frac{d^k}{dx^k}\left[\frac{x}{x^2+(Lh)^2}\right]\int_{\Sigma_V}\frac{dy}{\sigma(y)}\frac{2\pi Ly}{y^2+(Lh)^2}.
\end{align*}
Substituting (\ref{first derivative formula}) and (\ref{second derivative formula}) and estimating with (\ref{first derivative bound}) and (\ref{second derivative bound}),
\begin{align*}
\left|\int_{\Sigma_V}\frac{dy}{\sigma(y)}\frac{d^k}{dx^k}\left(\frac{\kappa_L(y)-\kappa_L(x)}{y-x}\right)\right| 
&\lesssim hL^2\frac{(Lh)^k+|x|^k}{(x^2+(Lh)^2)^{k+1}}+ L \frac{(Lh)^{k+1}+|x|^{k+1}}{(x^2+(Lh)^2)^{k+1}}.
\end{align*}
Exactly as above, this yields 
\begin{equation*}
\left|\int_{\Sigma_V}\frac{dy}{\sigma(y)}\frac{d^k}{dx^k}\left(\frac{\kappa_L(y)-\kappa_L(x)}{y-x}\right)\right|\lesssim \begin{cases}
\frac{1}{h}\frac{1}{(Lh)^k} & \text{if }|x|\leq Lh, \\
\frac{L}{|x|^{k+1}} & \text{if }|x|>Lh,
\end{cases}
\end{equation*}
which in turn yields the same estimates on the transport of $\kappa_L$.

Finally, we can use these estimates for the transport in $\Sigma_V$ to obtain transport estimates on $U \setminus \Sigma_V$ as well. As with Lemma \ref{Rescaled Transport Estimates}, it suffices to show that $|\psi^{(n)}(x)|\lesssim L$ outside of $\Sigma_V$; since $|x|\gtrsim 1$ outside of $\Sigma_V$, this will be sufficient for our purposes. Let $\xi$ denote either $\zeta_L$ or $\kappa_L$, and let $[\alpha,\beta]$ denote the interval of $\Sigma_V$ closest to $x$; without loss of generality, let $x>\beta$. Exactly as in the proof of Lemma \ref{Rescaled Transport Estimates} we can choose a $C^m$ extension $\tilde{\psi}$ with $\|\tilde{\psi}\|_{C^m\left(\left[-\frac{1}{4}a,\frac{1}{4}b\right]\right)}\lesssim\|\psi\|_{C^m\left(\left[-\frac{1}{4}a,\frac{1}{4}b\right]\right)}\lesssim L$ and bound 
\begin{align*}
|\psi^{(n)}(x)|&=\left|\tilde{\psi}^{(n)}(x)+\sum_{k=0}^n {n \choose k}\frac{d^{n-k}}{dx^{n-k}}\left(\frac{h(x)}{\sqrt{x-\beta}}\right)\frac{d^k}{dx^k}\left(\Xi_V[\tilde{\psi}](x)-c_\xi-\frac{\xi(x)}{2}\right)\right| \\
&\lesssim |\tilde{\psi}^{(n)}(x)|+\sum_{k=0}^n (x-\beta)^{-\frac{1}{2}-(n-k)}\left|\left(\Xi_V[\tilde{\psi}]\right)^{n+1}(z_k)+\frac{\xi^{(n+1)}(z_k)}{2}\right\|x-\beta|^{n-k+1} \\
&\lesssim L+\sum_{k=0}^n\sqrt{x-\beta}\left|\left(\Xi_V[\tilde{\psi}]\right)^{n+1}(z_k)+\frac{\xi^{(n+1)}(z_k)}{2}\right|
\end{align*}
for some collection $z_k \in [\beta,x]$, where in the penultimate line above we have used a Taylor expansion of the $k$th derivative of $\Xi_V[\tilde{\psi}]-c_\xi-\frac{\xi}{2}$ to $(n-k)+1$ order about $\beta$, removing all terms besides the error term using that $\Xi_V[\tilde{\psi}]-c_\xi-\frac{\xi(x)}{2}$ and its derivatives vanish at arbitrary order on $\Sigma_V$.

Let's look first at $\frac{\xi^{m)}(z)}{2}$ for $z \in U \setminus \Sigma_V$ and $m \geq 1$. Using (\ref{first derivative formula}), when $\xi=\zeta_L$ we obtain 
\begin{equation*}
\frac{d^m}{dz^m}\xi(z)=-2\pi hL^2 \frac{d^m}{dz^m}\left(\frac{1}{z^2+(Lh)^2}\right)=-2\pi hL^2 \frac{p_m^1(z)}{(z^2+(Lh)^2)^{m+1}}
\end{equation*}
and thus bound using (\ref{first derivative bound})
\begin{align*}
\left|\frac{\xi^{m)}(z)}{2}\right|&\lesssim hL^2 \frac{|p_m^1(z)|}{(z^2+(Lh)^2)^{m+1}} \lesssim hL^2 \frac{|z|^m+(Lh)^m}{(z^2+(Lh)^2)^{m+1}} \lesssim \frac{hL^2}{|z|^{m+2}}\lesssim L
\end{align*}
since $Lh <\min \left(\frac{a}{4}, \frac{b}{4}\right) \leq |z|$. Similarly, when $\xi=\kappa_L$ we obtain with (\ref{second derivative formula})
\begin{equation*}
\frac{d^m}{dz^m}\xi(z)=2\pi L \frac{d^m}{dz^m}\left(\frac{z}{z^2+(Lh)^2}\right)=2\pi L \frac{p_m^2(z)}{(z^2+(Lh)^2)^{m+1}}
\end{equation*}
and thus bound using (\ref{second derivative bound})
\begin{align*}
\left|\frac{\xi^{m)}(z)}{2}\right|&\lesssim L \frac{|p_m^2(z)|}{(z^2+(Lh)^2)^{m+1}} \lesssim L \frac{|z|^{m+1}+(Lh)^{m+1}}{(z^2+(Lh)^2)^{m+1}} \lesssim \frac{L}{|z|^{m+1}}\lesssim L
\end{align*}
since $Lh <\min \left(\frac{a}{4}, \frac{b}{4}\right) \leq |z|$.

We turn now to estimating $\Xi_V[\tilde{\psi}]^{n+1}(z)$ for $z \in U \setminus \Sigma_V$. Recalling the definition 
\begin{equation*}
\Xi_V[\tilde{\psi}](z)=-\frac{1}{2}\tilde{\psi}(z)V'(z)+\int \frac{\tilde{\psi}(z)-\tilde{\psi}(y)}{z-y}~d\muv(y),
\end{equation*}
we need only estimate derivatives of the integral term since $U$ is compact and so $V'(z)$ and its derivatives are controlled by $\|\tilde{\psi}^{(m)}\|_{C^m \left(\left[-\frac{1}{4}a, \frac{1}{4}b\right]\right)}\lesssim L$. Computing explicitly as before, we obtain 
\begin{equation*}
\frac{d^m}{dz^m}\int \frac{\tilde{\psi}(z)-\tilde{\psi}(y)}{z-y}~d\muv(y)=\int \frac{m!(\psi(y)-P_m^z(y))}{(y-z)^{m+1}}~d\muv(y),
\end{equation*}
where we have used $\tilde{\psi}=\psi$ on $\Sigma_V$ and recall that $P_m^z(y)=\tilde{\psi}(z)+\tilde{\psi}'(z)(y-z)+\cdots+\frac{\tilde{\psi}^{(m)}(z)}{m!}(y-z)^m$. Let us first consider the integral over intervals $[e,f]$, where $[e,f]$ is neither $[\alpha,\beta]$ nor the interval $[-a,b]$ containing zero (which, of course may be the same). On these, we can simply bound $|y-z|$ from below uniformly and $|\psi(y)|\lesssim L$ to obtain 
\begin{equation*}
\left|\int_e^f \frac{m!(\psi(y)-P_m^z(y))}{(y-z)^{m+1}}~d\muv(y)\right|\lesssim L+\sum_{k=0}^m|\tilde{\psi}^{(k)}(z)| \lesssim L.
\end{equation*}
Now, for the remaining intervals we first assume $[\alpha,\beta]\ne [-a,b]$. On $[-a,b]$, we can bound all terms besides $\psi(y)$ in the same way to obtain 
\begin{equation*}
\left|\int_{-a}^b \frac{m!P_m^z(y)}{(y-z)^{m+1}}~d\muv(y)\right|\lesssim \sum_{k=0}^m|\tilde{\psi}^{(k)}(z)| \lesssim L.
\end{equation*}
The $\psi(y)$ term is a bit trickier, since the $\frac{1}{|y|}$ decay from the bulk a priori will, as in the proof of Lemma \ref{Rescaled Transport Estimates}, give us a pesky logarithm. However, we are free to compute directly with the definition of $\psi$ on $\Sigma_V$. 

Let's start with $\zeta_L$. Computing explicitly, we find 
\begin{align*}
&\int_{-a}^b \frac{\psi(y)}{(y-z)^{m+1}}~d\muv(y)=\int_{-a}^b \frac{S(y)\sigma(y)}{(y-z)^{m+1}}\frac{-1}{2\pi^2 S(y)}\int_{\Sigma_V}\frac{\zeta_L(w)-\zeta_L(y)}{w-y}~dwdy \\
&=-\frac{hL^2}{\pi}\int_{-a}^b \frac{\sigma(y)}{(y-z)^{m+1}}\int_{\Sigma_V} \frac{1}{\sigma(w)}\frac{w+y}{(w^2+(Lh)^2)(y^2+(Lh)^2)}~dwdy \\
&=-\frac{c_1hL^2}{\pi}\int_{-a}^b\frac{\sigma(y)}{(y-z)^{m+1}(y^2+(Lh)^2)}~dy-\frac{c_2hL^2}{\pi}\int_{-a}^b \frac{\sigma(y)y}{(y-z)^{m+1}(y^2+(Lh)^2)}~dy,
\end{align*}
where 
\begin{equation*}
c_1=\int_{\Sigma_V}\frac{w}{\sigma(w)(w^2+(Lh)^2)}~dw \text{ and } c_2=\int_{\Sigma_V}\frac{1}{\sigma(w)(w^2+(Lh)^2)}~dw.
\end{equation*}
Modifying (\ref{first integral}) and (\ref{second integral}), we can control $c_1\lesssim1$ and $c_2 \lesssim \frac{1}{hL}$. Bounding $|y-z|\gtrsim |z|\gtrsim 1$, we can also control 
\begin{equation*}
\left|\int_{-a}^b\frac{\sigma(y)}{(y-z)^{m+1}(y^2+(Lh)^2)}~dy\right|\lesssim \frac{1}{hL}.
\end{equation*}
Using that the derivatives of $\frac{1}{(y-z)^{m+1}}$ are well controlled on the domain of integration, we can modify the principal value estimates in (\ref{second integral}) to control
\begin{equation*}
\left|\int_{-a}^b \frac{\sigma(y)y}{(y-z)^{m+1}(y^2+(Lh)^2)}~dy\right|\lesssim 1.
\end{equation*}
Combining all of these estimates yields 
\begin{equation*}
\left|\int_{-a}^b \frac{m!\psi(y)}{(y-z)^{m+1}}~d\muv(y)\right|\lesssim L
\end{equation*}
and hence 
\begin{equation*}
\left|\int_{-a}^b \frac{m!(\psi(y)-P_m^z(y))}{(y-z)^{m+1}}~d\muv(y)\right|\lesssim L.
\end{equation*}
We run the same computation for $\kappa_L$. Computing explicitly, we find 
\begin{align*}
&\int_{-a}^b \frac{\psi(y)}{(y-z)^{m+1}}~d\muv(y)=\int_{-a}^b \frac{S(y)\sigma(y)}{(y-z)^{m+1}}\frac{-1}{2\pi^2 S(y)}\int_{\Sigma_V}\frac{\kappa_L(w)-\kappa_L(y)}{w-y}~dwdy \\
&=\frac{-1}{\pi}\int_{-a}^b \frac{\sigma(y)}{(y-z)^{m+1}}\int_{\Sigma_V} \frac{1}{\sigma(w)}\frac{h^2L^3-Lwy}{(w^2+(Lh)^2)(y^2+(Lh)^2)}~dwdy \\
&=-\frac{c_2h^2L^3}{\pi}\int_{-a}^b\frac{\sigma(y)}{(y-z)^{m+1}(y^2+(Lh)^2}~dy+\frac{c_1L}{\pi}\int_{-a}^b \frac{\sigma(y)y}{(y-z)^{m+1}(y^2+(Lh)^2)}~dy,
\end{align*}
with $c_1$ and $c_2$ as above.
Using the above control on the integrals, we obtain 
\begin{equation*}
\left|\int_{-a}^b \frac{\sigma(y)y}{(y-z)^{m+1}(y^2+(Lh)^2)}~dy\right|\lesssim \frac{h^2L^3}{hL}\frac{1}{hL}+L\lesssim L
\end{equation*}
and hence 
\begin{equation*}
\left|\int_{-a}^b \frac{m!(\psi(y)-P_m^z(y))}{(y-z)^{m+1}}~d\muv(y)\right|\lesssim L.
\end{equation*}
Finally, we look at $[\alpha,\beta]$. We cannot bound $|y-z|$ uniformly from below (this was the problem with using the original formula directly!) but using Taylor's theorem as in the proof of Lemma \ref{Rescaled Transport Estimates} we can write for either $\kappa_L$ or $\zeta_L$ 
\begin{align*}
\left|\int_{\alpha}^\beta \frac{m!(\psi(y)-P_m^z(y))}{(y-z)^{m+1}}~d\muv(y)\right|
&\lesssim \int_\alpha^\beta \frac{\|\tilde{\psi}^m\|_{L^\infty(\alpha,z)}|y-z|^m}{|y-z|^{m+1}}~d\muv(y)
 \lesssim \|\tilde{\psi}^m\|_{L^\infty(\alpha,z)}.
\end{align*}
If $[\alpha,\beta]=[-a,b]$, then we can simply split the interval into two pieces; one near zero that we control as $[-a,b]$ above, and one near $b$ that we can control as $[\alpha,\beta]$ above. Of course, we could also do the same if $[\alpha,\beta]$ is not $[-a,b]$; take $[\alpha,\beta]$ and truncate it into $\left[\alpha, \frac{\beta}{2}\right]$ which is dealt with as $[e,f]$, and $\left[\frac{\beta}{2},\beta\right]$, which is dealt with as $[\alpha,\beta]$. This allows us to more easily rewrite the above estimate as 
\begin{equation*}
\left|\frac{d^m}{dz^m}\Xi_V[\tilde{\psi}](z)\right|\lesssim L+\|\tilde{\psi}^m\|_{L^\infty\left(U \setminus \left[-\frac{1}{4}a, \frac{1}{4}b\right]\right)} \lesssim L
\end{equation*}
by construction of $\tilde{\psi}$. Thus
\begin{align*}
|\psi^{(n)}(x)|&= |\tilde{\psi}^{(n)}(x)|+\sum_{k=0}^n\sqrt{x-\beta}\left|\left(\Xi_V[\tilde{\psi}]\right)^{n+1}(z_k)+\frac{\xi^{(n+1)}(z_k)}{2}\right| \lesssim L,
\end{align*}
which completes the argument.
\end{proof}


\subsection{Energy Transport Errors}
Using the above estimates on the transport $\psi$, we seek to obtain the requisite control on $\tau_t$. By directly exploiting the symmetry of the definition, we can improve upon the estimate found in \cite{BLS18}. 

\begin{proof}[Proof of Lemma \ref{Energy Difference Estimate}]
The argument begins with a Taylor expansion. We write 
\begin{align*}
\tau_t(x)&=\int -\log |\phi_t(x)-\phi_t(y)|~d\muv(y)+V_t(\phi_t(x))-\int -\log |x-y|~d\muv(y)-V(x)+\tilde{c}_t \\
&=\int -\log \frac{|\phi_t(x)-\phi_t(y)|}{|x-y|}~d\muv(y)+V(\phi_t(x))-V(x)+t\xi(\phi_t(x))+\tilde{c}_t
\end{align*}
Note that 
\begin{align*}
-\log \frac{|\phi_t(x)-\phi_t(y)|}{|x-y|}
&=-t\frac{\psi(x)-\psi(y)}{x-y}+\frac{t^2}{2\alpha_x(y)^2}\left(\frac{\psi(x)-\psi(y)}{x-y}\right)^2,
\end{align*}
where $\alpha_x(y)$ is some number between $1$ and $1+t\frac{\psi(x)-\psi(y)}{x-y}$, which is guaranteed to be in the interval $[0.5,1.5]$ for large $N$ by our assumptions on $t\psi$. Similarly, for the potential term we find
\begin{align*}
V_t(\phi_t(x))-V(x)&=V(x+t\psi(x))+t\xi(x+t\psi(x))-V(x) \\
&=t\psi(x)V'(x)+\frac{V''(\beta(x))}{2}t^2\psi(x)^2+t\xi(x)+t^2\psi(x)\xi'(\gamma(x)),
\end{align*}
where $\beta(x)$ and $\gamma(x)$ are numbers between $x$ and $x+t\psi(x)$. Finally, recall that we can write 
\begin{align*}
\tilde{c}_t=tc_\xi&=t\Xi_V[\psi]-t\xi(x) =-t\psi(x)V'(x)+t\int \frac{\psi(x)-\psi(y)}{x-y}~d\muv(y)-t\xi(x)
\end{align*}
for $x \in U$. Putting all of this together, we find that 
\begin{equation*}
\tau_t(x)=\int \frac{t^2}{2\alpha_x(y)^2}\left(\frac{\psi(x)-\psi(y)}{x-y}\right)^2~d\muv(y)+\frac{V''(\beta(x))}{2}t^2\psi(x)^2+\xi'(\gamma(x))t^2\psi(x)
\end{equation*}
for $x \in U$. 
Since $U$ has compact closure, this immediately yields (\ref{general tau bound}).

We next use (\ref{general tau bound}) to obtain the desired decay estimates on $\tau_t$ in the mesoscopic case, (\ref{rescaled tau bound}). We assume again without loss of generality that $\xi=\xi_{0,L}$ (the proof is the same for other values of $z \in \B$) and make the same assumptions that we did in the proof of Lemma \ref{Rescaled Transport Estimates}. For $x \in 2\supp(\xi)=[-2cL, 2dL]$ we immediately estimate using Lemma \ref{Rescaled Transport Estimates}
\begin{align*}
|\tau_t(x)|&\lesssim t^2 \int \left(\frac{\psi(x)-\psi(y)}{x-y}\right)^2~d\muv(y)+t^2\psi(x)^2+\frac{\|\theta\|_{C^1}}{L}t^2\psi(x) \\
&\lesssim t^2L\|\psi'\|_{L^\infty}^2+t^2\|\psi\|_{L^\infty}^2\frac{1}{L}+t^2\|\psi\|_{L^\infty}^2+\frac{\|\theta\|_{C^1}}{L}t^2\|\psi\|_{L^\infty} \lesssim \frac{t^2}{L}\max(\|\theta\|_{C^1}, \|\theta\|_{C^2}^2)
\end{align*}
Next, we turn to decay estimates. Without loss of generality, we assume $x>2dL$ and $x \in U$, since the argument for $x<-2cL$ is symmetric. The reason for choosing $x>2dL$ is that since $\frac{1}{LN}\rightarrow 0$ and $\gamma(x)$ is between $x$ and $x+t\psi(x)$, we can conclude that $\gamma(x)\notin \supp(\xi)$ and so the $\xi'(\gamma(x))$ term vanishes for large enough $N$. Furthermore, we can bound the $t^2\psi(x)^2$ using Lemma \ref{Rescaled Transport Estimates} by $\lesssim \frac{L^2t^2}{|x|^2}\|\theta\|_{C^1}^2$. The integral term is a bit more delicate, and we turn to that now. First,
\begin{align*}
t^2\int_{-\infty}^{-cL}\left(\frac{\psi(x)-\psi(y)}{x-y}\right)^2~d\muv(y)&\lesssim t^2 \int_{y \leq -cL}\frac{\psi(x)^2}{(x-y)^2}~dy+t^2 \int_{y \leq -cL}\frac{\psi(y)^2}{(x-y)^2}~dy \\
&\lesssim \frac{t^2L^2\|\theta\|_{C^1}^2}{|x|^2}\int_{y \leq -cL}\frac{dy}{y^2}\lesssim \frac{t^2L \|\theta\|_{C^1}^2}{|x|^2}.
\end{align*}

Next, a bit more roughly (but only for a piece of size $L$!)
\begin{align*}
t^2\int_{-cL}^{dL}\left(\frac{\psi(x)-\psi(y)}{x-y}\right)^2~d\muv(y)&\lesssim t^2 \int_{-cL}^{-dL}\frac{\psi(x)^2}{(x-y)^2}~dy+t^2 \int_{-cL}^{dL}\frac{\psi(y)^2}{(x-y)^2}~dy \\
&\lesssim \frac{t^2L^3\|\theta\|_{C^1}^2}{|x|^4}+\frac{t^2 L \|\theta\|_{C^1}^2}{|x|^2} \lesssim \frac{t^2 L \|\theta\|_{C^1}^2}{|x|^2} 
\end{align*}
since $|x|\gtrsim L$. Finally, we estimate $t^2 \int_{dL}^\infty \left(\frac{\psi(x)-\psi(y)}{x-y}\right)^2~d\muv(y)$. First, removing an interval of size $\frac{|x|}{2}$ around $x$ (this stays in the domain of integration since $\frac{|x|}{2}>dL$) we have for some $\beta(y)$ with $|\beta(y)-x|\leq \frac{|x|}{2}$ that 
\begin{align*}
t^2 \int_{\{|y-x|\leq \frac{|x|}{2}\}} \left(\frac{\psi(x)-\psi(y)}{x-y}\right)^2~d\muv(y)&= t^2 \int_{|y-x|\leq \frac{|x|}{2}}(\psi'(\beta(y))^2)~d\muv(y)  \\
&\lesssim t^2L^2\|\theta\|_{C^2}^2\left(\frac{1}{|x|^2}+\frac{1}{|x|^4}\right)|x|
\end{align*}
since $|\psi'(\beta(y))|\lesssim L\|\theta\|_{C^1}\left(\frac{1}{|\beta(y)|}+\frac{1}{|\beta(y)|^2}\right)\lesssim L\|\theta\|_{C^1}\left(\frac{1}{|x|}+\frac{1}{|x|^2}\right)$.  Furthermore, for large enough $N$ we have $|x|\lesssim \frac{1}{L}$ and we always have $|x|\gtrsim L$, so we can bound
\begin{equation*}
t^2 \int_{\{|y-x|\leq \frac{|x|}{2}\}} \left(\frac{\psi(x)-\psi(y)}{x-y}\right)^2~d\muv(y)\lesssim t^2L^2\|\theta\|_{C^2}^2\left(\frac{1}{|x|^2}+\frac{1}{|x|^4}\right)|x| \lesssim  \frac{t^2L\|\theta\|_{C^2}^2}{|x|^2}.
\end{equation*}
The rest we can simply estimate as with $y \leq -cL$:
\begin{align*}
&t^2 \int_{\{|y-x|\geq \frac{|x|}{2}\}\cap \{y \geq dL\}} \left(\frac{\psi(x)-\psi(y)}{x-y}\right)^2~d\muv(y) \lesssim t^2 \int_{\{|y-x|\geq \frac{|x|}{2}\}\cap \{y \geq dL\}}\frac{\psi(x)^2+\psi(y)^2}{(x-y)^2}~d\muv(y) \\
&\lesssim t^2 \frac{L^2\|\theta\|_{C^1}^2}{|x|^2L}+t^2 \int_{\{|y-x|\geq \frac{|x|}{2}\}\cap \{y \geq dL\}}\frac{L^2\|\theta\|_{C^1}^2}{(x-y)^2y^2}~d\muv(y) \\
&\lesssim t^2 \frac{L\|\theta\|_{C^1}^2}{|x|^2}.
\end{align*}
This establishes the requisite bound, Equation \ref{rescaled tau bound}.

The same argument works to achieve the desired control on $|\tau_t(x)|$ for the local laws test functions, since all of the arguments above mostly use information only about the transport. Again, we assume that $a=0$ without loss of generality. Let's start from the general estimate (\ref{general tau bound}). By our assumptions on $N$, $|\gamma(x)|\lesssim Lh$. When $\xi=\zeta_L$ we can estimate using (\ref{first derivative formula}) and (\ref{first derivative bound})
\begin{equation*}
|\xi'(z)|=2\pi hL^2 \left|\frac{d}{dz} \frac{1}{z^2+(Lh)^2}\right|\lesssim hL^2 \frac{|p_1^1(z)|}{(z^2+(Lh)^2)^2} \lesssim hL^2 \frac{Lh+|z|}{(z^2+(Lh)^2)^2}
\end{equation*}
and when $\xi=\kappa_L$ we can estimate with (\ref{second derivative formula}) and (\ref{second derivative bound})
\begin{equation*}
|\xi'(z)|=2\pi L\left|\frac{d}{dz} \frac{z}{z^2+(Lh)^2}\right| \lesssim L \frac{|p_1^2(z)|}{(z^2+(Lh)^2)^2} \lesssim L \frac{(Lh)^2+|z|^2}{(z^2+(Lh)^2)^2}.
\end{equation*}
Let's first look at $|x|\leq 2Lh$. In either case, since $|\gamma(x)|\lesssim Lh$ we have $|\xi'(\gamma(x))|\lesssim \frac{1}{Lh^2}$. Coupling the estimates of Lemma \ref{Local Laws Transport Estimates} with (\ref{general tau bound}) we find 
\begin{align*}
|\tau_t(x)|&\lesssim t^2 \int \left(\frac{\psi(x)-\psi(y)}{x-y}\right)^2~d\muv(y)+t^2\psi(x)^2+|\xi'(\gamma(x)|t^2\psi(x) \\
&\lesssim t^2\int_{-2Lh}^{2Lh}\|\psi'\|_{L^\infty}^2+t^2\int_{\Sigma_V\setminus [-2Lh,2Lh]}\frac{\|\psi\|_{L^\infty}^2}{(x-y)^2}~dy+t^2\|\psi\|_{L^\infty}^2+\frac{t^2}{Lh^2}\|\psi\|_{L^\infty} \\
&\lesssim \frac{t^2Lh}{L^2h^4}+\frac{t^2}{h^2Lh}+\frac{t^2}{Lh^3} \lesssim \frac{t^2}{Lh^3}.
\end{align*}
Now, decay estimates. Without loss of generality, we assume $x>2Lh$ and $x \in U$, since the argument for $x<-2Lh$ is symmetric. The reason for choosing $x>2Lh$ is that since $\gamma(x)$ is between $x$ and $x+t\psi(x)$ and $\|t\psi\|_{L^\infty}<Lh$, we can conclude that $\gamma(x)\geq Lh$ and $\gamma(x) \geq \frac{x}{2}$. In particular, using the above estimates on $|\xi'|$ for $\zeta_L$ and $\kappa_L$, we can conclude immediately here that $|\xi'(\gamma(x))| \lesssim \frac{L}{|x|^2}$ and so $\frac{|\xi'(\gamma(x))|}{2}t^2|\psi(x)|\lesssim \frac{t^2L}{|x|^2}\frac{L}{|x|} \lesssim \frac{t^2L^2}{|x|^3}$. Since $|x|\gtrsim Lh$, we can just bound this as $\frac{t^2L}{h|x|^2}$. Furthermore, we can bound $t^2\psi(x)^2$ using Lemma \ref{Local Laws Transport Estimates} by $\lesssim \frac{L^2t^2}{|x|^2}\lesssim \frac{Lt^2}{h|x|^2}$ as well, since $Lh \ll1$. 

We next turn to the integral term, which is estimated exactly as with the rescaled test function. First,
\begin{align*}
t^2\int_{-\infty}^{-Lh}\left(\frac{\psi(x)-\psi(y)}{x-y}\right)^2~d\muv(y)&\lesssim t^2 \int_{y \leq -Lh}\frac{\psi(x)^2}{(x-y)^2}~dy+t^2 \int_{y \leq -Lh}\frac{\psi(y)^2}{(x-y)^2}~dy \\
&\lesssim \frac{t^2L^2}{|x|^2}\int_{y \leq -cL}\frac{dy}{y^2}\lesssim \frac{t^2L }{h|x|^2}.
\end{align*}

Next, a bit more roughly (but again only for a piece of size $Lh$!)
\begin{align*}
t^2\int_{-Lh}^{Lh}\left(\frac{\psi(x)-\psi(y)}{x-y}\right)^2~d\muv(y)&\lesssim t^2 \int_{-Lh}^{-Lh}\frac{\psi(x)^2}{(x-y)^2}~dy+t^2 \int_{-Lh}^{Lh}\frac{\psi(y)^2}{(x-y)^2}~dy \\
&\lesssim \frac{t^2L^3h}{|x|^4}+\frac{t^2 Lh}{h^2|x|^2} \lesssim \frac{t^2 L }{h|x|^2} 
\end{align*}
since $|x|\gtrsim Lh$. Finally, we estimate $t^2 \int_{Lh}^\infty \left(\frac{\psi(x)-\psi(y)}{x-y}\right)^2~d\muv(y)$. First, removing an interval of size $\frac{|x|}{2}$ around $x$ (this stays in the domain of integration since $\frac{|x|}{2}>Lh$) we have for some $\beta(y)$ with $|\beta(y)-x|\leq \frac{|x|}{2}$ that 
\begin{align*}
t^2 \int_{\{|y-x|\leq \frac{|x|}{2}\}} \left(\frac{\psi(x)-\psi(y)}{x-y}\right)^2~d\muv(y)&= t^2 \int_{|y-x|\leq \frac{|x|}{2}}(\psi'(\beta(y))^2)~d\muv(y)  \\
&\lesssim t^2L^2\left(\frac{1}{|x|^2}+\frac{1}{|x|^4}\right)|x|
\end{align*}
since $|\psi'(\beta(y))|\lesssim L\left(\frac{1}{|\beta(y)|}+\frac{1}{|\beta(y)|^2}\right)\lesssim L\left(\frac{1}{|x|}+\frac{1}{|x|^2}\right)$.  Furthermore, for large enough $N$ we have $|x|\lesssim \frac{1}{Lh}$ and we always have $|x|\gtrsim Lh$, so we can bound
\begin{equation*}
t^2 \int_{\{|y-x|\leq \frac{|x|}{2}\}} \left(\frac{\psi(x)-\psi(y)}{x-y}\right)^2~d\muv(y)\lesssim t^2L^2\left(\frac{1}{|x|^2}+\frac{1}{|x|^4}\right)|x| \lesssim  \frac{t^2L}{h|x|^2}.
\end{equation*}
The rest we can simply estimate as with $y \leq -Lh$:
\begin{align*}
&t^2 \int_{\{|y-x|\geq \frac{|x|}{2}\}\cap \{y \geq Lh\}} \left(\frac{\psi(x)-\psi(y)}{x-y}\right)^2~d\muv(y)\lesssim t^2 \int_{\{|y-x|\geq \frac{|x|}{2}\}\cap \{y \geq Lh\}}\frac{\psi(x)^2+\psi(y)^2}{(x-y)^2}~d\muv(y) \\
&\lesssim t^2 \frac{L^2}{|x|^2Lh}+t^2 \int_{\{|y-x|\geq \frac{|x|}{2}\}\cap \{y \geq Lh\}}\frac{L^2}{(x-y)^2y^2}~d\muv(y) \\
&\lesssim t^2 \frac{L}{h|x|^2}.
\end{align*}
This establishes (\ref{local laws tau bound}).
\end{proof} 

In the case of the rescaled test function, we also need good control on $|\tau_t'(x)|$. For our purposes, a rough $L^\infty$ bound suffices.

\begin{proof}[Proof of Lemma \ref{Energy Difference Derivative Estimate}]
We recall from the proof of Lemma \ref{Energy Difference Estimate} that we can expand $\tau_t(x)$ as 
\begin{align*}
\tau_t(x)=f_1(x)+f_2(x)+f_3(x),
\end{align*}
where 
\begin{align*}
f_1(x)&=\int -\log \frac{|\phi_t(x)-\phi_t(y)|}{|x-y|}+t\frac{\psi(x)-\psi(y)}{x-y}~d\muv(y) \\
f_2(x)&=V(\phi_t(x))-V(x)-t\psi(x)V'(x) \\
f_3(x)&=t\xi(\phi_t(x))-t\xi(x).
\end{align*}
We estimate each term separately, differentiating and making use of a Taylor expansion. For $f_1(x)$, notice that $\frac{\phi_t(x)-\phi_t(y)}{x-y}=1+t\frac{\psi(x)-\psi(y)}{x-y}$, so that 
\begin{equation*}
f_1(x)=\int h(1+tz_y(x))-h'(1)tz_y(x)~d\muv(y),
\end{equation*}
with $h(z)=-\log z$ and $z_y(x)=\frac{\psi(x)-\psi(y)}{x-y}$. Differentiating, we have 
\begin{equation*}
f_1'(x)=\int (h'(1+tz_y(x))tz_y'(x)-h'(1)tz_y'(x))~d\muv(y),
\end{equation*}
so that with a Taylor expansion we observe 
\begin{equation*}
|f_1'(x)|\leq \int |tz_y'(x)\|h'(1+tz_y(x))-h'(1)|~d\muv(y) \lesssim \int t^2 |z_y'(x)\|z_y(x)|~d\muv(y)
\end{equation*}
since $|tz_y(x)|$ is small by assumption and so $h''$ is well controlled. Using our computations from the proof of Lemma \ref{Rescaled Transport Estimates}, we control
\begin{align*}
|f_1'(x)| \lesssim t^2 \int \left|\frac{\psi(y)-\psi(x)}{y-x}\right|\left|\frac{\psi(y)-(\psi(x)+\psi'(x)(y-x))}{(y-x)^2}\right|~d\muv(y).
\end{align*}
Using a mean value argument for small $|x-y|$ and simple $L^\infty$ bounds for large $|x-y|$, we control using Lemma \ref{Rescaled Transport Estimates}
\begin{align*}
|f_1'(x)|&\lesssim t^2 \int_{|x-y|\leq L} \|\psi'\|_{L^\infty}|\psi''|_{L^\infty}~d\muv(y)+t^2\int_{|x-y|>L}\frac{\|\psi\|_{L^\infty}^2}{|y-x|^3}+\frac{\|\psi\|_{L^\infty}\|\psi'\|_{L^\infty}}{|y-x|^2}~d\muv(y) \\
&\lesssim t^2L\frac{\|\theta\|_{C^{3}}\|\theta\|_{C^{4}}}{L^3}+t^2 \int_{|x-y|>L} \frac{\|\theta\|_{C^2}^2}{|y-x|^3}+\frac{\|\theta\|_{C^2}\|\theta\|_{C^3}}{L|y-x|^2}~dy \\
&\lesssim t^2 \frac{\|\theta\|_{C^{4}}^2}{L^2}.
\end{align*}
A similar argument handles $f_2'$ and $f_3'$. Differentiating $f_2$ yields 
\begin{align*}
f_2'(x)&=\frac{V'(x+t\psi(x))(1+t\psi'(x))-V'(x)-t\psi'(x)V'(x)-t\psi(x)V''(x)}{2} \\
&=\frac{(V'(x+t\psi(x))-V'(x)-V''(x)t\psi(x))+t\psi'(x)(V'(x+t\psi(x))-V'(x))}{2}.
\end{align*}
Using a mean value argument and Lemma \ref{Rescaled Transport Estimates} yields 
\begin{align*}
|f_2'(x)|&\lesssim t^2\psi(x)^2+t^2|\psi'(x)\|\psi(x)|\lesssim t^2\|\theta\|_{C^2}^2+t^2\frac{\|\theta\|_{C^3}\|\theta\|_{C^2}}{L},
\end{align*}
since the compact closure of $U$ guarantees that $V$ is controlled. This is of course also controlled by $t^2 \frac{\|\theta\|_{C^{4}}^2}{L^2}$. Finally, we can differentiate $f_3$ to find
\begin{align*}
f_3'(x)&=\frac{t\xi'(x+t\psi(x))(1+t\psi'(x))-t\xi'(x)}{2} \\
&=\frac{t(\xi'(x+t\psi(x))-\xi'(x))+t^2\xi'(x+t\psi(x))\psi'(x)}{2}.
\end{align*}
Again using a mean value argument and Lemma \ref{Rescaled Transport Estimates}, we find 
\begin{align*}
|f_3'(x)|&\lesssim t^2\|\xi''\|_{L^\infty}|\psi(x)|+t^2\|\xi'\|_{L^\infty}|\psi'(x)| \lesssim t^2 \frac{\|\theta\|_{C^2}^2}{L^2}+t^2 \frac{\|\theta\|_{C^1}\|\theta\|_{C^3}}{L^2},
\end{align*}
which is of course again controlled by $t^2 \frac{\|\theta\|_{C^{4}}^2}{L^2}$, completing the argument.
\end{proof}

\subsection{Technical Estimates for $\textsf{Error}_2$}
Here we prove control of the various $H^{1/2}$ norms required to estimate $\textsf{Error}_2$ in $\S \ref{Section CLT Estimate}$ (see (\ref{first anisotropy function})). Recall that $H^{1/2}$ is defined in Equation (\ref{H1/2 definition}) as the $L^2$ norm of the gradient of the harmonic extension, which is all the minimum $L^2$ gradient of all extensions. Thus, in particular, we could define an extension $\tilde{\chi}(x,y)=\xi(x)\chi(y)$ where $\chi(y)$ is a function that is identically $1$ for $|y|\leq 1$ and vanishes outside of $[-(\mathfrak{h}+1),\mathfrak{h}+1]$ with derivative bounded by $\frac{1}{\mathfrak{h}}$. Using the minimality of the harmonic extension and estimating the gradient of this extension in $L^2$, we have the bound
\begin{equation}\label{extension estimate}
\|\xi\|_{H^{1/2}}\lesssim\frac{1}{\sqrt{\mathfrak{h}}}\|\xi\|_{L^2}+\sqrt{\mathfrak{h}}\|\xi'\|_{L^2}.
\end{equation}
Optimizing this estimate over $h$ is often a convenient way of estimating a test function in $H^{1/2}$. 

First, we prove (\ref{first anisotropy function}).
\begin{proof}[Proof of (\ref{first anisotropy function})]
Here we are interesting in controlling 
\begin{equation*}
\xi(\cdot)=e^{im \cdot}\chi(\cdot),
\end{equation*}
where $\chi(x)$ is supported on $I_i \pm 2^{i}L$ (where $I_i=\square_{2^iL}$)and identically one on $I_i \pm 2^{i-1}L$ with $C^k$ norm controlled by $\frac{1}{(2^iL)^k}$ (for instance, by rescaling a compactly supported bump)
We use (\ref{extension estimate}) on real and imaginary parts, allowing us to write 
\begin{equation*}
\left\|e^{im \cdot}\chi(\cdot)\right\|_{H^{1/2}}\leq \frac{1}{\sqrt{\mathfrak{h}}}\|e^{im \cdot}\chi(\cdot)\|_{L^2}+\sqrt{\mathfrak{h}}\|(e^{im \cdot}\chi(\cdot))'\|_{L^2}.
\end{equation*}
Computing, we find 
\begin{equation*}
\|e^{im \cdot}\chi(\cdot)\|^2_{L^2} \lesssim \int_{-2^iL}^{2^iL} |\chi(x)|^2~dx \lesssim 2^iL
\end{equation*}
and 
\begin{equation*}
\|(e^{im \cdot}\chi(\cdot))'\|^2_{L^2} \lesssim \int_{-2^iL}^{2^iL} m^2~dx+\int_{-2^iL}^{2^iL} |\chi'(x)|^2~dx \lesssim (2^iL)m^2+\frac{1}{2^iL}=\frac{1+m^2(2^iL)^2}{2^iL}.
\end{equation*}
Choosing $\mathfrak{h}=\frac{\sqrt{2^iL}}{(1+m^2L^2)^{1/4}}$ and squaring, we obtain (\ref{first anisotropy function}).
\end{proof}

We turn now to proving the control (\ref{second anisotropy function}). 
 
\begin{proof}[Proof of (\ref{second anisotropy function})]
Here we are interested in controlling
\begin{equation*}
\xi(\cdot)=e^{i\lambda \cdot}\chi(\cdot)g(\cdot)
\end{equation*}
where $g$ is given by 
\begin{equation*}
g(x)=\int \frac{1-\chi(y)}{x-y}~d\fluct_N(y),
\end{equation*}
and $\chi$ is as in the previous proof. The tricky part of this computation lies in controlling $g$ and $g'$, which will require using the local energy estimate of Lemma \ref{CLT Energy Estimate}. Let $\{\zeta_i\}$ denote the same partition of unity that we have been using from Propositions \ref{easy Error2 bound - local laws} and \ref{easy Error2 bound - rescaled test function}; we will use this to respect the decay over scales in our estimate. So,
\begin{equation*}
|g(x)|\leq \sum_{k=1}^{k_*} \left|\int_{I_k} \frac{1-\chi(y)}{x-y}\zeta_k(y)~d\fluct_N(y)\right|+\left|\int_{I_{k_*+1}} \frac{1-\chi(y)}{x-y}\zeta_{k_*+1}(y)~d\fluct_N(y)\right|.
\end{equation*}
Let us start with the sum over $k \leq k_*$. Notice that $1-\chi(y)$ vanishes on $I_i \pm 2^{i-1}L$, so we can both remove the interval $I_i$ and also observe that $|x-y|\gtrsim 2^{i-1}L$, which will be of use when estimating the contributions to the fluctuation integral from $I_{i-1}$ and $I_{i+1}$. Furthermore, $1-\chi(y) \equiv 1$ for all $I_k$ with $|k-i|\geq 2$. Hence, using the energy estimate Lemma \ref{CLT Energy Estimate} and our energy control on $\mathcal{G}_N$, we find
\begin{align}
\label{first k}&\sum_{k=1}^{k_*} \left|\int_{I_k} \frac{1-\chi(y)}{x-y}\zeta_k(y)~d\fluct_N(y)\right| \\
\nonumber&\lesssim \sum_{k \ne i} \left\|\left(\frac{1-\chi(y)}{x-y}\zeta_k(y)\right)'\right\|_{L^\infty(I_k)}(|I_k|+N^{-3/4}|I_k|^{1/4}\sqrt{2^kLN})+ \\
\nonumber&\sum_{k \ne i} \left(\sqrt{\mathfrak{h}_k}\left\|\left(\frac{1-\chi(y)}{x-y}\zeta_k(y)\right)'\right\|_{L^2(I_k)}+\frac{1}{\sqrt{\mathfrak{h}_k}}\left\|\frac{1-\chi(y)}{x-y}\zeta_k(y)\right\|_{L^2(I_k)}\right)\sqrt{2^kLN}.
\end{align}
For $k=i \pm1$, we bound crudely 
\begin{equation*}
\left\|\left(\frac{1-\chi(y)}{x-y}\zeta_k(y)\right)'\right\|_{L^\infty(I_k)}\lesssim \frac{1}{(2^iL)^2}
\end{equation*}
and 
\begin{equation*}
\left\|\frac{1-\chi(y)}{x-y}\zeta_k(y)\right\|_{L^2(I_k)}^2 \lesssim \frac{2^iL}{(2^iL)^2}=\frac{1}{2^iL}, \hspace{3mm} \left\|\left(\frac{1-\chi(y)}{x-y}\zeta_k(y)\right)'\right\|_{L^2(I_k)}^2 \lesssim \frac{2^iL}{(2^iL)^4}=\frac{1}{(2^iL)^3}.
\end{equation*}
For $|k-i| \geq 2$, we observe that $1-\chi(y)\equiv 1$ and correspondingly control 
\begin{equation*}
\left\|\frac{1-\chi(y)}{x-y}\zeta_k(y)\right\|_{L^2(I_k)}^2\lesssim \int_{I_k}\frac{1}{(x-y)^2}~dy \lesssim \frac{1}{\dist(x,I_k)}\lesssim\frac{1}{\max(2^{i+1}, 2^{k-1})L}
\end{equation*}
since $x \in I_i$, and 
\begin{equation*}
\left\|\left(\frac{1-\chi(y)}{x-y}\zeta_k(y)\right)'\right\|_{L^2(I_k)}^2\lesssim \int_{I_k}\frac{1}{(x-y)^4}~dy \lesssim \frac{1}{\dist(x,I_k)^3}\lesssim\frac{1}{\max(2^{i+1}, 2^{k-1})^3L^3}.
\end{equation*}
Furthermore, 
\begin{equation*}
\left\|\left(\frac{1-\chi(y)}{x-y}\zeta_k(y)\right)'\right\|_{L^\infty(I_k)}\lesssim \frac{1}{\dist(x,I_k)^2}\lesssim \frac{1}{(\max(2^{i+1}, 2^{k-1})L)^2}.
\end{equation*}
We control then the sum \ref{first k} over $k$ in pieces. We first have
\begin{align*}
& \sum_{k=0}^{i+1} \frac{1}{(2^iL)^2}(|2^kL|+N^{-3/4}|2^kL|^{1/4}\sqrt{2^kLN})+\sum_{k=0}^{i+1} \left(\sqrt{\mathfrak{h}_k}\frac{1}{(2^iL)^{3/2}}+\frac{1}{\sqrt{\mathfrak{h}_k}}\frac{1}{\sqrt{2^iL}}\right)\sqrt{2^kLN} \\
&\lesssim \frac{1}{2^iL}+\frac{1}{(2^iL)^{5/4}N^{1/4}}+\frac{\sqrt{N}}{\sqrt{2^iL}}
\end{align*}
by choosing $\mathfrak{h}_k=2^{i}L$, since the $2^iL$ terms are dominant in the sum. The middle term can be neglected, since $\frac{1}{(2^iL)^{5/4}N^{1/4}}=\frac{1}{2^iL}\frac{1}{(2^iLN)^{1/4}}\lesssim \frac{1}{2^iL}$ for large enough $N$, as $\frac{1}{LN}\rightarrow 0$. Next, we have 
\begin{align*}
& \sum_{k=i+1}^{k_*} \frac{1}{(2^kL)^2}(|2^kL|+N^{-3/4}|2^kL|^{1/4}\sqrt{2^kLN})+\sum_{k=i+1}^{k_*} \left(\sqrt{\frac{h_k}{N}}\frac{1}{(2^kL)^{3/2}}+\frac{1}{\sqrt{\frac{h_k}{N}}}\frac{1}{\sqrt{2^kL}}\right)\sqrt{2^kLN} \\
&\lesssim \sum_{k=i+1}^{\log (1/L)}\left( \frac{1}{2^kL}+\frac{1}{(2^kL)^{5/4}N^{1/4}}+\frac{\sqrt{N}}{\sqrt{2^kL}}\right) \lesssim \frac{1}{2^iL}+\frac{\sqrt{N}}{\sqrt{2^iL}}
\end{align*}
by choosing $h_k=2^kLN$, since $k_* \sim \log \left(\frac{1}{L}\right)$. Finally, $\frac{1}{LN}\rightarrow 0$ so for large enough $N$ we have for any $i$ that $\frac{1}{\sqrt{2^iLN}}<1$, and in particular $\frac{1}{2^iL}=\frac{\sqrt{N}}{\sqrt{2^iL}\sqrt{2^iLN}}<\frac{\sqrt{N}}{\sqrt{2^iL}}$. 

For $k_*+1$ we cannot use the local laws of Theorem \ref{Local Law}, but we have already reached the macroscale so we can accomplish our goals with global estimates. Using Lemma \ref{CLT Energy Estimate} and the global energy bound, we have
\begin{align*}
&\left|\int_{I_{k_*+1}} \frac{1-\chi(y)}{x-y}\zeta_{k_*+1}(y)~d\fluct_N(y)\right| \\
&\lesssim  \left\|\left(\frac{1-\chi(y)}{x-y}\zeta_{k_*+1}(y)\right)'\right\|_{L^\infty(I_{k_*+1})}(|I_{k_*+1}|+N^{-3/4}|I_{k_*+1}|^{1/4}\sqrt{N})+ \\
& \left(\sqrt{\mathfrak{h}_{k_*+1}}\left\|\left(\frac{1-\chi(y)}{x-y}\zeta_{k_*+1}(y)\right)'\right\|_{L^2(I_{k_*+1})}+\frac{1}{\sqrt{\mathfrak{h}_{k_*+1}}}\left\|\frac{1-\chi(y)}{x-y}\zeta_{k_*+1}(y)\right\|_{L^2(I_{k_*+1})}\right)\sqrt{N} \\
&\lesssim 1
\end{align*}
choosing $\mathfrak{h}_{k_*+1}=1$, since all norms are controlled at order $1$. This is smaller than the other terms in \ref{first k}, so we conclude
\begin{equation}\label{control of g}
|g(x)|\lesssim \frac{\sqrt{N}}{\sqrt{2^iL}}.
\end{equation}

Similarly, we observe that $g'(x)=\int -\frac{1-\chi(y)}{(x-y)^2}~d\fluct_N(y)$. Exactly as above, we can write
\begin{equation*}
|g'(x)|\leq \sum_{k=1}^{k_*} \left|\int_{I_k} -\frac{1-\chi(y)}{(x-y)^2}\zeta_k(y)~d\fluct_N(y)\right|+\left|\int_{I_{k_*+1}} -\frac{1-\chi(y)}{(x-y)^2}\zeta_{k_*+1}(y)~d\fluct_N(y)\right|.
\end{equation*}
Estimating exactly as above, we can write on $\mathcal{G}_N$ using Lemma \ref{CLT Energy Estimate} that the first $k_*$ terms are controlled by 
\begin{align}
\label{first k round two}&\sum_{k=1}^{k_*} \left\|\left(\frac{1-\chi(y)}{(x-y)^2}\zeta_k(y)\right)'\right\|_{L^\infty(I_k)}(|I_k|+N^{-3/4}|I_k|^{1/4}\sqrt{2^kLN})+ \\
\nonumber&\sum_{k=1}^{k_*} \left(\sqrt{\mathfrak{h}_k}\left\|\left(\frac{1-\chi(y)}{(x-y)^2}\zeta_k(y)\right)'\right\|_{L^2(I_k)}+\frac{1}{\sqrt{\mathfrak{h}_k}}\left\|\frac{1-\chi(y)}{(x-y)^2}\zeta_k(y)\right\|_{L^2(I_k)}\right)\sqrt{2^kLN}.
\end{align}
For $k=i \pm1$, we bound crudely 
\begin{equation*}
\left\|\left(\frac{1-\chi(y)}{(x-y)^2}\zeta_k(y)\right)'\right\|_{L^\infty(I_k)}\lesssim \frac{1}{(2^iL)^3}
\end{equation*} 
and 
\begin{equation*}
\left\|\frac{1-\chi(y)}{(x-y)^2}\zeta_k(y)\right\|_{L^2(I_k)}^2 \lesssim \frac{2^iL}{(2^iL)^4}=\frac{1}{(2^iL)^3},\hspace{3mm} \left\|\left(\frac{1-\chi(y)}{(x-y)^2}\zeta_k(y)\right)'\right\|_{L^2(I_k)}^2 \lesssim \frac{2^iL}{(2^iL)^6}=\frac{1}{(2^iL)^5}.
\end{equation*} 
For $|k-i| \geq 2$, we observe that $1-\chi(y)\equiv 1$ and correspondingly control 
\begin{equation*}
\left\|\frac{1-\chi(y)}{(x-y)^2}\zeta_k(y)\right\|_{L^2(I_k)}^2\lesssim \int_{I_k}\frac{1}{(x-y)^4}~dy \lesssim \frac{1}{\dist(x,I_k)^3}\lesssim\frac{1}{(\max(2^{i+1}, 2^{k-1})L)^3}
\end{equation*}
since $x \in I_i$, and 
\begin{equation*}
\left\|\left(\frac{1-\chi(y)}{(x-y)^2}\zeta_k(y)\right)'\right\|_{L^2(I_k)}^2\lesssim \int_{I_k}\frac{1}{(x-y)^6}~dy \lesssim \frac{1}{\dist(x,I_k)^5}\lesssim\frac{1}{\max(2^{i+1}, 2^{k-1})^5L^5}.
\end{equation*}
Furthermore, $\left\|\left(\frac{1-\chi(y)}{(x-y)^2}\zeta_k(y)\right)'\right\|_{L^\infty(I_k)}\lesssim \frac{1}{\dist(x,I_k)^3}\lesssim \frac{1}{(\max(2^{i+1}, 2^{k-1})L)^3}$. We control then the sum (\ref{first k round two}) over $k$ in pieces. We first have
\begin{align*}
& \sum_{k=0}^{i+1} \frac{1}{(2^iL)^3}(|2^kL|+N^{-3/4}|2^kL|^{1/4}\sqrt{2^kLN})+\sum_{k=0}^{i+1} \left(\sqrt{\mathfrak{h}_k}\frac{1}{(2^iL)^{5/2}}+\frac{1}{\sqrt{\mathfrak{h}_k}}\frac{1}{(2^iL)^{3/2}}\right)\sqrt{2^kLN} \\
&\lesssim \frac{1}{(2^iL)^2}+\frac{1}{(2^iL)^{9/4}N^{1/4}}+\frac{\sqrt{N}}{(2^iL)^{3/2}}
\end{align*}
by choosing $\mathfrak{h}_k=2^{i}L$, since the $2^iL$ terms are dominant in the sum. The middle term can be neglected, since $\frac{1}{(2^iL)^{9/4}N^{1/4}}=\frac{1}{(2^iL)^2}\frac{1}{(2^iLN)^{1/4}}\lesssim \frac{1}{(2^iL)^2}$ for large enough $N$, as $\frac{1}{LN}\rightarrow 0$. Next, we have 
\begin{align*}
& \sum_{k=i+1}^{k_*} \frac{1}{(2^kL)^3}(|2^kL|+N^{-3/4}|2^kL|^{1/4}\sqrt{2^kLN})+\sum_{k=i+1}^{k_*} \left(\sqrt{\mathfrak{h}_k}\frac{1}{(2^kL)^{5/2}}+\frac{1}{\sqrt{\mathfrak{h}_k}}\frac{1}{(2^kL)^{3/2}}\right)\sqrt{2^kLN} \\
&\lesssim \sum_{k=i+1}^{\log (1/L)}\left( \frac{1}{(2^kL)^2}+\frac{1}{(2^kL)^{9/4}N^{1/4}}+\frac{\sqrt{N}}{(2^kL)^{3/2}}\right) \lesssim \frac{1}{(2^iL)^2}+\frac{\sqrt{N}}{(2^iL)^{3/2}}.
\end{align*}
by choosing $\mathfrak{h}_k=2^kL$, since $k_*\sim \log \left(\frac{1}{L}\right)$. Finally, $\frac{1}{LN}\rightarrow 0$ so for large enough $N$ we have for any $i$ that $\frac{1}{\sqrt{2^iLN}}<1$, and in particular $\frac{1}{2^iL}=\frac{\sqrt{N}}{\sqrt{2^iL}\sqrt{2^iLN}}<\frac{\sqrt{N}}{\sqrt{2^iL}}$. 

For $k_*+1$ we again cannot use the local laws of Theorem \ref{Local Law}, but we have already reached the macroscale so we can accomplish our goals with global estimates. Using Lemma \ref{CLT Energy Estimate} and the global energy bound, we have
\begin{align*}
&\left|\int_{I_{k_*+1}} -\frac{1-\chi(y)}{(x-y)^2}\zeta_{k_*+1}(y)~d\fluct_N(y)\right| \\
&\lesssim  \left\|\left(\frac{1-\chi(y)}{(x-y)^2}\zeta_{k_*+1}(y)\right)'\right\|_{L^\infty(I_{k_*+1})}(|I_{k_*+1}|+N^{-3/4}|I_{k_*+1}|^{1/4}\sqrt{N})+ \\
& \left(\sqrt{\mathfrak{h}_{k_*+1}}\left\|\left(\frac{1-\chi(y)}{(x-y)^2}\zeta_{k_*+1}(y)\right)'\right\|_{L^2(I_{k_*+1})}+\frac{1}{\sqrt{\mathfrak{h}_{k_*+1}}}\left\|\frac{1-\chi(y)}{(x-y)^2}\zeta_{k_*+1}(y)\right\|_{L^2(I_{k_*+1})}\right)\sqrt{N} \\
&\lesssim 1
\end{align*}
choosing $\mathfrak{h}_{k_*+1}=1$, since all norms are controlled at order $1$. This is smaller than the other terms in \ref{first k}, so we conclude
\begin{equation}\label{control of g'}
|g'(x)|\lesssim \frac{\sqrt{N}}{(2^iL)^{3/2}}.
\end{equation}
With (\ref{control of g}) and (\ref{control of g'}) in tow, we can control 
\begin{equation*}
\|\xi\|_{L^2}^2 \lesssim \int_{-2^iL}^{2^iL} |g(x)|^2\lesssim 2^iL \frac{N}{2^iL}=N
\end{equation*}
and 
\begin{equation*}
\|\xi'\|_{L^2}^2 \lesssim \lambda ^2 \int_{-2^iL}^{2^iL} |g(x)|^2+\frac{1}{(2^iL)^2}\int_{-2^iL}^{2^iL} |g(x)|^2+\int_{-2^iL}^{2^iL}|g'(x)|^2 \lesssim \left(\lambda^2+\frac{1}{(2^iL)^2}\right)N.
\end{equation*}
Using the extension estimate (\ref{extension estimate})
\begin{equation*}
\|\xi\|_{H^{1/2}} \leq \sqrt{\mathfrak{h}}\|\xi\|_{L^2}+\frac{1}{\sqrt{\mathfrak{h}}}\|\xi'\|_{L^2}
\end{equation*}
and choosing $\mathfrak{h}=\left(\lambda^2+\frac{1}{(2^iL)^2}\right)^{1/4},$ we obtain (\ref{second anisotropy function}).
\end{proof}

\setcounter{section}{4}
\setcounter{subsection}{0}
\setcounter{prop}{0}
\setcounter{equation}{0}
\section*{Appendix D: Proof of Corollary \ref{pseudo-Gaussian moments}}\label{Appendix Moment Bounds}
In this section we prove Corollary \ref{pseudo-Gaussian moments}. We first prove that if there are constants $a,b,c>0$ and an event $\mathcal{G}$ such that for all $s$
\begin{equation}\label{pseudo-Gaussian}
\Esp \left[\exp(sX)\indic_{\mathcal{G}}\right] \leq \exp \left(a|s|+bs^2+c|s|^3\right)
\end{equation}
then for all $k \in \mathbb{N}$,
\begin{equation}\label{moment bound}
\Esp|X|^k\indic_{\mathcal{G}} \lesssim a^k+ b^{\frac{k}{2}}+c^{\frac{k}{3}}
\end{equation}
for a constant only dependent on $k$. This type of control and proof is standard for sub-Gaussian random variables, but we cannot find it for our exact case so we include it here for completeness. 

First, 
\begin{align*}
\Esp[|X|^k\indic_{\mathcal{G}}]&=\int_0^\infty \P(|X|^k\indic_{\mathcal{G}} >t)~dt =\int_0^\infty \P(\{|X|>t^{1/k} \}\cap \mathcal{G})~dt.
\end{align*}
Now let $t$ be arbitrary, and we attempt to recover the typical subGaussian tails. Using a Chernoff bound, we write for any $s>0$ that
\begin{align*}
\PNbeta(\{X>t \} \cap \mathcal{G})&=\PNbeta\left(\left\{\exp \left(sX\right)>e^{st} \right\} \cap \mathcal{G}\right) \\
&\leq \exp \left(as+bs^2+cs^3-st\right).
\end{align*}
We have three separate cases.

If $t \leq a$, then we just bound $\PNbeta(\{X>t \} \cap \mathcal{G}) \leq 1$.

Next, if $s \leq \frac{b}{c}$, then we can bound 
\begin{equation*}
\PNbeta(\{X>t \} \cap \mathcal{G}) \leq \exp (2bs^2-s(t-a))
\end{equation*}
and optimizing with $s=\frac{t-a}{4b}$ we find 
\begin{equation*}
\PNbeta(\{X>t \} \cap \mathcal{G}) \leq \exp \left(-\frac{(t-a)^2}{8b}\right).
\end{equation*}
We have $s \leq \frac{b}{c}$ though only for $t \leq a+ \frac{4b^2}{c}$. Thus, we use this bound for $a \leq t \leq a+\frac{4b^2}{c}$. 

For $t>a+\frac{4b^2}{c}$, we choose $s=\sqrt{\frac{t-a}{4c}}$ and observe that as soon as $t>a+\frac{4b^2}{c}$ we have $s>\frac{b}{c}$ and can estimate
\begin{equation*}
as+bs^2+cs^3-st \leq 2cs^3-s(t-a)=\frac{2c(t-a)^{\frac{3}{2}}}{(4c)^{\frac{3}{2}}}-\frac{(t-a)^{\frac{3}{2}}}{\sqrt{4c}}=-\frac{(t-a)^{\frac{3}{2}}}{4\sqrt{c}}.
\end{equation*}
Thus, with $s=\sqrt{\frac{t-a}{4c}}$
\begin{equation*}
\PNbeta(\{X>t \} \cap \mathcal{G}) \leq \exp (as+bs^2+cs^3-st) \leq \exp \left(-\frac{(t-a)^{\frac{3}{2}}}{4\sqrt{c}}\right).
\end{equation*}
The analogous computation works as well for $X<-t$. So, inserting these estimates into the distributional integral above, we find
\begin{align*}
&\int_0^\infty \PNbeta(\{|X|>t^{1/k} \}\cap \mathcal{G})~dt \\
&\leq 2\int_0^{a^k}~dt+ 2\int_{a^k}^{\left(a+\frac{4b^2}{c}\right)^k}\exp \left(-\frac{\left(t^{\frac{1}{k}}-a\right)^2}{8b}\right)~dt+\int_{\left(a+\frac{4b^2}{c}\right)^k}^\infty  \exp \left(-\frac{\left(t^{\frac{1}{k}}-a\right)^{\frac{3}{2}}}{4\sqrt{c}}\right)~dt \\
&\lesssim a^k+ \int_0^{\infty}\frac{4bk(a+\sqrt{8bu})^{k-1}}{\sqrt{8bu}}e^{-u}~du+\int_0^\infty \frac{8k\sqrt{c}\left((a+(4w\sqrt{c})^{\frac{2}{3}}\right)^{k-1}}{3(4w\sqrt{c})^{\frac{1}{3}}}~dw \\
&\lesssim a^k+a^{k-1}\sqrt{b}\int_0^\infty \frac{e^{-u}}{\sqrt{u}}~du+b^{\frac{k}{2}}\int_0^\infty u^{\frac{k-2}{2}}e^{-u}~du+a^{k-1}c^{\frac{1}{3}}\int_0^\infty \frac{e^{-w}}{w^{\frac{1}{3}}}~dw+c^{\frac{k}{3}}\int_0^\infty w^{\frac{2}{3}k-1}e^{-w}~dw.
\end{align*}
via substituting $u=\frac{\left(t^{\frac{1}{k}}-a\right)^2}{8b}$ and $w=\frac{\left(t^{\frac{1}{k}}-a\right)^{\frac{3}{2}}}{4\sqrt{c}}$. Since all of the integrals in the final line are finite, we have with Young's inequality that 
\begin{align*}
\int_0^\infty \PNbeta(\{|X|>t^{1/k} \}\cap \mathcal{G})~dt \lesssim a^k+a^{k-1}\sqrt{b}+b^{\frac{k}{2}}+a^{k-1}c^{\frac{1}{3}}+c^{\frac{k}{3}} \lesssim a^k+b^{\frac{k}{2}}+c^{\frac{k}{3}}.
\end{align*}
This immediately implies (\ref{moment bound}).

\begin{proof}[Proof of Corollary \ref{pseudo-Gaussian moments}]

Notice that $X=\Fluct_N(\xi)$ satisfies (\ref{pseudo-Gaussian}) with $\mathcal{G}=\mathcal{G}_N$ and
\begin{align*}
a&=C\left(\|\theta\|_{C^3}\left(1+\left|1-\frac{1}{\beta}\right|\right)\right), \\
b&=\frac{\|\xi\|_{H^{1/2}}^2+\max(\|\theta\|_{C^1}, \|\theta\|_{C^3}^2)}{\beta}, \\
c&=\frac{\max(\|\theta\|_{C^2}^2, \|\theta\|_{C^3}^3)}{LN\beta^2}.
\end{align*}
Rearranging (\ref{moment bound}) yields
\begin{align*}
\Esp \left|\Fluct_N(\xi)\right|^k\indic_{\mathcal{G}_N} &\lesssim \left(\|\theta\|_{C^3}\left(1+\left|1-\frac{1}{\beta}\right|\right)\right)^k+\left(\frac{\|\xi\|_{H^{1/2}}^2+\max(\|\theta\|_{C^1}, \|\theta\|_{C^3}^2)}{\beta}\right)^{\frac{k}{2}} \\
&+\left(\frac{\max(\|\theta\|_{C^2}^2, \|\theta\|_{C^3}^3)}{LN\beta^2}\right)^{\frac{k}{3}},
\end{align*}
as desired.
\end{proof} 

\bibliographystyle{amsplain}
\bibliography{Final_Draft}{}

\end{document}